\documentclass[11pt]{book}
\usepackage{amsfonts,amssymb,a4,bm,makeidx}
\usepackage{amsmath}
\usepackage[pdftex]{graphicx}
\usepackage{epstopdf}
\usepackage{graphicx}
\usepackage[all]{xy}

\newcommand{\zeq}{\setcounter{equation}{0}}

\newtheorem{teo}{Theorem}[chapter]
\newtheorem{lem}{Lemma}[chapter]
\newtheorem{pro}{Proposition}[chapter]
\newtheorem{defi}{Definition}[chapter]
\newtheorem{cor}{Corollary}[chapter]

\newtheorem{conj}{Conjecture}[chapter]
\newcommand{\card}[1]{\left|#1\right|}
\makeindex
\begin{document}

\def\ind#1{\1_{\{#1\}}}
\def\1{\rlap{\mbox{\small\rm 1}}\kern.15em 1}
\def\trp{\mathbb{T}}
\def\bydef{:=}
\def\qed{ \square}
\def\proof{\noindent{\bf Proof. }}

%%%%%%%%%%%%%%%% GRECO

\let\a=\alpha \let\b=\beta \let\c=\chi \let\d=\delta \let\e=\varepsilon
\let\f=\varphi \let\g=\gamma \let\h=\eta    \let\k=\kappa \let\l=\lambda
\let\m=\mu \let\n=\nu \let\om=\omega    \let\p=\pi \let\ph=\varphi
\let\r=\rho \let\s=\sigma \let\t=\tau \let\th=\vartheta
\let\y=\upsilon \let\x=\xi \let\z=\zeta
\let\D=\Delta \let\F=\Phi \let\G=\Gamma \let\L=\Lambda \let\Th=\Theta
\let\O=\Omega \let\P=\Pi \let\Ps=\Psi \let\Si=\Sigma \let\X=\Xi
\let\Y=\Upsilon
%%%%%%%%%%%%%%% EQUAZIONI CON NOMI SIMBOLICI
%% per assegnare un nome simbolico ad una equazione basta
%% scrivere \Eq(...) o, in \eqalignno, \eq(...) o,
%% nelle appendici, \Eqa(...) o \eqa(...):
%% dentro le parentesi e al posto dei ...
%% si puo' scrivere qualsiasi commento; per avere i nomi
%% simbolici segnati a sinistra delle formule si deve
%% dichiarare il documento come bozza, iniziando il testo con
%% \BOZZA. Sinonimi \Eq,\EQ.
%% All' inizio di ogni paragrafo si devono definire il
%% numero del paragrafo e della prima formula dichiarando
%% \numsec=... \numfor=... (brevetto Eckmannn).

\global\newcount\numsec\global\newcount\numfor
\gdef\profonditastruttura{\dp\strutbox}
\def\senondefinito#1{\expandafter\ifx\csname#1\endcsname\relax}
\def\SIA #1,#2,#3 {\senondefinito{#1#2}
\expandafter\xdef\csname #1#2\endcsname{#3} \else \write16{???? il
simbolo #2 e' gia' stato definito !!!!} \fi}
\def\etichetta(#1){(\veroparagrafo.\veraformula)
\SIA e,#1,(\veroparagrafo.\veraformula)
 \global\advance\numfor by 1
% \write15{@def@equ(#1){\equ(#1)} \%:: ha simbolo= #1 }
 \write16{ EQ \equ(#1) ha simbolo #1 }}
\def\etichettaa(#1){(A\veroparagrafo.\veraformula)
 \SIA e,#1,(A\veroparagrafo.\veraformula)
 \global\advance\numfor by 1\write16{ EQ \equ(#1) ha simbolo #1 }}
\def\BOZZA{\def\alato(##1){
 {\vtop to \profonditastruttura{\baselineskip
 \profonditastruttura\vss
 \rlap{\kern-\hsize\kern-1.2truecm{$\scriptstyle##1$}}}}}}
\def\alato(#1){}
\def\veroparagrafo{\number\numsec}\def\veraformula{\number\numfor}
\def\Eq(#1){\eqno{\etichetta(#1)\alato(#1)}}
\def\eq(#1){\etichetta(#1)\alato(#1)}
\def\Eqa(#1){\eqno{\etichettaa(#1)\alato(#1)}}
\def\eqa(#1){\etichettaa(#1)\alato(#1)}
\def\equ(#1){\senondefinito{e#1}$\clubsuit$#1\else\csname e#1\endcsname\fi}
\let\EQ=\Eq

%%%%%%%%%%%%%%% DEFINIZIONI LOCALI

\def\sqr#1#2{{\vcenter{\vbox{\hrule height.#2pt
\hbox{\vrule width.#2pt height#1pt \kern#1pt \vrule
width.#2pt}\hrule height.#2pt}}}}
\def\square{\mathchoice\sqr34\sqr34\sqr{2.1}3\sqr{1.5}3}

\let\nin\noindent
\def\fiat{{}}
\def\pagina{{\vfill\eject}} \def\\{\noindent}
\def\bra#1{{\langle#1|}} \def\ket#1{{|#1\rangle}}
\def\media#1{{\langle#1\rangle}} \def\ie{\hbox{\it i.e.\ }}
\let\ii=\int \let\ig=\int \let\io=\infty \let\i=\infty

\let\dpr=\partial \def\V#1{\vec#1} \def\Dp{\V\dpr}
\def\oo{{\V\om}} \def\OO{{\V\O}} \def\uu{{\V\y}} \def\xxi{{\V \xi}}
\def\xx{{\V x}} \def\yy{y} \def\kk{{\bf k}} \def\zz{{\V z}}
\def\rr{{\V r}} \def\zz{{\V z}} \def\ww{{\V w}}
\def\Fi{{\V \phi}}

\let\Rar=\Rightarrow
\let\rar=\rightarrow
\let\LRar=\Longrightarrow

\def\lh{\hat\l} \def\vh{\hat v}

\def\ul#1{\underline#1}
\def\ol#1{\overline#1}

\def\ps#1#2{\psi^{#1}_{#2}} \def\pst#1#2{\tilde\psi^{#1}_{#2}}
\def\pb{\bar\psi} \def\pt{\tilde\psi}

\def\E#1{{\cal E}_{(#1)}} \def\ET#1{{\cal E}^T_{(#1)}}
\def\LL{{\cal L}}\def\RR{{\cal R}}\def\SS{{\cal S}} \def\NN{{\cal N}}
\def\HH{{\cal H}}\def\GG{{\cal G}}\def\PP{{\cal P}} \def\AA{{\cal A}}
\def\BB{{\cal B}}\def\FF{{\cal F}}
\def\QQ{{\mathcal Q}}

\def\tende#1{\vtop{\ialign{##\crcr\rightarrowfill\crcr
              \noalign{\kern-1pt\nointerlineskip}
              \hskip3.pt${\scriptstyle #1}$\hskip3.pt\crcr}}}
\def\otto{{\kern-1.truept\leftarrow\kern-5.truept\to\kern-1.truept}}
\def\arm{{}}
\def\pp{{p}}
\def\xx{{x}}
\def\vv{{\bm v}}
\def\Z{\mathbb{Z}}
\def\I{{\rm I}}
\def\E{{\rm E}}

\def\La{\Lambda}
\def\TT{\mathcal{T}}
\def\bea{\begin{eqnarray}}
\def\eea{\end{eqnarray}}
\def\<{\langle}
\def\>{\rangle}
\def\0{\emptyset}
\def\vv{\vskip.2cm}
\def\N{{}\mathbb{N}}

\def\begn{\begin{aligned}}
\def\egn{\end{aligned}}

\title{Cluster expansion methods in rigorous statistical mechanics}
\author{Aldo Procacci}
\maketitle
\tableofcontents

%\BOZZA

\numsec=1\numfor=1
\part{Classical Continuous Systems}
\chapter{Ensembles in Continuous systems}
\vglue1.0cm

\section{The Hamiltonian and the equations of motion}
\vglue.2cm
\\
In this first part of the book we will deal with a system made by a large number $N$ of continuous
particles enclosed in a  box $\L\subset \mathbb{R}^d$ (we will
assume $\L$ to be cube  and denote $|\L|$ its volume)
performing a motion according to the laws of classical mechanics.
``Large number of particles" in physics lingo means typically $N\approx 10^{23}$.

\\We will restrict our discussion to  systems composed of identical particles with
no internal structure, i.e. just ``point" particles with a given
mass $m$. The position of the $i^{th}$ particle in the box $\L$ at
a given time $t$ is given  by a $d$-component coordinate vector, denoted by ${x_i}=x_i(t)$
with respect  to some system of orthogonal axis.

\\The momentum of the $i^{th}$ particle at time $t$ is also given by
a $d$-component vector denoted by ${p_i}=p_i(t)$ which is directly related to the velocity of the
particles, i.e. if $m$ is the mass of the particle, then $p_i
=m(dx_i/dt)$.

\\In principle, the laws of mechanics permit to know the evolution of such a system in
time, i.e. by these laws one should be able to determine which positions
$x_i=x_i(t)$  and momenta $p_i=p_i(t)$ the particles in the
system will have in the future and had in the past, provided one
knows the position of the particles  $x^0_i=x_i(t_0)$ and the
momenta of $p^0_i=p_i(t_0)$ at a given time $t_0$.

\\As a matter of fact, the time evolution of such system is described by a real valued function
$H(p_1, \dots ,p_N,x_1 ,\dots ,x_N)$
of particle positions and momenta (hence a function of $2dN$ real
variables) called the {\it Hamiltonian}.  In the case of isolated
systems, this function $H$ is assumed to have the form
$$
H(p_1, \dots ,p_N,x_1 ,\dots ,x_N)= \sum_{i=1}^N
{p_i^2\over 2m}+ U(x_1,\dots x_N)\Eq(1.1)
$$
The term $\sum_{i=1}^N {p_i^2\over 2m}$ is called the
kinetic energy of the system, while the term $ U(x_1,\dots
x_N)$ is the potential energy. Since we are assuming
that particles have to be enclosed in a box $\L$ we still have to
restrict $x_i\in \L$ for all $i=1,2,\dots , N$, while no
restriction is imposed on $p_i$, i.e. the momentum (and hence the velocity) of particles
can be arbitrarily large.

Once the Hamiltonian
\index{Hamiltonian} of the
system is given, one could in principle solve the
system of $2dN$ differential equations
$$
\begin{aligned}
{dx_i\over dt} & = {\partial H\over \partial  p_i}\\
{dp_i\over dt}& = -{\partial H\over \partial
x_i}
\end{aligned}~~~~~~~~~i=1,\dots, N\Eq(1.2)
$$
where ${\partial / \partial  x_i}$ and ${\partial /\partial  p_i}$ are  $d$-dimensional gradients.

\\This is  a system of $2dN$ first order differential equations.
The solution are $2dN$ functions $x_i(t)$, $p_i(t)$, $i=1,\dots,N$, given the positions $x_i(t_0)$ and momenta  $p_i(t_0)$ at some
initial time $t=t_0$.

\\It is convenient to introduce a $2dN$-dimensional space $\G_N(\L)$,
called the phase space \index{phase space} of the system (with $N$ particles) whose
points are determined by the coordinates $(\bm q,\bm p)$ with
$\bm q=(x_1,\dots ,x_N)$ and $\bm p=(p_1,\dots ,p_N)$ ($\bm q$ and $\bm p$
are both $dN$ vectors!) with the further condition that $x_i\in
\L$ for all $i=1,2,\dots ,N$. A point $(\bm q,\bm p)$ in the phase space
of the system is called a {\it microstate} of the system. With
these notations the evolution of the system during time can be
interpreted as the  evolution of  a point in the ``plane" $\bm q,\bm p$
(actually a $2dN$ dimension space, the phase space $\G_N(\L)$).

\\We finally want to remark that, by  \equ(1.1) (isolated system),
the value of the Hamiltonian $H$ is constant during the time
evolution of the system governed  by the Hamilton equations \equ(1.2). Namely $H(\bm q(t),\bm p(t))=
H(\bm q(0),\bm p(0))=E$. To see this, just calculate the total derivative of $H$ respect
to the time using equations \equ(1.2). This constant $E$ of motion is
called {\it energy} of the system. Thus the trajectory of the point
$\bm q, \bm p$ in the phase space $\G_N(\L)$ occurs in the surface $H(\bm p,\bm q)=E$.
\vskip.3cm
\section{ Gibbsian ensembles}\index{ensembles}
\vskip.1cm
\\It is substantially meaningless to look for the solution of  system \equ(1.2).

\\First, there is a technical reason.  Namely, the system contains an enormous
number of equations ($\approx 10^{23}$) which in general are
coupled (depending on the structure of $U$), so that it  is practically an impossible task
to find the
solution.

\\However, even supposing that some very powerful entity would  give us the solution,
this extremely detailed  description (a microscopic description)
would not be useful to describe the macroscopic properties of such
a system. Macroscopic properties which appear to us as the laws of
thermodynamic are due presumably by some mean effects of such
large systems and can in general be described in terms of  very
few parameters, e.g. temperature, volume, pressure, etc.. Hence we
have no means and also no desire to know the microscopic state of
the system at every instant (i.e.  to know the functions $\bm q(t),
\bm p(t)$). We thus shall adopt a statistical point of view in order
to describe the system.

\\We know  a priori some macroscopic properties of the system, e.g.
an isolated system occupies the volume $\L$, has $N$ particles and
has a fixed energy $E$.

\\We further know that macroscopic systems, if not perturbated from the exterior,
tend to stay in a situation of macroscopic (or thermodynamic)
equilibrium (i.e. a ``static"  situation), in which the values of
some thermodynamic parameters (e.g. pressure, temperature, etc.)
are well defined and fixed and do not change in time. Of course, in a system at the
thermodynamic equilibrium, the situation at the microscopic level
is desperately far form a static one. Particles in a gas at the
equilibrium do in general complicated and crazy  motions all the
time, nevertheless nothing seems to happen as time goes by at the
macroscopic level. So the thermodynamic equilibrium of a system
must be the effect of some mean behavior at the microscopic level.
This ``static" mean macroscopic behavior of systems composed by a
large number of particles must be produced in some way by the
microscopic interactions between particles and by the law of
mechanics.

\\Adopting the {statistical mechanics} point of view to describe
a macroscopic system at equilibrium means that we renounce to
understand how and why a system reach the thermodynamic
equilibrium starting from the microscopic level, and we just
assume that, at the thermodynamic equilibrium (characterized by
some thermodynamic parameters), the system could be found in {\it
any} microscopic state within a certain suitable set of
microstates  compatible with the fixed thermodynamic parameters. In other words, over long periods, the time spent by a system in some region of the phase space of microstates with
the same energy is proportional to the volume of this region, i.e., that all accessible microstates are equiprobable over a long period of time.
This is the so called {\it ergodic hypothesis}\index{ergodic hypothesis}.  Of course, in
order to have some hope that such a point of view will work, we
need to treat really ``macroscopic systems". So values such $N$ and
$V$  must be always thought as very large values (i.e. close to
$\infty$).

\\The statistical description of the macroscopic properties of the
system at equilibrium  (and in particular the laws of
thermodynamic) is done in two steps. \vskip.3cm
\\
{Step 1: \it fixing the Gibbsian ensemble (or the space of
configurations)}. We choose the phase space $\G_e$  and we assume that the system can be found  in any microstate
$(\bm q,\bm p)\in \G_e$. This set  $\G_e$ has to interpreted as  the set
of all microstates accessible by the system and it is called the
{\it Gibbsian ensemble} or {\it the space of configurations} of
the system. We will see later that several
choices are  possible for $\G_e$. We will then  think not on a single
system, but in an infinite number of mental copies of the same system, one copy for each element of $\G_e$.

\vskip.2cm

\\
{Step 2: \it fixing the Gibbs measure \index{Gibbs measure} in the Gibbisan ensemble}.

\\
We choose a function $\r(\bm p,\bm q)$ in $\G_e$ which will represent the
probability density in the Gibbsian ensemble. Namely, $\r(\bm p,\bm q)$ is
a function such that
$$
\int_{\G_e}\r(\bm p,\bm q)d\bm p d\bm q=1
$$
and
$d\m(\bm p,\bm q)=\r(\bm p,\bm q)d\bm pd\bm q$ represents the probability to find the
system in a microstate (or in the configuration) contained in an
infinitesimal volume $d\bm pd\bm q$ around the point $(\bm p,\bm q)\in \G_e$ where
$d\bm p~d\bm q$ is the usual Lebesgue measure in $\mathbb{R}^{2dN}$. The
\index{Gibbs measure} measure $\m(\bm p,\bm q)$ defined in $\G_e$ is called the {\it Gibbs
measure} of the system.

\vskip.4cm
\\Once a Gibbsian ensemble and a Gibbs measure are established, one can
begin to do statistic in order to describe the macroscopic state
of the system. When we look at a macroscopic system described via
certain Gibbs ensemble we do not know in which microstate the
system is at a given instant. All we know is that its microscopic
state must  be one of the microstates of the space configuration
$\G_e$ with probability density given by the Gibbs measure $d\m$.

\\For example, suppose
that $f(\bm p,\bm q)$ is a measurable function respect to the Gibbs
measure $d\m$, such as the energy, the kinetic energy per
particle, potential energy etc. Then we can calculate its mean
value in the Gibbsian ensemble that we have chosen by the formula:
$$
\<f\>~=~ \int  f(\bm p,\bm q)\r(\bm p,\bm q)d\bm p d\bm q
$$

\\We also recall the concept of mean relative square fluctuation of $f$
(a.k.a. standard deviation) denoted by
$\s_f$. This quantity measures  how spread is the probability
distribution of $f(\bm p,\bm q)$ around its mean value. It is  defined as
$$
\s_f~~=~~{\<(f-\<f\>)^2\>\over \<f\>^2}~~=~~{\<f^2\>-\<f\>^2\over
\<f\>^2}\Eq(1.4)
$$

\section{The Micro-Canonical ensemble}

\\
There are different possible choices for the ensembles, depending
on the different macroscopic situation of the system.

\\We start defining the micro-Canonical ensemble
\index{ensembles!Micro canonical}which is used to describe perfectly isolated
systems. Hence we suppose that our system is totally isolated from the
outside,  the $N$ particles are constrained to stay in the box
$\L$, and they do not exchange energy with the outside, so that
the system has a fixed energy $E$, occupies a fixed volume $|\L|$
and has a fixed number of particles $N$.

\\Thus, for such a system, we  can naturally say that the space of configuration $\G_{mc}$
is the set of points $\bm p,\bm q$ (with $\bm p\in \mathbb{R}^{dN}$ and $\bm q\in \L^{N}$)  with energy between
a given value $E$ and $E+\D E$ (where $\D E$ can  be interpreted
as the experimental error in the measure of the energy $E$). We have
$$
\G_{mc}(E,\L,N)~=~ \{ (\bm p,\bm q):\bm p\in \mathbb{R}^{dN}, ~ \bm q\in \L^{N}~{\rm and }~E< H(\bm p,\bm q)< E+\D
E\}
$$
We now choose the probability measure  in such way that any microstate
in the set of configurations above is equally probable, i.e. there
is no reason to assign different probabilities  to different
microstates.  This quite drastic hypothesis is  the so-called {\it
postulate of equal a priori distribution}, which is just a different formulation of the
ergodic hypothesis. \index{ergodic hypothesis} Hence
$$
\r(\bm p,\bm q)~~=~~\begin{cases}~ [\Psi_\L(E,N)]^{-1}& ~{\rm if}~ (\bm q,\bm p)\in \G_{mc}\\ 0 &
~{\rm otherwise}
\end{cases}
$$
where\index{partition function}
$$
\Psi_\L(E,N)~=~\int_{E< H(\bm p,\bm q)< E+\D E} d\bm p \,d\bm q\Eq(1.5)
$$is the normalization constant, which in this case coincides with  the $2dN$-dimensional ``volume" in the phase space occupied by
the space of configuration of the micro-canonical ensemble.
$\Psi_\L(E,N)$ is generally called the {\it partition function} of
the system in the Micro-Canonical ensemble.
\index{partition function!microcanonical}
\\The link between the Micro-Canonical ensemble and the thermodynamic is obtained via the
definition of the thermodynamic \index{entropy} entropy  of the system by
$$
S_\L(E,N)~~=~~k\ln\left[{1\over |\d|} \Psi_\L(E,N)\right]\Eq(1.5b)
$$
where $k$ is the Boltzmann's constant and $|\d|$  is  the
volume of some elementary phase cell $\d$  in  the phase space, so that
the pure number ${|\d|}^{-1}\Psi_\L(E,N)$ is the number of such
cells in the configuration space.

\\It may seem that the value of this constant $|\d|$
can be somewhat arbitrary, since it depends on our measure
instruments, and it could be done as small as we please by
improving our measure techniques. But the experimental measures tell us that this
constant must be fixed at the value $h^{dN}$ where $h$ is the Plank
constant.\index{Plank constant} So hereafter we will assume, unless differently
specificated,  that $|\d|$ is set at the value of the Plank constant.
It is important to stress that the presence of this constant in
the definition of the entropy has a very deep physical meaning.
Actually, it is a first clue of the quantum mechanics nature of
particle systems: things go as if the position $\bm p,\bm q$ of a
micro-state in the phase space could not be known exactly and one
can just say that the micro-state is in a small cube $d \bm p d \bm q$
centered at $\bm p,\bm q$ of the space phase with volume $h^{dN}$. This is
actually the Heisenberg indetermination principle.

\\The definition \equ(1.5b) gives a  beautiful probabilistic interpretation of the
second law of thermodynamic. A macroscopic system at the
equilibrium  will tend to stay in a state of maximal entropy, namely,
by \equ(1.5b), in the most probable macroscopic
state, i.e. a macroscopic state with thermodynamics parameters
fixed in such way that this state corresponds to the largest
number of microstates.  Namely, at equilibrium, the quantity
$\Psi_\L(E,N)$ should expect to reach a maximum value. Thus the
entropy (which by definition is just the logarithm of the number
of micro-states of a macro-state) is also expected to be maximum at
equilibrium. So the second law of thermodynamics stating that the
entropy of an isolated system always increases means in term of
statistical mechanics that the systems tends to evolve to
macrostates which are more probable, i.e. those with maximum
number of microstates.

\\It is also interesting to check that entropy  \equ(1.5b) is an ``extensive"
quantity. Namely, if the macroscopic system is composed by two
macroscopic subsystems whose entropies are, respectively $S_1$ and
$S_2$, the entropy of the total system must be $S_1+S_2$.

\section{The entropy is additive. An euristic discussion}

\\Suppose thus to consider a system made with two subsystems, one living in a phase space
$\G_1$ with coordinates $\bm p_1,\bm q_1$ occupying the volume $\L_1$ and
described  by the Hamiltonian $H_1(\bm p_1,\bm q_1)$ and the other in a
phase space $\G_2$ with coordinates $\bm p_2,\bm q_2$, occupying the volume
$\L_2$ and  described  by the Hamiltonian $H_2(\bm p_2,\bm q_2)$. We also
suppose that systems are isolated from each other. Consider first
the micro-canonical ensemble for each subsystem taken alone. The
energy of the first system will stay in an interval say $(E_1,
E_1+\D )$ while the second system will have an energy in $(E_2,
E_2+\D )$. The entropies of the subsystems will be respectively
$S(E_1)=k\ln \Psi_1(E_1)$ and $S(E_2)=k\ln \Psi_2(E_2)$, where
$\Psi_1(E_1)$ and $\Psi_2(E_2)$ are the volumes occupied by the two
ensembles in their respective phase spaces $\G_1$ and $\G_2$.
Consider now the micro-canonical ensemble of the total system made
by the two subsystems, The composite system lives in a phase space
$\G_1\times \G_2$ with coordinates $p_1,q_1, p_2,q_2$ occupying
the volume $V_1+V_2$ and it is described  by the Hamiltonian
$H_1(p_1,q_1)+H_2(p_2,q_2)$ (system are supposed isolated one from
each other). Let the total energy be in the interval say $(E,
E+2\D )$  ($\D\ll E$). This ensemble contains all the micro-states
of the composite system such that:

\\ a) $N_1$ particles with momenta and coordinates $p_1,q_1$ are in the volume $V_1$

\\ b)  $N_2$ particles with momenta and coordinates $p_2,q_2$ are in the volume $V_2$

\\ c) The energy $E_1$ and $E_2$ of the subsystems have values satisfying the condition
$$
E<E_1+E_2<E+2\D\Eq(1.13)
$$
\\We want to calculate the partition function $\Psi(E)$ of the composite system.
Clearly $\Psi_1(E_1)\Psi_2(E_2)$ is  the volume in the composite phase space $\G$ with
coordinate $(\bm p_1,\bm p_2, \bm q_1, \bm q_2)$ that corresponds to conditions a)
and b) with first system at energy $E_1$ and second system at
energy $E_2$ such that $E_1+E_2\in (E,E+2\D)$. Then
$$
\Psi(E)~=~\sum_{E_1,E_2\atop E<E_1+E_2<E+2\D}\Psi_1(E_1)\Psi_2(E_2)
$$
Since $E_1$ and $E_2$ are possible values of $H(\bm p_1,\bm q_1)$ and
$H(\bm p_2,\bm q_2)$ , suppose that $H(p_1,q_1)$ and $H(\bm p_2,\bm q_2)$ are
bounded below (as it will be always the case, see later) and for
simplicity let the joint lower bound be equal to $0$. Hence $E_1$ and $E_2$
both varies in the interval $[0,E]$. Suppose also that $E_1$ and
$E_2$ take discrete values $E_i=0,\D, 2\D,...$ so that in the interval
$(0,E)$ there are $E/\D$ of such intervals. Then
$$
\Psi(E)~=~\sum_{i=1}^{E/\D}\Psi_1(E_i)\Psi_2(E-E_i)\Eq(1.14)
$$
The entropy of the total system of $N=N_1+N_2$ particles, volume
$\L~=~\L_1\cup \L_2$ and energy $E$ is given by
$$
S_\L(E,N)~~=~~k\ln\left[\sum_{i=1}^{E/\D}\Psi_1(E_i)\Psi_2(E-E_i)\right]
$$
As subsystems are supposed macroscopic  ($N_1\to\infty$ and
$N_2\to\infty$) is easy to see that a single term in the sum
\equ(1.14) will dominate. Sum \equ(1.14) is a sum of positive
terms, let the largest of such terms be $\Psi_1(\bar E_1)\Psi_2(\bar
E_2)$ with $\bar E_1+\bar E_2~=~E$. Then we have the obvious
inequalities
$$
\Psi_1(\bar E_1)\Psi_2(\bar E_2)~\le~ \Psi(E)~\le~ {E\over \D}\Psi_1(\bar
E_1)\Psi_2(\bar E_2)
$$
or
$$
k\ln\left[\Psi_1(\bar E_1)\Psi_2(\bar E_2)\right] ~\le ~S_\L(E,N) ~\le~
k\ln\left[\Psi_1(\bar E_1)\Psi_2(\bar E_2)\right]+k\ln(E/\D)\Eq(1.16)
$$
We expect, as $N_1\to\infty$ and $N_2\to\infty$, that $\Psi_1\sim
C^{N_1}$ and  $\Psi_2\sim C^{N_2}$, thus  $\ln \Psi_1\propto N_1$ and
$\ln Z_2\propto N_2$. and also $E\sim N_1+N_2$. Hence factor
$\ln(E/\D)$ goes like $\ln N$ and can be neglected. Namely by this
discussion (just a counting argument) we get
$$
S_\L(E,N)~=~S_{\L_1}(\bar E_1,  N_1)+ S_{\L_2}(\bar E_2,  N_2)+O(\ln
N)\Eq(1.200)
$$
In other words the entropy is extensive, modulo terms of order
$\ln N$. Note that \equ(1.200) also means that the two subsystems
has a definite values $\bar E_1$ and $\bar E_2$ for the energy.
Namely $\bar E_1$ and $\bar E_2$ are the values that maximize the
number
$$
\Psi_1(E_1)\Psi_2(E_2)
$$
under the condition $E_1+E_2~=~E$. Using Lagrange multiplier is easy
to check that
$$
\left.{\partial \ln \Psi_1(E_1)\over\partial E_1}\right|_{E_1=\bar
E_1}~=~ \left.{\partial \ln \Psi_2(E_2)\over\partial
E_2}\right|_{E_2=\bar E_2}
$$
or
$$
\left.{\partial S(E_1)\over\partial E_1}\right|_{E_1=\bar E_1}~=~
\left.{\partial S(E_2)\over\partial E_2}\right|_{E_2=\bar E_2}
$$
Since thermodynamics tells us that ${\partial S(E,V)\over\partial
E}~=~{1\over T}$, we conclude that the two subsystems choose energy
$\bar E_1$ and $\bar E_2$ in such way to have the same
temperature. Thus the temperature in a macroscopic system can be
seen as the parameter governing the equilibrium between one part
of the system and the other.

\section{Entropy of the ideal gas}

\index{ideal gas}
If one can calculate the partition function of the Micro-canonical
ensemble, then it is possible to derive the thermodynamic
properties of the system. The Micro-Canonical ensemble is
difficult to be treated mathematically. As a matter of fact, a
direct calculation of the integral in r.h.s. of  \equ(1.5) is
generally very difficult, since involves integration over
complicated surfaces in high dimensions. According to elementary
calculus in $\mathbb{R}^n$ we can use a volume integral to
calculate the integral \equ(1.5). Volume integrals are easier to
deal with than surface integrals. Let $\om_\L(E,N)$ be the volume of
the phase space surrounded by the surface $H(\bm p,\bm q)=E$ (with of
course the further condition that particles are constrained to
stay in $\L$). Then
$$
\om_\L(E,N)~=~ \int_{H(\bm p,\bm q)\le E}d\bm p~d\bm q\Eq(1.6)
$$
Then, for small $\D E$
$$
\Psi_\L(E,N)~=~\om_\L(E+\D E,N)-\om_\L(E,N)\approx {\partial\om_\L(E,N)\over\partial
E}\D E \Eq(1.7)
$$
Let for example calculate the Micro-Canonical partition function
of an ideal gas, i.e. a gas of non interacting particles. with
Hamiltonian
$$
H~=~\sum_{i=1}^N{ p_i^2\over 2m}\Eq(1.8)
$$
We thus calculate $\om(N,\L,E)$ when $H(\bm  p,\bm q)$ is given by
\equ(1.8). In this case it is absolutely elementary to calculate
the integral $\int_\L dq$ which gives just $|\L|$ i.e. the volume
occupied by $\L$, and  therefore we get
$$
\begin{aligned}
\om_\L(E,N) & = \int_{H(\bm p,\bm q)\le E}d\bm p~d\bm q\\
& =|\L|^{N} \int_{\sum_{i=1}^N{
p_i^2\over 2m}\le E}dp_1\dots dp_N\\
& =|\L|^{N}\int_{\sum_{i=1}^N{p_i^2}\le 2mE}dp_1\dots dp_N
\end{aligned}
$$
the last integral is just the volume of a $dN$ dimensional sphere
of radius $\sqrt{2mE}$. Let us thus face this geometric problem.
The volume of a sphere of given radius $R$ in a $n$ dimensional
space is the integral
$$
V_n(R)~=~\int_{\sum_{i=1}^nx_1^2\le R^2}dx_1 \dots dx_n~=~C_n R^n
$$
In order to find $C_n$ consider the following integral
$$
\begin{aligned}
\int_{\mathbb{R}^n}e^{-(x_1^2+\dots +x_n^2)}dx_1 \dots dx_n & =
\int_{-\infty}^{+\infty}dx_1\dots \int_{-\infty}^{+\infty}dx_n
e^{-(x_1^2+\dots +x_n^2)}\\
&=\left(\int_{-\infty}^{+\infty}dx
e^{-x^2}\right)^n\\
&=\p^{n\over2}
\end{aligned}
$$
on the other hand, noting that the integrand above depends only on
$r=(x_1^2+\dots x_n^2)^{1/2}$, by a transformation to polar
coordinates in $n$ dimensions we can express the volume element
$dx_1\dots dx_n$ by spherical shells $dV_n(r)=dC_n
r^n=nC_nr^{n-1}dr$. Then integral above can also be calculated as

$$
\begin{aligned}
\int_{\mathbb{R}^n}e^{-(x_1^2+\dots +x_n^2)}dx_1 \dots dx_n & =
\int_0^\infty e^{-r^2}dV_n(r)\\\\
& =nC_n\int_0^\infty r^{n-1}e^{-r^2}dr\\\\
&={nC_n\over 2}\int_0^\infty x^{n/2-1}e^{-x}dx
\end{aligned}
\Eq(1.9)
$$
Recall now the definition of the gamma  function: for any $z>0$
$$
\G(z)~~=~~\int_0^\infty x^{z-1}e^{-x}dx
$$
Among properties of the gamma function we recall
$$
\G(n)~=~(n-1)!~~~~~~~~~~~~~~~~~~{\rm n~  positive~  integer}
$$
and
$$
z\G(z)~=~\G(z+1)~~~~~~~~~~~~~~ {\rm z\in \mathbb{R}^+}
$$
So Gamma function  is an extension of factorial in the whole
positive real axis.

\\Equation \equ(1.9) thus becomes
$$
{nC_n\over 2}\G({n/2})~=~ \pi^{n/2}
$$
and consequently of the volume $V_n(R)$ of a sphere on radius $R$
in $n$ dimensions.
$$
C_n~=~{\p^{n/2}\over \G({n\over 2}+1)},~~~~~~~~~~
V_n(R)~=~{\p^{n/2}\over \G({n\over 2}+1)}R^n
$$
Hence we get
$$
\om_\L(E,N)~=~{\p^{3N/2}\over {3N\over 2}\G({3N\over
2})}(2mE)^{3N\over 2}|\L|^N
$$
and, by \equ(1.7)
$$
\Psi_\L(E,N)~=~{\partial\om_\L(E,N)\over\partial E}\D E~=~\D E ~|\L|^N~
{\p^{3N/2}\over \G({3N\over 2})}(2m)^{3N/2}E^{{3N\over 2}-1}
$$
So, recalling definition \equ(1.5b), the entropy of an ideal gas
is given by \index{entropy!entropy of the ideal gas}
$$
S_\L(E,N)~=~k\ln\left[{\D E} ~|\L|^N~ {\p^{3N/2}\over
\G({3N\over 2})}(2m/h^2)^{3N/2}E^{{{3N\over 2}-1}}\right]
$$
Since $N$ is a very big number, we may write, for
$N\gg1$,
$$
E^{{3N\over 2}-1}\approx E^{3N\over 2},~~~~
\ln\left[\G\left({3N\over 2}\right)\right]\approx {3N\over 2}
\left(\ln  (3N/2)-1\right),~~~~ k\ln(\D E)=O(1)\approx 0
$$we used  the Stirling approximation for the
factorial: $n!\approx {n^n\over e^n}$ for large $n$) thus
$$
S_\L(E,N)~=~kN\left\{{3\over 2}+\ln\left[ ~|\L|~ \left({4m\pi E\over
3Nh^2}\right)^{3/2}\right]\right\}\Eq(1.10)
$$
This equation leads to the correct equation of state for a perfect
gas. In fact, according to the laws of thermodynamics, the inverse temperature of the system
is the derivative of the entropy respect to the energy, and the
pressure is the derivative of the entropy respect to the volume
$V=|\L|$ times the temperature, i.e.
$$
{1\over T}~=~{\partial S\over \partial E}~=~{3\over 2}{Nk\over
E}~~~~~~{\rm or} ~~~~~~ E~=~{3N\over 2}k T
$$
$$
{p\over T}~=~{\partial S\over \partial V}~=~{Nk\over V}~~~~~~{\rm or}
~~~~~~ P V~=~N k T\Eq(1.11e)
$$
Nevertheless \equ(1.10) cannot be the correct expression for the
entropy of an ideal gas. One can just observe that l.h.s. of
\equ(1.10) is not a purely extensive quantity, as the
thermodynamic entropy should be. There is some deep mistake in the
calculation of the entropy. Indeed, by accepting that   \equ(1.10) is the the correct expression of the thermodynamic entropy, one obtains the called Gibbs paradox:
an ideal gas confined in a box  consisting initially of two adjacent
volumes $V_A$ and $V_B$  separated by a  removable wall at the thermodynamic equilibrium has entropy $S$  (calulated via \equ(1.10)) and when the wall  is removed the entropy $S$ increases.

\section{ The Gibbs paradox}
\index{Gibbs paradox}
\\In this section we set $V=|\L|$. Consider thus the entropy of an ideal gas as a function
of the temperature $T$, the volume  $V$ and the numer of particles $N$. By \equ(1.10) and \equ(1.11e) we get
$$
S(T,V,N)~=~kN\left\{{3\over 2}+\ln\left[ ~V~ \left({2m\pi kT\over
h^2}\right)^{3/2}\right]\right\}\Eq(1.11)
$$
where, using \equ(1.11e), we have posed that  $E/N={3\over 2}kT$.
Consider now a closed system consisting initially of two adjacent
volumes $V_A$ and $V_B$  separated by a  wall. The volume $A$
contains an ideal gas with $N_A$ particles, and the volume $V_B$
contains another ideal gas with $N_B$ particles. The two subsystems
are kept at the same pressure $P$ and at the same temperature $T$.
The entropy of such system, according to \equ(1.11) is
$$
S^i_{\rm total}~=~S(T, V_A, N_A)+S(T, V_B, N_B)
$$
If we now remove the wall, the two ideal gases will mix, each
occupying the volume $V_A+V_B$. In the new equilibrium situation
the entropy is now
$$
S^f_{\rm total}~=~S(T, V_A+V_B, N_A)+S(T, V_A+V_B, N_B)
$$
The entropy difference between the initial state $i$ and the final
state $f$ is, according to \equ(1.11)
$$
\D S ~=~ S^f_{\rm total}-S^i_{\rm total}~=~N_A k \ln(1+V_B/V_A)+N_B k
\ln(1+V_A/V_B)\Eq(dS)
$$
So far everything seems to be fine, since $\D S>0$ as one should
expect for this irreversible process (the mixing on two ideal
gases). But let us now suppose that the two ideal gases are
actually identical. We could repeat the argument and again we will
find that the variation of the entropy is $\D S>0$. However this
cannot be correct since, after the removal of the wall, no
macroscopic changes happens at all in the system. We could put
back the wall ad we will return to the initial macroscopic
situation. This paradox is clearly related to the fact that we are
assuming the particles of the ideal gas as distinguishable. If
particle are considered as distinguishable, then also the
situation in which two identical perfect gases mixes is a
irreversible process. If we put back the wall particles in the
volume $1$ are not the same particles of the initial situation.
The paradox was resolved by Gibbs supposing that identical
particles are not distinguishable. With this hypothesis the number
of microstates involving $N$ particles should be reduced by a
factor $N!$, since there are exactly $N!$ ways to enumerate $N$
identical particles. Hence the correct definition of, e,g. $\om_\L(E,N)$ should be
$$
\om_\L(E,N)~=~ {1\over N!}\int_{H(\bm p,\bm q)\le E}dp~dq\Eq(1.6b)
$$
Thus,  instead of \equ(1.10), the correct entropy of a perfect gas
is
$$
S_\L(E,N)~=~kN\left\{{3\over 2}+\ln\left[ ~|\L|~ \left({2m\pi k T\over
h^2}\right)^{3/2}\right]\right\}-k\ln N!\Eq(1.10b)
$$
First observe that this new definition of the entropy does not
affect the equation of state of the perfect gas, since it differs
from \equ(1.10) by a term independent on $E$ and $V$, so
\equ(1.10b) leads to the same equations \equ(1.11e). and, for
$N\gg 1$,  by Stirling's formula $\ln N!\approx N\ln N-N$
$$
S_\L(E,N)~=~kN\left\{{5\over 2}+\ln\left[ ~{|\L|\over N}~
\left({2m\pi kT\over h^2}\right)^{3/2}\right]\right\}\Eq(1.10c)
$$
Thus entropy defined by \equ(1.10c) of a perfect gas is indeed a
purely extensive quantity. Let us also check that \equ(1.10c)
solve the Gibbs paradox.

Let us consider again the argumentation which leads us to the Gibbs
paradox with the new temperature dependent entropy (just using
that $E~=~3NKT/2$)
$$
S_\L(E,N)~=~kN\left\{{5\over 2}+\ln\left[ ~{|\L|\over N}~
\left({2m\pi kT\over h^2}\right)^{3/2}\right]\right\}\Eq(1.10d)
$$
The new variation in entropy is now
$$
\D S~=~ k(N_A+N_B)\Bigg\{{5\over 2}+\ln\Bigg[ ~{V_A+V_B\over
N_A+N_B}~ \left({2m\pi kT\over h^2}\right)^{3/2}\Bigg]\Bigg\}-
$$
$$
- kN_A\Bigg\{{5\over 2}+\ln\Bigg[ ~{V_A\over N_A}~ \Big({2m\pi
kT\over h^2}\Big)^{3/2}\Bigg]\Bigg\}- kN_B\Bigg\{{5\over
2}+\ln\Bigg[ ~{V_B\over N_B}~ \Big({2m\pi kT\over
h^2}\Big)^{3/2}\Bigg]\Bigg\}=
$$
$$
= kN_A\ln\left[{{V_A+V_B\over N_A+N_B}\over {V_A\over N_A}}\right]+ kN_B\ln\left[{{V_A+V_B\over N_A+N_B}\over {V_B\over N_B}}\right]\Eq(1.11f)
$$
For two different gases this formula gives something similar to \equ(dS). But if
gas $A$ and gas $B$ are identical, then, since in the initial
state and in the final state temperature and pressure ar
unchanged, we must have, by \equ(1.11e)
$$
{V_A\over N_A}~=~{V_B\over N_B}~=~{V_A+V_B\over N_A+N_B}~=~{kP\over T},
~~~~{\rm if ~gas ~A~ and~ gas~ B ~are~ identical}\Eq(dSi)
$$
(of course also for different gases ${V_A\over N_A}~=~{V_B\over
N_B}~=~{kT\over P}$ but ${V_A+V_B\over N_A+N_B}\neq {kT\over P}$).

Inserting formula \equ(dSi) in \equ(1.11f) we obtain
$$
\D S~=~0 ~~~~~~~~{\rm if ~gas ~A~ and~ gas~ B ~are~ identical}
$$
The necessity to divide by the factor $N!$ to escape from the
Gibbs paradox is a new symptom that classical mechanics is not
adequate.

\section{ The Canonical Ensemble}\index{ensembles!canonical}
The Micro Canonical Ensemble is suited for isolated systems where natural macroscopic
variables are the volume $|\L|$, the number of particles $N$ and the energy $E$.  We now define a new ensemble
which is appropriate to describe a system which is  not isolated,
but it is in thermal equilibrium with a larger system (the heat
reservoir), e.g. a gas kept in a box made by heat conducting walls
which is fully immersed in a larger box containing some other gas
at a fixed temperature $T$. Hence this system is constrained to
stay in a box $\L$ with a fixed volume $|\L|$, a fixed number of
particles $N$, at a fixed temperature $T$, but its energy is no
longer fixed, since the system is now allowed to exchange energy with
the heat reservoir through the walls.

We define the Canonical Ensemble for such a system as follows. The
space of configuration of the Canonical Ensemble is

$$
\G_{c}(\L,N)~=~ \{ (\bm p,\bm q):\bm p\in \mathbb{R}^{dN},~~  \bm q\in \L^{N}\}
$$

The probability measure of the Canonical Ensemble is
$$
d\m_c(\bm p,\bm q)~=~ {1\over Z_\L(\b,N)}{1\over N!} e^{ -\b H(\bm p,\bm q)} {d\bm pd\bm q\over h^{3N}}\Eq(1.25a)
$$
where $\b~=~(k T)^{-1}$ is a constant proportional to the inverse
temperature of the system ($k$ is again the Boltzmann constant) and the normalization constant
$$
Z_\L(\b,N)~=~{1\over N!} \int e^{ -\b H(\bm p,\bm q)} {d\bm pd\bm q\over
h^{3N}}\Eq(1.25b)
$$
is the called {\it the partition function of the system in the canonical ensemble}.\index{partition function!canonical}

\\Thermodynamics is recovered by the following definition.
The thermodynamic function called  free energy of the system is
obtained in the canonical ensemble by the formula
$$
F_\L(\b,N)~=~ -k T\ln Z_\L(\b,N)\Eq(1.20)
$$

\\We now remark that in  \equ(1.25b) we can perform for free the integration
over momenta.
$$
\begin{aligned}
Z_\L(\b,N) & ={1 \over h^{3N} N!}\int \!\!dp_1\dots \int \!\!dp_N
\int_{\L}\!\!dx_1\dots\int_{\L} \!\!dx_N
 e ^{-\b H(p_1,\dots, p_N,x_1,\dots ,x_N)}  \\\\
& = {1 \over h^{3N} N!}\int d\pp_1\dots \int d\pp_N
\int_{\L}d\xx_1\dots\int_{\L} d\xx_N
 e ^{-\b(\sum_{i=1}^N {\pp_i^2\over 2m}+ U(\xx_1,\dots \xx_N)) }\\\\
&= {1 \over h^{3N} N!}\int e^{ -\b {\pp_1^2\over 2m}} d\pp_1\dots
\int e^{ -\b {\pp_N^2\over 2m}}d\pp_N
\int_{\L}d\xx_1\dots\int_{\L} d\xx_N
 e ^{-\b U(\xx_1,\dots \xx_N) } \\\\
&= {\left[\int_{-\infty}^{+\infty}e^{-\b {x^2/ 2m}}dx \right]^{3N}
\over h^{3N} N!} \int_{\L}d\xx_1\dots\int_{\L} d\xx_N
 e ^{-\b U(\xx_1,\dots \xx_N) } \\\\
&= {\left[(2m\pi/\b h^2)^{3/2} \right]^{N} \over N!}
\int_{\L}d\xx_1\dots\int_{\L} d\xx_N
 e ^{-\b U(\xx_1,\dots \xx_N) }
 \end{aligned}
$$

\\The integral

$$
{1\over N!}\int_{\L}d\xx_1\dots\int_{\L} d\xx_N
 e ^{-\b U(\xx_1,\dots \xx_N) }
 $$ is called the ``configurational'' partition function
 of the system. Generally, it is very difficult to calculate explicitly this function for
 a real gas. But in the case of an ideal gas the situation is again immediate.
So we can now  give a justification a posteriori for the
definition \equ(1.25a) by considering again the case of the ideal
gas. Indeed, in the case of the  ideal gas
we have that $U(\xx_1,\dots \xx_N)=0$ and thus (recall that we have set $V=|\L|$)
$$
Z^{\rm ideal\, gas}_\L(\b,N)
~=~ {V^N\over N!}\left({2m\p\over h^2\b}\right)
^{3N/2}~~~~~~~~~~~~~~~~~~~~~~~~~~~~~~~~~~~~~~
$$
Hence, since $\b^{-1}~=~k T$ and using also Stirling
approximation for $\ln N!$
$$
F_\L(\b,N)~=~-kT\ln Z_\L(\b,N)~=~-k T N \left\{1+\ln \left[{V\over
N}\left({2\pi m kT\over h^2}\right)^{3/2}\right]\right\}
$$
From free energy we can calculate all thermodynamic quantities.
E.g.
$$
P~=~-{\partial F\over \partial V}~=~{NkT\over V}  ~~~~\Rightarrow ~~~~
P V~=~ NkT
$$
$$
S~=~-{\partial F\over \partial T}~=~ N k\left[ {5\over 2}+ \ln
\left[{V\over N}\left({2\pi m kT\over
h^2}\right)^{3/2}\right]\right]
$$
\index{entropy!entropy of the ideal gas}
$$
E~=~F+TS~=~{3\over 2}NkT
$$
Results are identical to the case of Micro Canonical Ensemble!

\\Note also that the energy $E$ in the canonical ensemble is not fixed and hence $E$ has to be interpreted as mean energy.
This suggests that it could be also calculated directly by the
formula
$$
E~=~\<H(\bm p,\bm q)\>~=~ Z^{-1}_\L(\b,N) {1\over N!} \int e^{ -\b H(\bm p,\bm q)}
H(\bm p,\bm q) {d\bm pd\bm q\over h^{3N}} =- {\partial\over \partial\b}\ln Z_\L(\b,N)
$$

\\{\bf Exercise}. Show that
$\<H(\bm p,\bm q)\>~=~{3\over 2}NkT$ as soon as $H(\bm p,\bm q)~=~\sum_{i=1}^N {\pp_i^2\over 2m}$.
\vglue.3cm

\section{Canonical
Ensemble: Energy fluctuations}
A system in the Canonical Ensemble can have in principle
microstates of all possible energies. This means that the energy fluctuates around its mean value $E=\<H(\bm p,\bm q)\>$. Let us thus check the
fluctuations of the energy in the Canonical Ensemble.
The mean energy in the Canonical Ensemble is given by
$$
E~=~\<H(\bm p,\bm q)\> ~=~ {\int d\bm pd\bm q H(\bm p,\bm q)e ^{-\b H(\bm p,\bm q)}\over \int d\bm pd\bm q e ^{-\b H(\bm p,\bm q)}}\Eq(i)
$$
Differentiating both side of \equ(i) respect to $\b$ we get
$$
\begin{aligned}
{\dpr E\over \dpr \b} & = -~{\int d\bm pd\bm q H^2(\bm p,\bm q)e ^{-\b H(\bm p,\bm q)}
 \int d\bm pd\bm q e ^{-\b H(\bm p,\bm q)}
\over (\int d\bm pd\bm q e ^{-\b H(\bm p,\bm q)})^2}\\
& +~{\int d\bm pd\bm q H(\bm p,\bm q) e
^{-\b H(\bm p,\bm q)} \int d\bm pd\bm q H(\bm p,\bm q) e ^{-\b H(\bm p,\bm q)}\over (\int d\bm pd\bm q e
^{-\b H(\bm p,\bm q)})^2}\\
&~=~ \left[{\int d\bm pd\bm q H(\bm p,\bm q) e ^{-\b H(\bm p,\bm q)}\over \int d\bm pd\bm q e
^{-\b H(\bm p,\bm q)}}\right]^2~  -~ {\int d\bm pd\bm q H^2(\bm p,\bm q)e ^{-\b H(\bm p,\bm q)}
\over \int d\bm pd\bm q e ^{-\b
H(\bm p,\bm q)}}
\end{aligned}
$$
Hence we get the relation, for the standard deviation of
$E~=~\<H(\bm p,\bm q)\>$
$$
\<H^2 (\bm p,\bm q)\>-\<H(\bm p,\bm q)\>^2~=~ - {\dpr E\over \dpr \b}~=~kT^2 {\dpr
E\over \dpr T}
$$
From thermodynamics ${\dpr E\over \dpr T}=C_V$ where $C_V$ is the
heat capacity. In general  $C_V\propto N$ as also $E\propto N$
(see e.g. the case of the perfect gas where $C_V={3\over 2}Nk$ and
$E={3\over 2}NkT$). Hence
$$
{\sqrt{\<H^2 (\bm p,\bm q)\>-\<H(\bm p,\bm q)\>^2}\over \<H(\bm p,\bm q)\>}  ~=~{\sqrt{ kT^2
{\dpr E\over \dpr T}} \over E}\approx {1\over \sqrt N}\ll 1
$$
Thus the (relative) fluctuation of the energy around its mean value in the canonical ensemble are
``macroscopically'' small (in the sense that they are of the order
of $1/\sqrt{N}$ with $N$ being a very large number). This means
that in the canonical Ensemble it is highly probable to find the
system in microstates with energy equal or very close to the mean energy $E~=~
\<H(\bm p,\bm q)\>$. So Canonical Ensemble is ``nearly" a micro-canonical
Ensemble. The energy is not exactly fixed, but it can fluctuate
around a fixed value with relative  fluctuations of order
$N^{-1/2}\approx 10^{ -12}$. \vskip.5cm

\section{The Grand Canonical Ensemble}
\index{ensembles!grand canonical}
The Micro-Canonical Ensemble applies to isolated systems with
fixed $N$, $V$ and $E$, while the Canonical Ensemble describes
systems with fixed $N$, $V$ and $T$ and energy variable (e.g.
systems in heat bath). The Canonical Ensemble appears more
realistic than the Micro Canonical Ensemble. It is very difficult
to construct a perfectly isolated system, as demanded in the Micro
Canonical Ensemble. So systems whose energy is not known exactly
(hence not perfectly isolated) are easier to construct
experimentally.

\\On the other hand, in the canonical ensemble is still  demanded a severe ``microscopic"
condition:  the number of particles must fixed, i.e. system
confined in $\L$ cannot exchange matter with the outside. This is
also a very difficult situation to create experimentally. We
generally deal with systems where, besides the energy, also the
number of particles is not known exactly. We now thus define the
Ensemble suitable to describe systems in thermodynamic equilibrium
in which matter and energy can be exchanged with the exterior. The
fixed thermodynamics parameters for such a system are the volume
$V$, the temperature $T$ and the chemical potential $\m$.

\\Hence the configuration space of the Grand canonical Ensemble is
$$
\G_{GC}(\L)~=~ \bigcup_{N\ge 0}\G_c(\L, N)
$$
(by convention $\G_0$ represent the single micro-state in which no
particle is present in the volume $\L$). The (restriction to $\G_c(\L, N)$
of the) probability measure of the Grand Canonical Ensemble is
$$
d\m_{GC}(\bm p,\bm q)~=~ {1\over \Xi_\L (\b,\m)}{e^{N\b\m}\over N!\, h^{3N}} e^{
-\b H(\bm p,\bm q)} {d\bm p \,d\bm q}\Eq(1.25ca)
$$
where again  $\b=1/kT$ is the inverse temperature, $\mu$ is the chemical potential and the normalization constant
\index{partition function!grand canonical}
$$
\Xi_\L (\b,\m)=\sum_{N=0}^{\infty}{z^N \over h^{3N} N!}\int\!\!
d\pp_1\dots \int\!\! d\pp_N \int_{\L}\!\!d\xx_1\dots\int_{\L} \!\!d\xx_N e
^{-\b H(\pp_1,\dots, \pp_N,\xx_1,\dots ,\xx_N)}\Eq(GC1)
$$
is called the {\it partition function in the grand canonical ensemble}. Here above $z~\doteq~e^{\b\m }$ is called "activity" or "fugacity" \index{activity}of the
system. The term $N=0$ is put conventionally $=1$ in
the sum above while for the term $N=1$ we have
$H(\pp_1,\xx_1)={\pp_1^2/ 2m}$.

\\Hence $d\m_{GC}(\bm p,\bm q)$ is the probability\index{Gibbs measure!grand canonical} to find the system in
a micro-state with exactly $N$ particles, with momenta and
positions in the small volume ${d\bm p \,d\bm q}$ centered at $(\bm p,\bm q)
$ of the phase space $\G_N$. By convention
$$
d\m(\G_0)~=~{1\over \Xi (T,\L,\m)}
$$
is the probability to find the system in the micro-state where no
particle in $\L$ is present.

\\Again, since $H(\pp_1,\dots, \pp_N,\xx_1,\dots ,\xx_N)=\sum_{i=1}^N \frac{p_i^2}{2m}+U(x_1,\dots, x_N)$,  the integration
over momenta can be done explicitly and one gets
$$
\Xi_\L(\b,\l)~=~\sum_{N=0}^{\infty}{\l^N \over N!}
\int_{\L}d\xx_1\dots\int_{\L} d\xx_{N} ~e ^{-\b U(\xx_1,\dots
,\xx_N)} \Eq(GC2)
$$
where
$$
\l~=~e^{\b\m }\left({2\pi m\over \b
h^2}\right)^{3/2}\Eq(acty)
$$

\\The parameter $\l$ is called configurational activity (or simply activity when
it will be clear from the context).

\\The connection with thermodynamic in the Grand-canonical ensemble is defined via the formula
$$
\b P_\L(\b,\l) ~=~{1\over |\L|}\ln \Xi_\L(\b,\l)\Eq(pvf)
$$
and the function $P_\L(\b,\l)$ is identified with the thermodynamical {\it pressure}
\index{pressure} of the system.
Another important quantity that can be calculated  from the Grand
Canonical partition function is the {\it mean density}. The mean density is obtained by computing,
at fixed volume $|\L|$, temperature $T$ and chemical potential $\m$
the mean number of particles in the system.
$$
\begin{aligned}
\<N\> & =~{1\over  \Xi_\L(\b,\l)}\sum_{N=0}^{\infty}{N\l^N \over N!}
\int_{\L}d\xx_1\dots\int_{\L} d\xx_{N} ~e ^{-\b U(\xx_1,\dots
,\xx_N)}\\
&=~{\l}{\partial\over \dpr \l}{\ln \Xi_\L(\b,\l)}
\end{aligned}
\Eq(1.30)
$$
Hence calling $\r_\L(\b,\l)~=~{\<N\>\over|\L|}$
$$
\r_\L(\b,\l)~=~{1\over |\L|}{\l}{\partial\over \dpr \l}{\ln
\Xi_\L(\b,\l)}\Eq(1.30b)
$$

\section{The ideal gas in the Grand Canonical Ensemble}
To conclude this brief introduction let us consider the case of perfect gas in the
Grand Canonical Ensemble. It is very easy to calculate the Grand
Canonical partition function in this case, where $U(\xx_1,\dots
,\xx_N)=0$. E.g., by \equ(GC2) we get \index{ideal gas}
$$
\Xi^{ideal\, gas}_\L(\b,\l)~=~ \sum_{N=0}^{\infty}{\l^N \over N!} ~
\int_{\L}d\xx_1\dots\int_{\L} d\xx_{N}~=~ ~\sum_{N=0}^{\infty}{(\l
|\L|)^N \over N!} ~=~ e^{\l |\L|}
$$
Hence, \equ(pvf) and \equ(1.30b) become
$$
\b P^{ideal\, gas}~=~\l\Eq(1.31p)
$$
$$
\r^{ideal\, gas}~=~\l\Eq(1.32p)
$$
In particular \equ(1.32p) says that the activity $\l$ of a perfect
gas coincides with the its density $\<N\>/|\L|$. Putting \equ(1.32p) in
\equ(1.31p) we get
$$
\b P~=~\r
$$
which is again the equation of state of a perfect gas.

\\{\it Exercise:} Calculate the fluctuation $\<N^2\>-\<N\>^2$ in
the case of the perfect gas and show that it is of order $\<N\>$.

\section{The Thermodynamic limit}
The dependence of the density $\r_\L(\b,\l)$ from the volume $\L$ in \equ(1.30b)
must be a residual one. In fact we  expect that, by increasing  the
volume $\L$ of our system but  keeping fixed the value of the
chemical potential $\m$ and the inverse temperature $\b$, the density of the system does not vary in a
sensible way. Values of $|\L|$ that one can take in thermodynamics
are macroscopic, hence very large. We thus may think that the
volume is arbitrarily large (which is the rigorous formalization
of "macroscopically" large) and define
$$
\b P(\b,\l) ~=~\lim_{\L\to \infty}{1\over |\L|}\ln \Xi_\L(\b,\l)\Eq(1.31)
$$
$$
\r(\b,\l)~=~\lim_{\L\to\infty}{1\over |\L|}{\l}{\partial\over \dpr
\l}{\ln \Xi_\L(\b,\l)}\Eq(1.32)
$$
The limit $\L\to\infty$ (where the way in which $\L$ goes to
infinity has to be specified in a precise sense) is called the
{\it thermodynamic limit} of the Grand Canonical Ensemble. The
exact thermodynamic behavior of the system is recovered at the
thermodynamic limit. This limit can be understood in the physical
sense as ``the volume macroscopically large".

\\Note that, by \equ(1.32) it is possible to express the activity of the system
$\l$ as a function of the density $\r$ and of the (inverse)
temperature $\b$. So the pressure of the system can be expressed
in terms $\r$ and $\b$. This is very easy to do in the case of the ideal gas. When real gases
are concerned (i.e. gases for which the potential energy $U$ is non-zero), the
formula giving the pressure of the system in terms of the density is called the
{\it virial equation of state}.

\\Thermodynamic limit can also be done in the Micro Canonical and Canonical Ensemble.
In this case one has to fix a given density $\r=N/|\L|$ for the
system and then take the limit $\L\to\infty$, $N\to\infty$ in such
way that $N/|\L|$ is kept constant at the value $\r$.

\\In the Micro Canonical Ensemble, in place of \equ(1.5b) one define
$$
S(\r,E)~=~\lim_{\L\to\infty, N\to \infty,\atop N/|\L|,~E/|\L| ~{\rm fixed}
}{1\over |\L|} k\ln\left[{1\over h^{3N}} \Psi_\L(E,N)\right]\Eq(1.5B)
$$
where $S(\r,E)$ is the infinite-volume specific   entropy (i.e. entropy per unit volume)
which is an intensive quantity. In the Canonical Ensemble one can
consider in place of \equ(1.20)
$$
F(\r,\b)~=~ -\lim_{\L\to\infty, N\to \infty,\atop N/|\L|~=~\r}{1\over |\L|}k
T\ln Z_\L(\b,N)\Eq(1.20B)
$$
where $F(\r,\b)$ is the infinite volume Gibbs free energy per unit volume.

\\In principle the three ensembles that we have considered are equivalent only at the thermodynamic
limit, when the effect of the boundary are removed. Thus equations
of the thermodynamics are exactly recovered at the thermodynamic
limit.

\\Typical mathematical problems in statistical
mechanics are thus to show the existence of limits \equ(1.20B),
\equ(1.5B) and \equ(1.32) and to show that they produce the same
thermodynamic (otherwise something would be seriously wrong in the
picture of the statistical mechanics).

\\In the following we will  focus our attention mainly on the Grand Canonical Ensemble and we
will investigate the existence of the limit \equ(1.32) and the
property of this limit as a function of $\b$ and $z$.

\chapter{The  Grand Canonical Ensemble}
\numsec=2\numfor=1
\section{Conditions on the potential energy}
A system of point particles in a volume $\L$ in the Grand
Canonical Ensemble is described by a probability measure on $\cup
_{N} \G_N(\L)$ where
$$
\G_{N}(\L)~=~\{(\xx_1,\dots \xx_N)\in
\mathbb{R}^{dN}: \xx_i\in\L\}.
$$
The restriction of this
probability measure to $\G_N(\L)$ is  called the
{\it configurational Gibbs measure} \index{Gibbs measure!configurational}(we have already integrated over
momenta)\index{ensembles!grand canonical}
$$
d\m(\xx_1,\dots,\xx_N)~=~ {1\over \Xi_\L (\b,\l)}{\l^{N}\over N!}
e^{ -\b U(\xx_1,\dots ,\xx_N)} d\xx_1\dots d\xx_N \Eq(2.1)
$$
where $\b~=~(kT)^{-1}$ with $T$ absolute temperature and $k$
Boltzmann constant, while the activity $\l$ is given in
\equ(acty). $d\m(\xx_1,\dots,\xx_N)$ is the probability to find
the system in the micro-state in which exactly $N$ particles are
present and, for $i~=~1,2,\dots , N$,  the  $i^{th}$ particle  is in
the small volume $d^3\xx_i$ centered at the point $\xx_i\in \L$.
The normalization constant  $\Xi_\L(\b,\l)$ is called the Grand
Canonical partition function of the system and it is given by
$$
\Xi_\L (\b,\l)~=~1+ \sum_{N~=~1}^{\infty}{\l^N \over N!}
\int_{\L}d\xx_1\dots\int_{\L} d\xx_{N} ~e ^{-\b U(\xx_1,\dots
,\xx_N)} \Eq(2.2)
$$
\index{partition function!grand canonical}
The factor $1$ in the sum above corresponds to the micro-state in which
no particle is present, which hence can occur with probability
$\Xi^{-1} (\b,\L,\l)$.
\index{pair potential}
\\The potential energy $U(\xx_1,\dots ,\xx_N)$ is assumed to be a function
$U: (\mathbb{R}^d)^N\to (\mathbb{R}\cup\{+\infty\})$. We will
suppose from now on that  $U(\xx_1,\dots ,\xx_N)$ has the following form
$$
U(\xx_1,\dots ,\xx_N)~=~\sum_{1\le i<j\le N}V(\xx_i -\xx_j)
+\sum_{i=1}^{N} \Phi_{\rm e}(\xx_i)
$$
where $V(\xx)$ is a function $V: \mathbb{R}^d\to
\mathbb{R}\cup\{+\infty\}$ and $\Phi_{\rm e}(\xx)$ is a function
$\Phi_{\rm e}: \mathbb{R}^d\to \mathbb{R}$. We will always assume that
$V(\xx)$ is such that $V(-\xx)~=~V(\xx)$. We let  $|x|$ to denote the Euclidean norm of $x$. Physically the
assumption above on the potential energy means that we are
restricting to the case of particles interacting via a  translational invariant {\it pair potential} $V$ plus an external potential $\Phi_{\rm
e}$. The  interaction $\sum_{1\le i<j\le N}V(\xx_i -\xx_j)$  is
``internal" in the sense that it depends only on mutual positions
in space of particles (i.e. only from vectors $\xx_i-\xx_j$). The
interaction $\sum_{i=1}^{N}\Phi_{\rm e}(\xx_i)$ depends instead on
the absolute positions of particles in space and $\Phi_{\rm
e}(\xx_i)$ is interpreted as the effect of the world ``outside"
the boundary of $\L$ on the $i^{th}$ particle confined in $\L$.
For instance, suppose that outside $\L$ there are $M$ particles
in fixed positions $y_1,\dots y_M$, then $\Phi_{\rm
e}(\xx_i)~=~\sum_{j~=~1}^M V(\xx_i -y_j)$.

\\The choice of $\Phi_{\rm e}$ is somehow arbitrary, in the sense
that it will depend on conditions we are supposing outside $\L$. A
given choice of $\Phi_{\rm e}$ is called generally a {\it boundary
condition}. For continuous systems the
question of the effect of the  boundary conditions  at the thermodynamic limit
is rather difficult.
The pressure is expected to be independent on the
choice of the boundary condition but other quantities (such as some derivatives of the pressure) may not be independent
on boundary conditions.  In
this section we will make the mathematically simplest choice $\Phi_{\rm e}=0$ (i.e.
no influence at all on particles inside $\L$ from the world
outside) which is called {\it free boundary condition}. We will
consider in this chapter just free boundary conditions, hence we
will suppose
$$
U(\xx_1,\dots ,\xx_N)~=~\sum_{1\le i<j\le N}V(\xx_i -\xx_j) \Eq(2.3)
$$

\\By \equ(2.3), interaction energy between particles is known once we
have specified the two body potential $V(\xx)$.
We immediately see that the function $U(\xx_1,\dots ,\xx_N)$ defined in \equ(2.3) has
the following properties.
\vskip.3cm
\\
$i)$ {\it $U$ is symmetric for the exchange of particles}.

\\ Let $\{\s(1),\s(2),\dots,\s(N)\}$ be a permutation of the set $\{1,2,\dots ,N\}$ then
$$
U(\xx_{\s(1)},\dots ,\xx_{\s(N)})~=~U(\xx_1,\dots ,\xx_N)  \Eq(2.4)
$$

\\
$ii)$ {\it $U$ is translational invariant}.

\\
Namely, if $(x_1,\dots,x_N)$  and $(x'_1,\dots,x'_N)$ are two configurations which differs only by a translation
then

$$
U(\xx_1,\dots ,\xx_N)~=~U(\xx'_1,\dots ,\xx'_N)
\Eq(2.5)
$$

\\Some further conditions  on the potential $V$ must be imposed. Stability and temperedness are
commonly  considered as minimal   conditions to  guarantee a good
statistical mechanics behavior of the  system (see, e.g., \cite{Ru}
and \cite{Ga}).

\index{pair potential!stability}
\begin{defi}
A pair potential $V(x)$ is said to be stable if there exists $C\ge 0$ such that, for all ${n}\in \mathbb{N}$ such that
$n\ge 2$ and all $(x_1,\dots,x_n)\in \mathbb{R}^{dn}$,
$$
\sum_{1\le i< j\le n}V(x_i-x_j)\ge -n C \Eq(2.6)
$$
\end{defi}

\index{pair potential! regularity}
\begin{defi}\label{regular}
A pair potential $V(x)$ is said to be regular
$$
\int_{\mathbb{R}^d} |e^{-\b V(x)}-1|dx< \infty \Eq(2.7)
$$
\end{defi}
A typical example of a regular pair potential $V(x)$ is  a potential which is absolutely integrable when $|x|$ is outside a sphere of radius $r_0>0$ and it is positive for $|x|\le r_0$.
Then for such a potential, letting $S_d(r_0)$ be the volume of the $d$-dimensional sphere of radius $r_0$ and letting $C=\int_{|x|> r_0}|V(x)| dx$,
$$
\begin{aligned}
\int_{\mathbb{R}^d} |e^{-\b V(x)}-1|dx & = \int_{|x|\le r_0} |e^{-\b V(x)}-1|dx + \int_{|x|> r_0} |e^{-\b V(x)}-1|dx\\
& \le S_d(r_0) + \int_{|x|> r_0}\left[\sum_{n=1}^\infty{|-\b V(x)|^n\over n!}\right]dx \\
& \le S_d(r_0) + \sum_{n=1}^\infty {\b^n\over n!}\int_{|x|> r_0}|V(x)|^n dx \\
&\le  S_d(r_0) + \sum_{n=1}^\infty {\b^n\over n!}\left(\int_{|x|> r_0}|V(x)| dx\right)^n\\
& \le  S_d(r_0) + \sum_{n=1}^\infty {\b^n\over n!}C^n\\
& = S_d(r_0)+ e^{\b C}-1\\
& <+\infty
\end{aligned}
$$

\\A stable and regular  pair potential $V(x)$ will be also called {\it admissible}.
\index{pair potential!stability constant}
\begin{defi}
\\  Given $V(x)$ admissible, and given $n\ge 2$ let
$$
B_n=\sup_{ (x_1,\dots,x_n)\in \mathbb{R}^{dn}}-{1\over n}\sum_{1\le i<j\le n}V(|x_i-x_j|)
\Eq(bn)
$$
then
the nonnegative number
$$
B=\sup_{n\ge 2} B_n
\Eq(b)
$$
is called the {\it stability constant of the potential $V(x)$}.
\end{defi}
Note that   $B$ is non-negative  and $B=0$ if and only if $V(|x|)\ge 0$ (i.e. ``repulsive" potential).
Stability and  regularity
are actually deeply interconnected and the lack of one of them always produces non thermodynamic or catastrophic  behaviors (see ahead).

\\The conditions $i)$ and $ii)$  are motivated by physical considerations. They
originate from the observation that most of
the physical interactions are indeed symmetric under exchange of
particles and translational invariant.

\\Stability and regularity are from the physical point of view
more difficult to understand. Concerning in particular regularity, it is quite natural to assume that potential must vanish at
large distances since particles far away are expected to interact
in a negligible way. On the other hand, it is not clear why the rate of decay has
to be such that \equ(2.7) is satisfied. Concerning the stability condition, we will see that it
is somehow related to the fact that particles are not allowed to  be (or pay a high price to be)  at
short distance from each other.

\\We will see below that
the grand-canonical partition defined in \equ(1.1) is  a holomorphic function of $\l$ if the potential $V(|x|)$ is stable. Moreover, under
very mild additional conditions on the potential  (upper-continuity) it can be proved that  the converse is also true (see \cite{Ru}). In other words
$\Xi_{\La}(\b,\l)$ converges
if and only if the potential $V$ is stable. So, in some sense, stability is a  {\it sine qua non}  condition to construct a    consistent  statistical mechanics
for continuous particle systems.

\vv
\begin{pro}\label{Zolo}
Let $U(\xx_1,\dots ,\xx_N)$  be a stable interaction with stability constant $B$
and let $\l$ in \equ(2.2) be allowed to vary in $\mathbb{C}$, then
series in the r.h.s. of \equ(2.2) converges absolutely for all $\l\in
\mathbb{C}$, all $\b\in \mathbb{R}^{+}$ and all $\L$ Lebesgue
measurable set in $\mathbb{R}^{d}$, or in other words  the function
$\Xi (\b,\L,\l)$, as a function of $\l$ in the complex plane, is
holomorphic.
\end{pro}
\vv
\\{\bf Proof}.
$$
\begin{aligned}
|\Xi_\L (\b,\l)| & \le 1+ \sum_{N=1}^{\infty}{|\l|^N \over N!}
\int_{\L}d\xx_1\dots\int_{\L} d\xx_{N} ~e ^{-\b U(\xx_1,\dots
,\xx_N)} \\
&\le~
1+ \sum_{N=1}^{\infty}{|\l|^N\over N!}
\int_{\L}d\xx_1\dots\int_{\L} d\xx_{N}e ^{+\b B N}\\
& =~ \sum_{N=0}^{\infty}{(|\L||\l|e ^{\b B})^N \over N!}\\
&=~\exp\{|\L||\l|e ^{\b
B}\}
\end{aligned}
$$
$\Box$
\vskip.2cm

\\The stability condition is  therefore  a sufficient condition for the
absolute convergence of the Grand Canonical partition function. As mentioned above, it
can be also shown, that the converse also holds, i.e. stability
condition is also a necessary condition (see \cite{Ru}). But in that case some further
conditions on the function $V$ are
needed. We rather prefer here to show by mean of examples how violation of stability
and/or  temperness produce non thermodynamic behaviors in
the system. In all examples to follow we will assume to work  in $\mathbb{R}^3$.

\section{Potentials too attractive at short distances}
A simple way to violate stability
is by choosing a potential $V(\xx)$ which is negative in
the neighborhood of $\xx=0$. For example, let us consider the potential
$V_1^{\rm bad}(|\xx|)$ as in figure 1: a continuous, not decreasing, compactly supported function of $|x|$. This potential is
continuous, tempered  (i.e. satisfies \equ(2.7)), bounded (i.e.
$|V_1^{\rm bad}(|\xx|)|\le \a$, for some $\a>0$),  and strictly negative
around the origin $\xx~=~0$ (i.e.  $\exists \d>0$ and $b>0$ such
that $V_1^{\rm bad}(|\xx|)\le -b$ whenever $|\xx|\le ~\d$).
This potential $V_1^{\rm bad}(|\xx|)$ is not stable. Indeed, if we place $N$
particles in positions $\xx_1,\dots, \xx_N$  so close one to each
other that $|\xx_i-\xx_j|\le \,\d$ (for all $i,j~=~1,2,\dots, N$), then
$U(\xx_1,\dots, \xx_N)\le -b{N(N-1)/ 2}$.
\begin{figure}
\begin{center}
\includegraphics[width=12cm,height=8cm]{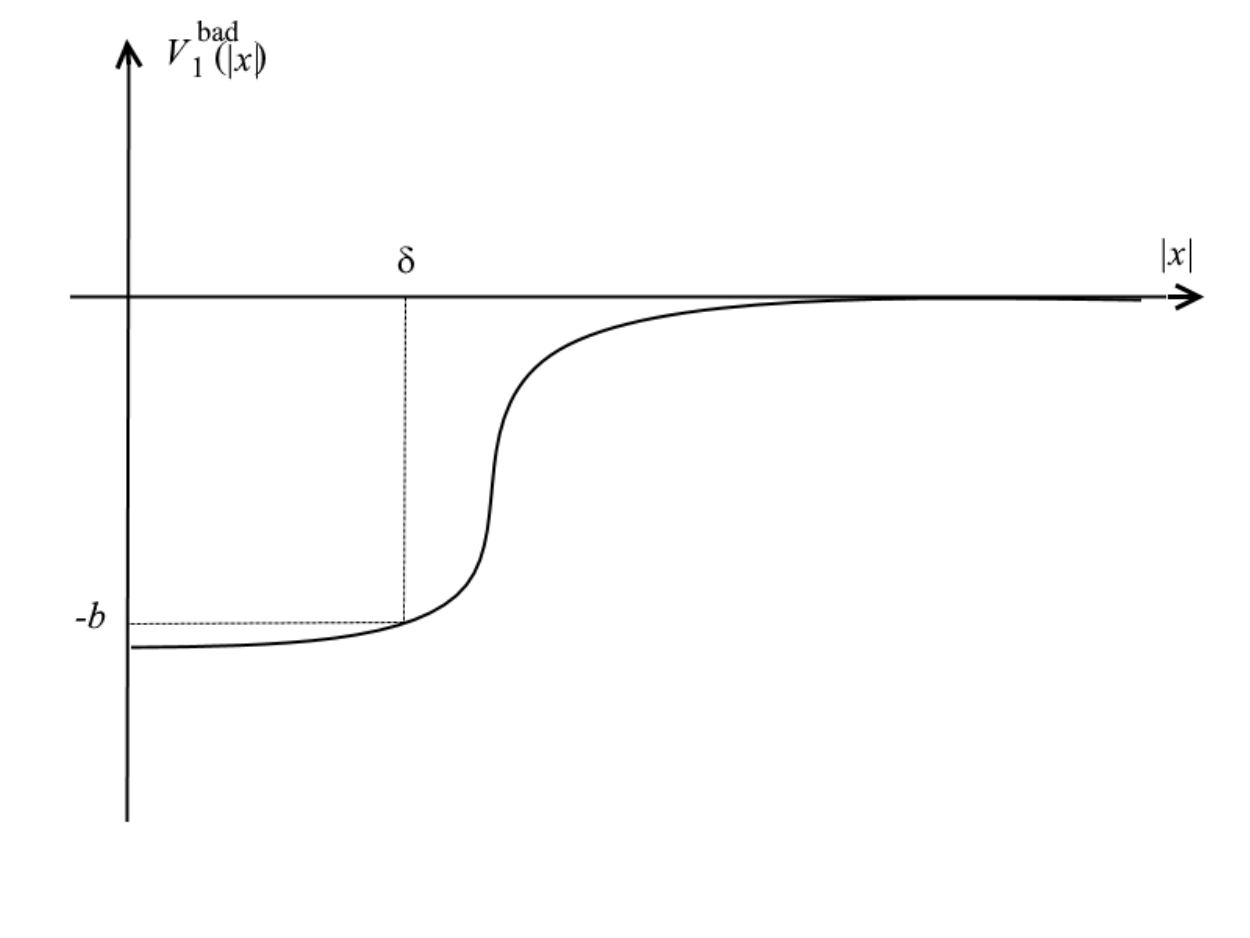}
\end{center}
\begin{center}
Figure 1. The potential $V^{\rm bad}_1(|\xx|)$
\end{center}
\end{figure}
As we said above, it can be shown that the lack of stability destroys the convergence of the series
for $\Xi_\L (\b,\l)$. I.e. it is possible to show that
the grand canonical  partition function with the potential $V_1^{\rm bad}(|\xx|)$
is a divergent series. Even so, one may argue that the partition function in the canonical ensemble for the potential
$V_1^{\rm bad}(|\xx|)$ is still well defined and calculations could then be performed in this ensemble.

\\Let us therefore do these calculations in the (configurational) Canonical Ensemble with $\b$ and $N$ fixed.

\\First we consider a  catastrophic situation in which the $N$ particles collapse in a small region
inside $\L$.
Namely, we calculate
the probability to find the system in a micro-state with the  $N$
particles  being {\it all} located in a small sphere $S_\d\subset \L$ of
radius $\d/2$ (so that they are all at distance less or equal to $\d$). By
\equ(2.1) a lower bound for such probability is given by
$$
\begin{aligned}
P_{\rm bad}(N)& =~{1\over Z_\L (\b,N )} \int_{S_\d}d\xx_1\dots
\int_{S_\d}d\xx_N {1\over N!} e^{ -\b U(\xx_1,\dots ,\xx_N)}\\
& \ge ~{1\over Z_\L(\b,N)} \left[{\pi \d^3\over 6}\right]^N
{1\over N!} e^{ +\b b{N(N-1)\over 2 } }
\end{aligned}
$$
Now consider configurations ``macroscopically correct", i.e.
micro-states with $N$ particles in positions $\xx_1, \dots \xx_N$
uniformly distributed in the box $\L$, with density equal to the
density $\r~=~{N/|\L|}$ fixed by the parameters $N$ and $|\L|$
in the  canonical ensemble. If the potential $V_1^{\rm bad}$ is tempered
then it is not difficult to see that for such configurations
$|U(\xx_1, \dots ,\xx_N)|\le  C N\r$ where
$2C~=~\int_{\mathbb{R}^3}d\xx|V(\xx)|$. As a matter of fact,
let us consider the box $\L$ as the disjoint union of small cubes $\D$  (with volume  $|\D|$) so that
each of the particles $x_1\dots, x_N$ belongs to one of these
small cubes $\D$ (with volume  $|\D|$). Hence, given a configuration  $x_1,\dots, x_N$ (recall that $V_1^{\rm bad}(\xx)$ non-positive by assumption)
$$
\begin{aligned}
\sum_{j\in \{1,2,\dots, N\}\atop j\neq i}V_1^{\rm bad}(|\xx_j-\xx_i|)& =~\sum_{\D\subset \L}
\sum_{j\in \{1,2,\dots, N\}\atop \;\xx_j\in \D,~j\neq i}V_1^{\rm bad}(|\xx_j-\xx_i|)|\\
& \ge
\sum_{\D\subset \L}V_1^{\rm bad}(r_\D)\sum_{j\in \{1,2,\dots, N\}\atop \;\xx_j\in \D, \;j\neq i}1
\end{aligned}
$$
where in the last line to get the inequality we have used the assumption that $V_1^{\rm bad}$ is non decreasing
and $\sum_{\D\subset \L}$ runs over all small cubes  whose disjoint union is $\L$ and $r_\D$ denotes  the
(minimal) distance of a point inside the cube $\D$ from the point $\xx_i$. Now, since we
are assuming that particles are uniformly distributed in $\L$ and
choosing the dimensions of the small cubes sufficiently large in
order to still consider these cubes macroscopic, so that  the
particles in a small cube $\D$ are still uniformly distributed
with density $\r$ or very close to $\r$. Hence we may assume that there exists $\e>0$  such that, for each $\D\subset \L$
$$
\sum_{j\in \{1,2,\dots,N\}\atop \;\xx_j\in \D\;j\neq i}1~\le (1+\e) \r |\D|
$$
so that
$$
\begin{aligned}
\sum_{j\in \{1,2,\dots, N\}:\;j\neq i}V_1^{\rm bad}(|\xx_j-\xx_i|) & \ge \r(1+\e)\,\mbox{$\sum_{\D}$}V_1^{\rm bad}(r_\D) \D\\
&\ge \r(1+\e)\int
_{\L}V_1^{\rm bad}(|\xx-\xx_i|)d\xx\\
& \ge \r(1+\e)\int _{\mathbb{R}^3}V_1^{\rm bad}(\xx)d\xx
\end{aligned}
\Eq(string)
$$
Hence, for such configurations we have
$$
\begin{aligned}
\sum_{1\le i<j\le N}V_1^{\rm bad}(\xx_i -\xx_j)&=~{1\over 2}\sum_{i=1}^{N}
\sum_{j: ~j\neq i}V_1^{\rm bad}(\xx_i -\xx_j)\\
&\ge ~{ N\r(1+\e)\over 2} \int
_{\mathbb{R}^3}V_1^{\rm bad}(\xx) d\xx\\
&=~ -N C \r
\end{aligned}
$$
with
$$
C={(1+\e)\over 2}\left[-\int_{\mathbb{R}^3}V_1^{\rm bad}(\xx)dx\right]
$$
 positive  (again, recall that $V_1^{\rm bad}(\xx)$ non-positive by assumption).
Then an upper bound for the  probability for such configuration to
occur is, according with \equ(2.1)
$$
\begin{aligned}
P_{\rm good}(N)& =~{1\over Z_\L (\b,N)} \int_{\L}d\xx_1\dots
\int_{\L}d\xx_N {1\over N!} e^{ -\b U(\xx_1,\dots ,\xx_N)}\\
&\le~ {1\over Z_\L(\b,N)} |\L|^N {1\over N!} e^{ +\b \r CN}\\
&= ~{1\over Z_\L (\b,N)} N^N\r^{-N} {1\over N!} e^{ +\b C\r N }
\end{aligned}
$$

\\Hence a lower bound for the ratio between the  probability of bad
configurations and good configurations is
$$
{P_{\rm bad}(N)\over P_{\rm good}(N)}~\ge ~{ \left[{\pi\over 6}
\r\d^3\right]^N {1\over N!} e^{ +\b b{N(N-1)\over 2 } }\over
N^N {1\over N!} e^{ +\b C\r N }}~=~ \left[{\pi\over 6}
\r\d^3\over e^{\b ({b\over 2}+C\r)} \right]^N {e^{ +\b {b\over 2}{N^2} }\over N^N}
$$
and, no matter how small $\d$ and/or $b$ are  we have that
$$
\lim_{N\to \infty} {P_{\rm bad}(N)\over P_{\rm good}(N)}~=~+ \infty
$$
This means that it is far more probable to find the system in a
micro-state in which  all particles are all contained within a small sphere of diameter $\d$ in some place of $\L$
rather than in a micro-state ``macroscopically correct", i.e. a  configuration with particles
uniformly distributed in $\L$ with a constant density $\r=N/|\L|$. \vskip.3cm

\section{ Two examples of infrared catastrophe }
In the previous section we consider a  non-stable (but regular) pair potential  $V_1^{\rm bad}(|\xx|)$ which was
not preventing  particles to accumulate in arbitrarily small regions of the space.
We now show that the lack of stability and/or regularity of a pair potential,  even  preventing particles to accumulate in arbitrarily small regions of the space,
yields to non-thermodynamic situations.

\subsection{Potential with hard core too attractive at large distances}
We  consider a second
case of ``bad" potential which this time does not allow particles to come together at arbitrarily short distances  but it is too {\it
attractive} at large distances.  Let  the space dimension be set at $d=3$ and let $a>0$, $3>\e>0$ and define,
$$
V_2^{\rm bad}(|\xx|)~=~
\begin{cases}
+ \infty & {\rm if}~ |\xx|\le a\\
{-|\xx|^{-3+\e}} &{\rm otherwise}
\end{cases}
\Eq(2.8)
$$
\begin{figure}
\begin{center}
\includegraphics[width=16cm,height=12cm]{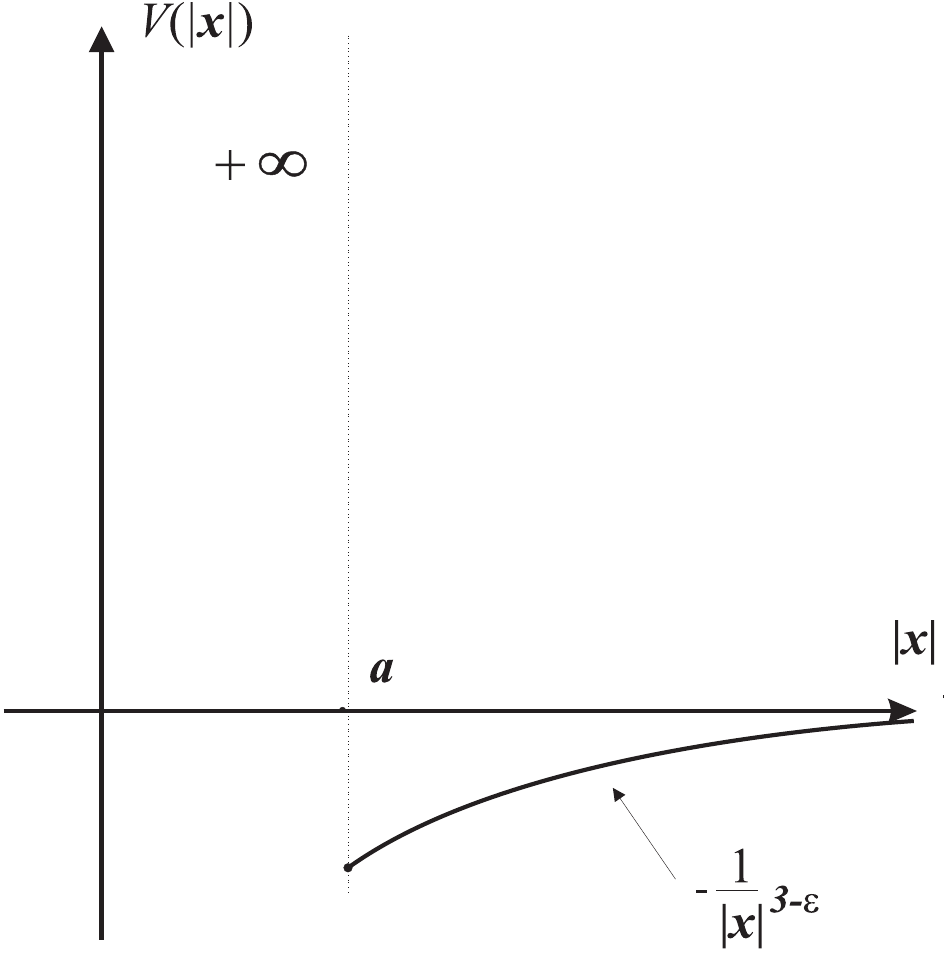}
\end{center}
\begin{center}
Figure 2. The potential $V^{\rm bad}_2(|\xx|)$
\end{center}
\end{figure}

\\It can be proved that this potential, illustrated in
figure 2, is neither stable nor tempered.  This is a first example of a so called  {\it hard-core type potential} (where the hard core condition
is $V_2^{\rm bad}(x)= +\infty $ if $|\xx|\le a$).
\index{pair potential!hard core}
It describes a
system of interacting hard spheres of radius $a$. In fact, since
$V_2^{\rm bad}(|\xx|)$ is $+\infty$ whenever $|\xx|\le a$, then
$U(\xx_1,\dots, \xx_N) =+\infty$ whenever $|\xx_i -\xx_j|\le a$
for some $i,j$, thus the Gibbs factor for such configuration (i.e
where some particles are at distances less or equal $a$) is $
\exp\{ -\b U(\xx_1,\dots, \xx_N)\}~=~0$ and hence it has zero
probability to occur.

\\This means that such system cannot take densities greater than a certain density
$\r_{\rm cp}=c/a^3$ called the close-packing density (here $c$ is a suitable constant\footnote{The constant for the $d=3$ case is actually known: $\p/2\sqrt{3}$.}), where particles are
as near as possible one to each other compatibly with the hard
core condition. E.g., particles are  arranged  (in a close-packed way)
at vertices of a face-centered cubic lattice with nearest neighbors at distance $a$.

\\Let us again consider  the system in the Canonical Ensemble at
fixed inverse temperature $\b$, fixed volume $|\L|$ and fixed number of particles $N$, hence at fixed $\r~=~N/V$. We consider the case $d=3$ and choose $N$ and $|\L|$
in such way that $\r$ is much smaller that the
close-packing density $\r_{\rm cp}$, i.e. $\r / \r_{\rm cp}\ll 1$.

\\We now compare the probability to find the system in a region of the space-phase constituted by micro-states
near the close-packing situation. To be specific, let $y_1,\dots y_N$ be sites of a face-centered cubic lattice
with nearest neighbors at distance $a+\d$ (with $\d$ much smaller than $a$) organized in a close-packed configuration in such way that they fill
completely a cube $\L_{\rm cp}$ of size $L_{\rm cp}$. We
now  suppose that the $N$ particles are at positions $x_1,\dots x_N$ such that $|x_i-y_i|<{\d\over 2}$ for $i=1,\dots N$.
In other words the $i^{\rm th}$ particle can move freely inside the sphere $S_\d^i$ of radius $\d/2$ centered at the vertex $y_i$ of the fcc lattice
without violating the hard core condition of the pair potential. In
this set of configurations the density can vary from the maximum
$\r_{\rm cp}= {c\over a^3}$ to a minimum
density which is surely greater than $\tilde\r_{\rm cp}={c\over (a+2\d)^3}=(1+{2\d\over a})^{-3}\r_{\rm cp}$.
In such configurations  the system does not occupy uniformly all the
available volume $|\L|$, but it rather occupies the fraction ${\r\over\r_{\rm cp}}$ of the available volume $\L$ (since $N=\r_{\rm cp}|\L_{\rm cp}|=\r|\L|$ and thus
$|\L_{\rm cp}|~=~|\L|{\r\over\r_{\rm cp}}$)
leaving  a region
(with volume $|\L|{(\r_{\rm cp}-\r)/ \r_{\rm cp} }$) empty
inside $\L$. Of course such configurations are non thermodynamics. To fix the ideas and simplify the calculations let us suppose that $\L$ is a
sphere of radius $R$ while we recall that by assumption
$\L_{\rm cp}$ is a  cube of size $L_{\rm cp}$. Let of course further assume that $a$ is much smaller  than $L_{\rm cp}$ (actually
$a<{L_{\rm cp}\over 2^{1/\e}}$ will do).

\\In these configurations the potential energy $U$ is strongly negative.
An upper bound for the value of $U$ for such type of
configurations can be obtained by a reasoning similar to that  took us to
the \equ(string).

\\Indeed, suppose that particles are arranged in a configuration near the close packing (in the sense specified above). Then as before
$$
\begin{aligned}
\sum_{j\in \{1,2,\dots, N\}\atop j\neq i}V_2^{\rm bad}(|\xx_j-\xx_i|)&=~\sum_{\D\subset \L}
\sum_{j\in \{1,2,\dots, N\}\atop \;j\neq i,\,\xx_j\in \D}V_2^{\rm bad}(|\xx_j-\xx_i|)|\\
&\le~
\sum_{\D\subset \L}V_2^{\rm bad}(r^{\rm max}_\D)\sum_{j\in \{1,2,\dots, n\}\atop j\neq i,\;\xx_j\in \D}1
\end{aligned}
$$
where this time $r^{\rm max}_\D$ represent the maximal distance of a point in the small cube $\D$ from $x_i$.
Now we have that
$$
\sum_{j\in \{1,2,\dots, n\}\atop j\neq i,\;\xx_j\in \D}1~\ge  \tilde\r_{\rm cp} |\D|
$$
Moreover there exists surely an $\eta_\D$ (depending on $\D$) such that
$$
\sum_{\D\subset \L}V_2^{\rm bad}(r^{\rm max}_\D)|\D|~\le ~-(1-\eta_\D)B_\e(\L_{\rm cp},x_i)
$$
where $\L_{\rm cp}$ denotes the region inside $\L$ where the
$N$ close-packed particles are situated (which, recall, is assumed to be a cube contained in the sphere $\L$) and
$$
B_\e(\L_{\rm cp},x_i)~=~\int_{\xx\in \L_{\rm
cp},~|\xx-x_i|>a}{1\over |\xx-x_i|^{3-\e}}d\xx
$$
Now, since we are assuming that $\L_{\rm cp}$ is a square then, wherever $x_i$ is situated inside $\L_{\rm cp}$, there is always
a quarter of sphere, say $U$,
of radius $L_{\rm cp}/2$ all contained in $\L_{\rm cp}$ such that $x_i$ is the center of this sphere
and hence
$$
\begin{aligned}
B_\e(\L_{\rm cp},x_i)&=~\int_{\xx\in \L_{\rm
cp},~|\xx-x_i|>a}{1\over |\xx-x_i|^{3-\e}}d\xx\\
&\ge~ \int_{x\in U,~|\xx|>a} {1\over |\xx|^{3-\e}}d\xx\\
& =~\int_0^{\pi\over 2}d\varphi\int_0^{\pi\over 2} \sin \theta d\theta\int_a^{L_{\rm cp}\over 2} {r^2\over r^{3-\e}}  dr\\
&=~
{\pi\over 2\e}(L_{\rm cp}^\e-a^\e)\\
&\ge~{\pi\over 4\e}L_{\rm cp}^\e
\end{aligned}
$$
Therefore we get, observing that $L^\e_{\rm cp}=\L_{\rm cp}^{\e\over 3}$
$$
\sum_{j\in \{1,2,\dots, N\}\atop j\neq i}V_2^{\rm bad}(|\xx_j-\xx_i|)~\le~
-(1-\eta_\D)\tilde\r_{\rm cp}{\pi\over 4\e} |\L_{\rm cp}|^{\e\over 3}
$$
and therefore
$$
\begin{aligned}
U(\xx_1,\dots, \xx_N)&=~{1\over 2}\sum_{i=1}^N\sum_{j\in \{1,2,\dots, N\}\atop j\neq i}V_2^{\rm bad}(|\xx_j-\xx_i|)\\
& \le~  - {N}C\r_{\rm cp} |\L_{\rm cp}|^{\e\over 3}
\end{aligned}
$$
with
$$
C= {1\over 2}(1+{2\d\over a})^{-3}(1-\eta_\D){\pi\over 4\e}\Eq(CI)
$$

\\This allows to bound from below the probability $P_{\rm bad}(N)$ to find the system
 (in the canonical ensemble) in such  bad configurations near the close-packing as follows.
$$
\begin{aligned}
P_{\rm bad}(N)& =~{1\over Z_\L (\b,N)} \int_{S^1_\d}d\xx_1\dots
\int_{S^N_\d}d\xx_N {1\over N!} e^{ -\b U(\xx_1,\dots ,\xx_N)}\\
& \ge {1\over Z_\L (\b,N)} \left[{4\over 3}\pi \d^3\right]^N
{1\over N!} e^{ +\b  {N}C\r_{\rm cp}|\L_{\rm cp}|^{\e\over 3} }
\end{aligned}
$$
As far as good configurations (i.e. those with uniform  density $\r= N/|\L|$) are concerned,  proceeding similarly one can bound
$$
\begin{aligned}
\sum_{j\in \{1,2,\dots, N\}\atop j\neq i}V_2^{\rm bad}(|\xx_j-\xx_i|)&=~\sum_{\D\subset \L}
\sum_{j\in \{1,2,\dots, N\}\atop \;\xx_j\in \D,~j\neq i}V_2^{\rm bad}(|\xx_j-\xx_i|)|\\
&\ge
\sum_{\D\subset \L}V_2^{\rm bad}(r^{\rm min}_\D)\sum_{j\in \{1,2,\dots, N\}\atop \;\xx_j\in \D, \;j\neq i}1
\end{aligned}
$$
where this time $r^{\rm min}_\D$ represent the minimal distance of a point in the small cube $\D$ from $x_i$. Since, as before,
 we
are assuming that particles are uniformly distributed in $\L$ with density $\r$ we
may assume that there exists a small $\g>0$  such that, for each $\D\subset \L$
$$
\sum_{j\in \{1,2,\dots,N\}\atop \;\xx_j\in \D\;j\neq i}1~\le (1+\g) \r |\D|
$$
so that
$$
\sum_{j\in \{1,2,\dots, N\}\atop j\neq i}V_2^{\rm bad}(|\xx_j-\xx_i|)~\ge ~\r(1+\g)\,
\sum_{\D\subset\L}V_1^{\rm bad}(r^{\rm min}_\D) |\D|
$$
Now we have
$$
\sum_{\D\subset\L}V_1^{\rm bad}(r^{\rm min}_\D) |\D| \ge  -B_\e(\L,x_i)
$$
with  (recall that $\L$ is a sphere of radius $R$ and thus $R=(3/4\pi)^{1\over 3} |\L|^{1\over 3}$)
$$
\begin{aligned}
B_\e(\L,x_i) & = ~\int_{\xx\in \L,~|\xx-x_i|>a}{1\over |\xx-x_i|^{3-\e}}d\xx\\
&\le~
\int_{\xx\in \L,~|\xx|>a}{1\over |\xx|^{3-\e}}d\xx\\
& ={4\pi\over \e}(R^\e-a^\e)\\
&\le {3^{\e\over 3}(4\pi)^{1+{\e\over 3}}\over \e}|\L|^{\e\over 3}
\end{aligned}
$$
Hence
$$
\sum_{j\in \{1,2,\dots, N\}\atop j\neq i}V_2^{\rm bad}(|\xx_j-\xx_i|)~\ge ~-\r(1+\g)\,
{3(4\pi)^{1+{\e\over 3}}\over \e}|\L|^{\e\over 3}
$$
and therefore for the ``good" configurations $x_1,\dots,x_N$ we have
$$
\begin{aligned}
U(x_1,\dots,x_N)& =~\sum_{1\le i<j\le N}V_1^{\rm bad}(\xx_i -\xx_j)\\
&=~{1\over 2}\sum_{i=1}^{N}
\sum_{j: ~j\neq i}V_1^{\rm bad}(\xx_i -\xx_j)\\
&\ge ~-N~{(1+\g)\over 2}{3(4\pi)^{1+{\e\over 3}}\over \e}\r|\L|^{\e\over 3}\\
&= ~ -N C' \r |\L|^{\e\over 3}
\end{aligned}
$$
with
$$
C'={(1+\g)\over 2}{3(4\pi)^{1+{\e\over 3}}\over \e}\Eq(CIP)
$$
Thus an upper bound for the ``good" configurations is
$$
\begin{aligned}
P_{\rm good}(N)& =~{1\over Z_\L(\b,N)} \int_{\L}d\xx_1\dots
\int_{\L}d\xx_N {1\over N!} e^{ -\b U(\xx_1,\dots ,\xx_N)}\\
&\le~ {1\over Z_\L(\b,N)} |\L|^N {1\over N!} e^{ +\b {N}C'\r |\L|^{\e\over 3}}
\end{aligned}
$$

\\Hence the ratio between the  probability of bad and good configurations is
$$
{P_{\rm bad}(N)\over P_{\rm good}(N)}\ge { \left[{4\over 3}\pi
\d^3\right]^N
 e^{ +\b  NC\r_{\rm cp}  |\L_{\rm cp}|^{\e\over 3} }
\over |\L|^N
 e^{ +\b N C'\r |\L|^{\e\over 3} }}
$$
recalling that
$$
|\L|~=~N/\r ~~~~~~|\L_{\rm cp}|~=~N/\r_{\rm cp} ~$$
we get
$$
{P_{\rm bad}(N)\over P_{\rm good}(N)}\ge \left[{4\over 3}\pi
\d^3\r\right]^N { e^{ +\b C N^{1+{\e\over 3}}\left[\r^{1-{\e\over 3}}_{\rm cp}  - {C'\over C}\r^{1-{\e\over
3}}\right] } \over N^N}
$$
Now observe that factor $\left[\r^{1-{\e\over 3}}_{\rm
cp}  - {C'\over C}\r^{1-{\e\over 3}}\right]$ in the exponential is positive
if $\r$ is sufficiently smaller than $\r_{\rm cp}$, i.e. , recalling \equ(CI) and \equ(CIP), if
$$
{\r\over \r_{\rm cp}}<{1-\eta_\D\over 1+\g} {1\over 16 (12\p)^{\e\over 3}(1+{2\d\over a})^3}
$$
Hence, calling
$$
C_1~=~
 e^{ +\b C\left[\r^{1-{\e\over 3}}_{\rm cp}  -
{C'\over C}\r^{1-{\e\over 3}}\right] },~~~~C_2 ~=~ \left[{4\over 3}\pi
\d^3\r\right]^{-1}
$$
and noting that $C_1>1$ (for $\r$ sufficiently smaller than $\r_{\rm cp}$) we get
$$
{P_{\rm bad}(N)\over P_{\rm good}(N)}\ge { C_1^{ N^{1+{\e\over
3}}} \over C_2^N N^N}
$$
It is just a simple exercise to show that, if $C_1>1$, this ratio
goes to infinity as $N\to \infty$ (write e.g. $N^N~=~e^{N\ln N}$) no matter how large is $C_2$.

\\It is interesting to stress that gravitational interaction behaves at large distances
exactly as $\sim |\xx|^{-3+\e}$ with $\e~=~2$. Hence we can expect
that the matter in the universe  does not obey the laws of
thermodynamics and in particular it is not distributed as a
homogeneous low-density gas (and indeed it is really the case!).

\subsection{Potential with hard core too repuslive at large distances}

\\Consider now a similar case where the pair potential has the same decay as in \equ(2.8), but
now is purely repulsive, i.e. suppose
$$
V_3^{\rm bad}(x)~=~
\begin{cases}
+ \infty &{\text if}  |\xx|\le a\\
{1\over |\xx|^{3-\e}} & {\rm otherwise}
\end{cases}
\Eq(2.9)
$$
\begin{figure}
\begin{center}
\includegraphics[width=10cm,height=7cm]{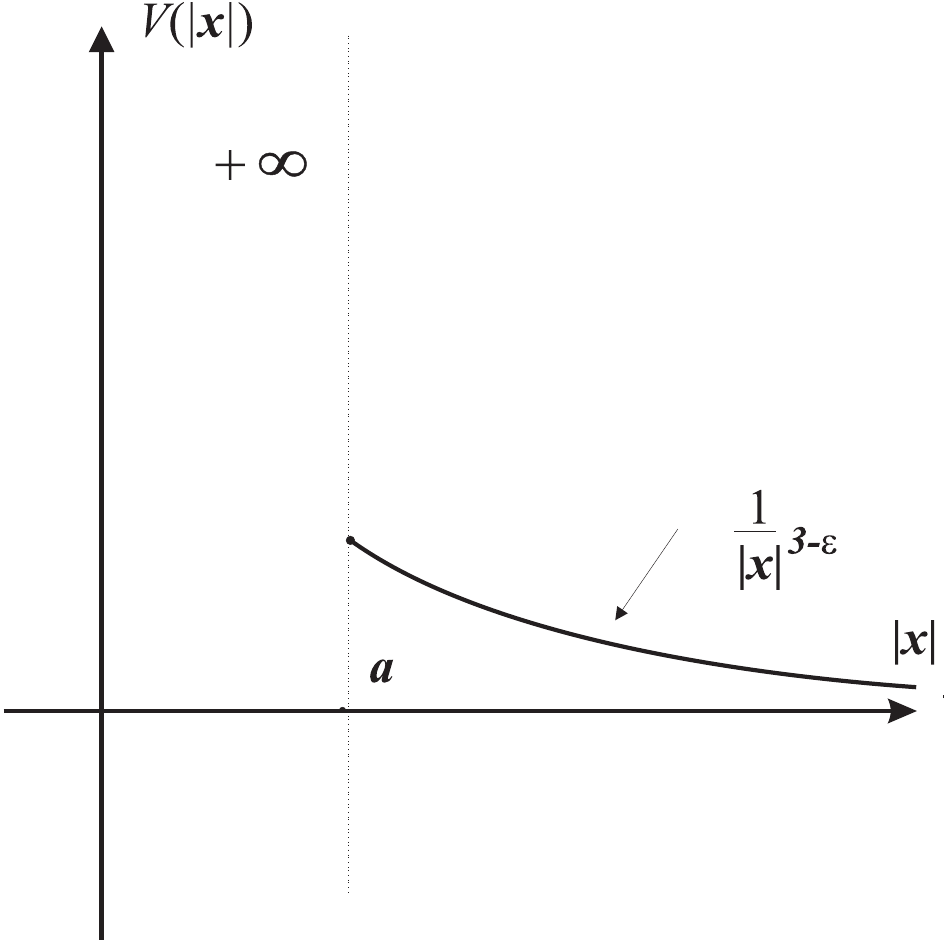}
\end{center}
\begin{center}
Figure 3. The potential $V^{\rm bad}_{3}(|\xx|)$
\end{center}
\end{figure}
This is indeed a stable potential (since it is positive!) but not regular. It will produce with high probability
 bad configurations in which particles tend
to accumulate at the boundary of $\L$ in a close-packed
configuration hence forming a layer. By an argument identical to
the one of case 2, supposing $\r\ll\r_{\rm cp}$, one again shows
that ``close-packed" configurations are far more probable than ``thermodynamic
configurations" (with particles uniformly distributed in $\L$).
Indeed, it is possible to show that particles interacting via a potential of type \equ(2.9)
will tend to leave the center of the container and accumulate in a
layer at the boundary of the container.

\\To simplify the calculations let us suppose that the volume $\L$ enclosing the system is a sphere of radius $L$ and
let us estimate the probability of a bad configuration in which the $N$ particles are in configurations
$x_1,\dots ,x_N$ which are nearly close-packed
(i.e. the  $i^{\rm the}$ particle can move freely inside the sphere $S_\d^i$ of radius $\d/2$ centered at the vertex
$y_i$ of the fcc lattice of step $a+\d$) occupying a region  $\L_{\rm cp}$
which is now a layer stick  at the boundary of the box $\L$ with  thickness  $\D L$.
%Since
%we are assuming that $\L$ is a sphere of radius $L$ we have that the volume of
%the occupied region $\L_{\rm cp}$ is $|\L_{\rm cp}|= {4\over 3} \p (L^3-(L-\D L)^3)\approx 4\pi L^2\D L$ for
%$\r\ll\r_{\rm cp}$. Now, since $|\L_{\rm cp}|=N/\r_{\rm cp}$ and   $|\L|=N/\r$,
%we get $|\L_{\rm cp}|\r_{\rm cp}=|\L|\r$ and thus $ 4\pi L^2\D L\r_{\rm cp}={4\over 3}\pi L^3\r$ and finally
%$$
%\D L= {\r\over 3\r_{\rm cp}}L\Eq(DL)
%$$
Now,
by the same argument seen above we have that for such configurations  (recall that now the pair potential is everywhere positive)
we have
$$
\begin{aligned}
\sum_{j\in \{1,2,\dots, N\}\atop j\neq i}V_3^{\rm bad}(|\xx_j-\xx_i|)&=~\sum_{\D\subset \L}
\sum_{j\in \{1,2,\dots, N\}\atop \;j\neq i,\,\xx_j\in \D}V_3^{\rm bad}(|\xx_j-\xx_i|)|\\
& \le~
\sum_{\D\subset \L}V_3^{\rm bad}(r^{\rm min}_\D)\sum_{j\in \{1,2,\dots, n\}\atop j\neq i,\;\xx_j\in \D}1
\end{aligned}
$$
where this time $r^{\rm min}_\D$ represent the minimal distance of a point in the small cube $\D$ from $x_i$.
Now we have that
$$
\sum_{j\in \{1,2,\dots, n\}\atop j\neq i,\;\xx_j\in \D}1~\le  \r_{\rm cp} |\D|
$$
Moreover there exists surely an $\eta_\D$ (depending on $\D$) such that
$$
\sum_{\D\subset \L}V_3^{\rm bad}(r^{\rm max}_\D)|\D|\le (1+\eta_\D)\int_{\xx\in \L_{\rm
cp},~|\xx-x_i|>a}{1\over |\xx-x_i|^{3-\e}}d\xx
$$
where $\L_{\rm cp}$ denotes the region  (a layer now) inside $\L$ where the
$N$ close-packed particles are situated.
Now we have that
$$
\int_{\xx\in \L_{\rm
cp},~|\xx-x_i|>a}{1\over |\xx-x_i|^{3-\e}}d\xx\le \int_{\xx\in \L_{\rm
cp}}{1\over |\xx-x_i|^{3-\e}}d\xx \le
\int_{\xx\in \tilde\L_{\rm cp}}{1\over |\xx-x_i|^{3-\e}}d\xx
$$
where $\tilde\L_{\rm cp}$ is a sphere centered at the $x_i$ with volume $|\tilde\L_{\rm cp}|= |\L_{\rm cp}|$ (i.e. the sphere $\tilde\L_{\rm cp}$
has the same volume of the layer $\tilde\L_{\rm cp}$ and hence with radius $\tilde R= (3|\L_{\rm cp}|/4\pi))^{1\over 3}$. Therefore
$$
B_\e(\L_{\rm cp},x_i)\le \int_{\xx\in \tilde\L_{\rm cp}}{1\over |\xx-x_i|^{3-\e}}d\xx=
{4\pi\over \e} (3/4\pi))^{\e\over 3} |\L_{\rm cp}|^{\e\over 3} \le {4\pi\over \e} |\L_{\rm cp}|^{\e\over 3}
$$
and thus
$$
\sum_{j\in \{1,2,\dots, N\}\atop j\neq i}V_2^{\rm bad}(|\xx_j-\xx_i|)~\le~
{4\pi\over \e} |\L_{\rm cp}|^{\e\over 3}\r_{\rm cp}
$$
Therefore, recalling that $\r_{\rm cp}|\L_{\rm cp}|=N$
$$
U(\xx_1,\dots, \xx_N)\le   {N\over 2}(1+\eta_\D){4\pi\over \e} |\L_{\rm cp}|^{\e\over 3}\r_{\rm cp} =
{N^{1+{\e\over 3}}\over 2}(1+\eta_\D){4\pi\over \e}\r_{\rm cp}^{2\e\over 3}
$$
and hence
$$
P_{\rm bad}(N)~=~{1\over Z_\L(\b,N)} \int_{S^1_\d}d\xx_1\dots
\int_{S^N_\d}d\xx_N {1\over N!} e^{ -\b U(\xx_1,\dots ,\xx_N)}
\ge
$$
$$
\ge {1\over Z_\L(\b,N)} \left[{4\over 3}\pi \d^3\right]^N
{1\over N!} e^{ -\b   {N^{1+{\e\over 3}}\over 2}(1+\eta_\D){4\pi\over \e}\r_{\rm cp}^{2\e\over 3}}
$$
On the other hand, for ``good" configurations $x_1,\dots, x_N$ in which the particles are uniformly distributed in $\L$ with
density $\r=N/|\L|$ (and, say,  never lower than $(1-\e)\r$ for some small $\e$)
we have
$$
\begin{aligned}
\sum_{j\in \{1,2,\dots, N\}\atop j\neq i}V_3^{\rm bad}(|\xx_j-\xx_i|)&=~\sum_{\D\subset \L}
\sum_{j\in \{1,2,\dots, N\}\atop \;j\neq i,\,\xx_j\in \D}V_3^{\rm bad}(|\xx_j-\xx_i|)|\\
& \ge~
\sum_{\D\subset \L}V_3^{\rm bad}(r^{\rm max}_\D)\sum_{j\in \{1,2,\dots, n\}\atop j\neq i,\;\xx_j\in \D}1\\
& \ge~ \sum_{\D\subset \L}V_3^{\rm bad}(r^{\rm max}_\D)(1-\e)\r|\D|
\end{aligned}
$$
Now, as before
$$
\sum_{\D\subset \L}V_3^{\rm bad}(r^{\rm max}_\D)|\D|~\ge~ (1-\eta_\D)\int_{\L:|x-x_i|> a}{1\over |x-x_i|^{3-\e}}dx
$$
But $\int_{\L:|x-x_i|> a}{1\over |x-x_i|^{3-\e}}dx$ is minimal when $x_i$ is on the boundary of the sphere $\L$ and we can bound
$$
\begin{aligned}
\int_{\L:|x-x_i|> a}{1\over |x-x_i|^{3-\e}}dx & \ge \int_{0}^{2\pi}\int_{0}^{\pi/4}\int_a^{L} {r^2\over r^{3-\e}}drd\theta d\phi\\
&=
{\pi\over 2\e}\left[L^\e-a^\e\right]\\
& \ge~{\pi\over 4\e}L^\e\\
& \ge  {|\L|^{\e\over 3}\over 3\e}
\end{aligned}
$$
and thus
$$
U(\xx_1,\dots, \xx_N)\ge   {N\over 6} {1\over \e}|\L|^{\e\over 3}\r
$$
and so
$$
\begn
P_{\rm good}(N) & ={1\over Z_\L (\b,N)} \int_{\L}d\xx_1\dots
\int_{\L}d\xx_N {1\over N!} e^{ -\b U(\xx_1,\dots ,\xx_N)}
\\
& \le {1\over Z_\L (\b,N)} |\L|^N {1\over N!} e^{-\b  {N\over 6}{1\over \e}|\L|^{\e\over 3}\r}
\egn
$$

\\Therefore the ratio between the  probability of bad and good configurations is now
$$
\begn
{P_{\rm bad}(N)\over P_{\rm good}(N)} &  \ge  { \left[{4\over 3}\pi
\d^3\right]^N
 e^{ -\b   {N\over 2}(1+\eta_\D){4\pi\over \e} |\L_{\rm cp}|^{\e\over 3}\r_{\rm cp}}
\over |\L|^N
 e^{-\b  {N\over 6}{1\over \e}|\L|^{\e\over 3}\r}}\\
& =~ \left[{4\pi\d^3\over 3|\L|}\right]^N e^{\b {N^{1+{\e\over 3}}\over \e}\left[4\p(1+\eta_\D)\r_{\rm cp}^{4\over 3}-{\r^{4\over 3}\over 6} \right]}
\egn
$$
where  recall that $|\L|=\r N $ and $|\L_{\rm cp}|=\r_{\rm cp} N $. Hence we get
$$
\begn
{P_{\rm bad}(N)\over P_{\rm good}(N)} & \ge
 \left[{4\pi\d^3\over 3|\L|}\right]^N e^{\b {N^{1+{\e\over 3}}\over \e}
 \left[4\p(1+\eta_\D)\r_{\rm cp}^{4\over 3}-{\r^{4\over 3}\over 6} \right]}\\
& \ge \left[{4\pi\d^3\over 3|\L|}\right]^N e^{\b {N^{1+{\e\over 3}}\over \e}\left[\r_{\rm cp}^{4\over 3}-{\r^{4\over 3}} \right]}
\egn
$$
which diverges as $N\to \infty$.

\\Again physics gives us an example of  potential such as \equ(2.9). That is, the purely repulsive
Coulomb potential (for which $\e=2$) between charged particles {\it with the same
charge}. Indeed electrons in excess inside a conductor tend to
accumulate at the boundary of the conductor forming layers.

\section{The Ruelle example}
In Example 1, $V_1^{\rm bad}$ was tempered but not stable. It was a potential
strictly negative at the origin. Therefore a necessary condition for a pair potential $V(x)$ to be stable is
$V(0)\ge 0$. In the two examples of the previous section we have seen that potentials with a hard-core
preventing that particles accumulate in an arbitrarily  small region of $\mathbb{R}^d$
but with a two slow decay at large distances (i.e non absolutely integrable and hence non regular)  yield to non thermodynamic behaviors.
Failure of stability however can occur also for regular
potential which are strictly positive in the neighborhood of the origin. We will
consider here a very interesting and surprising example of a potential in $d=3$ dimensions which is strongly
positive near the origin and  regular according to Definition \ref{regular}, that nevertheless
is a non stable potential. This subtle example, originally due to
Ruelle, illustrates very well the intuitive fact that the stability condition is there to avoid  the collapse
of many particles into a bounded region of $\mathbb{R}^d$ and it also
shows the key role played by the continuum  where we have always the
possibility to put an arbitrary number of particles in a small region of $\mathbb{R}^d$.

\\Let $R>0$ and let $\d>0$ and let
$$
V^{\rm bad}_4(\xx)~=~
\begin{cases}
11 & {\rm if} ~|\xx|<R-\d\\
-1 & {\rm if}~
R-\d\le|\xx|\le R+\d\\
0 & {\rm otherwise}
\end{cases}
\Eq(2.10b)
$$
This is clearly a  regular potential (actually it is bounded and {\it finite
range}: particles at distances greater than $R+\d$ do
not interact at all). Neverthless, we will show that this is a non-stable
potential by proving that the grand canonical partition function
diverges for such a potential.
\begin{figure}
\begin{center}
\includegraphics[width=7cm,height=7cm]{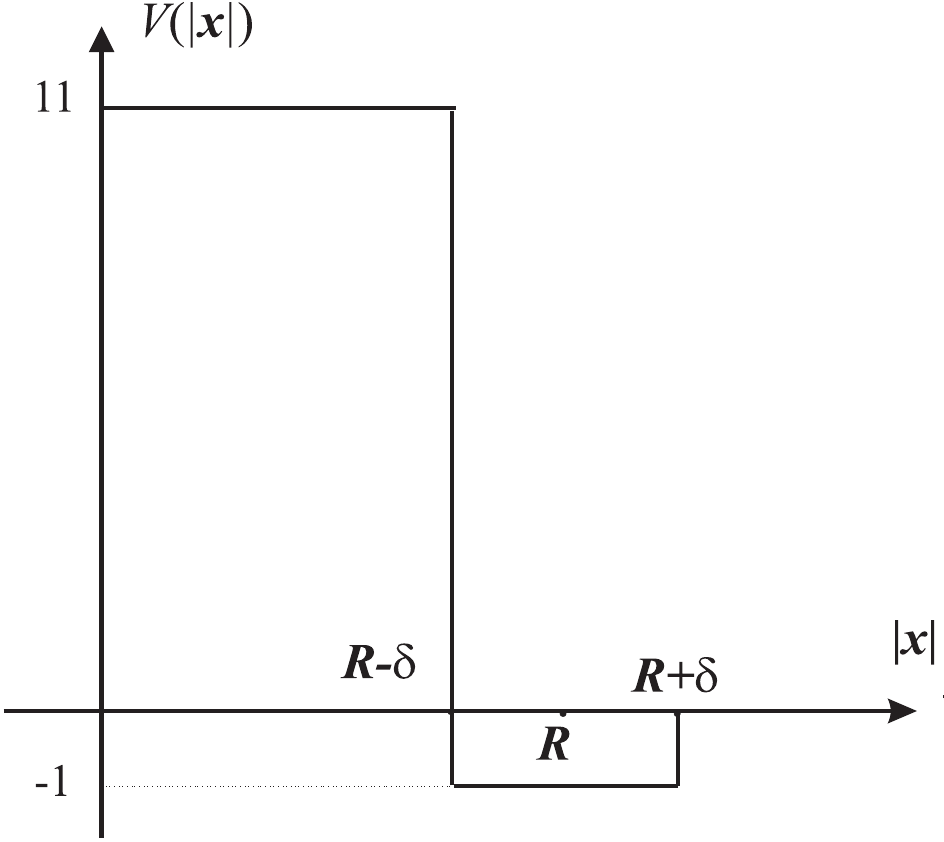}
\end{center}
\begin{center}
Figure 4. A catastrophic potential. The Ruelle potential
\end{center}
\end{figure}

\\Fix an integer $n$ and
let $\tilde\xx_1,\dots ,\tilde \xx_n$ be sites of a face-centered
cubic lattice in three dimensions with nearest neighborhoods at
distance $R$. We recall that a face-centered cubic lattice is a lattice
whose unit cells are cubes, with lattice points
at the center of each face of the cube, as well as at the vertices.

\\Let $B(\tilde\xx_1,\dots ,\tilde \xx_n)=\{\{i,j\}:
|\tilde\xx_i-\tilde\xx_j|=R\}$ the set of nearest neighborhood
bonds in $\tilde\xx_1,\dots ,\tilde \xx_n$ and let
$B_n~=~|B(\tilde\xx_1,\dots ,\tilde \xx_n)|$ the cardinality of this
set. Suppose that $\tilde\xx_1,\dots ,\tilde \xx_n$ are arranged
in such way to maximize $B_n$, hence in a ``close-packed"
configuration. We remind that in the face-centered cubic lattice
every site has $12$ nearest neighborhoods. If we take
$\tilde\xx_1,\dots ,\tilde \xx_n$ to be close-packed, then if $n$
is sufficiently large, the number of nearest neighborhoods bonds
are of the order $B_n\sim 6n$. In fact  each site is the vertex of
12 nearest neighborhoods bonds, and each nearest neighborhood bond is
shared between two sites. If the fixed integer  $n$ is chosen sufficiently large then
it is surely possible to find a close-packed configuration
$\tilde\xx_1,\dots ,\tilde \xx_n$ in such way that
\def\zz{{\bf z}}
$$
B_n>{11\over 2}n +\e \Eq(2.11)
$$
for some $\e>0$. Suppose thus that $n$ is chosen so large such that
the close-packed face centered configuration $\tilde\xx_1,\dots ,\tilde \xx_n$  is such that   \equ(2.11) is satisfied. Then,
$$
\sum_{i=1}^{n}\sum_{j=1}^{n}V_4^{\rm bad}(|\tilde\xx_i-\tilde\xx_j|)~=~n
V_4^{\rm bad}(0)+2B_n V_4^{\rm bad}(R)~=~11n-2B_n<-2\e<0
$$
Consider now the
 function
$$
\Phi: \mathbb{R}^{6n}\to \mathbb{R}: (\yy_1,\dots
,\yy_n,z_1,\dots, z_n)\mapsto
\sum_{i=1}^{n}\sum_{j=1}^{n}V_4^{\rm bad}(|\yy_i-z_j|)
$$
We have that $\Phi(\tilde\xx_1,\dots ,\tilde
\xx_n,\tilde\xx_1,\dots ,\tilde \xx_n)<-2\e$. Moreover, if
$$
|\yy_1 -\tilde\xx_1|<{\d\over 2},~~~\dots,~~~ |\yy_n
-\tilde\xx_n|<{\d\over 2},~~~ |z_1 -\tilde\xx_1|<{\d\over 2},
~~~\dots ~~~|z_n -\tilde\xx_n|<{\d\over 2}
$$
with $\d$ being the constant appearing in \equ(2.10), then also
$$
\Phi(\yy_1,\dots
,\yy_n,z_1,\dots, z_n)=\sum_{i=1}^{n}\sum_{j~=~1}^{n}V_4^{\rm bad}(|\yy_i-z_j|)<-2\e\Eq(2.12)
$$
Let now   $S_{\d}^i~=~\{\xx\in
\mathbb{R}^3:|\xx-\tilde\xx_i|<\d/2\}$ be the open sphere
in
$\mathbb{R}^3$ with radius $\d/2$ and center in $\tilde \xx_i$ and
define $\L_\d=\cup_{i=1}^{n}S_\d^i$ ($\L_\d$ is of course a subset
of $\mathbb{R}^3$).

\\Let $s$ be a positive integer and  define $M_s$ as the following subset of $\mathbb{R}^{3sn}$
$$
M_s~=~\{(\xx_1,\dots ,\xx_{sn})\in \mathbb{R}^{3sn}:
|\xx_{(i-1)s+p}-\tilde \xx_i|<\d/2,~~ p=1,...,s
~~i~=1,..., n\}
$$
Namely, $(\xx_1,\dots ,\xx_{sn})\in M_s$ means that  the $sn$-uple
$(\xx_1,\dots ,\xx_{sn})$ is such that the first $s$ variables
$\xx_1,\dots, \xx_s$ of the $sn$-uple  are all inside the sphere
$S_{\d}^1$, the variables $\xx_{s+1},\dots, \xx_{2s}$ are all
inside the sphere  $S_{\d}^2$, the variables $\xx_{2s+1},\dots,
\xx_{3s}$ are all inside the sphere  $S_{\d}^3$, and so on until
the last $s$ variables of the $sn$-uple, which are
$\xx_{(n-1)s+1},\dots, \xx_{sn}$, and are all inside the sphere
$S_{\d}^n$.

\\Now, if $(\xx_1,\dots ,\xx_{sn})\in M_s$, then,
$$
\begn
U(\xx_1,... ,\xx_{sn}) & =\sum_{1\le i<j\le sn}V_4^{\rm bad}(|\xx_i-\xx_j|)\\
&=
{1\over 2}\left[\sum_{i=1}^{sn}\sum_{j=1}^{sn}V_4^{\rm bad}(|\xx_i-\xx_j|)- s
n V_4^{\rm bad}(0)\right]\\
&= {1\over
2}\left[\sum_{p=1}^s\sum_{p'=1}^s\sum_{i=1}^{n}\sum_{j=1}^{n}
V_4^{\rm bad}(|\xx_{(i-1)s+p}-\xx_{(j-1)s+p'}|)- s n V_4^{\rm bad}(0)\right]
\egn
$$
Now, for fixed $p$ and $p'$ call $\xx_{(i-1)s+p}=\yy_i$ and
$\xx_{(j-1)s+p'}~=~z_j$. Then,  by definition of $M_s$ we have that
$|\yy_i-\tilde\xx_i|<\d/2$ and  $|z_j-\tilde\xx_j|<\d/2$. Hence
by \equ(2.12)
$$
\begn
\sum_{i=1}^{n}\sum_{j~=~1}^{n}
V_4^{\rm bad}(|\xx_{(i-1)s+p}-\xx_{(j-1)s+p'}|)& =\sum_{i=1}^{n}\sum_{j~=~1}^{n}
V_4^{\rm bad}(|\yy_{i}-z_{j}|)\\
& <-2\e
\egn
$$
hence we conclude that
$$
U(\xx_1,\dots ,\xx_{sn})< -\left(s^2\e+ {11\over 2} s
n\right)~~~~~~~~~~{\rm whenever}~~ (\xx_1,\dots ,\xx_{sn})\in M_s
$$
Therefore, if $V_\d$ denote the volume of the sphere of radius
$\d/2$ in $\mathbb{R}^3$ we have
$$
\Xi_\L(\b,\l)~=~1+ \sum_{N=1}^{\infty}{\l^N \over N!}
\int_{\L}d\xx_1\dots\int_{\L} d\xx_{N} ~e ^{-\b U(\xx_1,\dots
,\xx_N)}\ge
$$
$$
\ge 1+ \sum_{s=1}^{\infty}{\l^{sn} \over (sn)!}
\int_{\L}d\xx_1\dots  \int_{\L} d\xx_{sn} ~e ^{-\b U(\xx_1,\dots
,\xx_{sn})} \ge
$$
$$
\ge 1+ \sum_{s=1}^{\infty}{\l^{sn} \over (sn)!}
\int_{M_s}d\xx_1\dots  d\xx_{sn} ~e ^{-\b U(\xx_1,\dots
,\xx_{sn})}~\ge
$$
$$
\ge~\sum_{s=1}^{\infty}{\l^{sn} \over (sn)!} V_\d^{sn}
~e ^{{\b}(s^2\e+ {11\over 2} s n)  }~
=~ \sum_{s=1}^{\infty}{\left[\l V_\d e^{{11\b\over 2}}\right]^{sn}
\over (sn)!}\left(e^{\b\e}\right)^{s^2} ~=~\sum_{s=1}^\infty a_s
$$
This last  series is a series with positive terms whose term
$s^{\rm th}$ term is given by
$$
a_s~=~ {\left[\l V_\d e^{{11\b\over 2}}\right]^{sn} \over
(sn)!}\left(e^{\b\e}\right)^{s^2}
$$
In this formula recall that $n$ is a fixed value while $s$ is the variable
integer. It is now easy to show that $\sum_{s=1}^\infty a_s$
diverges. As a matter of fact, by the ratio test for positive term
series, we have that
$$
\lim_{s\to \infty} {a_{s+1}\over a_s}~=~\lim_{s\to
\infty}{(sn)!\over [(s+1)n)]!} \left[\l V_\d e^{{11\b\over
2}}\right]^n \left(e^{\b\e}\right)^{2s+1}\ge
$$
$$
\ge \lim_{s\to \infty} \left[{\l V_\d e^{{11\b\over 2}}\over
(s+1)n}\right]^n \left(e^{2\b\e}\right)^{s}~=~\left[{\l V_\d
e^{{11\b\over 2}}\over n}\right]^n \lim_{s\to\infty}
{\left(e^{2\b\e}\right)^{s}\over (s+1)^n}~=~+\infty
$$
~

\newpage

\section{Admissible potentials}
Let us start by discussing first three classes of potentials which are both tempered and
stable. \vskip.2cm
\\1 - {``Repulsive"  temperate potentials}

$$
V\ge 0 ~~~{\rm and} ~~~~\int _{|x|\ge r_0}V(x) dx <+\infty~~~~~\mbox{for some $r_0\ge 0$}
$$
\vskip.2cm
\\2 -   {Positive type potentials}:  absolutely integrable potentials which are the Fourier Transform of a positive function
$$
V(\xx)~=~\int\tilde V(\kk)e^{i\kk\cdot\xx}d^3\kk, ~~~~~~ \tilde
V(\kk)\ge 0
$$
\vskip.2cm
\\3 - {Lennard-Jones Potentials}: there exist strictly positive constants $a$, $C_1$, $C_2$ and $\e$ such that
$$
V(\xx)\ge {C_1\over|\xx|^{3+\e}}~~{\rm for}~~ x\le a, ~~~~
|V(\xx)|\le {C_2\over|\xx|^{3+\e}}~~{\rm for}~~x> a
$$
\vskip.2cm\index{pair potential!repulsive}\index{pair potential!positive type}
\index{pair potential!Lennard-Jones}

\\The potentials in the class 1, which are automatically stable are called ``repulsive" in a somewhat improper manner: a positive pair potential should be monotonically decreasing in order to be really purely repulsive.
\\We remark that tempered potentials with hard core ($V(\xx)~=~+\infty$ if $|\xx|\le a$) can be included
in case 3 or in case 1.
\begin{figure}
\begin{center}
\includegraphics[width=7cm,height=7cm]{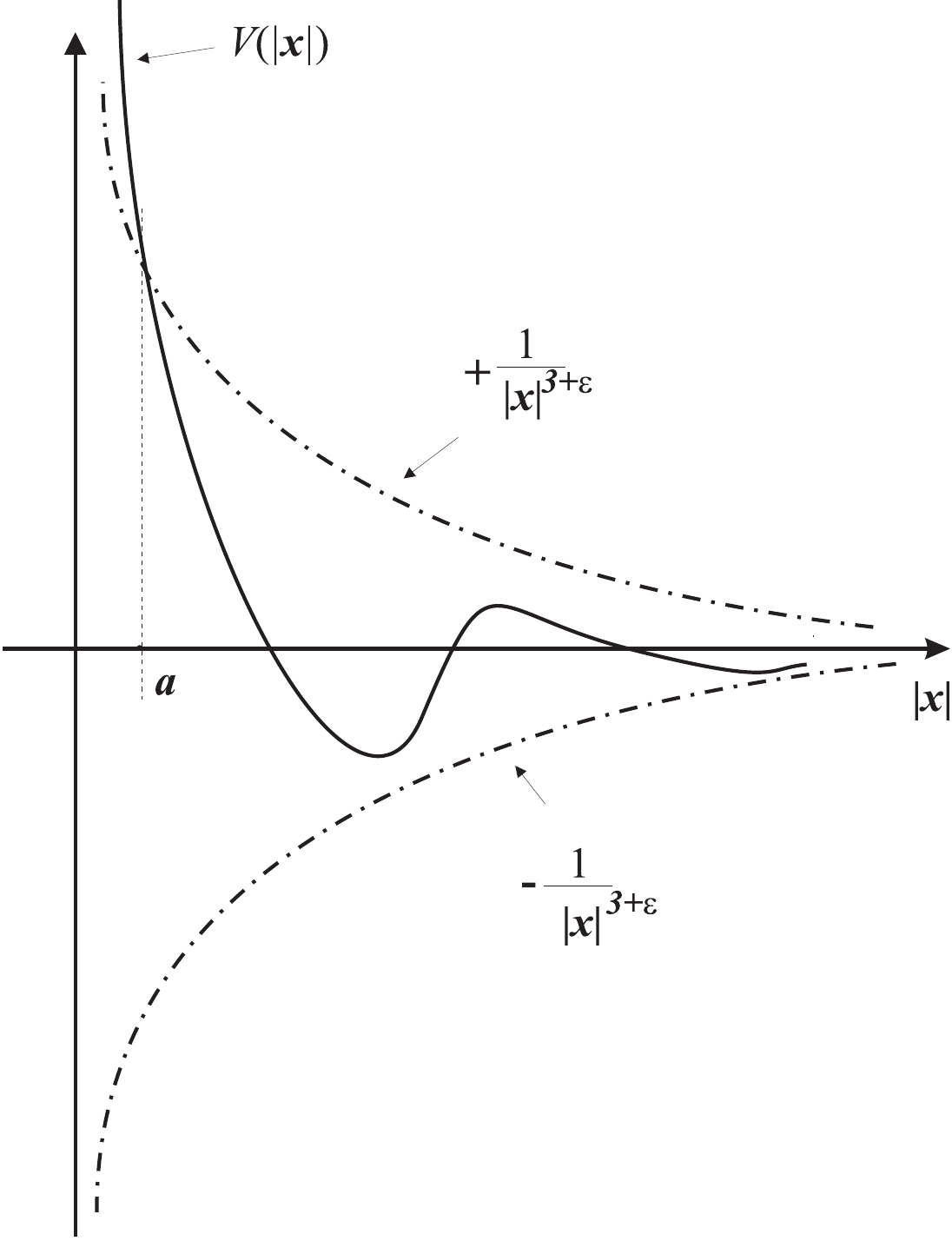}
\end{center}
\begin{center}
Figure 5. A Lennard-Jones potential
\end{center}
\end{figure}

\\It is immediate to see that a potential of type 1 is admissible. It is tempered
(i.e. it satisfies \equ(2.7)) by definition and it is stable,
since from $V(\xx)\ge 0$ we get
$$
U(\xx_1,\dots ,\xx_n)\ge 0~~~~~~\forall n,\xx_1,\dots,\xx_n
$$
hence \equ(2.6) is verified. Moreover, since $\inf_{n, \xx_1,\dots,\xx_n}  U(\xx_1,\dots ,\xx_n)=0$, we have that $B_V~=~0$.

\\Let us show now that a potential of type 2 is admissible. A positive type
potential $V$ is stable because
$$
\begn
U(\xx_1,\dots ,\xx_n)& =\sum_{1\le i<j\le n}V(\xx_i-\xx_j)\\
&=~ {1\over
2} \sum_{ i\neq j}V(\xx_i-\xx_j)\\
&={1\over 2} \sum_{ i,
j}V(\xx_i-\xx_j)-{n\over 2}V({0})\\
&= {1\over 2} \sum_{ i, j}\int e^{ {\rm i}\kk\cdot (\xx_i-\xx_j)} \tilde
V(\kk)d^3\kk -{n\over 2}V(0)\\
&=
~ {1\over 2} \int\left[ \sum_{ i,
j}e^{ {\rm i}\kk\cdot (\xx_i-\xx_j)}\right] \tilde V(\kk)d^3\kk -{n\over
2}V(0)\\
& =~{1\over 2} \int \left|\sum_{ i=1}^ne^{ {\rm i}\kk\cdot \xx_i}\right|^2
\tilde V(\kk)d^3\kk -{n\over 2}V(0)\ge -{n\over 2}V(0)
\egn
$$
where the last inequality follows from the assumption $\tilde
V(\kk)\ge 0$.

\\Hence the stability condition
\equ(2.6) is satisfied by choosing $B~=~{1\over 2}V(0)$. Potential 2
is tempered since it is supposed to admit Fourier transform, hence
it needs to be absolutely integrable. \vskip.5cm

\\The example 3, the Lennard Jones type potential, is of major interest in  applications, since it  is the most popular and used by  physicists and chemists
to model the interactions between molecules in real gases.

\\Let us  show that Lennard-Jones type potential is admissible.
Such a potential is indeed regular by definition, thus we have
just to show that it is also stable. In a given configuration
$\xx_1,\dots, \xx_n$, denote with the number $r_{\rm
min}~=~\min_{i\neq j}|\xx_i-\xx_j|$ the minimum distance between
particles.

\\We distinguish two cases: 1) $r_{\rm min}< {a/2}$; and  2) $r_{\rm min}\ge{a/2}$.
\vskip.2cm
\\Case 1)  $r_{\rm min}< a/2$.
Suppose without loss of generality that $|\xx_1 -\xx_2|=r_{\rm
min}$ and all other distances are at distances greater or equal
than $r_{\rm min}$. We have
$$
\begn
U(\xx_1,\dots ,\xx_n) & =V(\xx_1-\xx_2)+ \sum_{j=3}^n V(\xx_1
-\xx_j)+ U(\xx_2, \dots ,\xx_n)\\
&\ge {C_1\over r_{\rm min}^{3+\e}} + \sum_{j\in \{3,\dots, n\}\atop
|\xx_1-\xx_j|\ge a} V(\xx_1 -\xx_j)+ U(\xx_2, \dots ,\xx_n)
\egn
$$
Note that the second sum of the last inequality above is over those
particles with index $j\in \{2,3,\dots,n\}$ at distance greater or
equal to $a$. The inequality follows from the fact that  $V(\xx_1
-\xx_j)\ge 0$ if $|\xx_1 -\xx_j|<a$.

\\We will now get a lower bound for the term
$ \sum_{j\in \{3,\dots, n\}:~ |\xx_1-\xx_j|\ge a} V(\xx_1
-\xx_j)$. First note that, since in this sum $|\xx_1-\xx_j|\ge a$
for all $j$, we can bound
$$
\sum_{j\in \{3,\dots, n\}:\atop |\xx_1-\xx_j|\ge a} V(\xx_1
-\xx_j)\ge- \sum_{j\in \{3,\dots, n\}:\atop |\xx_1-\xx_j|\ge a}
{C_2\over |\xx_1 -\xx_j|^{3+\e}}
$$
We then proceed  as follows. Draw around each $\xx_j$ a cube $Q_j$
with side $r_{\rm min}/ {\sqrt 12}$ (in such way that its maximal
diagonal is $r_{\rm min}/2$) with $\xx_j$ being the vertex
farthest away from $\xx_1$. Since any two points among $\xx_3
,\dots, \xx_n$ are at mutual distances $\ge r_{\rm min}$ the cubes
so constructed do not overlap. Moreover, if we consider the open
sphere $S_1 ~=~\{\xx\in \mathbb{R}^3: |\xx-\xx_1|< {a\over 2}\}$ and
denote by $S_1^c$ its complementar in $ \mathbb{R}^3$,  then by
construction (recall that $r_{\rm min}<a/2$ and $|x_1-x_j|\ge a$) all cubes $Q_j\subset
S_1^c$, i.e., they lay outside the open sphere with center $\xx_i$
and radius $a/2$. Furthermore, we have
$$
{1\over |\xx_1 -\xx_j|^{3+\e}}\le {(\sqrt{12})^3\over {r_{\rm
min}^3}}\int _{Q_j} {1\over |\xx_1 -\xx|^{3+\e}} d\xx
$$
recall in fact that the cube $Q_j$ is chosen in such way that
$|\xx -\xx_1|\le |\xx_1-\xx_j|$ for all $\xx\in Q_j$.
\\Therefore

$$
\begn
\sum_{j\in \{3,\dots, n\}:\atop |\xx_1-\xx_j|\ge a}^n V(\xx_1
-\xx_j)
& \ge- C_2{(\sqrt{12})^3\over {r_{\rm min}^3}}\sum_{j\in
\{3,\dots, n\}:\atop |\xx_1-\xx_j|\ge a}^n \int _{Q_j} {1\over
|\xx_1 -\xx|^{3+\e}} d^3\xx\\
& =
- C_2{(\sqrt{12})^3\over {r_{\rm min}^3}}\int _{\cup_j Q_j}
{1\over |\xx_1 -\xx|^{3+\e}} d^3\xx\\
& \ge
 - C_2{(\sqrt{12})^3\over
{r_{\rm min}^3}}\int _{|\xx-\xx_1|\ge a/2} {1\over |\xx-\xx_1|^{3+\e}}
d^3\xx\\
&= -C_2{(\sqrt{12})^3\over
{r_{\rm min}^3}}\int _{|\yy|\ge a/2} {1\over |\yy|^{3+\e}}
d^3\yy\\
& = {C_2K_a(\e)\over r_{\rm min}^3}
\egn
$$
where
$$
K_a(\e)~=~ {(\sqrt{12})^3}{4\pi2^\e\over \e a^\e}
$$
So
$$
U(\xx_1,\dots ,\xx_n)~ \ge~ {C_1\over r_{\rm min}^{3+\e}}
-{C_2K_a(\e)\over r_{\rm min}^3}+ U(\xx_2, \dots ,\xx_n)
$$
I.e.
$$
U(\xx_1,\dots ,\xx_n)\ge  -C_\e+ U(\xx_2, \dots ,\xx_n)
$$
where  $-C_\e$  is the minimum of the
function $F(x) ={C_1\over x^{3+\e}}- {C_2K_a(\e)\over x^{3}}$ for $x>0$.
Therefore, iterating this formula we
get
$$
U(\xx_1,\dots ,\xx_n)\ge  -C_\e n
$$

\vskip.2cm
\\Case 2)  $r_{\rm min}\ge a/2$. In this case we write
$$
U(\xx_1, \dots, \xx_n)~=~{1\over 2}\sum_{i=1}^n \sum_{j\in
\{1,2,\dots ,n\}:\,\,j\neq i}V(\xx_i -\xx_j)
$$
and we will get an estimate $\sum_{j\in \{1,2,\dots ,n\},\,j\neq
i}V(\xx_i -\xx_j)$. First we bound
$$
\sum_{j\in \{1,2,\dots ,n\}\atop j\neq i}V(\xx_i -\xx_j)\ge-
\sum_{j\in \{1,2,\dots ,n\}\atop |x_i-x_j|\ge a} {C_2\over |\xx_i
-\xx_j|^{3+\e}}
$$
We then proceed analogously as before.  This time, for fixed $i$,
we can draw around each $\xx_j$ a cube $Q_j$ with side $a/ {\sqrt
48}$ (such that the maximal diagonal of $Q_j$ is $a/4$) with
$\xx_j$ being the vertex farthest away from $\xx_i$. Since any two
points among $\xx_1 ,\dots, \xx_n$ are at mutual distances $\ge
a/2$ the cubes so constructed do not overlap. Moreover if we
consider the open  sphere $S_i ~=~\{\xx\in \mathbb{R}^3:
|\xx-\xx_i|< {a\over 2}\}$ and denote by $S_i^c$ its complementar
in $ \mathbb{R}^3$,  then by construction $\cup_{j\neq i}
Q_j\subset S_i^c$. In other words all cubes $Q_j$ lay outside the
open sphere with center $\xx_i$ and radius $a/2$. Furthermore we
have
$$
{1\over |\xx_i -\xx_j|^{3+\e}}\le {(\sqrt{48})^3\over a^3}\int
_{Q_j} {1\over |\xx_i -\xx|^{3+\e}} d^3\xx
$$
recall in fact that the cube $Q_j$ is chosen in such way that
$|\xx -\xx_i|\le |\xx_i-\xx_j|$ for all $\xx\in Q_j$.

\\Therefore

$$
\begn
\sum_{j\in \{1,2,\dots ,n\}\atop |x_1-x_j|\ge a}V(\xx_i -\xx_j)& \ge~-
C_2{(\sqrt{48})^3\over a^3}\sum_{j\in \{1,2,\dots ,n\}\atop x_1-x_j|\ge a}\int _{Q_j} {1\over |\xx_i -\xx|^{3+\e}} d^3\xx\\
& =- C_2{(\sqrt{48})^3\over a^3}\int _{\cup_{j\neq i} Q_j} {1\over |\xx_i
-\xx|^{3+\e}} d^3\xx\\
& \ge - C_2{(\sqrt{48})^3\over a^3}\int
_{|\xx-x_i|\ge {a\over 2}} {1\over |\xx_i -\xx|^{3+\e}} d^3\xx\\
& \ge~ -{\bar K_a}(\e)
\egn$$
where
$$
\bar K_a(\e)~=~ C_2{(\sqrt{48})^3\over a^3}\int _{|\xx-x_i|\ge a/2} {1\over
|\xx_i -\xx|^{3+\e}} d^3\xx
$$
Hence
$$
\sum_{j\in \{1,2,\dots ,n\}\atop j\neq i}V(\xx_i -\xx_j) ~\ge~- {\bar K_a}(\e)\Eq(0)
$$
and
$$
U(\xx_1,\dots ,\xx_n)~\ge~ -{\bar K_a(\e)\over 2}n
$$
\subsection{Basuev Criteria}
We now present two very efficient criteria for stability of tempered potentials first proposed  by Basuev \cite{Ba1}. To do this we need to itroduce the following notation.
Given a potential $V(x)$ let
$$
V(a)= \inf_{|x|=a}V(x)\Eq(vudia)
$$
Given $n\in \mathbb{N}$,   let us consider the set of configurations defined by
$$
\mathcal{S}_n(a)=\{(x_1,\dots, x_n)\in \mathbb{R}^{dn}: ~|x_i-x_j|>a,~~~{\rm for}~ 1\le i<j\le n\}
$$
I.e. $\SS_n(a)$ is formed by all configurations of $n$ particles in $\mathbb{R}^{dn}$ such that any two particles are at mutual
distance greater than $a$.

\begin{defi}\label{defbas}
A potential $V(x)$  is called Basuev if there exists $a>0$ such that $V(a)>0$,
$V(x)\ge V(a)$ for all $|x|\le a$ and
$$
V(a)\ge \mu(a)\Eq(cba)
$$
where
$$
\mu(a)= \sup_{n\in \mathbb{N},\atop{(x_1,\dots, x_n)\in \mathcal{S}_n(a)}}\sum _{i=1}^n V^-(x_i)\Eq(mua)
$$
with
$$
V^-(x)=\max \{ - V(x), 0\}.
$$
V(x) is called strongly Basuev if
$$
V(a)\ge 2\mu(a)\Eq(strobasu)
$$
\end{defi}

\\We now prove the following theorem.
\begin{teo}\label{basta}
Let $V(x)$ be a Basuev pair  potential with Basuev constants $a$, $V(a)$ and $\mu(a)$ as in definition \ref{defbas}.
Then  $V(x)$ is stable with stability constant $B_V$ such that
$$
B_V\le \mu(a)
$$
\end{teo}

\\{\bf Remark}. Since $V(|x|)$ is translational invariant, observe that for any $x_0\in \mathbb{R}^d$ it holds that

$$
\sup_{n\in \mathbb{N},\atop{(x_1,\dots, x_n)\in \mathbb{R}^{dn}\atop |x_i-x_j|>a}}\sum _{i=1}^n V^-(x_0-x_i)=\m(a) \Eq(mua0)
$$

\begin{proof} Let $n\in \mathbb{N}$ be fixed and let us first consider the case in which we have  a configuration of $n$ particles in positions
$x_1,\dots,x_n$ such that all $n$ particles are mutually at distance greater than $a$, i.e let us consider first the set of configurations
in $\SS_n(a)$.
We have that, for any $n\in \mathbb{N}$
$$
\begin{aligned}
\inf_{(x_1,\dots, x_n)\in \SS_n(a)}{1\over n} U(x_1,\dots, x_n) & =
\inf_{(x_1,\dots, x_n)\in \SS_n(a)}{1\over n}
\sum_{1\le k<s\le n}V(x_k-x_s)\\
& =
\inf_{(x_1,\dots, x_n)\in \SS_n(a)}{1\over 2n}\sum_{k=1}^n\sum_{s\neq k}V(x_k-x_s)\\
& \ge
-{1\over 2n}\sum_{k=1}^n\sup_{(x_1,\dots, x_n)\in \SS_n(a)}\sum_{s\neq k}V^-(x_k-x_s)\\
&\ge
-{1\over 2n}\sum_{k=1}^n \sup_{s\in \mathbb{N}\atop{(x_1,\dots, x_s)\in \SS_n(a)}} \sum_{j=1}^s V^-(x_j)= -{\mu(a)\over 2}
\end{aligned}
$$
where in the last inequality we have used \equ(mua0) and \equ(mua).
We can therefore limit ourselves to configurations in which there are particles at distance $a$ or smaller than $a$.
Consider thus  a configuration $(x_1,\dots,x_n)$ such that there exists $\{i,j\}\subset [n]$ such that
$|x_i-x_j|\le a$. Thus there is a particle, which, without loss of generality,  we can assume to be the particle indexed by 1 at position $x_1=0$ (i.e. at the origin),
which has the maximum number of particles among $x_2,\dots,x_n$ at distance less than or equal $a$.
We have to estimate, for any $n\in \mathbb{N}$
$$
U(x_1,\dots,x_n)=E_1+ U(x_2,\dots,x_n)
$$
where
$$
E_1= \sum_{j=2}^nV(x_j)
$$
Say that the number of these particles close to $x_1$ less or equal to $a$ is $l$ (clearly $l\ge 1$ by assumption). The energy $E_1$ of the particle at position $x_1$
is thus
$$
E_1\ge l V(a) -\sum_{k\in [n]\atop|x_k|>a} V^-(x_k)\Eq(E1)
$$

\\To bound  the sum
$$
\sum_{k\in [n]\atop|x_k-x_1|>a} V^-(x_k))
$$
 observe that we are supposing that each particle has at most $l$ other particles at distance less or equal than $a$.
Thus take the  $k\in [n]$ such that $V^-(x_k)$ is maximum. Again, without loss of generality we can suppose $k=2$. In the sphere with center $x_2$ and radius $a$ there are at most
$l+1$ particles (the particle at position $x_2$, for which the value $V^-(|x_2|)$ is maximum plus at most $l$ other particles,  say $x_{i_1},\dots , x_{i_l}$, for which $V^-(|x_{i_s}|)\le V^-(x_2) $ with $s\in[l]$).
Hence
$$
\sum_{k\in [n]\atop |x_k-x_1|>a} V^-(|x_k|)\le  (l+1)V^-(x_2)+\sum_{k\in [n]\atop |x_k-x_1|>a, |x_k-x_2|>a} V^-(x_k)
$$
Now we have to control the sum
$$
\sum_{k\in [n]\atop |x_k-x_1|>a, |x_k-x_2|>a} V^-(|x_k|)
$$
 Note that in this sum all particles are at distance greater than $a$ for the particle in $x_1$
and also from particle at position $x_2$ moreover each of the particle in the sum
 has at most $l$ particles at distance less or equal than $a$. Suppose without loss of generality
that the particle at position $x_3$ is such that
 $$
 V^-(x_3)=\max_{k\in [n]\atop |x_k-x_1|>a, |x_k-x_2|>a} V^-(x_k)
 $$
and this particle at $x_3$ has at most $l$ particles at distance less or equal than $a$. Therefore
$$
\sum_{k\in [n]\atop |x_k-x_1|>a, |x_k-x_2|>a} V^-(x_k)\le (l+1) V^-(x_3)+\sum_{k\in [n]\atop |x_k-x_1|>a, |x_k-x_2|>a, \,|x_k-x_3|>a} V^-(x_k)
$$
Iterating  we get
$$
\sum_{k\in [n]\atop |x_k-x_1|>a} V^-(x_k)\le (l+1)\sum_{k\in [n]\atop  |x_i-x_j|>a} V^-(x_k)\Eq(lpu2)
$$
where now in the sum in the r.h.s of \equ(lpu) all pairs of particles are at  distance greater than $a$ to each other. Therefore, recalling definition \equ(mua)
we have that
$$
\sum_{k\in [n]\atop  |x_i-x_j|>a} V^-(x_k)\le \m(a)
$$
so that
$$
\sum_{k\in [n]\atop |x_k-x_1|>a} V^-(x_k)\le (l+1)\m(a)
$$

\\Therefore, recalling \equ(E1), we have that
$$
E_1\ge l V(a) -(l+1)\m(a)= l(V(a)-\m(a))-\m(a)\ge -\m(a)
$$
So we have obtained
$$
U(x_1,\dots,x_n)\ge  -\m(a)+ U(x_2,\dots,x_n)
$$
Iterating we get
$$
U(x_1,\dots,x_n)\ge  -n\mu(a)
$$
$\Box$
\end{proof}
\vv
\vv
\\{\bf Example}. Let us  consider the potential
$$
V(\xx)~=~
\begin{cases} A & {\rm if}~ |\xx|\le R\\
-1 &{\rm if}~R<|\xx|\le R+\d\\
0& {\rm otherwise}
\end{cases}
\Eq(2.10)
$$
Let us prove that $V(x)$ is stable if $A\ge 12$ and $\d$ sufficiently small.

\\It is known that the maximum number of points that can be put on the surface of the  sphere $S(R)$ of radius
$R$  in such way that the distance between any pair of such points
is greater or equal $R$ is 12.   Let us now show that if $\d$ is sufficiently small 12 is also the maximum number of points that can be fitted in the layer $S_\d(R)=\{x\in \mathbb{R}^3: R\le |x|\le R+\d\}$.
Indeed, suppose by absurd that  for any $\d>0$ one can find 13 points on $S_\d$ such that the distance between any pair of these points is greater or equal than $R$.  Then by choosing $\d={1\over n}$, for $n\in \mathbb{N}$
one can
construct 13 sequences $x_n^1,\dots, x_n^{13}$  all contained in $S_\d$ such that, for any $n$  and for any $i,j\in\{1,2,\dots, 13\}$ it holds that $|x_n^i-x_n^j|\ge R$ since $S_\d$ is compact
and $\lim_{n\to \infty}S_\d(R)=S(R)$, for each of the sequences
$x^i_n$ there exists a subsequence convergent to some point $y^i\in S(R)$ such that for any pair  $|y^i-y^j|\ge R$. So we have found 13 points of the sphere $S(R)$
of radius $R$ all at distance from one to another greater that $R$, contradicting the statement that the maximal number of such points is 12.
This means that the maximum number of points we can put on  the layer $S_\d$ in such way that the distance between any pair of such points
is greater or equal $R$ is 12 and therefore the potential of Figure 4  with  11 is replaced by any number $A\ge 12$ and $\d$ is sufficiently small satisfies the hypothesis of  Theorem \ref{basta} and hence is stable.

\begin{teo}\label{basu1}
Let $V(x)$ be a strongly Basuev  and tempered pair potential, with costants $a$, $V(a)$ and $\mu(a)$ as in definition \ref{defbas}.
Then  $V(x)$ is stable with stability constant $B_V\le {1\over 2} \mu(a)$. Moreover the representation
$$
V(x)= V_a(x) + K_a(x)
$$
with
$$
V_a(x)=
\begin{cases}
V(a) & {\rm if} ~ |x|\le a\\
V(x)  & {\rm if}~ |x|>a
\end{cases}
\Eq(va)
$$
and
$$
K_a(x)=
\begin{cases}
V(x)-V(a) & {\text if} ~|x|\le a\\
0 &{\text if}~ |x|>a
\end{cases}
\Eq(ka)
$$
is such that
the potential
$V_a(x)$ defined in \equ(va) is also stable and it has the same stability constant $B_V$ of the full potential  $V(x)$ and
the potential $K_a(x)$ is positive
and supported in $[0,a]$.
\end{teo}

\begin{proof} The thesis is trivial if $V^-(x)=0$ (i.e. if $V(x)$ is purely repulsive). So we may assume that  $V^-(x)\neq0$.
For any $(x_1,\dots,x_n)\in \mathbb{R}^{nd}$ and any $i\in [n]$, let
$$
E_i(x_1,\dots,x_n)=\sum_{j\in [n]:\,j\neq i}V(x_i-x_j)
$$
so that $U(x_1,\dots,x_n)={1\over 2}\sum_{i\in [n]}E_i(x_1,\dots,x_n)$.
Note that when $V^-\neq 0$ then surely $B>0$ and thus configurations $x_1,\dots,x_n$ which minimize the energy must be searched
between those such that $U(x_1,\dots,x_n)<0$.
Let now
$(x_1,\dots,x_n)\in \mathbb{R}^{nd}$ be a configuration such that $U(x_1,\dots,x_n)<0$
but there is a particle in position say  $x_1$ (without loss of generality) such that $E_1(x_1,\dots,x_n)\ge 0$. Then
$$
U(x_1,\dots,x_n)= E_1(x_1,\dots,x_n)+U(x_2,\dots,x_n)
$$
and, since we are assuming $E_1(x_1,\dots,x_n)\ge 0$ and $U(x_1,\dots,x_n)<0$, we have that $U(x_2,\dots, x_n)<0$ and
$$
U(x_1,\dots,x_n)\ge  U(x_2,\dots,x_n)
$$
i.e.
$$
-U(x_1,\dots,x_n)\le - U(x_2,\dots,x_n)
$$
with both sides of the last inequality positive.
Therefore, since ${1\over n}<{1\over n-1}$,
$$
-{1\over n}U(x_1,\dots,x_n)<  -{1\over n-1}U(x_2,\dots,x_n)
$$
Thus the configuration $(x_2,\dots,x_n)$  produces a value $-{1\over n-1}U(x_2,\dots,x_n)$ which is nearer to $B_V$ than
$-U(x_1,\dots,x_n)/n$. Whence we can look for minimal energy configurations $(x_1,\dots, x_n)$ limiting ourselves to those configurations in which
the energy per particle $E_i(x_1,\dots,x_n)$ is negative for all $i\in [n]$.

\\Now let us consider the system of
particles interacting via the pair potential $V_a(x)$ defined in \equ(va) and let us
assume that conditions \equ(cba) and  \equ(strobasu) hold. Note first that, recalling  definition \equ(va),
$V_a^-(|x|)=\max\{0, - V_a(x)\}=V^-(x)$.
Consider then  a configuration $(x_1,\dots,x_n)$ such that there exists $\{i,j\}\subset [n]$ such that
$|x_i-x_j|\le a$. Thus there is at least a particle, (which, without loss of generality,  we can assume to be the particle indexed by 1 at position $x_1$),
which has the maximum number of particles among $x_2,\dots,x_n$ at distance less than or equal than $a$.
Say that the number of these particles close to $x_1$ less or equal to $a$ is $l$ (clearly $l\ge 1$ by assumption). The energy $E_1$ of the particle at position $x_1$
can thus be estimated as follows.
$$
E_1(x_1,\dots,x_n)\ge l V(a) -\sum_{k\in [n]\atop|x_k-x_1|>a} V^-(x_k)
$$
To control the sum
$$
\sum_{k\in [n]\atop|x_k-x_1|>a} V^-(x_k))
$$
 observe that we are supposing that each particle has at most $l$ other particles at distance less or equal than $a$.
Thus take the  $k\in [n]$ such that $V^-(x_k)$ is maximum. Again, without loss of generality we can suppose $k=2$. In the sphere with center $x_2$ and radius $a$ there are at most
$l+1$ particles (the particle at position $x_2$, for which the value $V^-(|x_2|)$ is maximum plus at most $l$ other particles).
Hence
$$
\sum_{k\in [n]\atop |x_k-x_1|>a} V^-(x_k|)\le  (l+1)V^-(x_2)+\sum_{k\in [n]\atop |x_k-x_1|>a, |x_k-x_2|>a} V^-(x_k)
$$
Now we have to control the sum
$$
\sum_{k\in [n]\atop |x_k-x_1|>a, |x_k-x_2|>a} V^-(x_k)
$$
 Note that in this sum all particles are at distance greater than $a$ for the particle in $x_1$
and also from the particle at position $x_2$. Moreover,s each of the particle in the sum
 has at most $l$ particles at distance less or equal than $a$. Suppose without loss of generality
that the particle at position $x_3$ is such that
 $$
 V^-(x_3)=\max_{k\in [n]\atop |x_k-x_1|>a, |x_k-x_2|>a} V^-(x_k)
 $$
and this particle at $x_3$ has at most $l$ particles at distance less or equal than $a$. Therefore
$$
\sum_{k\in [n]\atop |x_k-x_1|>a, |x_k-x_2|>a} V^-(x_k)\le (l+1) V^-(x_3)+\sum_{k\in [n]\atop |x_k-x_1|>a, |x_k-x_2|>a, \,|x_k-x_3|>a} V^-(x_k)
$$
Iterating  we get
$$
\sum_{k\in [n]\atop |x_k-x_1|>a} V^-(x_k)\le (l+1)\sum_{k\in [n]\atop  |x_i-x_j|>a} V^-(x_k)\Eq(lpu)
$$
where now in the sum in the r.h.s of \equ(lpu) all pairs of particles are at  distance greater than $a$ to each other. Therefore, recalling definition \equ(mua)
we have that
$$
\sum_{k\in [n]\atop  |x_i-x_j|>a} V^-(x_k)\le \m(a)
$$
and hence
$$
E_1(x_1,\dots,x_n)\ge l V(a) -(l+1)\m(a)
$$
I.e. we have that $E_1>0$ whenever
$$
V(a)>{l+1\over l} \mu(a)
$$
Using now assumption \equ(strobasu) and since ${l+1\over l}<2$ we get
$$
E_1(x_1,\dots,x_n)>0
$$
The conclusion is  that if a configuration $(x_1,\dots,x_n)$ is such that some particles are at distance less or equal than $a$
then there is at least a  particle  whose energy is positive. Hence the minimal energy configurations for $V_a(x)$ must be searched among
those configurations in which all particles are at distance greater than $a$ from each other. But for these configurations $V_a(|x|)= V(|x|)$ which
implies that $V_a$ and $V$, if stable, have the same stability constant $B$ (and hence also the same $\bar B$). It is now easy to see that
 $V_a(|x|)$ is  stable.
Indeed observe that, for any configuration $(x_1,\dots,x_n)$ for which particles are at distance greater than $a$ from each other we have
$$
U_a(x_1,\dots,x_n)={1\over 2}  \sum_{i=1}^n\sum_{j\in [n]\atop j\neq i} V(x_i-x_j)\ge-{1\over 2}n\m(a)
$$
which implies that $V_a(x)$ is stable with stability constant $B\le { \m(a)\over 2}$. $\Box$
\end{proof}
\vv
\\We have shown above that a Lennard-Jones type potential $V(|x|)$ is
stable. Here below we prove that a Lennard-Jones
type potential is also Basuev.
\begin{teo}\label{pro2}
Let $V(x)$ be a pair  potential on  $\mathbb{R}^d$ such that
there exist constants $w,r_1, r_2>0$, with $r_1\le r_2$,  and  non-negative monotonic decreasing functions
$\xi(|x|)$, $\eta(|x|)$ such that
$$
V(x)
\begin{cases}
\ge \xi(|x|) & {\rm if}~ |x|\le r_1\\
\ge - w & {\rm if} ~r_1< |x|< r_2\\
\ge -\eta(|x|) & {\rm if}~ |x|\ge r_2
\end{cases}
\Eq(condLJg)
$$
with
$$
\lim_{x\to 0} \x (x) x^d =+\infty \Eq(ult)
$$
and
$$
\int_{|x|\ge r_2} \eta (|x|) dx <+\infty\Eq(long)
$$
Then $V(x)$ is strongly Basuev (i.e. it satisfies Theorem \ref{basu1}).
\end{teo}

\begin{proof} By Theorem \ref{basu1} we just need to show that there exists $a$ such that \equ(cba) and \equ(strobasu) are satisfied. Fix $a\in (0,r_1)$,
let $\bar w=\max\{w, \h(r_2)\}$
$$
\bar \h(|x|)=
\begin{cases}
\h(|x|) & {\text if}~ |x|> r_2 \\
\bar w &{\text if}~|x|\le r_2
\end{cases}
$$
Then, by construction  $\bar \h(|x|)$ is monotonic decreasing  and such that $\int_{\mathbb{R}^{d}}\bar \h(|x|)\le \infty$.
Moreover by conditions \equ(condLJg) we have that
$$
V^-(x)\le \bar \h(|x|)
$$
Hence, recalling \equ(mua)  and considering also that,  since we took $a\in (0,r_1)$, by hypothesis $V^-(|x|)=0$ for all $|x|\le a$, we have
$$
\m(a)\le\sup_{n\in \mathbb{N},~(x_1,\dots, x_n)\in \mathbb{R}^{dn}\atop |x_i-x_j|>a,~|x_i|>a}\sum _{i=1}^n \bar \h(|x_i|)
$$
To  estimate from above $\sum _{i=1}^n \bar \h(|x_i|)$, having in mind that  all particles are at mutual
distances greater than  $a$ and are at
distance greater than  $a$ from the origin, we proceed as follows.
We  draw for each $x_j$ a hypercube $Q_j$ with side $a/ 2\sqrt{d}$
(such that the maximal diagonal of $Q_j$ is $a/2$) in such a way that
$x_j$ is  a  vertex of the cube $Q_j$  and at the same time is the point $x\in Q_j$ which is the farthest away from the origin $0$. Since any two
points among $x_1 ,\dots, x_n$ are at mutual distances $\ge
a$ the cubes so constructed do not overlap.  Furthermore, using the fact that $\bar\h(|x|)$ is monotonic decreasing,  we
have
$$
\bar\h(|x_j|)\le {(4d)^{d\over 2}\over a^d}\int_{Q_j} \bar\h(|x|)dx
$$
recall in fact that the cube $Q_j$ is chosen in such way that
$|x|\le |x_j|$ for all $x\in Q_j$.
Therefore
$$
\begn
\sum _{i=1}^n \bar \h(|x_i|)& \le
{(4d)^{d\over 2}\over a^d}\sum_{i=1}^n\int_{Q_i} \bar \h(|x|)dx\\
& ={(4d)^{d\over 2}\over a^d}\int _{\cup_i Q_i} \bar\h(|x|)dx\\
& \le
{{(4d)^{d\over 2}\over a^d}}\int
_{\mathbb{R}^d}\bar\h(|x|)dx\\
& ={C_d\over a^d}
\egn
$$
where
$C_d=~ {(4d)^{d\over 2}}\int_{\mathbb{R}^d}\bar\h(|x|)dx$.
Hence
we get
$$
\mu(a)\le {C_d\over a^d}
$$
Now, in view of  condition \equ(ult), we can always choose $a$ such that $\x(a)a^d > 2C_d$. Thus we get
$$
\x(a)a^d > 2C_d ~~~\Longrightarrow~~ \x(a)> 2{C_d\over a^d}~~~\Longrightarrow~~~V(a)> 2{C_d\over a^d} ~~~\Longrightarrow~~~ V(a)> 2\mu(a)
$$
$\Box$
\end{proof}
\vv
\\Exercise: prove that a Lennard-Jones pair potential satisfies Theorem \ref{pro2} and hence Theorem \ref{basu1} (thus a Lennard-Jones potential can be written as the sum
of a positive potential plus an absolutely integrable and stable potential with stability constant $B_V$ equal to the one of the full potential).

\section{The infinite volume limit}
We will now start to consider the problem of the existence of the thermodynamic limit
\index{thermodynamic limit}
for the pressure of a system of particles in the grand canonical
ensemble interacting via a pair potential stable and tempered. The
mathematical problem is to show the existence of the infinite
volume pressure defined as
$$
\b p(\b,\l)~=~\lim_{\L\to \infty} {1\over |\L|}\ln \Xi_\L(\b,\l)\Eq(2.13)
$$
where we recall that $|\L|$ denotes the volume of $\L$.

\\First of all we need to give a mathematical meaning to the notation
$\L\to \infty$ in \equ(2.13). We know that $\L$ is a finite region
of $ \mathbb{R}^3$ and it can tend to infinity (namely its volume
tends to infinity) in various ways. For instance, $\L$ could be a
cylinder of fixed base $A$ and height $h$ and we could let $h\to
\infty$. I.e. like a cigar with increasing length. It is obvious
that such a system (particularly if $A$ is very small) is not
expected to have a thermodynamic behavior even if $h\to\infty$.

\\Thus $\L\to \infty$ cannot simply be $|\L|\to \infty$, since we want to exclude
cases like ``the cigar". We need thus that $\L\to \infty$
roughly in such way that $\L$ is big in every direction, e.g. a
sphere of increasing radius, a cube of increasing size etc. We
will review the following definitions
\vskip.3cm
\begin{defi}\label{vanhove}
$\L$ is said to go to infinity in the sense of {\it Van Hove}
if the following occurs.

\\Cover $\L$ with small cubes of size $a$ and let $N_+(\L, a)$ the number of cubes with non void
intersection with $\L$ and  $N_-(\L, a)$ the number of cubes
strictly included in $\L$. Then $\L\to\infty$ in the sense of Van
Hove if

\vskip.5cm
$$N_-(\L, a)\to \infty,~~~~~{\rm and}~~~~
{N_-(\L, a)\over N_+(\L, a)}\to 1~~~~~~~~ \forall a$$
\end{defi}
\\See figure 7.

\vskip.3cm
\\As an example, it is easy to show that if $\L\to \infty$ as in the case of the ``cigar'' seen before, then
$\L$ is not tending to infinity in the sense of Van Hove. Consider
thus, for sake of simplicity in the plane, a rectangle $R$ with
sizes $l$ fixed and $t$ variable and going to infinity. Let
consider squares of size $a$ to cover the rectangle. Then
$$
N_-(R,a)~=~{l\over a}{t\over a},~~~~~~~~N_+(R,a)~=~{l\over a}{t\over
a} + 2 {t\over a}+ 2{l\over a}
$$
thus
$$
\lim_{t\to \infty}{N_-(R,a)\over N_+(R,a)}~=~{l\over l+2a}\neq 1
$$

\\On the other hand, let us check that a rectangle of sizes $(\sqrt{t}, t)$ and
tends to infinity, as $t\to\infty$, in the sense of Van Hove. As a
matter of fact
$$
N_-(R,a)~=~{\sqrt{t}\over a}{t\over
a},~~~~~~~~N_+(R,a)~=~{\sqrt{t}\over a}{t\over a} + 2 {t\over a}+
2{\sqrt{t}\over a}
$$
thus
$$
\lim_{t\to \infty}{N_-(R,a)\over
N_+(R,a)}~=~\lim_{t\to\infty}{t\sqrt{t}\over t\sqrt{t}+2at +
2a\sqrt{t}}~=~ 1
$$
\vv\vv
\\We now give a second and more stringent definition of $\L\to \infty$.
For a given $\L$  let $\L_h$ denotes  the
set of points at distance less or equal to $h$ from the boundary of $\L$ and let $|\L_h|$ denotes its volume. Let finally
$d(\L)$ denote the diameter of $\L$ (i.e. $d(\L)~=~\sup_{x,y\in
\L}\{|x-y|\}$) .
\begin{figure}
\begin{center}
\includegraphics[width=7cm,height=1cm]{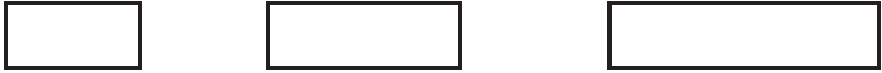}
\end{center}
\begin{center}
Figure 6. A box $\L$ going to infinity in a ``non-thermodynamic"
way
\end{center}
\end{figure}
\vskip.3cm
\begin{defi}\label{fisher}
We say that $\L$ tends to infinity in the sense of Fischer if
$|\L|\to \infty$ and it exists a positive function $\p(\a)$ such
that $\lim_{\a\to 0}\p(\a)~=~0$ and for $\a$ sufficiently small and  for
all $\L$
$$
{|\L_{\a d(\L)}|\over |\L|}\le \p(\a)
$$
\end{defi}
\begin{figure}
\begin{center}
\includegraphics[width=7cm,height=5cm]{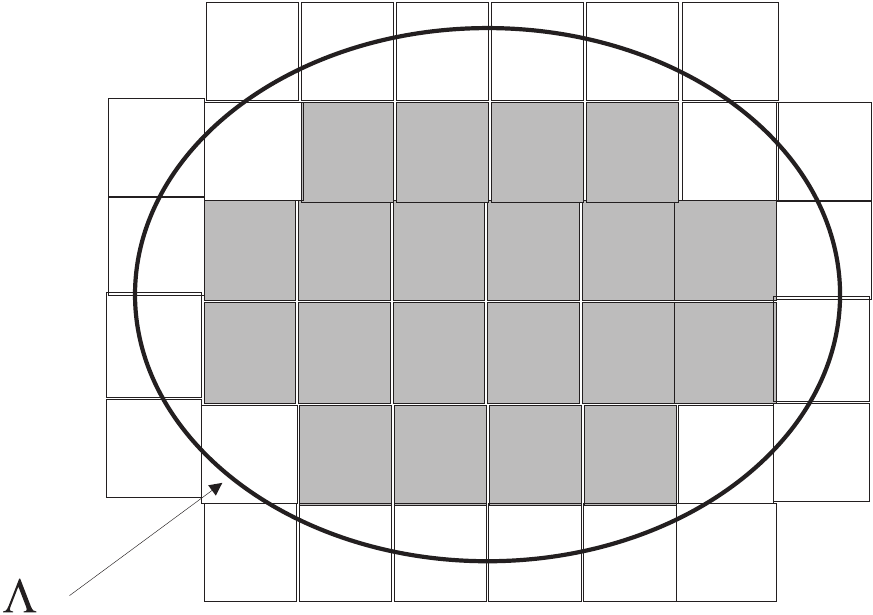}
\end{center}
\begin{center}
Figure 7. A set $\L$ with $N_+ ~=~44 $ and $N_- ~=~20$
\end{center}
\end{figure}

\\A rectangle $R$ of sizes $f(t),t$
such that $\lim_{t\to\infty}{f(t)\over t}~=~0$ does not go to
infinity in the sense of Fischer. As a matter of fact
$$
d(R)~=~[t^2 + f^2(t)]^{1/2}
$$
$$
|\L_{\a d(R)}|~=~ 2\a [t^2 + f^2(t)]^{1/2}[t-2\a [t^2 +
f^2(t)]^{1/2}+f(t)]
$$
$$
{|\L_{\a d(R)}|\over |R|}~=~{ 2\a [t^2 + f^2(t)]^{1/2}[t-2\a [t^2 +
f^2(t)]^{1/2}+f(t)]\over tf(t)}
$$
$$
~=~{ 2\a t^2\sqrt{1 + {f^2(t)\over t^2}}[1-2\a \sqrt{1 +
{f^2(t)\over t^2}}+{f(t)\over t}]\over tf(t)}
$$
For any fixed $\a$ the quantity above can be make  large at will,
as $t\to\infty$. Thus it is not possible to find any $\p(\a)$ such
that ${V_{\a d(R)}(R)/ V(R)}\le \p(\a)$ for all $R$ (i.e. for all
$t$).

\\On the other hand, the square $S$ of size $(t,t)$ goes to infinity in the sense of Fischer
$$
d(S)~=~\sqrt{2}t,~~~~~~~~~~~~~~~~~~~~ |\L_{\a d(S)}|~=~ 4\a
\sqrt{2}t[t-\a \sqrt{2}t]
$$
$$
{|\L_{\a d(S)}|\over |S|}~=~{4\a \sqrt{2}\,t[t-\a \sqrt{2}t]\over
t^2}~=~4\a \sqrt{2}[1-\a \sqrt{2}]\le 4\a \sqrt{2}
$$
Hence one can choose $\p(\a)~=~ 4\a \sqrt{2}$.

\def\rb{{\bar r}}
\section{Example: finite range potentials}
We will prove in this section the existence of the thermodynamic limit for the function
$\b p(\b,\l)$ defined in \equ(2.13), but  we will not treat the
general case, namely particles interacting via a pair potential
stable and tempered and $\L$ going to infinity as e.g. Van Hove.
We will rather give our proof in a simpler case. Namely, we will put ourselves in $d=3$ dimensions  and we will
assume that particles interact through a  stable and tempered pair
potential $V(\xx)$ with the further property that it exists $\rb
>0$ such that $V(\xx)\le 0$ if $|\xx|\ge \rb$, i.e. the potential
is negative for distances greater than $\rb$.

\\In order to make things even simpler we will also suppose that $\L$ is a cube of size
$L$ and $\L\to \infty~ \Leftrightarrow~L\to\infty$. We now choose
two particular sequences $\L_1,\L_2,\dots,\L_n,\dots$ and
$\tilde\L_1,\tilde\L_2,\dots,\tilde\L_n,\dots$ of cubes. Let
$\L_1$ be a cube of size $L_1$ with volume $V_1$ while
$\tilde\L_1$ is a new cube of size  $\tilde L_1~=~L_1+\rb$ which
consists of $\L_1$ plus a frame of thickness ${\rb\over 2}$. We
denote $\tilde V_1$ it volume (of course $\tilde V_1> V_1$) $\L_2$
is a cube of size $L_2~=~ 2 L_1+\rb$,  thus in $\L_2$ we can arrange
$2^3~=~8$ cubes $\L_1$ with  frames of thickness $\rb/2$ in such
way that any point in a given cube $\L_1$ inside $\L_2$ is at
distance greater than $\rb$ from any point in any other cube
$\L_1$ inside $\L_2$. Of course $\tilde \L_2$ is a cube of size
$2\tilde L_1$, i.e. is the cube $\L_2$ plus a frame of thickness
$\rb/2$. See Figure 8.

\\In general $\L_{n+1}$ is a cube of size $L_{n+1}= 2 L_{n}+\rb$
and $\tilde \L_{n+1}$ is a cube of size $\tilde L_{n+1}=2\tilde
L_n ~=~2 L_{n}+2\rb$.
 Note that $|\L_{n+1}|>8 |\L_n|$ and $|\tilde\L_{n+1}|=8|\tilde \L_n|$ and $\lim_{n\to\infty}{|\tilde \L_n|/ |\L_n|}~=~1$.

\\We will now  show that the sequence of functions
$$
P_n(\b,\l)~=~  {1\over |\L_n |}\ln \Xi_{\L_n}
(\b,\l)\Eq(2.14)
$$
tends to a limit.
\begin{figure}
\begin{center}
\includegraphics[width=7cm,height=7cm]{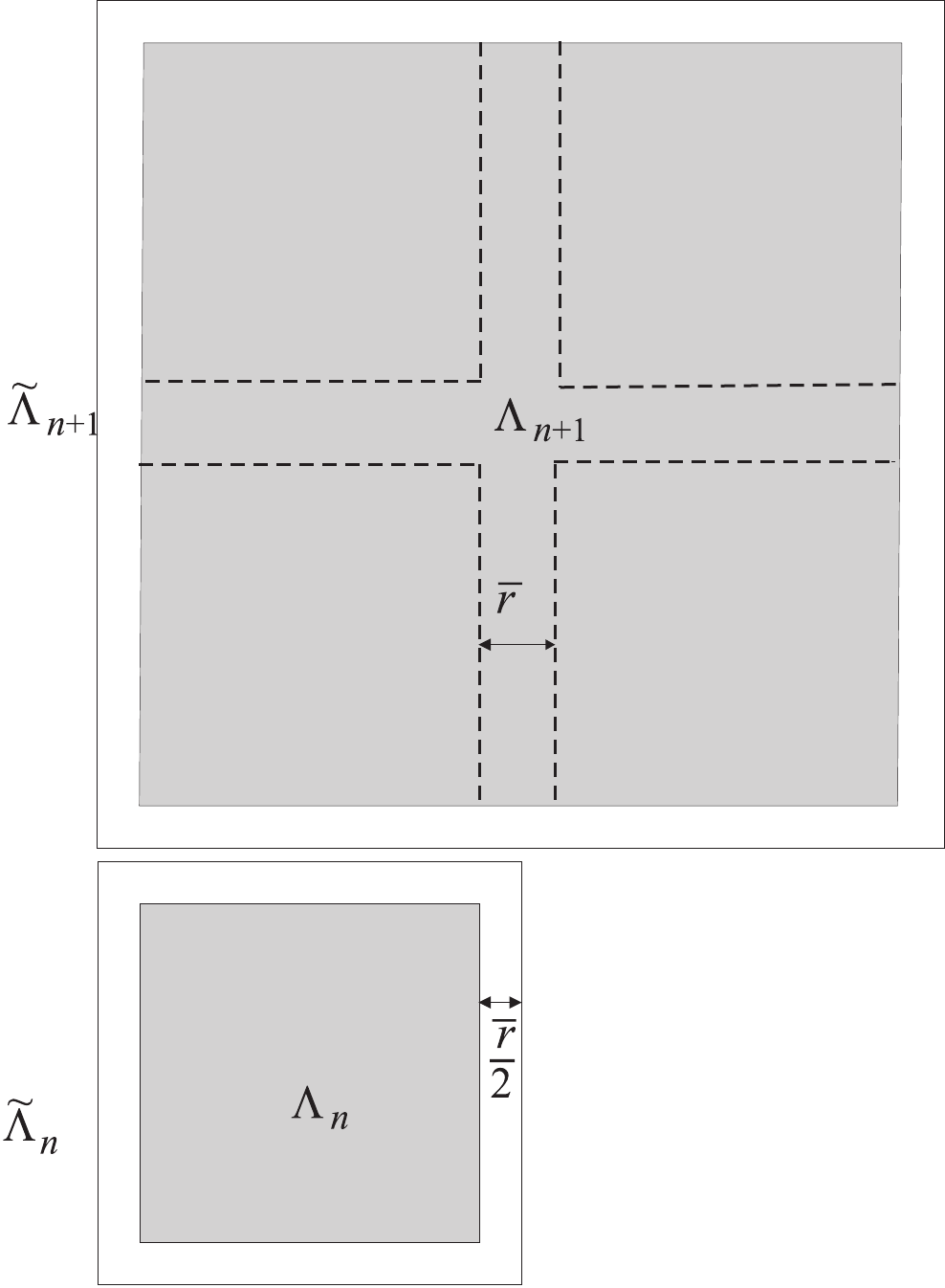}
\end{center}
\begin{center}
Figure 8. The cubes $\L_{n+1}$, $\tilde\L_{n+1}$ and the cubes
$\L_n$, $\tilde\L_n$
\end{center}
\end{figure}
Define the sequence
$$
\tilde P_n~=~  {1\over |\tilde \L_n|}\ln \Xi_{\L_n}(\b,\l)\Eq(2.15)
$$
We will first show that this sequence converges to a limit.
Consider the sequence $\Xi_n ~=~ \Xi_{\L_n}(\b,\l)$. We have
$$
\Xi_{n+1}~=~\sum_{N=0}^{\infty} {\l^N\over N!}\int_{\L_{n+1}^N}
d\xx_1\dots d\xx_N e^{ -\b U(\xx_1,\dots ,\xx_N)}
$$
We now think $\L_{n+1}$ as  the union of $8$ cubes $\L^j_n$
($j=1,2,\dots ,8$) plus the (internal) frames. Obviously, if we calculate
$\Xi_{n+1}$ eliminating the configurations in which particles can
stay in the frames we are underestimating $\Xi_{n+1}$. I.e.
$$
\Xi_{n+1}>\sum_{N=0}^{\infty} {\l^N\over N!}\int_{\cup_j\L^j_{n}}
d\xx_1\dots d\xx_N e^{ -\b U(\xx_1,\dots ,\xx_N)}
$$
Thus the $N$ particles are arranged in such way that $N_1$ are in
the box $\L_n^1$, $N_2$ are in the box $\L_n^2$..., and $N_8$
are in the box $\L_n^8$. Thus we can rename coordinates $\xx_1,
\dots, \xx_N$ as $\xx_1^{1},\dots \xx_{N_1}^{1},\dots,
\xx_1^{8},\dots \xx_{N_8}^{8}$. The interaction between particles
in different boxes is surely non positive (since they are at
distances greater or equal to $\rb$) then
$$
U(\xx_1,\dots ,\xx_N)~=~ U(\xx_1^{1},\dots \xx_{N_1}^{1},\dots,
\xx_1^{8},\dots \xx_{N_8}^{8})\le
$$
$$~~~~~~~~~~~~~~~~~~~~~~~~
\le U(\xx_1^{1},\dots \xx_{N_1}^{1})~+~\dots~ +~ U(\xx_1^{8}\dots
\xx_{N_8}^{8})
$$
and hence
$$
e^{-\b U(\xx_1,\dots ,\xx_N)}\ge e^{-\b \left[U(\xx_1^{1},\dots
\xx_{N_1}^{1})~+~\dots~ +~ U(\xx_1^{8}\dots \xx_{N_8}^{8})\right]}
$$

\\The number of ways in which such arrangement can occur is
$$
{N\choose N_1}{N-N_1\choose N_2}\dots {N-N_1 -N_2 - \dots -
N_7\choose N_8}~=~ {N!\over N_1!\dots N_8!}
$$
Hence
$$
\begn
\Xi_{n+1} & >\sum_{N_1,\dots ,N_8}{\l^{N_1+\dots +N_8}\over N!}
{N!\over N_1!\dots N_8!} \int_{\L^1_{n}}d\xx^{1}_1\dots
\int_{\L^1_{n}}d\xx^{1}_{N_1}\cdots\\
& \cdots
\int_{\L^8_{n}}d\xx^{8}_1\dots \int_{\L^8_{n}}d\xx^{8}_{N_8}
 e^{ -\b U(\xx^{1}_1,\dots ,\xx^{1}_{N_1})}
e^{ -\b U(\xx^{2}_1,\dots ,\xx^{2}_{N_2})} \dots e^{ -\b
U(\xx^{8}_1,\dots ,\xx^{8}_{N_8})}\\
&=(\Xi_{n})^8
\egn
$$
So we have shown that
$$
\Xi_{n+1}~> ~(\Xi_n)^8
$$
Hence, since $f(x)=\ln x$ is a monotonic increasing function, we
also get
$$
\ln\Xi_{n+1}> 8\ln (\Xi_n)
$$
and
$$
{1\over \tilde V(\L_{n+1})}\ln\Xi_{n+1}> {8\over \tilde
V(\L_{n+1})}\ln (\Xi_n)~=~ {1\over \tilde V(\L_{n})}\ln\Xi_{n}
$$
So the sequence $\tilde P_n$ is monotonic increasing with $n$. On
the other hand by stability we have that.
$$
\Xi_{n}~=~\sum_{N~=~0}^{\infty} {\l^N\over N!}\int_{\L_{n}^N}
d^3\xx_1\dots d^3\xx_N e^{ -\b U(\xx_1,\dots ,\xx_N)}~\le
$$
$$
~~~~~~~~\le~
\sum_{N~=~0}^{\infty} {\l^N\over N!}\int_{\L_{n}^N} d^3\xx_1\dots
d^3\xx_N e^{ +\b B N} ~=~ \exp\{\l V_n e^{ \b B}\}
$$
therefore,
$$
\ln\Xi_{n}~\le~ {\l V_n e^{ \b B}}
$$
and
$$
\tilde P_n ~=~{1\over \tilde V_n} \ln\Xi_{n}\le {1\over \tilde V_n}
{\l V_n e^{ \b B}}\le {\l  e^{ \b B}}
$$
This means that the sequence $\tilde P_n$ is monotonic increasing
and bounded above. Hence $\lim_{n\to \infty}\tilde P_n~=~P$ exists.
But now
$$
P_n ~=~ {1\over V_n}\ln\Xi_n ~=~ {\tilde V_n\over V_n} \tilde P_n
$$
whence
$$
\lim_{n\to \i} P_n ~=~ \lim_{n\to \i} {\tilde V_n\over V_n} \tilde
P_n~=~
 \lim_{n\to \i}{\tilde V_n\over V_n} \lim_{n\to \i} \tilde P_n ~=~P
$$
Therefore we show the existence of the thermodynamic limit for a
class of systems interacting via a potential which, beyond being
tempered and stable, has the further property to be non positive
for $|\xx|\ge \rb$, when $\L$ goes to infinity along the sequence
of cubes $\L_n$. It is not difficult from here to show that the
existence of such limit implies also that the limit exists if
$\L$ is a cube which goes at infinity in sense homothetic (i.e.
the size $\L\to \infty$).

\\Actually, the existence of limits can be proved for much more general cases, see e.g.
Theorem 3.3.12  in \cite{Ru}

\section{Properties the pressure}

\\Let us now show some general properties of the limit for the pressure  $\b p(\b,\l)$
in \equ(2.13). The pressure $p(\b,\l)$ is a function of two
variables $\l$ and $\b$  which are the two independent
thermodynamic parameters describing the macroscopic equilibrium
state in the Grand canonical ensemble.\index{pressure}

\\We are  interested to study the function $\b p(\b,\l)$ only
for   the ``physically" admissible  values of $\l$ and $\b$. These
physical values  are:  $\l$ real number in the interval $(0,+\infty)$ and  $\b$  real number in the interval $(0,+\infty)$ (recall definition \equ(acty)). The Grand canonical
partition function $ \Xi (\b,\L,\l)$ defined in \equ(2.2) where we
are supposing of course that $U(\xx_1,\dots, \xx_N)$ is derived
from a stable and tempered pair potential has the following
structure

$$
\Xi_\L(\b,\l)~=~1+Z_1(\L, \b)\l+ {Z_2(\L,\b)}\l^2+
{Z_3(\L,\b)}\l^3 + \dots ~=$$
$$
=~\sum_{n=0}^{\infty} {Z_n(\L,\b)}\l^n
~~~~~~~~~~~~~~~~~~~~~~~$$I.e.
is a power series in $\l$ with convergence radius $R=\infty$ (this
is true for any $\L$ such that $V(\L)<\infty$), i.e.,
$\Xi_\L(\b,\l)$ is analytic as a function of $\l$ in the whole
complex plane. Hence {\it a fortiori} $\Xi_\L(\b,\l)$ is analytic
for all $\l\in (0,+\infty)$. This is true for all $\L$ such that
$V(\L)<\infty$.

\\The coefficients $Z_n(\L,\b)$ are explicitly given by
$$
Z_n(\L,\b)~=~{1\over n!}\int_{\L}d\xx_1\dots\int_{\L} d\xx_{n}~e ^{-\b
U(\xx_1,\dots ,\xx_n)}\Eq(bu)
$$
They are clearly all positive numbers and due to stability (recall
proposition 1) they admit the upper bound $Z_n(\L,\b)\le
[V(\L)]^n\,e^{nB\b}/n!$. Moreover as functions of $\b$ the
coefficients $Z_n(\L,\b)$ are analytic in $\b$ in the whole
complex plane and hence {\it a fortiori} for all $\b\in
(0,\infty)$.

\\So,  $\Xi_\L(\b,\l)$, for all $\L$ such that $|\L|<\infty$, and for all
$\l\in (0,+\infty)$ is also analytic as a function of $\b$ in the
whole complex plane. Hence {\it a fortiori} $\Xi_\L(\b,\l)$ is
analytic for all $\b\in (0,+\infty)$.

\\Now the function $\ln \Xi_\L(\b,\l)$ has no reason to continue
analytic in $\l$ and $\b$ in the whole complex plane, but it is
indeed  analytic in $\l$  for any $\l\in (0,+\infty)$ and it is
analytic in $\b$ for all $\b\in (0,+\infty)$. This is due to the
fact that coefficients $Z_n(\L,\b)$ are positive numbers. Hence $
\Xi_\L(\b,\l)$ has no zeroes in the intervals $\l\in (0,+\infty)$
and  $\b\in (0,+\infty)$. This means that its logarithm is
analytic for such intervals.

\\In conclusion, we can state that the function
$$
\b p_{\L}(\b,\l)~=~{1\over |\L|}\ln \Xi_\L(\b,\l)~\doteq~ f_{\L}(\b,\l)\Eq(2.17)
$$
is analytic in $\l$ for all $\l\in (0,+\infty)$  and  it is also
analytic in $\b$ for all $\b\in (0,+\infty)$ for any finite
box $\L$.

\\This fact of course does not imply that also in the limit $\L\to \infty$ the function
$\b p(\b,\l)$ will stay analytic in the whole physical intervals
$\l\in (0,+\infty)$  and $\b\in (0,+\infty)$.

\\Let us now list some properties of the function $f_{\L}(\b,\l)$ defined by \equ(2.17).
\vskip.7cm

\\{\bf Property 0a}.  {\it $ f_\L(\b,\l)$ defined in \equ(2.17) is continuous  as a function of $\l$
and all its  derivatives are continuous as functions of $\l$ in the
whole interval $\l\in (0,+\i)$ and for all $\b\in (0,+\io)$ and
for all $\L$ such that $V(\L)<\i$}.

\\{\bf Property 0b}.  {\it $ f_\L(\b,\l)$ defined in \equ(2.17) is continuous  as a function of $\b$
and all its  derivatives are continuous as functions of $\b$ in the
whole interval $\b\in (0,+\i)$ and for all $\l\in (0,+\io)$ and
for all $\L$ such that $V(\L)<\i$}.

\\These properties follow trivially from the fact that $f_{\L}(\b,\l)$ is analytic in $\l$
and $\b$ when they vary in the interval $(0,+\i)$.

\\{\bf Property 1}. {\it $ f_\L(\b,\l)$ defined in \equ(2.17)
is monotonic increasing  as a function of $\l$ in the interval
$\l\in (0,+\i)$, for all $\b\in (0,+\io)$ and for all $\L$ such
that $V(\L)<\i$}

\vskip.2cm
\\In order to show the property 1 it is sufficient
to show that ${\dpr f_\L(\b,\l)\over\dpr \l}\ge 0$. But
$$
{\dpr f_\L(\b,\l)\over\dpr \l}~=~{1\over |\L|}
{{\dpr\Xi_\L(\b,\l)/\dpr \l}\over \Xi_\L(\b,\l)}~=~{\r_\L(\b,\l)\over
\l}\Eq(2.18)
$$
where  $\r_\L(\b,\l)={\<N\>\over |\L|}$ is the density and
$\<N\>$ is the mean number of particles in the grand
canonical  ensemble at fixed values of $\l$, $\b$ and $\L$.
Explicitly $\<N\>$ is given by
$$
\<N\>~=~{\sum_{N=0}^{\infty}  {\l^N}N
Z_N(\L,\b)\over \sum_{N=0}^{\infty} {\l^N} Z_N(\L,\b)}
$$
Hence, since $\<N\>$ is surely a positive number for  $\l > 0$, we
get
$$
{\dpr f_\L(\b,\l)\over\dpr \l}>0
$$

\\{\bf Property 2}. {\it $f_\L(\b,\l)$ defined in \equ(2.17) is convex   as a function of
$\ln\l$
\footnote{Reminding \equ(acty) observe that $\ln\l$ is proportional to the chemical potential.}
 in the interval $\l\in (0,+\i)$, for all $\b\in (0,+\io)$
and for all $\L$ such that $V(\L)<\i$. Moreover the finite volume
density $\r_{\L}(\b,\l)={\dpr \over \dpr (\ln \l)}f_\L(\b,\l)$
is a monotonic increasing function of $\ln \l$.}

\\As a matter of fact
$$
{\dpr \over \dpr (\ln \l)}f_\L(\b,\l)~=~\l {\dpr \over \dpr
\l}f_\L(\b,\l) ~=~\r_\L(\b,\l)
$$
Last line follows by \equ(2.18). Moreover
$$
{\dpr^2 \over\dpr (\ln \l)^2}f_\L(\b,\l) ~=~\l {\dpr \over \dpr
\l}\r_\L(\b,\l)~=~ {\l\over |\L|}{\dpr \over \dpr \l}\<N\>~=
$$
$$
=~
{\l\over |\L|}{\dpr \over \dpr \l}\left[{\sum_{N=0}^{\infty}
{\l^N}N Z_N(\L,\b)\over \sum_{N~=~0}^{\infty} {\l^N}
Z_N(\L,\b)}\right]
~=~{1\over |\L|}\left(   \<N^2\>- \<N\>^2\right)~=
$$
$$
=~
{\<(N-\<N\>)^2\>\over |\L|}~~~~~~~~~~~~ ~\Eq(2.19)
$$
where
$$
 \<N^2\>~=~{\sum_{N=0}^{\infty}  {\l^N}N^2 Z_N(\L,\b)\over
\sum_{N=0}^{\infty} {\l^N} Z_N(\L,\b)}
$$
thus, since the factor $\<(N-\<N\>)^2\>$ is always positive we get
$$
{\dpr^2 \over\dpr (\ln \l)^2}f_\L(\b,\l)~=~ {\dpr \over\dpr (\ln
\l)}\r_\L(\b,\l)> 0
$$
This prove that $f_\L(\b,\l)$ is convex in the variable $\ln \l$
and that $\r_\L(\b,\l)$ is monotonic increasing in $\ln \l$.
\vskip.7cm

\\{\bf Property 3}. {\it $f_\L(\b,\l)$ defined in \equ(2.17) is convex   as a function of
$\b$ in the interval $\b\in (0,+\i)$, for all $\l\in (0,+\i)$ and
for all $\L$ such that $V(\L)<\i$}.

\\As a matter of fact
$$
{\dpr \over \dpr \b}f_\L(\b,\l)~=~ {1 \over |\L|}{\dpr \over \dpr
\b}\ln \Xi(\L,\b,\l) ~=~ {-\<U \>\over |\L|}
$$
where
$$
\<U \>~=~ {1\over \Xi(\L,\b,\l)}\sum_{N~=~0}^{\infty}{\l^N\over N!}
\int_{\L}d\xx_1\dots\int_{\L}d\xx_NU(\xx_1,\dots ,\xx_N)e^{-\b
U(\xx_1,\dots ,\xx_N)}
$$
Deriving one more time respect to $\b$
$$
{\dpr^2 \over \dpr \b^2}f_\L(\b,\l)~=~ {\<U^2 \>-\<U\>^2 \over
|\L|}
$$
where
$$
\<U^2 \>~=~ {1\over \Xi(\L,\b,\l)}\sum_{N~=~0}^{\infty}{\l^N\over N!}
\int_{\L}d\xx_1\dots\int_{\L}d\xx_N [U(\xx_1,\dots
,\xx_N)]^2\,e^{-\b U(\xx_1,\dots ,\xx_N)}
$$
and since $\<U^2 \>-\<U\>^2~=~ \<(U-\<U\>)^2\>\ge 0$ we obtain
$$
{\dpr^2 \over \dpr \b^2}f_\L(\b,\l)\ge 0
$$
and $f_\L(\b,\l)$ is a convex function of $\b$ for all $\L$ such
that $V(\L)$ is finite.

\vskip.7cm
\\{\bf Property 4}. {\it $\b p(\b,\l)~=~\lim_{\L\to \infty}\b p_{\L}(\b,\l)$ is convex
as a function of $\b$ and $\ln\l$  in the interval $\b,\l\in
(0,+\i)$}.

\\This very important property of the pressure in the thermodynamic
limit follows trivially from the fact that  the limit of a pointwise
converging sequence of convex functions is also a convex  function.

\vskip.7cm
\\{\bf Property 5}. {\it It is possible to express, for any $\L$ finite, the pressure
$p_\L$ as a function of $\r_\L$ and  $\b$, i.e.
$$
p_\L~=~g_\L(\r_\L, \b)
$$
moreover the function $g_\L(\r_\L, \b)$ is monotonic increasing as
a function of $\r_\L$.}

\\The finite volume density is
$$
\r_\L ~=~ \r_\L (\b,\l)~=~  \r_\L (e^{\ln\l}, \b)~=~ F_\L(\ln \l ,\b)
$$
Since, by Property 2 the function $F_\L(x ,\b)$ is strictly
increasing as a function of $x$ for any $\L$ finite and any $\b\in
(0,+\i)$, then it admits, as function of $x$, an inverse, say $x ~=~
G_\L(\r_\L,\b)$. Hence
$$
\ln \l ~=~ G_\L(\r_\L,\b), ~~~~~\l ~=~ e^{G_\L(\r_\L,\b)}
$$
Thus the function $g_\L(\r_{\L},\b)$ can be indeed constructed and is
given explicitly by
$$
p_\L ~=~{1\over \b}f_\L( e^{G_\L(\r_\L,\b)},\b)~=~ g_\L(\r_\L,\b)
$$
It is now easy to check that this function is monotonic
increasing.
As a matter of fact
$$
\begn
{\dpr\over \dpr\r_\L}p_\L & ={1\over \b} {\dpr\over \dpr\r_\L}(\b
p_\L)\\
& ={1\over \b} {\dpr(\ln \l)\over \dpr \r_{\L}}{\dpr \over \dpr
(\ln \l)} (\b p_\L)\\
&={1\over \b} \left[{\dpr\r_{\L}\over
\dpr (\ln \l)}\right]^{-1}~\l{\dpr \over \dpr
\l} f_\L(\b, \l)\\
& ={1\over \b} \left[{\dpr\r_{\L}\over
\dpr (\ln \l)}\right]^{-1}{\r_{\L}}\\
&={1\over \b}
{\r_{\L}}\left[\l{\dpr\r_{\L}\over \dpr  \l}\right]^{-1}
\egn
$$
recalling now that $\r_\L~=~\<N\>/V(\L)$ and \equ(2.19) we obtain
$$
{\dpr\over \dpr\r_\L}p_\L ~=~{1\over \b}
{\r_{\L}}\left[\l{\dpr\r_{\L}\over \dpr  \l}\right]^{-1}~=~{1\over
\b} {\<N\>\over \<N^2\>-\<N\>^2}\Eq(2.20)
$$
Formula \equ(2.20) shows that ${(\dpr/ \dpr\r_\L)}p_\L$ is
always positive. Actually \equ(2.20) tells us also that the value
of  $(\dpr/ \dpr\r_\L)p_\L$ is ${\<N\>/\b(\<N^2\>-\<N\>^2)}$.
Thus if we are able to prove that ${\<N\>(\<N^2\>-\<N\>^2)^{-1}}$
stay bounded away from $+\i$ for any $\L$ we can conclude that
$p(\r,\b)=\lim_{\L\to \i}p_\L(\r_\L,\b)$  is continuous a function
of the density $\r~=\lim_{\L\to\i}\r_\L$.

\vskip.7cm

\\{\bf Property 6}.
{\it $p=g(\r, \b)$ is  monotonic increasing as a function of $\r$
(hence monotonic decreasing as a function of $\r^{-1}$)}
\vskip.2cm
\\The monotonicity follows trivially  from the fact that $p_\L(\r_\L,\b)$ is monotonic increasing
for any $\L$.

\\Note now that

$$
{\dpr\over \dpr\r_\L}p_\L ~=~{1\over \b} {\<N\>\over
\<N^2\>-\<N\>^2}~=~C_\L\ge 0 \Eq(CL)
$$
$C_\L$ is a constant in general  depending on $\L$.

\section{Continuity of the pressure}
\index{pressure!continuity}
Experimentally the thermodynamic pressure (i.e. the infinite volume limit
of the finite volume pressure $p_\L$) is not only increasing as a function
of the density $\r$ but it appears furthermore to have no (jump)
discontinuities. A general proof of this fact is still lacking. It
has been proven that under suitable conditions on the potential
(super-stability) the pressure is indeed continuous as a function
of the density \cite{Ru70}. We  prove here this fact in  a much simpler case, namely
we assume that the pair potential is either hard core, or non
negative (i.e. purely repulsive). Our strategy will consist in
proving that the constant $C_\L$ in equation \equ(CL) is bounded
uniformly in $\L$. This will allow us to conclude that
$p_\L(\r_\L,\b)$ has a bounded derivative in $\r_\L$ uniformly in
$\L$, so the limit $p(\r,\b)=\lim_{\L\to\i}p_\L(\r_\L,\b)$, cannot
have vertical jumps as a function of $\r$, i.e.  $p(\r,\b)$ is
continuous as a function of the density.

\\Let us thus assume that the pair potential between particles is hard core or
purely repulsive. Under these conditions we will then prove the following theorem.

\begin{teo}\label{pcont}
Let $V$ be a tempered and stable pair potential.
If $V$ is either  positive ($V\ge 0$) or has an hard core (i.e.  $\exists
\;a$ such that $V(\xx)=+\i$ whenever $|\xx|\le a$), then
$$
{\<N\>\over \<N^2\>-\<N\>^2}\le (1+D\l)\Eq(2.21)
$$
where $D$ is uniform in $\L$.
\end{teo}

\\{\bf Proof}.  We will use the following short notations
$$
X_N ~=~ \xx_1,\dots ,\xx_N,~~~~~~~~~~~~ dX_N ~=~ d\xx_1\dots d\xx_N~,
~$$
$$
 U(X_N)~=~U(\xx_1,\dots ,\xx_N),~~~~~~~~~~~~~~
W(\xx,X_N)~=~ \sum_{j=1}^N V(\xx_j - \xx)
$$
$$
Z_N~=~ \int_{\L}d\xx_1 \dots \int_{\L}d\xx_N \, e^{ -\b
U(\xx_1,\dots ,\xx_N)}~=~ \int_{\L^N}dX_N  \, e^{ -\b U(X_N)}
$$
The partition function $\Xi_\L(\l,\b)$ (denoted shortly by $\Xi$)
can  thus be rewritten as
$$
\Xi ~=~\sum_{N=0}^{\i} {\l^N\over N!} Z_N\Eq(X)
$$
With these definitions we have,
$$
Z^2_{N+1}~=~ \left[\int_{\L^N}dX_N  \, e^{ -\b U(X_N)}\int_\L d\xx
e^{-\b W(\xx,X_N)}\right]^2~
~=
$$
$$
~=~
 \left[\int_{\L^N}dX_N  \, e^{ -{\b\over 2} U(X_N)}\int_\L d\xx
e^{ -{\b\over 2} U(X_N)}\, e^{-\b W(\xx,X_N)}\right]^2~=~
$$
$$=~
 \left[\int_{\L^N}dX_N  \, F(X_N)G(X_N)\right]^2
$$
where
$$
F(X_N)~=~ e^{ -{\b\over 2} U(X_N)}~~~~~~~~~~G(X_N)~=~ \int_\L d\xx e^{
-{\b\over 2} U(X_N)}\, e^{-\b W(\xx,X_N)}
$$
Using Schwartz inequality we get
$$
\left[\int_{\L^N}dX_N  \, F(X_N)G(X_N)\right]^2\le \int_{\L^N}dX_N
\, F^2(X_N) \int_{\L^N}dX_N  \, G^2(X_N)
$$
thus
$$
\begin{aligned}
Z^2_{N+1} & \le\int_{\L^N}\!\!dX_N  \, e^{ -\b U(X_N)} \!\!\int_{\L^N}\!\!dX_N
\int_{\L^2} d\xx d\yy e^{ -{\b} U(X_N)}\, e^{-\b W(\xx,X_N)} e^{-\b
W(\yy,X_N)}\\
& = Z_N \int_{\L^N}dX_N  \,\int_{\L^2} d\xx d\yy e^{ -{\b}
U(X_N,\xx,\yy)}\, e^{+\b V(\xx-\yy)}\\
& = Z_N \int_{\L^N}dX_N  \,\int_{\L^2} d\xx d\yy e^{ -{\b}
U(X_N,\xx,\yy)}\,\left ( e^{+\b V(\xx-\yy)}-1+1\right)\\
&= Z_N Z_{N+2} + Z_N \int_{\L^N}dX_N  \,\int_\L d\xx\int_\L d\yy
e^{ -{\b} U(X_N,\xx,\yy)}\,\left ( e^{+\b V(\xx-\yy)}-1\right)\\
&= Z_N  Z_{N+2} + Z_N \!\!\!\int_{\L^N}\!\!dX_N\!\! \int_\L\!\! d\xx e^{ -{\b}
U(X_N,\xx)} \!\!\int_\L \!\!d\yy e^{ -{\b} W(X_N,\yy)} [ 1-e^{-\b
V(\xx-\yy)}]
\end{aligned}
$$
Now, since we are assuming $V(\xx)$ hard core or positive we have
$$
e^{ -{\b} W(X_N,\yy)}\le K(\b) ~=~
\begin{cases} 1 &{\text if}~ V\ge 0\\
e^{ \b \mu(a)} & {\text if}~ V~{\rm  has ~an ~hard~ core}~ a
\end{cases}
$$
where $\mu(a)$ is a constant defined in \equ(mua).

\\Thus
$$
Z_{N+1}^2\le Z_N  Z_{N+2} ~+ ~K(\b) Z_N\!\!\! \int_{\L^N}dX_N \int_\L
\!\!d\xx e^{ -{\b} U(X_N,\xx)} \!\!\int_\L \!\!d\yy \,\left ( 1-e^{-\b
V(\xx-\yy)}\right)\le
$$
$$
\le Z_N \, Z_{N+2} + K(\b) Z_N Z_{N+1} \int_\L d\yy \,\left|
1-e^{-\b V(\xx-\yy)}\right|\
$$
Note now that in the case in which $V$ is positive or have hard
core, and also recalling that  $V$ must also be tempered, we have
$$
\int_\L d\yy \,\left| 1-e^{-\b V(\xx-\yy)}\right|\le \int_{
\mathbb{R}^3} d\xx | 1-e^{-\b V(\xx)}|\le C(\b)<+\i
$$
and calling $D(\b) ~=~ K(\b) C(\b)$ we get finally
$$
Z_{N+1}^2\le Z_N \, Z_{N+2} + D(\b) Z_N Z_{N+1}    ~
\Longleftrightarrow ~ {Z^2_{N+1}\over Z_{N}}\le {Z_{N+2}}+
Z_{N+1}D(\b)\Eq(Pi)
$$
Now consider
$$
\left(\<N\>[1+\l D(\b)]\right)^2~=~ \left(\sum_{N=0}^{\i}{\l^N\over
N!}N {Z_N\over \Xi}[1+\l D(\b)] \right)^2 ~=~
$$
$$
=~ \left(\sum_{N=0}^{\i}{\l^N\over N!} {Z_N\over \Xi}
\left[{\l^{N+1}Z_{N+1}\over \l^N Z_N}+N\l D(\b)\right]
\right)^2~=
$$
$$=
~~ \left(\sum_{N=0}^{\i}{\l^N\over N!} {Z_N\over \Xi}
\left[{\l\,Z_{N+1}\over Z_N}+N\l D(\b)\right] \right)^2
$$
where in the second line we have used
$$
\sum_{N=0}^{\i}  {\l^N\over N!} N Z_N ~=~ \sum_{N=1}^{\i}  {\l^N\over (N-1)!}  Z_N ~=~\sum_{N=0}^{\i}
{\l^{N+1}\over N!}  Z_{N+1}
$$
and we will use below that, by definition \equ(X)
$$
\sum_{N=0}^{\i}  {\l^N\over N!}  {Z_N\over \Xi} ~=~1
$$
Now
$$
\left(\sum_{N=0}^{\i}{\l^N\over N!} {Z_N\over \Xi} \left[{\l
Z_{N+1}\over  Z_N}+N\l D(\b)\right]
\right)^2~=~\left(\sum_{N=0}^{\i} F_N \cdot G_N\right)^2
$$
where
$$
F_N ~=~ \left({\l^N\over N!} {Z_N\over \Xi}\right) ^{1/2}
~~~~~~~~~ G_N~=~  \left({\l^N\over N!} {Z_N\over \Xi}\right)
^{1/2} \left[{\l Z_{N+1}\over  Z_N}+N\l D(\b)\right]
$$
using thus again Schwartz inequality we get
$$
\left(\sum_{N=0}^{\i} F_N \cdot G_N\right)^2\le
(\sum_{N=0}^{\i}F_N^2)(\sum_{N=0}^{\i}G_N^2) ~=~
$$
$$=~
1\times
\sum_{N=0}^{\i}  \left({\l^N\over N!} {Z_N\over \Xi}\right)
\left[\l {Z_{N+1}\over Z_N}+N\l D(\b)\right]^2~=~
$$
$$
~=~ \sum_{N=0}^{\i}  \left({\l^N\over N!} {1\over \Xi}\right)
\left[\l ^2 {Z^2_{N+1}\over Z_N}+ 2 N\l^2 D(\b)  {Z_{N+1}}+ Z_N
N^2\l ^2D^2(\b)\right]
$$
hence using also \equ(Pi)
$$
\left(\<N\>[1+\l D(\b)]\right)^2\le \sum_{N~=~0}^{\i}
\left({\l^N\over N!} {1\over \Xi}\right) \l ^2\Bigg[ {Z_{N+2}}~+~~~~~~~~~~~~~~~~~~~
$$
$$
~~~~~~~~~~~~~~~~~~~~+~
Z_{N+1}D(\b)+ 2 ND(\b)  {Z_{N+1}}+ Z_N N^2D^2(\b)\Bigg]
$$
~~ Now observe that
$$
 \sum_{N=0}^{\i}  \left({\l^N\over N!} {1\over \Xi}\right)
\l ^2 {Z_{N+2}}~=~\sum_{N=2}^{\i}  \left({\l^N{Z_{N}}\over (N-2)!} {1\over \Xi}\right)
~=~\<N(N-1)\>
$$
$$
 \sum_{N=0}^{\i}  \left({\l^N\over N!} {1\over \Xi}\right)
\l^2 Z_{N+1}D(\b)~=~\l D(\b) \<N\>
$$

$$
 \sum_{N=0}^{\i}  \left({\l^N\over N!} {1\over \Xi}\right)
2 N\l^2 D(\b)  {Z_{N+1}}~=~ 2D(\b)\l \<N(N-1)\>
$$
$$
 \sum_{N=0}^{\i}  \left({\l^N\over N!} {1\over \Xi}\right)
Z_N N^2\l ^2D^2(\b)~=~ D^2(\b)\l^2\<N^2\>
$$
thus
$$
\left(\<N\>[1+\l D(\b)]\right)^2\le
$$
$$
\le ~\<N(N-1)\>+ \l D(\b) \<N\> ~ +~ 2D(\b)\l \<N(N-1)\> ~ + ~
D^2(\b)\l^2\<N^2\>~=~
$$

$$
~=~ \<N^2\> -\<N\>+ \l D(\b) \<N\>  ~+ ~2D(\b)\l (\<N^2\> -\<N\>)~+~
D^2(\b)\l^2\<N^2\>~=~
$$

$$
~=~ \<N^2\>(1+ D(\b)\l)^2 - \<N\>(1+\l D(\b))
$$

\\
Thus we are arrived at the inequality
$$
\<N\>^2[1+\l D(\b)]^2\le ~\<N^2\>(1+ D(\b)\l)^2 - \<N\>(1+\l
D(\b))
$$
i.e.
$$
\<N\> ~~\le  ~(\<N^2\>-\<N\>^2) (1+ \l D(\b) )
$$
which is as to say
$$
{\<N\>\over \<N^2\>-\<N\>^2} \le 1+\l D(\b)
$$
and the proof is completed. $\Box$

\section {Analiticity of the pressure}
\vskip.1cm
\\We have seen by property 0a that the pressure at finite volume $p_\L(\b,\l)$ is analytic as a
function of its parameters $\b$ and $\l$ in the whole physical
domain $\l>0$, $\b>0$. We can now ask if the infinite volume
pressure $p(\b,\l)$ is also analytic in its parameters. If this
were the case then we would be really in trouble and we should
conclude that statistical mechanics is not sufficient to describe
the  macroscopic behaviour of a system with a large number of
particles.
\index{pressure!analyticity}
As a matter of fact experiments tells us that the physical
pressure can indeed be non analytic. For example the graphic of
the pressure versus the density at constant temperature (if the
temperature is not too high, i.e. not above the critical point)
for a real gas is drawn below.

\\When $\r$ reaches the value $\r_0$ the gas starts to condensate to its liquid phase
and during the whole interval $[\r_0, \r_1]$ the gas performs a
phase transitions. i.e. from its gas phase to its liquid phase, at
the same pressure $p_0$. Above $\r_1$ the system is totally in its
liquid phase. The change occurs abruptly and is usually
characterized by singular behaviour in thermodynamic functions.

Hence, in spite of the fact that pressure is continuous even in
the limit, its derivatives may be not continuous in the infinite
volume limit. Note that this fact substantially justifies the
necessity to take the thermodynamic limit. Until we stay at finite
volume, all thermodynamic functions are analytic, hence in order
to describe a phenomenon like phase transition we are forced to
consider the infinite volume limit. %Of course the infinite volume
%is always taken in such way that $\r ~=~ N/V$ is fixed.
We can thus give the following mathematical definition for phase transition
\vskip.5cm
\begin{defi}
Any non analytic point of the grand
canonical pressure defined in \equ(2.13) occurring for real
positive $\b$ or $\l$ is called a phase transition point.
\end{defi}

\\People believe in general that the pressure $p(\b, \l)$ is a piecewise analytic function of
its parameters in the physical intervals $\l>0$ $\b>0$.

\\Hence it is very important to see which are the values of parameter $\l$ and $\b$ for which
the pressure is analytic.  Guided by the physical intuition, for  low
values of  $\l$ and $\b$   one expects that the
pressure is indeed analytic. In fact $\l$ low means that the
system is at low density, while $\b$ low means that the system is
at high temperature (e.g. above the critical point thus so high
that the system is always a gas and never condensates). For such
values the system is indeed a gas and in general very near to a
perfect gas. I.e. for temperature sufficiently high and/or density
sufficiently low the system should be in the gas phase and no phase
transition occurs.

\\Hence it should exist a theorem stating that the pressure $p(\b,\l)$ is analytic for
$\b$ and/or $\l$ sufficiently small.  We would see later the such a
theorem can effectively be proved.

\begin{figure}
\begin{center}
\includegraphics[width=7cm,height=7cm]{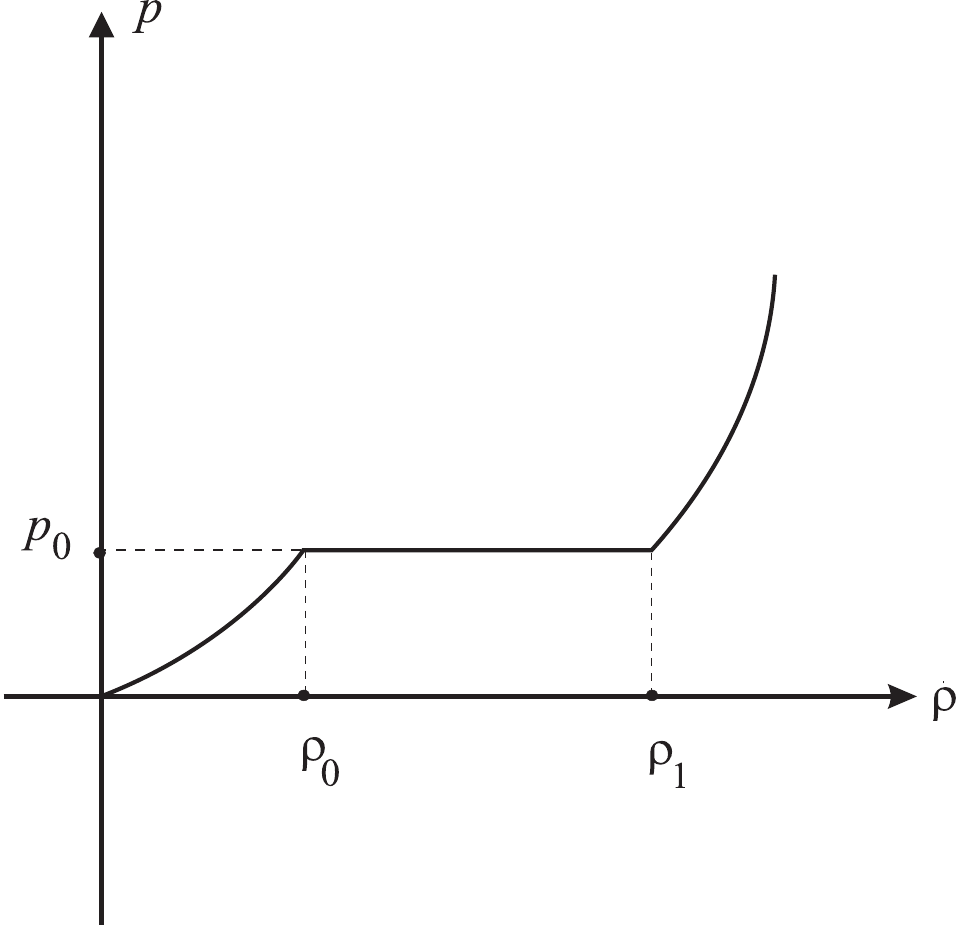}
\end{center}
\begin{center}
Figure 9. Pressure versus density for a physical gas. The
gas-liquid phase transition
\end{center}
\end{figure}

\\We may ask the following question. We know that $p(\b,\r)$ is continuous while e.g.
${\dpr p\over\dpr \r}$ may be not. But if ${\dpr p\over\dpr \r}$
is not continuous in some point it means that at that point $\r_0$
the function can take two values. So what is the thermodynamic
limit when $\r~=~\r_0$? Or, in other words which of the possible
values of ${\dpr p\over\dpr \r}$ the system chooses at the
thermodynamic limit?

\\The answer to this question resides  in the concept of boundary conditions.
Up to now boundary conditions were ``open", i.e. we were studying
a system of particles enclosed in a box $\L$ supposing that
outside $\L$ there was nothing, just empty space.

\\However, we could also
have done things in a different way, or in more proper words we
could have put a different boundary condition. For instance, we
can put $n$ particles outside the box  $\L$  at fixed point $\yy_1,\dots,\yy_n$.\index{boundary conditions}

\\In this case the grand canonical partition function looks as
$$
\Xi^{\yy}_\L (\b,\l)~=~1~+ \sum_{N=1}^{\infty}{\l^N \over N!}
\int_{\L}\!\!d\xx_1\dots\!\!\int_{\L} \!\!d\xx_{N} ~e ^{-\b U(\xx_1,\dots
,\xx_N)} e^{-\b W (\xx_1,\dots ,\xx_N,\yy_1,\dots ,\yy_n)}
\Eq(2.2bc)
$$
where
$$
W (\xx_1,\dots ,\xx_N,\yy_1,\dots ,\yy_n)~=~\sum_{i=1}^N\sum_{j=1}^n
V(\xx_i -\yy_j)
$$
Hence the function $\Xi_\L (\b,\l)$ may depend also on boundary
conditions. This means that also $p_\L$ the pressure at finite
volume depends on boundary conditions.

\\Does the infinite volume pressure depend on boundary conditions?

\\The answer to this question must be no (always guided by physical intuition), at least
for not too strange systems and/or not too strange boundary
conditions.

\\It is possible to show this quite easily for a finite range (with
range $\bar r$) hard core (with hard core $a$)  potential. I.e.,
particles must stay at distances $a$ or higher and they do not
interact if the distance is greater than $\bar r$.

\\In this case, the particles outside $\L$ that give contribution to the
partition function \equ(2.2bc) are  those which are contained  in a layer of thickness
$\bar r$ outside $\L$. Supposing that $\L$ is a cube in $\mathbb{R}^3$ of size $L$, the
volume of this layer is of the order $L^2 \bar r$.
\begin{figure}
\begin{center}
\includegraphics[width=5cm,height=5cm]{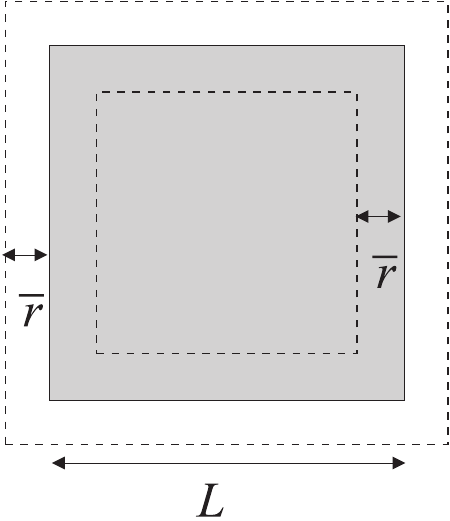}
\end{center}
\begin{center}
Figure 10.
\end{center}
\end{figure}
Any particle inside $\L$ can interact just with particles inside a
sphere of radius $\bar r$ centered at the particles position, and
due to the hard core condition in this sphere one can arrange at
most of the order of  $(\bar r / a)^3$. This means that a particle
inside can interact at most with $(\bar r / a)^3$ particles
outside $\L$.

\\On the other hand  particles inside $\L$  that can interact with particles
outside $\L$  are also contained in a (internal) layer of size
$\bar r$, and the maximum number of particles inside $\L$(hence
contributing to $W$ in \equ(2.2bc)) is of the order
$$
L^2 \bar r\over a^3
$$
With these observations it is not difficult to see that
$$
|W (\xx_1,\dots ,\xx_N,\yy_1,\dots ,\yy_n)|\le   {\rm Const.}~{L^2
\bar r\over a ^3 } \left({\bar r \over a}\right)^3
$$
and therefore
$$
 e^{-{\rm Const.}~{L^2
\bar r\over a ^3 } \left({\bar r \over a}\right)^3} \Xi^{\rm open}_\L (\b,\l)~\le ~\Xi^{\yy}_\L (\b,\l)~\le ~ e^{{\rm Const.}~{L^2
\bar r\over a ^3 } \left({\bar r \over a}\right)^3} \Xi^{\rm open}_\L (\b,\l)
$$
Thus
$$
{1\over L^3}\ln \Xi^{y}_\L(\b, \l)~=~ {1\over L^3}\ln \Xi^{\rm
open}_\L(\b,\l) \pm  {1\over L^3}{\rm Const.}~{L^2 \bar r\over
a ^3 } \left({\bar r \over a}\right)^3
$$
and the factor
$$
{1\over L^3}{\rm Const.}~{L^2 \bar r\over  a ^3 } \left({\bar r
\over a}\right)^3
$$ goes to zero as $L\to\i$. I.e. the pressure does not depend on boundary conditions.

\\But now  in general the derivative of the pressure may depend on boundary condition.
This happens precisely when the derivative are not continuous. In
such points different boundary condition may force different
values of derivatives. Changing boundary condition one can thus
change the value on a derivative in a discontinuity point. This
can be interpreted as an alternative definition of phase
transition. \vskip.1cm
\index{phase transition}
\begin{defi}
A phase transition point is a point in which the value of
some derivative of the infinite volume pressure depends on
boundary conditions even at the thermodynamic limit.
\end{defi}
Thus when the system is sensible to change of boundary conditions  we say that  there is a phase transition.

\\By this rough discussion, we see that lack of analyticity of
the pressure or sensitivity of the system to change in boundary
conditions are two ways to characterize a phase transition.

\newpage
\numfor=1\numsec=3
\chapter{High temperature low density expansion}
\numsec=3\numfor=1
\section{The Mayer series}\label{sec3}
\subsection{Some Notations about graphs}
\\We first give some definitions about abstract graphs which will be useful below.

\index{graph}
\\Let $A$ be any finite set, we denote by $|A|$ the number of
elements of $A$. We denote by $P(A)$ the power set of $A$ (i.e.
the set of all subsets of $A$). We denote by $P^k(A)~=~\{U\subset A:
|U|~=~k\}$ (i.e. the set of all subsets of $A$ with cardinality $k$). If
$A~=~\{1,2,\dots ,n\}$, we shortly put  $P(\{1,2,\dots ,n\})\equiv
P_n$ and $P^k(\{1,2,\dots ,n\})\equiv P^k_n$.

\begin{defi}
A graph $g$ is a pair $g~=~(V_g,E_g)$ where
$V_g$ is a countable set and $E_g\subset P^2(V_g)$. The set $V_g$
is called the {\it vertex set} of $g$ and the elements of $V_g$
are called {\it vertices} of the graph $g$, while $E_g$ is called
the {\it edge set} of $g$ and the elements of $E_g$ are
called {\it edges} of the graph $g$.
\end{defi}

\\%Given a graph $g~=~(V_g,E_g)$,
%the set ${\rm Supp}\, g~=~\{x\in V_g: \exists e\in E_g: x\in e\}$ is
%called the {\it support} of $g$. Note that ${\rm supp}\ g\subset
%V_g$.
Given two graphs $g~=~(V_g,E_g)$ and $f~=~(V_f,E_f)$ we say that
$f\subset g$ if  $V_f\subset V_g$ and $E_f\subset E_g$.

\begin{defi}
A graph $g~=~(V_g,E_g)$
is said to be {\it connected} if for any pair $B, C$ of  subsets
of $V_g$ such that $B\cup C ~=~V_g$ and $B\cap C ~=~\emptyset$, there is
a $e\in E_g$ such that $e\cap B\neq\emptyset$ and $e\cap
C\neq\emptyset$.
\end{defi}

\\We denote by $\GG_{V}$ the set of all graphs with vertex set $V$ and
by $G_{V}$ the set of all connected graphs with vertex set $V$. We
will shortly denote $\GG_n$ the set $\GG_{\{1,2,\dots, n\}}$ of
all graphs with vertex set $\{1,2,\dots ,n\}$ and by $G_n$ the set
$G_{\{1,2,\dots, n\}}$ of all connected graphs with vertex set
$\{1,2,\dots ,n\}$. We will also denote shortly
$$
[n]=\{1,2,\dots, n\}\Eq(Ien)
$$
$$
\E_n=\{\{i,j\}\subset [n]\}\Eq(Een)
$$
In particular $\E_n$ is the set of all unordered pairs $\{i,j\}$ of the set $[n]$.

\begin{defi}
A  graph $\t\in G_{V}$ such
that $|E_\t|~=~|V|-1$ is called a tree graph.
\end{defi}

\\We denote $T_{V}$ the set of all tree graphs with vertex set $V$. Note that
$T_V\subset G_V$. The set of all the tree graph over $[n]$
will be denoted by $T_n$.

\begin{defi}
Let $g~=~(V_g,E_g)$ be a graph and let $x\in V_g$. Then
the degree (or number of incidence)  $d_x$ of the vertex $x\in V_g$
in $g$ is the number of edges $e\in E_g$ such that $x\in e$.
\end{defi}

\index{graph!connected}\index{graph!tree graph}
\\Observe that in a tree $\t\in T_n$ the numbers of incidence
$d_1, \dots ,d_n$ at vertices $1,\dots, n$ satisfies the identity
$$
d_1+d_2+\dots +d_n~=~ 2n-2\Eq(dipiu)
$$
This identity holds because any tree $\t$ has $n-1$ edges and each
edge has two vertices and so each edge is counted twice in the sum of the l.h.s. of \equ(dipiu).
Moreover, we have clearly the following
bound for the number of incidence $d_i$ in a vertex of tree $\t$.
$$
1\le d_i\le n-1\Eq(dimeno)
$$
This is again because any tree   $\t\in T_n$ is connected (hence $d_i\ge 1$ for all $i\in [n]$) and has $n-1$ edges (hence $d_i\le n-1$ for  all $i\in [n]$).

\\The number of trees in $T_n$ is explicitly computable using the
Cayley formulas.  \index{Cayley formulas} These formulas can be presented by the
following lemma. \vskip.1cm

\begin{lem}
The number of trees in $T_n$ with
numbers of incidence in vertices $1,2,\dots,n$ fixed at the values
$d_1 ,d_2 ,\dots , d_n$ is given by
$$ \sum_{\t\in T_{n}\atop d_1 ,d_2 ,\dots ,
d_n ~{\rm fixed}}1~=~ {(n-2)!\over \prod_{i~=~1}^{n}(d_i -1)!}
\Eq(cay1)
$$
Moreover
$$ |T_n|~=~\sum_{\t\in
T_{n}}1~=~n^{n-2}\Eq(cay2)
$$
\end{lem}

\\{\it Proof.}
We first show \equ(cay1) by induction. Formula \equ(cay1) is
trivially true if $n=2$. We suppose now that \equ(cay1) is true
for all trees with $n$ vertices and with any incidence numbers
$d_1,\dots ,d_n$ and we will prove that \equ(cay1) holds also for
all trees with $n+1$ vertices and with any incidence numbers
$d_1,\dots ,d_n,d_{n+1}$.

\\Take a tree $\t$ with $n+1$ vertices and numbers of incidence in vertices fixed at
the values $d_1,\dots ,d_{n+1}$. Such tree has at least two
vertices with number of incidence equal to 1. Without lost in
generality let us suppose that $d_{n+1}=1$ (otherwise we can
always rename the vertices of $\t$). Now if $d_{n+1}=1$, then
there is an edge in $\t$ which links the vertex $n+1$ with some
vertex  $j\in \{1,\dots ,n\}$. This vertex $j$ has surely $d_j\ge
2$, i.e. there is at least another edge starting from $j$ which
ends in some other vertex $k\in \{1,\dots ,n\}$, since $\t$ is
connected. So we can count trees with $n+1$ vertices and fixed
number of incidence proceeding as follows. For each time we fix
the edge joining $n+1$ to some vertex $j\in \{1,\dots ,n\}$, we
count all trees in $\{1, \dots ,n\}$  with number of incidence
fixed at the values $d_1, \dots, d_{j-1}, d_j -1, d_{j+1}, \dots
d_n$, by using the induction hypothesis. Namely,
$$
\begin{aligned}
\sum_{\t\in T_{n+1}\atop d_1, \dots , d_{n+1} ~{\rm
fixed}}1& =\!\!\!\!\!\!\sum_{j\in \{1,2,\dots,n\}:\atop d_j\ge 2}\!\!\!\! {(n-2)!\over
(d_1-1)!\cdots (d_{j-1}-1)!(d_j-2)!(d_{j+1}-1)!\cdots (d_n-1)!}\\
&=\sum_{j\in \{1,2,\dots,n\}:\atop d_j\ge 2} {(n-2)!(d_j-1)\over
(d_1-1)!\cdots(d_j-1)!\cdots  (d_n-1)!(d_{n+1}-1)!}\\
& = {(n-2)!\over
\prod_{i=1}^{n+1}(d_i-1)!}\sum_{j=1}^n (d_j-1)
\end{aligned}
$$
where in the second line we just rewrite
$$
{1\over (d_j-2)!}~=~ { (d_j-1)\over (d_j-1)!(d_{n+1} -1)!}
$$
using the assumption that $d_{n+1}=1$ and hence  $(d_{n+1}
-1)!=0!=1$. Now using again that $d_{n+1}-1=0$ and by \equ(dipiu) with $n+1$ in place of $n$, we get
$$
\sum_{j=1}^n (d_j-1)= \sum_{j= 1}^{n+1} (d_j-1)~=~ \sum_{j= 1}^{n+1} d_j -(n+1)~=~
2(n+1)-2 -(n+1)~=~ n-1
$$
hence
$$
\sum_{\t\in T_{n+1}\atop d_1 ,d_2 ,\dots , d_{n+1} ~{\rm fixed}}1~=~
{(n-2)!\over \prod_{i=1}^{n+1}(d_i-1)!}(n-1)~=~ {(n-1)!\over
\prod_{i=1}^{n+1}(d_i-1)!}
$$
and  \equ(cay1) is proved. Now  \equ(cay2) is an easy consequence of \equ(cay1). As a matter of fact, just observe,  using
\equ(dipiu), \equ(dimeno), and \equ(cay1), that
$$
\sum_{\t\in T_{n}}1~=~\sum_{d_1, \dots , d_n: ~1\le d_i\le n-1\atop
d_1+\dots +d_n=2n-2}  {(n-2)!\over \prod_{i=1}^{n}(d_i -1)!}~=~
\sum_{s_1, \dots , s_n:~ 0\le s_i\le n-2\atop s_1+\dots +s_n~=~n-2}
{(n-2)!\over \prod_{i=1}^{n}s_i!}~=~ n^{n-2}
$$
$\Box$
\vskip.1cm
\subsection{Mayer Series: definition}\label{mayersec}
We now come back to the claim  that a system of particles,
interacting via a  reasonable potential energy (e.g. defined via a
stable and tempered pair potential) should be  in the  a gas phase  at
sufficiently low density and/or sufficiently high temperature,
hence the pressure of such system should be analytic in the
thermodynamic parameters in this region.  Since we are considering
just the Grand Canonical ensemble, the region of high temperature
and low density will be in the case the region of $\l$ small and
$\b$ small.

\\Let us thus consider again the grand Canonical partition function
$$
\Xi_\L(\b,\l)~=~1+ \sum_{n=1}^{\infty}{\l^n \over n!}
\int_{\L}d\xx_1\dots\int_{\L} d\xx_{n} ~e ^{-\b U(\xx_1,\dots
,\xx_n)} \Eq(n2.2)
$$
with
$$
U(\xx_1,\dots ,\xx_n)~=~\sum_{1\le i<j\le n}V(\xx_i-\xx_j)\Eq(n2.3)
$$
where $V(\xx)$ is a stable and tempered pair potential. We denote shortly $\xx_{[n]}$ the set of coordinates
$\xx_1,\dots,\xx_n$ and if $I\subset [n]$ then $\xx_{I}$ will denote the set of coordinates $\xx_i$ with $i\in I$. Hence
for $I\subset [n]$
$$
U(\xx_{I})=\sum_{\{i,j\}\subset I} V(\xx_i-\xx_j)
$$

\\We now look for an expansion of the log of such function valid for $\l$ small.
This is clearly possible, due to the structure of $\Xi_\L(\b,\l)$. In fact, as a power series of $\l$ the partition
function has the form
$$
\Xi_\L(\b,\l)~=~1 + c_1\l + c_2\l^2 + c_ 3\l^3 + \dots ~=~ 1+ O(\l)
$$
and $\ln( 1+ O(\l))$ can indeed be expanded also in power series
of $\l$. Let us show that there exists a formal expansion of $\ln
\Xi $ in powers of $\l$ called the Mayer series of the
pressure. We thus prove now the following theorem.

\begin{teo}\label{mayer}
Let $\Xi_\L(\b,\l)$ be defined as in \equ(n2.2) and \equ(n2.3).
Then
$$
{1\over |\L|}\ln \Xi_\L(\b,\l) ~=~ \sum_{n=1}^{\i}C_n(\b,\L)\l^n\Eq(pressm)
$$
where
$$
C_n(\b,\L)~=~{1\over n!}{1\over |\L|}\int_{\L}d\xx_1
\dots \int_{\L} d\xx_n \Phi^T(\xx_1,\dots,\xx_n) \Eq(ursm)
$$
with $\Phi^T(\xx_1,\dots,\xx_n)$ given by  the following  expressions:
\begin{itemize}
\item[A)]
$$
\Phi^T(\xx_1,\dots,\xx_n)~=~
\begin{cases}\sum\limits_{g\in G_{n}}~
\prod\limits_{\{i,j\}\in E_g}\left[  e^{ -\b V(\xx_i -\xx_j)} -1\right] &{\rm  if}~
n\ge 2\\ 1 &{\rm if}~ n=1
\end{cases}
\Eq(urse)$$
where recall   that $G_n$ is the set of all connected graphs in
$\{1,2,\dots, n\}$ and if $g\in G_n$ then its edge set is denoted by $E_g$.

\item[B)] Given $[n]$ and $1\le k\le n$, let $\Pi^k_n$ is the set of all partitions of $[n]$ in $k$ blocks. Then
$$
\Phi^T(\xx_1,\dots,\xx_n)~=~ \sum_{k=1}^n (-1)^{k-1}(k-1)!\sum_{\{I_1,I_2, \dots, I_k\}\in \P^k_{n}}e^{-\b\sum_{\a=1}^k U(\xx_{I_\a})} \Eq(Basuev)
$$
where  $I=\{i_1,\dots, i_{|I|}\}\subset [n]$ and $x_I=(x_{i_1},\dots, x_{i_{|I|}})$ with $ U(\xx_{I})=U(x_{i_1},\dots, x_{i_{|I|}})$.
\end{itemize}

\end{teo}

\\The series in the r.h.s. of \equ(pressm) is called the Mayer series\index{Mayer series} and the coefficient
$C_n(\b,\L)$ ($n\ge 1$) is called Mayer (or Ursell) coeffcient of order $n$. For the moment this series is merely
a formal series, since Theorem \ref{mayer} does not say if the series
converges or not.

\\{\it Proof of A}.

$$
\begin{aligned}
\Xi_\L(\b,\l)& =~1+ \sum_{n=1}^{\infty}{\l^n \over n!}
\int_{\L}d\xx_1\dots\int_{\L} d\xx_{n} ~e ^{-\b \sum_{1\le i<j\le n}V(\xx_i -\xx_j)}\\
& =~1+ \sum_{n=1}^{\infty}{\l^n \over n!}
\int_{\L}d\xx_1\dots\int_{\L} d\xx_{n} ~\prod_{1\le i<j\le n}e
^{-\b V(\xx_i -\xx_j)}\\
&=~1+ \sum_{n=1}^{\infty}{\l^n \over n!}
\int_{\L}d\xx_1\dots\int_{\L} d\xx_{n} ~\prod_{1\le i<j\le
n}\left[\left(e ^{-\b V(\xx_i -\xx_j)}-1\right)+1\right]
\end{aligned}
$$
Develop  now  the product in  the factor
$$
\prod_{1\le i<j\le n}\left[\left(e ^{-\b V(\xx_i
-\xx_j)}-1\right)+1\right]
$$
Then it is not difficult to see that it
can be rewritten as
$$
\prod_{1\le i<j\le n}\left[\left(e ^{-\b V(\xx_i
-\xx_j)}-1\right)+1\right]~=~ \sum_{g\in \GG_n} \prod_{\{i,j\}\in
E_g}\left[ e^{ -\b V(\xx_i -\xx_j)} -1\right]
$$
where recall that $\GG_n$ is the set of all graphs (connected or
not connected) in the set $\{1,2\dots ,n\}$.  In $\GG_n$ we also
include the empty graph, i.e the graph $g$ such that
$E_g~=~\emptyset$, and its contribution is the factor $1$ in the
development of the product.

\\We now reorganize the sum over graphs in $\GG_n$.

\\For $1\le k\le n$,
let $\{I_1,I_2, \dots, I_k\}$ denote a partition of the set
$[n]$. Namely, for any $i,j ~=~1,2\dots ,k$ we have that
$I_i\neq\emptyset$, $I_i\cap I_j~=~\emptyset$ and
$\cup_{j~=~1}^{k}I_j~=~ [n]$. We denote by $\xx_{I_j}$ the
set of coordinates $\xx_i$ with $i\in I_j$. We also denote with
$\P_{n}$ the set of all partitions of the set $[n]$.

\\Then it is not difficult to see that
$$
\sum_{g\in \GG_n} \prod_{\{i,j\}\in E_g}\left[ e^{ -\b V(\xx_i
-\xx_j)} -1\right] ~=~ \sum_{k=1}^n\sum_{\{I_1,I_2, \dots,
I_k\}\in \P^k_{n}}\prod_{j=1}^k \Phi^T(\xx_{I_{j}})\Eq(2.2s)
$$
where
$$
\Phi^T(\xx_{I_{j}})~=~
\begin{cases}\sum\limits_{g\in G_{I_j}} \prod\limits_{\{l,s\}\in
E_g}\left[  e^{ -\b V(\xx_l -\xx_s)} -1\right] &{\rm if}~ |I_j|\ge 2\\
1 &{\rm if}~ |I_j|~=~ 1
\end{cases}
$$
Note now that $\sum_{g\in G_{I_j}} $ runs over all {\it connected}
graphs in the set $I_j$.

\\Hence
$$
\Xi_\L (\b,\l)~=~1+ \sum_{n=1}^{\infty}{\l^n \over n!}
\int_{\L}d\xx_1\dots\int_{\L} d\xx_{n}
\sum_{k=1}^n~\sum_{\{I_1, \dots, I_k\}\in \P^k_{n}}~\prod_{j=1}^k
\Phi^T(\xx_{I_{j}})
$$
Now observe that  each $\Phi^T(\xx_{I_{j}})$ depends only on the
set of coordinates $\xx_i$ with $i\in I_j$ and since $I_1, \dots
I_k$ is a partition of $\{1,\dots ,n\}$ we can write
$$
\int_{\L}d\xx_1\dots\int_{\L} d\xx_{n} ~=~
\int_{\L^{|I_1|}}d\xx_{I_1}\dots \int_{\L^{|I_k|}}d\xx_{I_k}
$$
where of course $d\xx_{I_i}~=~\prod_{i\in I_j}d\xx_i$. Hence
$$
\Xi_\L (\b,\l)~=~1+ \sum_{n=1}^{\infty}{\l^n \over n!}
\sum_{k=1}^n\sum_{\{I_1, \dots, I_k\}\in \P_{n}}\prod_{j=1}^k
\int_{\L^{|I_j|}}d\xx_{I_j}\Phi^T(\xx_{I_{j}})
$$
Observe now that the factor
$$
\int_{\L^{|I_j|}}d\xx_{I_j}\Phi^T(\xx_{I_{j}})
$$
depends only on $|I_j|$ and not anymore from $\xx_i$ (since space
coordinates are integrated) and not even on $I_j$ since $i\in I_j$
is just an index attached to a mute variable. I.e.

$$
\int_{\L^{|I_j|}}d\xx_{I_j}\Phi^T(\xx_{I_{j}})~=~
\int_{\L}d\xx_1\dots \int_{\L}d\xx_{|I_j|} \Phi^T(\xx_1, \dots
\xx_{|I_{j}|}) ~=~  \phi(|I_j|)$$ Note that the numbers  $|I_j|$ are
positive integers subjected  to the condition $\sum_{j=1}^k |I_j|
=n$. Hence
$$
\Xi_\L (\b,\l)~=~1+ \sum_{n=1}^{\infty}{\l^n \over n!}
\sum_{k=1}^n\sum_{\{I_1, \dots, I_k\}\in \P^k_{n}}\prod_{j=1}^k
\phi(|I_j|)
$$

\\ Now observe that
$$
\begin{aligned}
\sum_{\{I_1,I_2, \dots, I_k\}\in \P^k_{n}}\prod_{j=1}^k \phi(|I_j|)&=~
{1\over k!} \sum_{m_1, \dots ,m_k:~\, m_i\ge 1\atop m_i+\dots
+m_k=n} \sum_{I_1,\dots ,I_k:\atop
|I_1|=m_1,\dots, |I_k|=m_k}\prod_{j=1}^k\phi(m_j)\\
&=\,\,{1\over k!} \sum_{m_1, \dots ,m_k:~\, m_i\ge 1\atop m_1+\dots
+m_k=n} \prod_{j=1}^k\phi(m_j) \sum_{I_1,\dots ,I_k:\atop
|I_1|=m_1,\dots, |I_k|=m_k}\!\!\!\!1\end{aligned}
$$
here above we have also to divide by $k!$
because the same partition $I_1, \dots I_k$ of $[n]$ appears
exactly $k!$ times in the sum
$$
\sum_{m_1, \dots ,m_k:\, m_i\ge 1\atop m_1+\dots
m_k=n} \sum_{I_1,\dots ,I_k:\atop
|I_1|=m_1,\dots, |I_k|=m_k}
$$
Therefore we get
$$
\sum_{\{I_1,I_2, \dots, I_k\}\in \P^k_{n}}\prod_{j=1}^k \phi(|I_j|)=
{1\over k!}\sum_{m_1, \dots ,m_k:~\, m_i\ge 1\atop m_1+\dots
m_k=n}\prod_{j=1}^k \phi(m_j){n!\over m_1!\dots m_k!}\Eq(aiou)
$$
Indeed, ${n!/(m_1!\dots m_k!)}$ is the number
of partitions of $[n]$ in $k$ subsets $I_1 , ..., I_k$
such that the numbers $|I_1|,...,|I_k|$ are fixed at the values
$m_1,\dots ,m_k$ respectively.
\\Hence we can write
$$
\begin{aligned}
\Xi_\L(\b,\l)&=~1+ \sum_{n=1}^{\infty}{\l^n \over n!}
\sum_{k=1}^n{1\over k!}\sum_{m_1, \dots,m_k: \,m_i\ge
1\atop m_1+\dots + m_k=n} {n!\over m_1!\dots m_k!}
\prod_{j=1}^k \phi(m_j)\\
&=~1+ \sum_{n=1}^{\infty}  \sum_{k=1}^n{1\over
k!}\sum_{m_1, \dots,m_k: \,m_i\ge 1\atop m_1+\dots +
m_k=n} {\l^{m_1}\cdots\l^{m_k}\over m_1!\dots m_k!} \prod_{j=1}^k \phi(m_j)\\
&=~1+  \sum_{k=1}^\infty{1\over
k!}\sum_{m_1, \dots,m_k \atop \,m_i\ge 1} {\l^{m_1}\cdots\l^{m_k}\over m_1!\dots m_k!} \prod_{j=1}^k \phi(m_j)\\
&=~ 1+ \sum_{k=1}^\i {1\over
k!}\left[\sum_{m=1}^{\i} {\l^{m}\over m!}  \phi(m)\right]^k
\end{aligned}
$$
i.e we have found
$$
\Xi_\L(\b,\l)~=~1+ \sum_{k=1}^\i {1\over k!}\left[\sum_{n=1}^{\i}
{\l^{n}\over n!}  \phi(n)\right]^k~=~
\exp\left\{\left[\sum_{n=1}^{\i} {\l^{n}\over n!}
\phi(n)\right]\right\}
$$
Hence
$$
\begin{aligned}
\ln \Xi_\L(\b,\l)&=~ \sum_{n=1}^{\i} {\l^{n}\over n!}  \phi(n)\\
&=~\sum_{n=1}^{\i} {\l^{n}\over n!} \int_{\L}d\xx_1\dots
\int_{\L}d\xx_{n} \Phi^T(\xx_1, \dots, \xx_{n})\\
 &=~ \sum_{n=1}^{\i}
{\l^{n}\over n!} \int_{\L}d\xx_1\dots \int_{\L}d\xx_{n} \sum_{g\in
G_{n}} \prod_{\{i,j\}\in E_g}\left[  e^{ -\b V(\xx_i -\xx_j)}
-1\right]
\end{aligned}
$$
which concludes the proof of part A) of the theorem. $\Box$

\vskip.3cm
\\{\it Proof of B)}. We start by writing
$$
Z_{n}= \int_{\L}d\xx_1\dots\int_{\L} d\xx_{n} ~e ^{-\b U(\xx_{[n]})}
$$
and in general, if $I\subset [n]$,
$$
Z_{I}= \int_{\L^{|I|}}d\xx_{I} ~e ^{-\b U(\xx_{I})}= Z_{|I|}
$$
so that
$$
\Xi_\L (\b,\l)~=~1+ \sum_{n=1}^{\infty}{\l^n \over n!}
Z_n
$$
Hence, formally
$$
\begin{aligned}
\ln\Xi_\L (\b,\l)& =~\sum_{k=1}^\infty {(-1)^{k-1}\over k} \left[\sum_{n=1}^{\infty}{\l^n \over n!}Z_n\right]^k\\
&= ~
\sum_{k=1}^\infty {(-1)^{k-1}\over k} \sum_{m_1, \dots ,m_k\atop  m_i\ge 1}\prod_{j=1}^k {\l^{m_j}Z_{m_j}\over m_j!}\\
&= ~\sum_{k=1}^\infty {(-1)^{k-1}\over k}\sum_{n=k}^\infty{\l^n\over n!} \sum_{m_1, \dots ,m_k:\;  m_i\ge 1\atop m_1+\dots+
m_k=n}(\prod_{j=1}^k {Z_{m_j}}){n!\over m_1!\cdots m_k!}
\end{aligned}
$$
Now, by equation \equ(aiou) and denoting shortly
$$
\int_\L d\xx_1\dots\int_\L d\xx_n=\int_{\L^n} d\xx_{[n]}
$$
 we have that
$$
\begin{aligned}
\sum_{m_1, \dots ,m_k:\;  m_i\ge 1\atop m_1+\dots+\,
m_k=n}(\prod_{j=1}^k {Z_{m_j}}){n!\over m_1!\cdots m_k!} & = k! \sum_{\{I_1,I_2, \dots, I_k\}\in \P^k_{n}}\prod_{j=1}^k Z_{(|I_j|)}\\
&
=~
k!\int_{\L^n} d\xx_{[n]} \sum_{\{I_1, \dots, I_k\}\in \P^k_{n}}e ^{-\sum_{\a=1}^k\b U(\xx_{I_\a})}
\end{aligned}
$$
Thus we get
$$
\ln\Xi_\L (\b,\l)= \sum_{k=1}^\infty {(-1)^{k-1}(k-1)!}\sum_{n=k}^\infty{\l^n\over n!}\int_{\L^n} d\xx_{[n]} \sum_{\{I_1, \dots, I_k\}\in \P^k_{n}} e ^{-\b\sum_{\a=1}^k U(\xx_{I_\a})}
$$
Finally, exchanging the sum over $k$ with the sum over $n$ we get
$$
\ln\Xi_\L (\b,\l)= \sum_{n=1}^\infty{\l^n\over n!}\int_{\L^n} d\xx_{[n]} \sum_{k=1}^n {(-1)^{k-1}(k-1)!} \sum_{\{I_1, \dots, I_k\}\in \P_{n}} e ^{-\b\sum_{\a=1}^k U(\xx_{I_\a})}
$$
whence \equ(Basuev) follows.
$\Box$

\subsection{The combinatorial problem}\label{combpro}
The very structure of the Mayer series for $\ln\Xi_\L (\b,\l)$  hides a very
hard combinatorial problem. To understand this problem let us try to find
 a bound uniform in $\L$ for the coefficient of order $n$ of
the Mayer series, as given by formula \equ(ursm), i.e.
$$
C_n(\b,\L) ~=~  {1\over n!}{1\over |\L|} \int_{\L}d\xx_1\dots \int_{\L}d\xx_{n}\Phi^T(x_1,\dots,x_n)~=~~~~~~~~~~~~~~~~~~~~~~~~~~~~
$$
$$
=~
 {1\over n!}{1\over |\L|} \int_{\L}d\xx_1\dots \int_{\L}d\xx_{n} \sum_{g\in G_{n}}
\prod_{\{i,j\}\in E_g}\left[  e^{ -\b V(\xx_i -\xx_j)} -1\right]\Eq(cbeta2)
$$
\\Note that $C_n(\b,\L)$ is a function of $\b$ and $\L$.
The simplest way to bound  $|C_n(\b,\L)|$ is as follows
$$
\begn
|C_n(\b,\L)| & \le
 {1\over n!}{1\over |\L|} \int_{\L}d\xx_1\dots \int_{\L}d\xx_{n}\Big|  \sum_{g\in G_{n}}
\prod_{\{i,j\}\in E_g}\left[  e^{ -\b V(\xx_i -\xx_j)} -1\right]
\Big|\\
&\le
 {1\over n!}{1\over |\L|}  \sum_{g\in G_{n}}  \int_{\L}d\xx_1\dots \int_{\L}d\xx_{n}
\prod_{\{i,j\}\in E_g}\left|  e^{ -\b V(\xx_i -\xx_j)} -1\right |
\egn
$$
Now observe that, if $V(\xx)$ is stable then  necessarily it is
bounded below by $2B$, i.e.
$$
V(\xx_i -\xx_j)\ge -2B
$$
which is simply the stability condition $U(\xx_1, \dots ,\xx_n)\ge
-B n$ for the case in which $n~=~2$. Thus factor
$$
|  e^{ -\b V(\xx_i -\xx_j)} -1 |\le \max\{1, e^{2\b B}-1\}\le
1\Eq(3.1)
$$
for $\b$ sufficiently small.

\\Now observe that any connected graph $g\in G_n$ contains at least
one tree $\t\in T_n$. Let us thus choose for each $g$ a tree $\t$ such
that $\t\subset g$. It is allowed to choose the same tree for
different $g$'s.

\\So, by \equ(3.1) we have

$$
\prod_{\{i,j\}\in E_g}\left|  e^{ -\b V(\xx_i -\xx_j)} -1\right |\le
\prod_{\{i,j\}\in E_{\t}}\left|  e^{ -\b V(\xx_i -\xx_j)} -1\right
|
$$
We now prove the following
\begin{pro}\label{integr}
For any  $\t\in T_n$ it holds the following inequality
$$
\int_{\L}d\xx_1\dots \int_{\L}d\xx_{n} \prod_{\{i,j\}\in E_{\t}}
\left|  e^{ -\b V(\xx_i -\xx_j)} -1\right |~ \le~
|\L|\left[\int_{\mathbb{R}^3} |   e^{ -\b V(\xx)}
-1|d\xx\right]^{n-1} \Eq(treintg)
$$
\end{pro}

\\{\bf Proof}.
Let us denote shortly
$$
C_\t(\b)= \int_{\L} d\xx_1\dots \int_{\L} d\xx_{n}  \prod_{\{i,j\}\in E_{\t}}
\left|  e^{ -\b V(\xx_i -\xx_j)} -1\right |
$$
and
to fix ideas  let  $E_{\t}~=~\{i_1,j_1\},\dots \{i_{n-1},j_{n-1}\}$, so
that
$$
C_\t(\b)=
\int_{\L}d\xx_1\dots\int_{\L}d\xx_{n} \prod_{k=1}^{n-1}
\left|  e^{ -\b V(\xx_{i_k} -\xx_{j_k})} -1\right |
$$

\\Define a change of variables in the integral as follows
$$
y_k~=~ x_{i_k}-x_{j_k}, ~~~~~~\forall k ~=~2\dots ,n
$$
$$
y_1 ~=~x_1
$$
The Jacobian of this transformation is clearly 1, then
$$
\begn
C_\t(\b) & \le
\int_{\L}d\xx_1\int_{\mathbb{R}^3}d\yy_2\dots
\int_{\mathbb{R}^3}\!\!d\yy_{n}
 \prod_{j~=~2}^n
\left|  e^{ -\b V(\yy_j)} -1\right |\\
&=~
|\L|\left[\int_{\mathbb{R}^3} |   e^{ -\b V(\xx)}
-1|d\xx\right]^{n-1}\\
\egn
$$
$\Box$
\vv
\vv
\\It is now a simple exercise to see  that stability and temperness
imply that
$$
\int_{\mathbb{R}^3} |   e^{ -\b V(\xx)} -1|d\xx~\doteq~C(\b)< \i
$$
Actually, it is not even necessary for the potential to be
stable. A sufficient condition for the finiteness of $C(\b)$ is,
if $S(a)$ is a sphere of radius $a$ and center in $\xx~=~0$,
$V(\xx)$ integrable in $ \mathbb{R}^3 \backslash S(a)$ and non
negative for $|\xx|\le a$.

\\So we get

$$
|C_n(\b,\L)| \le {[C(\b)]^{n-1}\over n!} \sum_{g\in G_{n}} 1 ~=~
{[C(\b)]^{n-1}\over n!} B_n
$$
where $B_n$ is the number of connected graphs in the set
$\{1,2,\dots ,n\}$. Hence $|C_n(\b,\L)|$  has as  upper bound  a quantity uniform in $\L$, which is good,
 hence the
coefficents for the Mayer series of the pressure indeed admit a
bound uniform in $\L$!

\\Thus looking for the absolute convergence of the Mayer series for the
finite volume pressure we get
$$
|\b p_{\L}(\b,\l)|\le \sum_{n=1}^{\i} |C_n((\b,\L))
\l^n|\le \sum_{n=1}^{\i} {[C(\b)]^{n-1}\over n!} B_n   \l^n
$$
If we could now get a good bound for $B_n$, the number of
connected graphs in a finite set of $n$ elements, e.g. a bound at
worst such as $n! B^n$, we would have the proof that the infinite
volume  pressure is analytic in the region
$$
|\l|C(\b)B< 1
$$
with the further restriction \equ(3.1) on $\b$. But it is not the
case since it is easy to show that
for  all $n\ge 2$
$$
B_n\ge 2^{{(n-1)(n-2)\over 2} }\Eq(bad)
$$

\\Indeed, first observe that  the number of graphs, denoted with ${\cal B}_n $, (connected or not connected) in the
set $\{1,2,\dots ,n\}$ is
$$
{\cal B}_n ~=~\sum_{g\in {\cal G}_n} 1 ~=~2^{{1\over 2}n(n-1)}\Eq(3.2)
$$
As a matter of fact, we can construct a one-to-one correspondence
between the set of graphs in $\GG_n$ and the set $\Omega_n$ of
sequences $\{\s_{\{i,j\}}\}_{\{i,j\}\subset \{1,\dots ,n\}}$
with $\s_{\{i,j\}}~=~0,1$, with the rule that, $\s_{\{i,j\}}=1$ if
$\{i,j\}\in E_g$ and $\s_{\{i,j\}}=0$ if $\{i,j\}\notin E_g$. Hence a
graph $g\in \GG_n$ can be viewed as an ordered sequence of
$n(n-1)/2$  numbers which can be either 0 or 1. Thus the total
number of graphs equals the total number of such sequences which
is clearly $2^{n(n-1)/2}$.

\\It is now simple to get a bound for $B_n =|G_n|$. Consider the subset of
$\tilde\Omega_n\subset\Omega_n$ formed by sequences
$\{\s_{\{i,j\}}\}_{\{i,j\}\subset \{1,\dots ,n\}}$
such that
$\s_{\{1,2\}}=1 ,\s_{\{2,3\}}=1, \dots , \s_{\{n-1,n\}}=1$, hence
$n-1$ links are fixed while $n(n-1)/2 -(n-1) = (n-1)(n-2)/2$ are
arbitrary. Clearly, any graphs corresponding to a sequence in
$\tilde\Omega_n$ is connected by construction, since it contains
the tree $\t=\left\{\{{1,2\}},\{2,3\} ...,\{n-1,n\}\right\}$ and
$|\tilde\Omega_n|=2^{(n-1)(n-2)/2}$. Thus
$$
B_n \ge 2^{(n-1)(n-2)/2}
$$
\def\ba{\begin{array}}
\def\ea{\end{array}}  \def \eea {\end {eqnarray}}\def \bea {\begin {eqnarray}}

\\Such a bound shows that we have no hope to control the (absolute) convergence of
the Mayer series if we don't exploit cancellations hidden in the
factor
$$
\left|  \sum_{g\in G_{n}} \prod_{\{i,j\}\in E_g}\left[  e^{ -\b
V(\xx_i -\xx_j)} -1\right] \right|
$$
One may think to use the alternative expression \equ(Basuev) for the Ursell coefficients $\Phi^T(\xx_1,\dots,\xx_n)$ which looks more well behaved as a function of $n$ and also looks more suitable to
take advantage of stability. The problem is that
the expression \equ(Basuev), if bounded  naively, behaves badly in the volume $|\L|$. Using naively \equ(Basuev) we get the following bound for the absolute value of
$C_n(\b,\L)$.
$$
\begn
|C_n(\b,\L)| & \le {1\over n!}{1\over |\L|} \int_{\L}d\xx_1\dots \int_{\L}d\xx_{n} |\Phi^T(\xx_1,\dots,\xx_n)|\\
& \le{1\over |\L|}\sum_{k=1}^n (k-1)!\sum_{\{I_1, \dots, I_k\}\in \P_{n}}{1\over n!} \int_{\L}d\xx_1\dots \int_{\L}d\xx_{n} e^{-\b\sum_{\a=1}^k U(\xx_{I_\a})}
\egn
$$
and using then stability we get
$$
\begn
|C_n(\b,\L)| & \le {1\over |\L|} \sum_{k=1}^n (k-1)!\sum_{\{I_1, \dots, I_k\}\in \P_{n}}{1\over n!} \int_{\L}d\xx_1\dots \int_{\L}d\xx_{n}e^{+\b\sum_{\a=1}^kB|I_\a|}\\
&\le |\L|^{n-1}  e^{\b Bn}   \sum_{k=1}^n {k!\over n!}\sum_{\{I_1,I_2, \dots, I_k\}\in \P_{n}}1\\
 & = e^{+\b B}(|\L|e^{\b B})^{n-1}   \sum_{k=1}^n {k!\over n!}S(n,k)
\egn
$$
where the terms $S(n,k)$ are the Stirling numbers of the second kind\footnote{$S(n,k)$  is the number of ways to partition a set of $n$ objects into $k$ non-empty subsets.}.
It holds
$$
\begn
S(n,k)& =~\sum_{\{I_1,I_2, \dots, I_k\}\in \P_{n}}1\\
&=~{1\over k!}\sum_{m_1, \dots ,m_k:~\, m_i\ge 1\atop m_1+\dots
m_k=n}{n!\over m_1!\dots m_k!}\\
&=~{n!\over k!(n-k)!}\sum_{s_1, \dots ,s_k:~\, s_i\ge 0\atop s_1+\dots
s_k=n-k}{(n-k)!\over (s_1+1)!\dots (s_k+1)!}\\
&\le ~{n\choose k} \sum_{s_1, \dots ,s_k:~\, s_i\ge 0\atop s_1+\dots
s_k=n-k}{(n-k)!\over s_1!\dots s_k!}\\
&=~ {n\choose k}k^{n-k}
\egn
$$
Therefore we obtain
$$
\begn
|C_n(\b,\L)| & \le {e^{\b B}}{(|\L|e^{\b B})^{n-1}}  \sum_{k=1}^n {k!\over n!}{n\choose k}k^{n-k}\\
&=~{e^{\b B}}{(|\L|e^{\b B})^{n-1}}   \sum_{k=1}^n {k^{n-k}\over (n-k)!}\\
&=~{e^{\b B}}{(|\L|e^{\b B})^{n-1}}\sum_{s=0}^{n-1}{(n-s)^{s}\over s!}\\
&\le~ {e^{\b B}}{(|\L|e^{\b B})^{n-1}}\sum_{s=0}^{n-1}{{n}^{s}\over s!}\\
&\le~ {e^{\b B}}{(|\L|e^{\b B})^{n-1}}\sum_{s=0}^{\infty}{{n}^{s}\over s!}\\
&= ~ |\L|^{n-1}{(e^{+\b B+1})^n }
\egn
$$
This bound for $|C_n(\b,\L)| $ has a well  combinational behavior in $n$, but it  grows as $|\L|^n$ in the volume $\L$ and hence the (lower) bound for the convergence radius of the Mayer series obtained
from this bound shrinks to zero as $\L\to \infty$.

\section{Convergence of the Mayer series}
The best rigorous upper bound on $|C_n(\b,\L)|$ until recently (and hence the best lower bound on the convergence radius
of the Mayer series)  for
stable and tempered pair  potentials was that
obtained by Penrose and Ruelle in 1963 \cite{Pen63,Ru63}.

\begin{teo}[Penrose-Ruelle]\label{peru}
 Let $V$ be a stable and tempered pair  potential with  stability constant $B$. Then
the  $n$-order Mayer  coefficient $C_n(\b,\L)$ defined in \equ(cbeta2)
is bounded by
$$
|C_n(\b,\L)|\le  e^{2\b B (n-2)}n^{n-2} {[C(\b)]^{n-1}\over n!}\Eq(PerRu)
$$
where
$$
C(\b)=\int_{\mathbb{R}^{d}} dx ~ |e^{-\b V(x)}-1|\Eq(cbetapr)
$$
Therefore the Mayer series \equ(pressm) converges absolutely, uniformly in $\L$,
 for any complex  $\l$ inside the disk
$$
|\l| <{1\over e^{2\b B+1} C(\b)}\Eq(RadPR)
$$
I.e. the convergence radius of the Mayer series \equ(pressm) admits the following lower bound
$$
R_{V}\ge {1\over e^{2\b B+1} C(\b)}\Eq(convrad)
$$
\end{teo}
The estimate \equ(PerRu) leading to the lower bound \equ(convrad) was obtained by attacking the problem with a rather
indirect approach. Namely Penrose and Ruelle  looked at an infinite class of functions of the systems,
the so called  correlation functions, and showed that they can be expressed via an absolute convergent
expansion. The next section is entirely   devoted to the proof of Theorem \ref{peru}.

\subsection{Kirkwood-Salsburg equations:  the original proof of Theorem \ref{peru} (according to Penrose)}
We are considering a system of particles enclosed in a box $\L$ and interacting via a pair potential
$V(x-y)$ stable and tempered.
We thus define the correlation functions of the system as follows. For any $n\ge 1$
$$
\r_n(\xx_1, \dots \xx_{n};\l)={1\over \Xi_\L(\b,\l)}\sum_{m=0}^\infty {\l^{n+m}\over m!}\int_\L d\yy_1\dots\int_\L d\yy_m\,
e^{-\b U(\xx_1,\dots,\xx_n,\yy_1,\dots,\yy_m)}\Eq(corr)
$$
with the convention that
$$
U(\xx_1)=0\Eq(conv0)
$$
and for $n=0$
$$
\r_0(\emptyset;\l)=1\Eq(conv)
$$
The numbers  $\r_n(\xx_1, \dots \xx_{n};\l)$ represent the probability density of finding
in the system $n$ particles at positions $\xx_1, \dots \xx_{n}$ irrespective of where the other particle are.
It is important to observe that the functions  $\r_n(\xx_1, \dots \xx_{n};\l)$ are by construction symmetric
under permutation of positions  $\xx_1, \dots \xx_{n}$. I.e. if $\s:\{1,2,\dots,n\}\to \{1,2,\dots,n\}:i\mapsto \s(i)$
is one-to-one (i.e. a permutation) then
$$
\r_n(\xx_1, \dots, \xx_{n};\l)=\r_n(\xx_{\s(1)}, \dots ,\xx_{\s(n)};\l)\Eq(sym)
$$

\\Observe also that, due to stability, the numerator and the denominator of \equ(corr) are holomorphic functions
of $\l$ (see proposition \ref{Zolo}), hence $\r_n(\xx_1, \dots \xx_{n};\l)$ is a meromorphic function of $\l$.
The series expansion $\r(\xx_1, \dots \xx_{n};\l)$ in powers of $\l$ around $\l=0$  has convergence radius at least
equal to the convergence radius of $\log\Xi_\L(\b,\l)$ (i.e. where in the region where the denominator
$\Xi_\L(\b,\l)$ in l.h.s of \equ(corr)
is free of zeros). Let us write this series expansion as
$$
\r_n(\xx_1, \dots, \xx_{n};\l)= \sum_{\ell=0} \r_{n,\ell\,}(\xx_1, \dots, \xx_{n})\l^{n+\ell}\Eq(expan)
$$
Observe that if we consider the one-point correlation function $\r_1(\xx_1)$, it is easy to see that
the coefficients $\r_{1,\ell}(x_1)$ of its series expansion in power of $\l$,  integrated over the volume $\L$, are related with
the Ursell coefficient of the Mayer expansion for $\log\Xi_\L(\b,\l)$. As a matter of fact,
since
$$
\r_1(\xx_1;\l)= {1\over \Xi_\L(\b,\l)}\sum_{m=0}^\infty {\l^{1+m}\over m!}\int_\L d\yy_1\dots\int_\L d\yy_m\,
e^{-\b U(\xx_1,\yy_1,\dots,\yy_m)}\Eq(relru)
$$
it is easy to check that
$$
\int_{\L} \r_1(\xx_1;\l)d\xx_1=\l\,\:{d(\log\Xi_\L(\b,\l))\over d\l}
$$
Thus
$$
\begn
{1\over |\L|}\int_{\L} \r_{1,\ell-1\,}(\xx_1)d\xx_1& = {1\over (\ell-1)!}{1\over |\L|}\int_{\L}d\xx_1
\dots \int_{\L} d\xx_\ell~ \Phi^T(\xx_1,\dots,\xx_\ell)\\
& =\ell C_l(\b,\L)\egn
\Eq(impo)
$$
One can now calculate easily the first term in the expansion \equ(expan) by  division
in \equ(corr). As a matter of fact, since
$\Xi_\L(\b,\l)= 1+ O(\l)$, then also $\Xi^{-1}_\L(\b,\l)= 1+ O(\l)$ and the first term
of \equ(corr) is thus the first term of the numerator. Therefore
$$
\r_{n,0\,}(\xx_1, \dots \xx_{n})=e^{-\b U(\xx_1.\dots,\xx_n)}\Eq(ron0)
$$
Moreover observe that, by  the convention \equ(conv), we also have
$$
\r_{0,\ell\,}(\0)=\d_{0,\ell}\Eq(conv2)
$$
where $\d_{ij}$ is the Kronecker symbol, i.e $\d_{ij}=1$ if $i=j$ and $\d_{ij}=0$ if $i\neq j$. Let us now  find recurrence relations for the higher coefficients
$\r_{n,\ell\,}(\xx_1, \dots \xx_{n})$ for $n\ge 1$. It is easy to see that the functions $\r_n(\xx_1, \dots \xx_{n};\l)$
satisfies a set of integral functions named the Kirkwood-Salsburg equations. These equations
are obtained first by decomposing the energy $U(\xx_1,\dots,\xx_n,\yy_1,\dots,\yy_m)$ appearing in l.h.s. of
\equ(corr) as
$$
U(\xx_1,\dots,\xx_n,\yy_1,\dots,\yy_m)=W(\xx_1;\xx_2,\dots,\xx_n)+~~~~~~~~~~~~~~~~~~~~~~~~~~~~~~~~~~~~~~~~~~~~~~~~~~~~~~~~~
$$
$$~~~~~~~~~~~~~~~~~~+\sum_{j=1}^mV(\xx_1-\yy_j)+ U(\xx_2,\dots,\xx_n,\yy_1,\dots,\yy_m)\Eq(dec)
$$
where
$$
W(\xx_1;\xx_2,\dots,\xx_n)=\sum_{i=2}^n V(\xx_1-\xx_i)
$$

\\Putting \equ(dec) into \equ(corr) one gets
$$
\begn
\r_n(\xx_1, \dots \xx_{n};\l) & =
{\l e^{-\b W(\xx_1;\xx_2,\dots,\xx_n)}\over \Xi_\L(\b,\l)}
\sum_{m=0}^\infty {\l^{n-1+m}\over m!}\int_\L d\yy_1\dots\int_\L d\yy_m\\
&~~~\times e^{-\b (\sum_{j=1}^mV(\xx_1-\yy_j)+ U(\xx_2,\dots,\xx_n,\yy_1,\dots,\yy_m))}\\
&= ~{\l e^{-\b W(\xx_1;\xx_2,\dots,\xx_n)}\over \Xi_\L(\b,\l)}
\sum_{m=0}^\infty {\l^{n-1+m}\over m!}\int_\L d\yy_1\dots\int_\L d\yy_m\\
&~~~\times \prod_{j=1}^m \Big[\Big(e^{-\b V(\xx_1-\yy_j)}-1\Big)+1\Big] ~
e^{-\b U(\xx_2,\dots,\xx_n,\yy_1,\dots,\yy_m)}\\
& = ~{\l e^{-\b W(\xx_1;\xx_2,\dots,\xx_n)}\over \Xi_\L(\b,\l)}
\sum_{m=0}^\infty {\l^{n-1+m}\over m!}\int_\L d\yy_1\dots\int_\L d\yy_m\\
&~~~\times \sum_{s=0}^m\sum_{j_1,\dots j_s\atop
1\le j_1<\dots<j_s\le m}
\prod_{k=1}^s(e^{-\b V(\xx_1-\yy_{j_k})}-1) e^{-\b U(\xx_2,\dots,\xx_n,\yy_1,\dots,\yy_m)}
\egn
$$

\\Here above the term with $s=0$ correspond to 1.
Since the variables $\yy$ are dummy indices, for any term of the sum over indices
$j_1,\dots j_s$ we can relabel the variables $\yy_1, \dots \yy_m$ in such way that
$\yy_{j_1}=\yy_1,\dots, \yy_{j_s}=\yy_s$ and keeping in account that
$$
\sum_{j_1,\dots j_s\atop
1\le j_1<j_2<\dots<j_s\le m}1={m\choose s}={m!\over s!(m-s)!}
$$
and the fact that $U(z_1,\dots z_p)$ is symmetric under permutations of its variables, we obtain
$$
\r_n(\xx_1, \dots \xx_{n};\l)= {\l e^{-\b W(\xx_1;\xx_2,\dots,\xx_n)}\over \Xi_\L(\b,\l)}
\sum_{m=0}^\infty {\l^{n-1+m}}\int_\L d\yy_1\dots\int_\L d\yy_m\,
$$
$$
~~~~~~~~~~~~~~~~~~~~~~~\sum_{s=0}^m\,{1\over s!(m-s)!}
\prod_{k=1}^s(e^{-\b V(\xx_1-\yy_{k})}-1) e^{-\b U(\xx_2,\dots,\xx_n,\yy_1,\dots,\yy_m)}
$$
We can interchange the sum over $m$ and $s$  by observing that
$$
\sum_{m=0}^\infty\sum_{s=0}^m \dots=\sum_{s=0}^\infty\sum_{m=s}^\infty \dots
$$
and calling $t=m-s$. Hence we get, after a suitable renomination of the dummy indices $\yy$
$$
\r_n(\xx_1, \dots \xx_{n};\l)= {\l e^{-\b W(\xx_1;\xx_2,\dots,\xx_n)}\over \Xi_\L(\b,\l)}
\sum_{s=0}^\infty {1\over s!}\int_\L d\yy_1\dots\int_\L d\yy_s\,\prod_{k=1}^s(e^{-\b V(\xx_1-\yy_{k})}-1)
$$
$$
~~~~~~~~~~~~~~~~~~~~~~~\sum_{t=0}^\infty\,{{\l^{n-1+s+t}}\over t!} \int_\L d\yy'_1\dots\int_\L d\yy'_t
 e^{-\b U(\xx_2,\dots,\xx_n,\yy_1,\dots,\yy_s,\yy'_1,\dots,\yy'_t)}
$$
i.e.
$$
\r_n(\xx_1, \dots \xx_{n};\l)= {\l e^{-\b W(\xx_1;\xx_2,\dots,\xx_n)}}
\sum_{s=0}^\infty {1\over s!}\int_\L d\yy_1\dots\int_\L d\yy_s\,\prod_{k=1}^s(e^{-\b V(\xx_1-\yy_{k})}-1)
$$
$$
\r_{n-1+s}(\xx_2, \dots \xx_{n},\yy_1,\dots,\yy_s;\l)\Eq(K-S)
$$
The latter are the Kirkwood-Salsburg equations.
Substituting now \equ(expan) into \equ(K-S) we get
$$
\sum_{k=0}^\infty \r_{n,k\,}\!(\xx_1, \dots \xx_{n})\l^{k}= {e^{-\b W(\xx_1;\xx_2,\dots,\xx_n)}}\!\!
\sum_{s=0}^\infty {1\over s!}\!\!\int_\L \!\!d\yy_1\dots\!\!\!\int_\L \!\!d\yy_s\!\prod_{k=1}^s(e^{-\b V(\xx_1-\yy_{k})}-1)
$$
$$
\sum_{r=0}^\infty \r_{n-1+s,r\,}(\xx_2, \dots \xx_{n},\yy_1,\dots,\yy_s)\l^{s+r}
$$
and equating coefficients with the same power, say  $k$, of $\l$
$$
\r_{n,k\,}(\xx_1, \dots \xx_{n})= {e^{-\b W(\xx_1;\xx_2,\dots,\xx_n)}}\sum_{s=0}^k {1\over s!}
\int_\L d\yy_1\dots\int_\L d\yy_s\,\prod_{k=1}^s(e^{-\b V(\xx_1-\yy_{k})}-1)
$$
$$
\r_{n-1+s,k-s\,}(\xx_2, \dots \xx_{n},\yy_1,\dots,\yy_s)\Eq(KS2)
$$
\vv
\\Recalling \equ(conv2), this equation holds for all integers $n\ge 1$ and $k\ge 0$. Note that for $n=1$ the formula \equ(KS2) degenerates into
$$
\r_{1,k\,}(\xx_1)= \sum_{s=0}^k {1\over s!}
\int_\L d\yy_1\dots\int_\L d\yy_s\,\prod_{k=1}^s(e^{-\b V(\xx_1-\yy_{k})}-1)
\r_{s,k-s\,}(\yy_1,\dots,\yy_s)\Eq(KS2b)
$$

\\To estimate the convergence radius of the correlation functions
$\r_{n}(\xx_1, \dots \xx_{n};\l)$
and hence of the $\log\Xi_\L(\b,\l)$, we have to calculate efficient upper bounds
on the coefficients $\r_{n,k\,}(\xx_1, \dots \xx_{n})$ of the expansion of the correlations.

\\Looking  to the structure of the equation \equ(KS2) one can observe the following.
Given a pair of indices $(n,k)$,  let  $M=n+k$, then  \equ(KS2) says that
 $\r_{n,k}$ is a function of the $\r_{i,j}$ such that $i+j=M-1$.
Therefore it is absolutely natural
to look for a bound on $\r_{n,k}$ by induction on $n+k$. Namely, we look for
$$
|\r_{n,M-n\,}(\xx_1, \dots \xx_{n})|\le K_{n,M-n}~~~~~~(n=1,2,\dots,M) \Eq(indu)
$$
and  make the  induction on  the integer $M$.
By \equ(ron0) and \equ(conv0)
$$
K_{1,0}=1\Eq(K10)
$$
Therefore inequality \equ(indu) is satisfied for $M=1$ with $|\r_{1,0}(\xx_1)|\le K_{1,0}= 1$.
Now, for $M>1$, we proceed  by induction on $M$. Assume that \equ(indu) is true $ M-1$ and for all $n=1,\dots, M-1$, then, by \equ(KS2)
$$
|\r_{n,M-n\,}(\xx_1, \dots \xx_{n})|\le e^{-\b W(\xx_1;\xx_2,\dots,\xx_n)}\sum_{s=0}^{M-n} {1\over s!}
[C(\b)]^s
\,K_{n-1+s,M-n-s\,}\Eq(KS3)
$$
where
$$
C(\b)=\int_{\mathbb{R}^2} \Big|e^{-\b V(\xx)}-1\Big|\,d\xx \Eq(cbeta)
$$
In order to bound the factor $\exp\{-\b W(\xx_1;\xx_2,\dots,\xx_n)\}$ in \equ (KS3),
we explicitly make use of the symmetry of the function $\r_{n,M-n\,}(\xx_1, \dots \xx_{n})$
under permutation of $\xx_1,\dots,\xx_n$.

Let  $i\in \{1,2,\dots,n\}$ and  let $\s_{1\leftrightarrows i}$ be
a permutation in
$\{1,2,\dots,n\}$ such that
$\s_{1\leftrightarrows i}(1)=i$, $\s_{1\leftrightarrows i}(i)=1$ and $\s_{1\leftrightarrows i}(k)=k$
for all $k\neq 1,i$ (i.e. $\s_{1\leftrightarrows i}$ exchange only $i$ with 1). Then, by stability
for at least one $j\in \{1,2,\dots,n\}$, it holds
$$
W(\xx_{\s_{1\leftrightarrows j}(1)};\xx_{\s_{1\leftrightarrows j}(2)},\dots,\xx_{\s_{1\leftrightarrows j}(n)})\ge - 2B
\Eq(stab2)
$$
Indeed, by stability we have that
$$
\sum_{i=1}^n W(\xx_{\s_{1\leftrightarrows i}(1)};\xx_{\s_{1\leftrightarrows i}(2)},\dots,\xx_{\s_{1\leftrightarrows i}(n)})=
2U(\xx_1,\dots,\xx_n)\ge -2n B
$$
and this immediately implies \equ(stab2) for at least one $j$. Thus, choosing $j\in [n]$ such that the permutation
$\s_{1\leftrightarrows j}$ satisfies \equ(stab2), we have
$$
|\r_{n,M-n\,}(\xx_1, \dots \xx_{n})|=
|\r_{n,M-n\,}(\xx_{\s_{1\leftrightarrows j}(1)}, \dots \xx_{\s_{1\leftrightarrows j}(n)})|
\le
$$
$$
\le~e^{2\b B}\sum_{s=0}^{M-n} {1\over s!}
[C(\b)]^s
\,K_{n-1+s,M-n-s\,}\Eq(KS32)
$$

\\Hence \equ(indu) holds also for $M$ provided that, for $n=1,\dots,M$
$$
K_{n, M-n}\ge\,e^{2\b B}\sum_{s=0}^{M-n} {1\over s!}
[C(\b)]^s
\,K_{n-1+s,M-n-s\,}\Eq(KS4)
$$
In conclusion we have proved by induction that \equ(indu) holds
for all $M\ge 1$ if  \equ(K10) holds and \equ(KS4) for $M>1$.
To find the best value for $K_{n,M-n}$
let us solve the  set of equations

$$
K_{n, M-n}=\,e^{2\b B}\sum_{s=0}^{M-n} {1\over s!}
[C(\b)]^s
\,K_{n-1+s,M-n-s\,}\Eq(rec)
$$
%with the ``initial condition''
%$$
%K_{0,\ell}=\d_{0,\ell} \Eq(aaz)
%$$
These equations are  recursive. I.e., for a fixed integer $M$ all coefficients $K_{n, M-n}$ with $n=1,\dots,M$ are functions of
coefficients $K_{n,M-1-n}$ with $n=1,\dots M-1$. So the initial condition \equ(K10) plus the equations \equ(rec)
determines uniquely all coefficients $K_{n, M-n}$.
It is worth  to check that its  solution is
$$
K_{n,\ell}=e^{2\b B (n+\ell-1)}n(n+\ell)^{\ell-1} {[C(\b)]^\ell\over \ell!}\Eq(soluz)
$$
In conclusion we have that the coefficients of the power series in $\l$ given by \equ(expan) can be bounded as
$$
|\r_{n,\ell\,}(\xx_1, \dots \xx_{n})|\le e^{2\b B (n+\ell-1)}n(n+\ell)^{\ell-1} {[C(\b)]^\ell\over \ell!}\Eq(OKK)
$$

\\This implies that \equ(expan) has convergence radius at least
$$
R\ge  {1\over e^{2\b B +1} C(\b)}
$$

\\As a matter of fact, the n-point correalation $|\r_{n}(\xx_1, \dots \xx_{n})|$, see \equ(expan) and
\equ(OKK) is less or equal to
$$
\begn
|\r_{n}(\xx_1, \dots, \xx_{n};\l)|& \le\sum_{\ell=0}^\infty
e^{2\b B (n+\ell-1)}n(n+\ell)^{\ell-1} {[C(\b)]^\ell\over \ell!}\l^{n+\ell}\\
&\le\sum_{\ell=0}^\infty
e^{2\b B (n+\ell-1)}(n+\ell)^{\ell}~ {[C(\b)]^\ell\over \ell!}\l^{n+\ell}\\
&=e^{2\b B (n-1)}\l^n
\sum_{\ell=0}^\infty
\Big(1+{n\over \ell}\Big)^{\ell} ~\,{\ell^{\ell} \over \ell!}~\Big [e^{2\b B }C(\b)\l\Big ]^\ell\\
&\le
e^{2\b B (n-1)}\l^n \sum_{\ell=0}^\infty e^n\,e^\ell \Big [e^{2\b B }C(\b)\l\Big ]^\ell\\
&\le
e^{(2\b B+1)n}\l^n \sum_{\ell=0}^\infty\Big [e^{2\b B +1}C(\b)\l\Big ]^\ell
\egn
$$
It is also interesting to calculate the bound we obtain for the Ursell coefficients
of the Mayer series. Observe first that in place of \equ(KS3) we have, recalling the special case
$n=1$ \equ(KS2b),
$$
|\r_{1,M-1}(\xx_1)|\le \sum_{s=0}^{M-1} {1\over s!}
[C(\b)]^s
\,K_{s,M-1-s\,}\Eq(KS33)
$$
(i.e.the factor $e^{-\b W(\xx_1;\xx_2,\dots,\xx_n)}$ is not present in the $n=1$ case).
So, using \equ(rec) and \equ(soluz) we get
$$
|\r_{1,\ell}(\xx_1)|\le e^{-2\b B}K_{1,\ell}= e^{2\b B (\ell-1)}(1+\ell)^{\ell-1} {[C(\b)]^\ell\over \ell!}
$$
Hence, recalling the definitions \equ(pressm) and \equ(ursm) and using \equ(impo), we get

$$
\begn
n|C_n(\b,\L)|&\le~ \left|n {\l^n\over n!}{1\over |\L|}\int_{\L}d\xx_1
\dots \int_{\L} d\xx_n~ \Phi^T(\xx_1,\dots,\xx_n)\right|\\
& =~ \Big|{1\over |\L|}\int_{\L} \r_{1,n-1\,}(\xx_1)d\xx_1\Big|\\
&\le~
e^{2\b B (n-2)}(n)^{n-2} {[C(\b)]^{n-1}\over (n-1)!}
\egn$$
I.e. we obtain the bound \equ(PerRu)

$$
|C_n(\b,\L)|\le e^{2\b B (n-2)}n^{n-2} {[C(\b)]^{n-1}\over n!}\Eq(bmaru)
$$
which completes the proof of Theorem \ref{peru}.

\\We stress once again that this bound is valid for any stable pair potential $V(\xx)$ such that
$C(\b)$ defined in \equ(cbeta) is finite. The bound is very efficient, but
it is indeed obtained in a rather involved and indirect way. We will now show
an alternative way to get the same bound \equ(bmaru) in a much more direct way,
i.e. obtaining directly a bound for the absolute values of the Ursell coefficients $|C_n(\b,\L)|$
starting from their explicit expression \equ(ursm).

\section{The Penrose Tree graph Identity}\label{secpen}
%
%We first show a  somehow forgotten but very deep tree identity obtained
%by Oliver  Penrose~\cite{Pe} in 1967.

The
%first
 tree-graph identity that we present was proposed by Penrose \cite{Pen67} in 1967 and
it was based on  the existence of a map from the set $T_n$ of the  trees with vertex set $[n]$ to the set  $G_n$ of the connected graphs with vertex set $[n]$ inducing a so-called {\it partition scheme} in  $G_n$.
This tree graph identity allows us to rewrite the Ursell coefficient defined in \equ(urse), whose expression we recall
$$
\Phi^T(\xx_1,\dots,\xx_n)~=~\sum\limits_{g\in G_{n}}~
\prod\limits_{\{i,j\}\in E_g}\left[  e^{ -\b V(\xx_i -\xx_j)} -1\right]
\Eq(urse2)
$$
in terms of a sum over trees rather than over connected graphs.
\begin{defi}\label{partschem} A map $\mathfrak{M}: T_n\to G_n$ is called a partition scheme in the set of the connected graphs $G_n$  if, for all $\tau\in T_n$,
$\tau\subset \mathfrak{M}(\t)$
and $G_n=\biguplus_{\tau\in {\mathcal T}}[\tau,\mathfrak{M}(\tau)]$
where $\biguplus$ means disjoint union and $[\tau,\mathfrak{M}(\tau)]=\{g\in G_n: \tau\subset g\subset \mathfrak{M}(\tau)\}$
 is a boolean interval (with respect to the set-inclusion).
\end{defi}

\\Once a  partition scheme in $G_n$ has been given, we have the following identity

\begin{teo}[General Penrose identity]\label{Penid} Let $n\ge 2$ and let $\{V_{ij}\}_{\{i,j\}\in E_n}$ be $n(n-1)/2$ real numbers
(one for each unordered pair $\{i,j\}\subset[n]$) taking values in $\mathbb{R}\cup\{+\infty\}$ and let
$$
\Phi^T(n)=\sum\limits_{g\in G_{n}}~
\prod\limits_{\{i,j\}\in E_g}\left[  e^{ - V_{ij}} -1\right]
$$
Let  $\mathfrak{M}: T_n\to G_n$ be a partition scheme in $G_n$.
Then the following identity  holds
$$
\Phi^T(n)~=~
\sum_{\t\in T_n}
e^{-\sum\limits_{\{i,j\}\in E_{\mathfrak{M}(\t)}\backslash E_\t}V_{ij}}
\prod_{\{i,j\}\in E_\t}\left(e^{- V_{ij}}-1\right)
\Eq(r.200)
$$
\end{teo}
\def\Ti{\mathfrak{T}}
\def\Mi{\mathfrak{M}}

\\{\bf Proof}. Since
$G_n$ is the disjoint union $G_n=\biguplus_{\tau\in {T}_n}[\tau,\mathfrak{M}(\tau)]$ we can write
$$
\begn
\Phi^T(n)& = \sum_{\t\in T_n} \sum_{g\in  [\tau,\mathfrak{M}(\tau)]} \prod_{\{i,j\}\in E_g}\left(e^{- V_{ij}}-1\right)\\
& =\sum_{\t\in T_n} \prod_{\{i,j\}\in E_\t}\left(e^{- V_{ij}}-1\right)\sum_{g\in [\tau,\mathfrak{M}(\tau)]}
 \prod_{\{i,j\}\in E_g\setminus E_\t}\left(e^{- V_{ij}}-1\right)\\
&= \sum_{\t\in T_n} \prod_{\{i,j\}\in E_\t}\left(e^{- V_{ij}}-1\right)\sum_{g\in G_n\atop E_\t\subset E_g\subset E_{\mathfrak{M}(\t)}}
 \prod_{\{i,j\}\in E_g\setminus E_\t}\left(e^{- V_{ij}}-1\right)\\
&=\sum_{\t\in T_n} \prod_{\{i,j\}\in E_\t}\left(e^{- V_{ij}}-1\right)\sum_{E\subset E_{\mathfrak{M}(\t)}\setminus E_\t}
 \prod_{\{i,j\}\in E}\left(e^{- V_{ij}}-1\right)\\
&= \sum_{\t\in T_n} \prod_{\{i,j\}\in E_\t}\left(e^{- V_{ij}}-1\right)
 \prod_{\{i,j\}\in E_{\mathfrak{M}(\t)}\setminus E_\t}\left[\left(e^{- V_{ij}}-1\right)+1\right]\\
& =\sum_{\t\in T_n} \prod_{\{i,j\}\in E_\t}\left(e^{- V_{ij}}-1\right)
 \prod_{\{i,j\}\in E_{\Mi(\t)}\setminus E_\t} e^{- V_{ij}}
 \egn
$$
which concludes the proof.
$\Box$

\vskip.2cm
\\In general  it is not so simple to check whether a given map $\mathfrak{M}:T_n\to G_n$ is a partition scheme. The proposition below
can be useful.

\begin{pro}\label{twomaps}
The following statements are equivalent.
\begin{itemize}
\item[1.]
There are  two maps $$\xymatrix{{\mathcal G}_n \ar@<.5ex>[r]^\Ti & {\mathcal T}_n \ar@<.5ex>[l]^\Mi}$$
such that $\Ti^{-1}(\tau)=\{g \in {\mathcal G}_n:\, \tau \subseteq g \subseteq \Mi(\tau)\}$ for every $\tau \in {\mathcal T}_n$.
\item[2.]
 $\Mi$ is a partition scheme in ${\mathcal G}_n$.
 \end{itemize}
\end{pro}

\\{\it Proof}. $1\Rightarrow2$.
Since $g \in \Ti^{-1}(\Ti(g))$, we have $\Ti(g) \subset g$ for all $g \in {\mathcal G}_n$. In particular,
for every tree $\tau$ we have $\Ti(\tau) \subset \tau$ which implies
$\Ti(\tau)=\tau$ because both are trees.  I.e.,  $\Ti$ is  surjective and thus the  intervals $\Ti^{-1}(\tau)$ are nonempty. This implies that
 ${\mathcal G}_n$ is the disjoint union of the intervals
 ${\mathcal G}_n = \bigcup_{\tau \in {\mathcal T}_n} \Ti^{-1}(\tau)=\bigcup_{\tau \in {\mathcal T}_n}[\t, \Mi(\t)]$.
 Hence in view
 of Definition \ref{partschem} we conclude  that  $\Mi$ is a partition scheme in ${\mathcal G}_n$.

 \\$2\Rightarrow1$. If
  $\Mi$ is a partition scheme then for any $g\in {\mathcal G}_n$ there exists a unique tree $\t\in {\mathcal T}_n$
  such that  $g\in [\t,\bm M(\t)]$. Therefore we can define the map $\bm T$ from ${\mathcal G}_n$ to ${\mathcal T}_n$
  such that, for all $g\in [\t,\bm M(\t)]$,  $\bm T(g)=\t$.
  $\Box$

\subsection{The original Penrose map}

\def\ti{{\rm\bf  t}}\def\mi{{\rm\bf m}}
\\ The original partition scheme proposed by Penrose was exclusively based on labels $1,\dots ,n$ (rather than $x_1,\dots,x_n$)
 and it  involves two explicit maps, say $\ti: G_n\to T_n$ and $\mi: T_n\to G_n$ satisfying
Proposition \ref{twomaps}. Let us first construct the map
$\ti : G_n\to T_n$.
To define this map we have first of all to choose a root among vertices ${1,2,\dots,n}$. So we
identify for example a vertex among $\{1,2,\dots,n\}$ as the root, e.g., to fix the ideas, let the root
be the vertex $1$ (as in the original paper of Penrose).
Once the root $1$ has been chosen,
let us denote, for any $g\in G_n$,
by $d_g(i)$ the graph distance of the vertex $i$ from the root $1$ in $g$.
Given thus $g\in G_n$, we construct the tree $\ti(g)$
as follows.
\vv
\vv
\vv
\\1) We first delete all edges $\{i,j\}$ in $E_g$
with $d_g(i)=d_g(j)$.
\vv
\\After this operation we are left
with  a connected graph $g'$ such that $d_{g'}(i)=d_{g}(i)$ for all vertices $i=1,\dots,n$.
Moreover,
each edge $\{i,j\}$ of $g'$ is such that $|d_{g'}(i)-d_{g'}(j)|=1$.
\vv
\\2) For any $i\neq 1$ let now delete from the graph $g'$ all edges $\{i,j\}$
in $E_{g'}$  such that $d_{g'}(j)=d_{g'}(i)-1$ except the one with $j$ minimal.
\vv
\\The resulting graph $g''\doteq \ti(g)$ is by construction a connected graph in $G_n$, i.e.
$\ti(g)\in G_n$, which is a subgraph of $g$, i.e. $\ti(g)\subset g$, and which has
no cycles, i.e. $\ti(g)\in T_n$. Observe that the map $\ti$ is a surjection from $G_n$ to $T_n$ (because $\ti(\t)=\t$ for all $\t\in T_n$).
\vv
\\In order to define the   map $\mi:T_n\to G_n$, i.e. the Penrose partition scheme, let us introduce some notations regarding rooted trees..
First, observe that if $\t\in T_n$
is thought as rooted in $1$,
then vertices of $\t\in T_n$ may be thought as {\it partially ordered} so that
{each vertex  {$i$} of $\t$} has a unique parent, which we denote by {$i'$}, and
{$s_i$} children, denoted by  {$i^1,\dots,i^{s_i}$}. The number {$s_i$} is also called the  branching factor of $i$.
Of course, the root has no parent. For any $\t\in T_n$ we also denote by {$d_\t(i)$}
the tree distance of the vertex $i$ from the root $1$ ($d_\t(i)$ is also called the
{ \it generation number} of the vertex $i$ in $\t$). If $i$ is a vertex of $\t$
such that { $s_i=0$} (i.e. $i$ has no children) then $i$ is called a {\it leaf} (or end-point) of $\t$.

\begin{defi}[Penrose partition scheme]\label{penspart} The Penrose partition scheme is
the map $\mi:T_n\to G_n$ such that to each tree $\tau\in T_n$ associates the graph
$\mi(\tau)\in G_n$ formed by adding  to $\tau$
all edges $\{i,j\}$ such that either:
\begin{itemize}
\item[(p1)]  $d_\t(i)=d_\t(j)$ (edges between vertices of the same generation), or
\item[(p2)]  $d_\t(j)=d_\t(i)-1$ and $j>i'$ (edges between vertices with generations differing by one).
\end{itemize}
\end{defi}

\\It is clear, by construction, that, for any $\t\in T_n$,    we have  $\ti(\mi(\t))=\t$ and if $g\in G_n$ such that $\ti(g)=\t$ then $g\in [\t,\mi(\t)]$. In other words,
the maps $\mi$ and $\ti$ satisfy Proposition \ref{twomaps} and the map $\mi$ is a partition scheme.

\medskip

\\The original  \emph{Penrose identity \cite{Pen67}}, involving the explicit map $\mi$, is the following
\begin{teo}\label{Penid2}
Let $V$ be a pair potential and let $\mi$ be the map described in
Definition \ref{penspart}. Then the following identity holds.
$$
\Phi^T(\xx_1,\dots,\xx_n)~=~
\sum_{\t\in T_n}
e^{-\b\sum_{\{i,j\}\in E_{\mi(\t)}\backslash E_\t}V(x_i-x_j)}
\prod_{\{i,j\}\in E_\t}\left(e^{- \b V(x_i-x_j)}-1\right)
\Eq(r.20)
$$
\end{teo}
The problem with this identity is that,  supposing stability of the potential $V(x)$, it is in general very hard (if not impossible) to efficiently estimate
the factor $-\sum_{\{i,j\}\in E_{\mi(\t)}\backslash E_\t}V(x_i-x_j)$ using stability. In other words, let us consider the following conjecture
\begin{conj}\label{stabpen}
Let $\mi$ be the map described  in Definition \ref{penspart}. Let $V(x)$ be a stable pair potential with stability constant $B$. Then there exists a constant
$\tilde B$ such that
for any $n\ge 2$, for any $\t\in T_n$ and for any $(x_1,\dots,x_n)\in \mathbb{R}^{dn}$ it holds
$$
\sum_{\{i,j\}\in E_{\mi(\t)}\backslash E_\t}V(x_i-x_j)\ge -\tilde B n
$$
\end{conj}
If the conjecture above is true, then we would immediately get from \equ(r.20) the inequality
$$
|\Phi^T(\xx_1,\dots,\xx_n)|~\le~ {e^{\b\tilde B n}}\sum_{\t\in T_n}
\prod_{\{i,j\}\in \t}\left|e^{- \b V(x_i-x_j)}-1\right|
$$
and thus, recalling \equ(cbeta2), the Mayer coefficients  would be bounded as follows
$$
C_n(\b,\L)\le {1\over |\L|}{ e^{\b\tilde B n}\over n!} \sum_{\t\in T_n} \int_\L dx_1\dots\int_\L dx_n \prod_{\{i,j\}\in E_\t}\left|e^{-\b V(x_i-x_j)}-1\right|\le
$$
$$
\le  { e^{\b\tilde B n}\over n!}\left[ \int_{\mathbb{R}^{d}} dx \left|e^{- \b V(x)}-1\right| \right]^{n-1} \sum_{\t\in T_n}1
~\le~ { e^{\b\tilde B n}\over n!}C(\b)^{n-1}n^{n-2}
$$
The latter estimates yield a lower bound for the convergence radius $R_V$ of the Mayer series as
$$
R_V\ge {1\over e^{\b \tilde B+1}C(\b)}
$$
which would have been better than the Penrose-Ruelle bound \equ(convrad) provided $\tilde B\le 2B$.

\\Unfortunately, Conjecture \ref{stabpen}, as far as the map $\mi$ of Definition \ref{penspart} is concerned,  has never been proven to be true. We will see later that using a
partition scheme different from $\mi$, it is possible to prove  the Conjecture \ref{stabpen}.

\\Let us conclude this section by showing a first consequence of the tree graph identity \equ(r.20). We show that
when the potential is positive the Mayer series of the pressure has the remarkable property
to be an alternate series. Using the tree graph identity \equ(r.20), it is
very easy to see the following. \vskip.1cm

\begin{pro}\label{alternate}
If $V(\xx)\ge 0$ then
$$
\Phi^T(\xx_1,\dots,\xx_n)=(-1)^{n-1}|\Phi^T(\xx_1,\dots,\xx_n)|
$$
and therefore the Mayer series of the
finite volume pressure is, for any $\L$, an alternate series, i.e.
$$
(-1)^{n-1}C_n(\b,\L)\ge 0
$$
\end{pro}

\\{\bf Proof}.
We have
$$
\Phi^T(\xx_1,\dots,\xx_n)~=~ \sum\limits_{g\in G_{n}}
\prod\limits_{\{i,j\}\in E_g}\left[  e^{ -\b V(\xx_i -\xx_j)}
-1\right]
$$
and, using the tree graph identity \equ(r.20), with, for any fixed $\xx_1
,\dots ,\xx_n$, we can write
$$
\begn
\Phi^T(\xx_1,\dots,\xx_n)& =
\sum_{\t\in T_n}\prod_{\{i,j\}\in E_\t}[e^{- \b V(x_i-x_j)}-1]e^{-\b\sum\limits_{\{i,j\}\in E_{\mi(\t)}\backslash E_\t}V(x_i-x_j)}\\
&~=~(-1) ^{n-1} \sum_{\t\in T_n}\prod_{\{i,j\}\in E_\t}[1-e^{- \b V(x_i-x_j)}]e^{-\b\sum\limits_{\{i,j\}\in E_{\mi(\t)}\backslash E_\t}\!\!\!\!\!\!V(x_i-x_j)}\\
&=(-1) ^{n-1} \left|\Phi^T(\xx_1,\dots,\xx_n)\right|
\egn
$$
$\Box$

\\An interesting consequence of such proposition is that for positive
potential the Mayer series with convergence radius $R$ has surely
a singularity at the point $\l~=~ -R$. This fact is quite
unpleasant, since the singularity occurs in a non physical region.

\\The Penrose partition scheme can be used to prove the following identity.

\begin{pro}\label{congraphs}
Let $G_n$ be the set of all connected graphs with vertex set $[n]$ and let, for any $g\in G_n$, be $E_g$ its edge set with cardinality $|E_g|$. Then
$$
\sum_{g\in G_n}(-1)^{|E_g|}=(-1)^{n-1}(n-1)!
$$
\end{pro}
{\bf Proof} Recalling  that $V_{ij}$ can take the value $+\infty$, let us consider the case in which $V_{ij}=+\infty$ for all $\{i,j\}\in E_n$ and  let us  apply the general Penrose identity
\equ(r.200) with $\mathfrak{M}=\mi$. We have in this case
$$
\begin{aligned}
\sum\limits_{g\in G_{n}}~
\prod\limits_{\{i,j\}\in E_g}\left[  e^{ - V_{ij}} -1\right] & = \sum_{g\in G_n}(-1)^{|E_g|} \\
&=
\sum_{\t\in T_n}
e^{-\sum\limits_{\{i,j\}\in E_{\mi(\t)}\backslash E_\t}V_{ij}}
\prod_{\{i,j\}\in E_\t}\left(e^{- V_{ij}}-1\right)\\
& = \sum_{\t\in T_n\atop \mi(\t)\backslash E_\t=\0}
\prod_{\{i,j\}\in E_\t}\left(e^{- V_{ij}}-1\right)\\
&= (-1)^{n-1} \sum_{\t\in T_n\atop \mi(\t)\backslash E_\t=\0}1\\
& = (-1)^{n-1} (n-1)!
\end{aligned}
$$
The last identity follows from the fact that $\sum_{\t\in T_n:~\mi(\t)\backslash E_\t=\0}1=(n-1)!$. Indeed the trees $\t\in T_n$ such that $ \mi(\t)\backslash E_\t=\0$ are
 the trees with incident numbers $d_1,\dots, d_n$ such that $d_1=1$, and $d_j\le 2$ for a single $j\in [n]\setminus \{1\}$ and $d_k=2$ for all $k\in [n]\setminus \{1,j\}$  (i.e. the linear trees rooted at 1), which are clearly $(n-1)!$  (i.e the permutations of the set $\{2,3,\dots,n\}$).

\section{The hard-sphere gas via Penrose identity}\label{sechard}
\def\Va{{V^a_{\rm h.c.}}}
The  original  Penrose identity \equ(r.20), although it has not been proven useful to deal with general stable potentials, has been proven to be useful and very efficient
for purely hard-core pair potentials. Namely for those potentials of the form
$$
V^a_{\rm h.c.}(x)=
\begin{cases}
+\infty & {\rm if}~ |x|\le a\\
0 & {\rm otherwise}
\end{cases}
\Eq(hasp)
$$where $a>0$.  We remind that particles interacting via the pair potential $V^a_{\rm h.c.}(x)$ are in fact a system of
free hard spheres of diameter $a>0$.

\\Observe that the potential $V^a_{\rm h.c.}$ is  stable with stability constant $B=0$ (since $V^a_{\rm h.c.}$ is  non-negative) and tempered (since
$\int_{x> a}V^a_{\rm h.c.}(|x|)dx=0$).
We can therefore apply  Theorem \ref{peru}  which gives us the following lower bound  the convergence radius of the Mayer series  for such a system
$$
R_{\Va}\ge {1\over e S_d(a)}\Eq(peruhc)
$$
where
$$
S_d(a)=\int_{\mathbb{R}^d}|e^{-\b V^a_{\rm h.c.}(|x|)}-1|dx =\int_{|x|\le a} 1 dx
$$
is the volume of the $d$-dimensional sphere of radius $a$.

\\Let us now use the Penrose tree graph identity \equ(r.20) to bound directly the Mayer coefficients of a system of free hard spheres
of diameter $a$  (i.e particles interacting via the potential \equ(hasp)).

\\Given $(x_1,\dots,x_n)\in \mathbb{R}^{dn}$, let us use the short notation $i\sim j$ when $|x_i-x_j|>a$ and  $i\nsim j$ when $|x_i-x_j|\le a$.

%\begin{defi}
%Let   $V^a_{\rm h.c.}(x)$  be a purely hard core pair potential as in formula \equ(hasp).
%Then, for any $n\ge 2$ and any $(x_1,\dots,x_n)\in \mathbb{R}^{dn}$, it is defined the graph $g_a(x_1,\dots,x_n)\in \GG_n$ with vertex set
%$[n]$ and edge set $E_{g_a(x_1,\dots,x_n)}=\{\{i,j\}\in [n]: ~V(x_i-x_j)=+\infty$ (i.e. $|x_i-x_j|\le a$).
%\end{defi}

\begin{defi}\label{mappen}
Let   $V^a_{\rm h.c.}(x)$  be a purely hard core pair potential as in formula \equ(hasp).
Then, for any $n\ge 2$ and any $(x_1,\dots,x_n)\in \mathbb{R}^{dn}$,
we define the set of Penrose trees $P_a(x_1,\dots,x_n)\subset T_n$
as follows. A tree $\t\in T_n$ belongs to $P_a(x_1,\dots,x_n)$ if the following conditions are satisfied

\begin{enumerate}

\item[(t0)]  if  $\{i,j\}\in E_\t$   then $i\nsim j$ (i.e. $ |x_i-x_j|\le a$)

\item[(t1)] if two vertices $i$ and $j$ are cousins in $\t$ (i.e. such that $d_\t(i)=d_\t(j)$), then $i\sim j$ (i.e. $|x_i-x_j|>a$);

\item[(t2)]  if two vertices $i$ and $j$ are such that $d_\t(j)=d_\t(i)-1$ and $j>i'$,  then $i\sim j$ (i.e. $|x_i-x_j|>a$);

%\\{\small   The property (t0) is simply the fact that $\t$ is a spanning tree of $g_a(x_1,\dots,x_n)$.}

\end{enumerate}
\end{defi}

\\Then we have the following

\def\Re{{\mathbb{R}}}
\begin{teo}\label{pencore}
For any purely hard core pair potential
$\Va$, for any $n\ge 2$, for any $(x_1,\dots,x_n)\in \Re^{dn}$ and for any $\b\in (0,+\infty)$
it holds the identity
$$
\sum_{g\in G_n} \prod_{\{i,j\}\in E_g}\left(e^{-\b\Va(x_i-x_j)}-1\right)~=~
(-1)^{n-1}\sum_{\t\in T_n}\ind{P_a(x_1,\dots,x_n)}(\t)
\Eq(r.20b)
$$
where
$$
\ind{P_a(x_1,\dots,x_n)}(\t)=
\begin{cases}
1 &{\text if} \t\in P_a(x_1,\dots,x_n)\\
0 &{\text otherwise}
\end{cases}
$$
\end{teo}

\\{\it Proof}. Observe that, by definition \equ(hasp) we have that $\b\Va(x)=\Va(x)$, for any $\b\in (0,+\infty)$ and any $x\in \mathbb{R}^d$.   Using now \equ(r.20) we have
$$
\sum_{g\in G_n} \prod_{\{i,j\}\in E_g}\left(e^{-\b\Va(x_i-x_j)}-1\right)
= \sum_{\t\in T_n} w_\t(x_1,\dots,x_n)
$$
where
$$
w_\t(x_1,\dots,x_n) =  e^{-\sum\limits_{\{i,j\}\in E_{\mi(\t)}\backslash E_\t}\Va(x_i-x_j)}
\prod_{\{i,j\}\in E_\t}\left(e^{- \Va(x_i-x_j)}-1\right)
$$
Now observe that
$$
\prod_{\{i,j\}\in E_\t}\left(e^{- \Va(x_i-x_j)}-1\right)=
\begin{cases}(-1)^{n-1} &{\rm  if} ~|x_i-x_j|\le a~ {\rm for\; all}~ \{i,j\}\in E_\t\\
0 & {\rm otherwise}
\end{cases}
$$
and
$$
e^{-\sum\limits_{\{i,j\}\in E_{\mi(\t)}\backslash E_\t}\Va(x_i-x_j)}=
\begin{cases}
1 &{\rm if}~ |x_i-x_j|>a~{\rm  for \;all}~ \{i,j\}\in E_{\mi(\t)}\backslash E_\t\\
0 &{\rm otherwise}
\end{cases}
$$
Therefore
$$
 w_\t(x_1,\dots,x_n)~=~
\begin{cases}
(-1)^{n-1} &{\rm if}~ |x_i-x_j|\le a \;{\rm for\;all\;}\{i,j\}\in E_\t~{\rm and}\atop
~~~~~ |x_i-x_j|>a \; {\rm for\; all\;}  \{i,j\}\in E_{\mi(\t)}\backslash E_\t\\\\
0 &{\rm  otherwise}
\end{cases}
$$
Now, recalling the Definition \ref{penspart} of the map $\mi$,  we have
$$
\{i,j\}\in E_{\mi(\t)}\setminus E_\t ~~~\Longrightarrow  ~~~{\rm either }~d_\t(i)=d_\t(j) \;{\rm or} \;d_\t(j)=d_\t(i)-1\; {\rm and}\; j>i'
$$
whence
$$
w_\t(x_1,\dots,x_n)= (-1)^{n-1}
\begin{cases}
1 &{\rm if}~ {{|x_i-x_j|\le a ~{\rm for\;all\;}\{i,j\}\in E_\t~{\rm and}}
\atop ~~~|x_i-x_j|>a \;
{{\rm for\; all\;} \; \{i,j\}\in E_{\mi(\t)}\setminus E_\t }}
\\ \\
0 &{\rm otherwise}
\end{cases}
$$
I.e., recalling the Definition \ref{mappen} concerning the Penrose trees,
$$
w_\t(x_1,\dots,x_n)=(-1)^{n-1}\ind{P_a(x_1,\dots,x_n)}(\t)\Eq(wtpen)
$$
$\Box$
\vv
\\We now derive  a useful inequality from \equ(r.20b). To do this we need one more definition.

\begin{defi}\label{mappenwe}
Let   $V^a_{\rm h.c.}(x)$  be a purely hard core pair potential as in formula \equ(hasp).
Then, for any $n\ge 2$ and any $(x_1,\dots,x_n)\in \mathbb{R}^{dn}$,
we define the set of weakly Penrose trees $P^*_a(x_1,\dots,x_n)\subset T_n$
as follows. A tree $\t\in T_n$ belongs to $P^*_a(x_1,\dots,x_n)$ if the following conditions are satisfied

\begin{enumerate}

\item[(t0)]  if {  $\{i,j\}\in E_\t$}   then $i\nsim j$ (i.e. $ |x_i-x_j|\le a$);

\item[(t1)] if two vertices $i$ and $j$ are siblings (i.e. $d_\t(i)=d_\t(j)$ and moreover they have the same parent $i'=j'$), then $i\sim j$ (i.e. $|x_i-x_j|>a$);

\end{enumerate}
\end{defi}

\\Then we have the following

\begin{teo}\label{pencore2} For any purely hard core pair potential
$V^a_{\rm h.c.}$
it holds the inequality
$$
|\sum_{g\in G_n} \prod_{\{i,j\}\in E_g}\left(e^{-\Va(x_i-x_j)}-1\right)|~\le~
\sum_{\t\in  T_n}w^*_\t(x_1,\dots,x_n)
\Eq(r.20c)
$$
with
$$
w^*_\t(x_1,\dots,x_n)=\ind{P^*_a(x_1,\dots, x_n)}(\t)
$$
\end{teo}
{\it Proof}. The inequality follows immediately from Theorem \ref{pencore} by noting that $P^*_a(x_1,\dots,x_n)\supset P_a(x_1,\dots,x_n)$. $\Box$
\vv
\\Observe now that
$$
w^*_\t(x_1,\dots,x_n)=
\begin{cases}
1 &{\rm if}~ {|x_i-x_j|\le a \;{\rm for\;all\;}\{i,j\}\in E_\t}~{\rm and}\atop
~~~~~ |x_i-x_j|>a \; {\rm for\; all\;}  i,j\;{\rm siblings\;in\;}\t\\\\
0 & {\rm  otherwise }
\end{cases}
\Eq(wtau)
$$
Therefore $n$-order Mayer coefficient  for a system interacting via the potential $\Va$  is bounded by
$$
|C_n(\b,\L)|\le {1\over n!}\sum_{\t\in T_n}S_\L(\t)\Eq(cbetahc)
$$
with
$$
S_\L(\t)= {1\over |\L|}\int_\L dx_1\dots\int_\L dx_n w^*_\t(x_1,\dots,x_n)
$$
By  \equ(wtau) we have
$$
S_\L(\t)\;\le\;   g_d(d_1)\prod_{i=2}^n  g_d(d_i-1)
\Eq(r.1)
$$
where $d_i$ is the degree of the vertex $i$ in $\t$,
$$
g_d(k)\;=\;\int_{|x_i|\le a\atop |x_i-x_j|>a} dx_1\dots dx_k\;=\;
a^{dk}  \int_{|y_i|\le 1\atop |y_i-y_j|>1} dy_1\dots dy_k
$$
for $k$ positive integer, and $g_d(0)=1$ by definition. Recalling that $S_d(a)$ denotes the volume of the $d$-dimensional sphere of radius $a$, it is convenient to write
$$
g_d(k)\;=\;[S_d(a)]^k\, \widetilde g_d(k)
\Eq(r.2)
$$
with
$$
\widetilde g_d(k)\;=\;
{1\over [S_d(1)]^k}\int_{|y_i|\le 1\atop |y_i-y_j|>1} dy_1\dots dy_k
\Eq(r.3)
$$
for $k$ positive integer and $\widetilde g_d(0)=1$.  We observe
that $\widetilde g_d(k)\le 1$ for all values of $k$.
From
\equ(r.1)--\equ(r.3) we conclude that
\begin{eqnarray*}
S_\L(\t) &\le& [S_d(a)]^{d_1}\, \widetilde g_d(d_1)
\prod_{i=2}^n    [S_d(a)]^{d_i-1} \,\widetilde g_d(d_i-1)\\
&=& [S_d(a)]^{n-1}\, \widetilde g_d(d_1)
\prod_{i=2}^n    \widetilde g_d(d_i-1)\;.
\end{eqnarray*}
The last identity follows from the fact that for every tree of $n$ vertices,
$d_1+\cdots+d_n=2n-2$.
The $\tau$-dependence of this last bound is only through the degree of the vertices,
hence it leads, upon insertion in \equ(cbetahc), to the inequality
$$
|C_n(\b,\L)|\;\le\;{[S_d(a)]^{n-1}\over n!}
\sum_{d_1,...,d_n: \,d_i\ge 1\atop d_1+...+d_n=2n-2}
\widetilde g_d(d_1)\,
\prod_{i=2}^n  \widetilde g_d(d_i-1)\,
{(n-2)!\over \prod_{i=1}^n (d_i-1)!}\;\le
$$
$$
\le {[S_d(a)]^{n-1}\over n}
\sum_{d_1,...,d_n: \,d_i\ge 1\atop d_1+...+d_n=2n-2}
{\widetilde {g}_d(d_1)\over d_1!}\,
\prod_{i=2}^n  {\widetilde{g}_d(d_i-1)\,
\over (d_i-1)!}\;
$$
where we recall that the quantity ${(n-2)!/[ \prod_{i=1}^n (d_i-1)!]}$ in the first line of inequality above is precisely
the number of trees with $n$ vertices and fixed degrees
$d_1,\ldots,d_n$, according to Cayley formula and in the second line  we have used the bound $d_1\le n-1$.

\\At this point, following \cite{PS},
we multiply and divide by $\m^{n-1}$ where $\m>0$ is a parameter to
be chosen in an optimal way.  This leads us to the inequality
\begin{eqnarray*}
|C_n(\b,\L)|&\le&
{1\over n}\left[{[S_d(a)]\over \m}\right]^{n-1}
\sum_{d_1,...,d_n:\,d_i\ge 1\atop d_1+...+d_n=2n-2}
 {{\widetilde g}_d(d_1)\,\m^{d_1}\over d_1!}
\prod_{i=2}^n   {\widetilde g_d(d_i-1)\,\m^{d_i-1}\over (d_i-1)!}\\
&\le& {1\over n}\left[{[S_d(a)]\over \m}\right]^{n-1}{\Bigg(\sum_{s\ge 0}
{\widetilde g_d(s)\m^s\over s!}\Bigg)^n}
\end{eqnarray*}
Therefore the convergence radius of the Mayer series of the of the gas of free hard spheres of diameter $a$ admit the following lower bound
$$
R_{\Va}\ge{1\over S_d(a)} \max_{\m>0}{\mu\over  C_d(\m)}
$$
where
$$
C_d(\m)\;=\;\sum_{s\ge 0}
{\widetilde g_d(s)\over s!}\,\m^s
$$
(attention! This is a polynomial in $\m$).
Let us show that for $d=2$
the quantitative improvement given by this condition
 with respect to the classical bound \equ(peruhc)
can be substantial.
In the $d=2$ case (i. e. the two-dimensional hard sphere gas)
$$
C_2(\m)\;=\;\sum_{s= 0}^5
{\widetilde g_2(s)\over s!}\m^s
$$
where, by definition, $\widetilde g_2(0)=\widetilde g_2(1)=1$.
The factor $\widetilde g_2(2)$ can be explicitly evaluated in terms of
straightforward integrals and we get

$$
\widetilde g_2(2)={1\over \pi^2}\int_{|x|\le 1}d^2x\int_{|x'|\le 1}d^2x'\Theta(|x-x'|>1)
$$
where $\Theta(|x-x'|>1)= 1$ if $|x-x'|>1$ and zero otherwise.
Using polar coordinates
$$
\widetilde g_2(2)={2\pi\over \pi^2}\int_0^1A(\r)\r d\r
$$
where $A(\r)$ is the area of the region $S_0\backslash S_\r$
with $S_0=\{(x,y)\in \mathbb{R}^2: x^2+y^2\le 1\}$ and
$S_\r=\{(x,y)\in \mathbb{R}^2: (x-\r)^2+y^2\le 1\}$.
We get
$$
A(\r)=2\left[\int_{-1}^{\r/2}\sqrt{1-x^2}dx-\int_{-1+\r}^{\r/2}\sqrt{1-(x-\r)^2}\,dx\right]=4\int_{0}^{\r/2}\sqrt{1-x^2}dx=
$$
$$
=2\left[\arcsin(\r/2)+{\r\over 2}\sqrt{1-{\r^2\over 4}}\right]
$$
Hence
$$
\widetilde g_2(2)={2\over \pi}\int_0^1A(\r)\r d\r= {4\over \pi}\int_0^1\left[\arcsin(\r/2)+{\r\over 2}\sqrt{1-{\r^2\over 4}}\right]\r d\r=
$$
$$
= {16\over \pi}\int_0^{1/2}\left[u\arcsin(u)+{u^2}\sqrt{1-{u^2}}\right] du
$$
But, since
$$
\int u\arcsin(u)du={1\over 4}\left[2u^2\arcsin(u) - \arcsin(u)+u\sqrt{1-u^2}\right]
$$
and
$$
\int u^2\sqrt{1-u^2}={1\over 8}\left[\arcsin(u)+u\sqrt{1-u^2}(1-2u^2)\right]
$$
We obtain
$$
\widetilde g_2(2)={3\sqrt{3}\over 4\pi}
$$
The other terms of the sum can be numerically evaluated (e.g. using Mathematica),
obtaining
$$
 \widetilde g_2(3)=0,0589\,\,\,\,\,\,\,\,\widetilde g_2(4)=0,0013\,\,\,\,\,\,\,\,\widetilde g_2(5)\le 0,0001
$$
Choosing $\m=\left[8\pi\over 3\sqrt{3}\right]^{1/2}$ (a value for which ${\m\over C_d(\m) }$ is close to its maximum) we get
$$
R_{\Va}\ge{1\over S_d(a)} \; 0.5107
$$
This should be compared with the bound $R_{\Va}\ge{1/e\over  S_d(a)} $ obtained
through the classical condition \equ(peruhc).

\section{Stable and tempered potentials}

The Penrose map $\bm m$ described above is not suitable to control the factor
$$
\exp\{-\b\sum_{\{i,j\}\in E_{\bm{m}(\t)}\backslash E^+_\t}V(x_i-x_j)\}
$$
appearing in \equ(r.20). The key point  is  thus to find another partition scheme $\bm M$ in such a way that is possible to find a good bound for such a factor.
\\We will construct explicitly our  new partition scheme $\bm M: {\mathcal T}_n\to {\mathcal G}_n$  by first defining an auxiliary map
$\bm T :{\mathcal G}_n\to {\mathcal T}_n$ and then deriving $\bm M$ from $\bm T$ according to  Proposition \ref{twomaps}.

\\ We thus start
by first defining the map $\bm T$  from ${\mathcal G}_n$ to
${\mathcal T}_n$.

\\For a given pair potential $V$ and a fixed $(x_1,\dots,x_n)\in \mathbb{R}^{dn}$,  call $V_{ij}=\b V(x_i-x_j)$ for all $\{i,j\}\in E_n$. Let us establish a total order $\succ$ in the set $E_n$ in such way that
 $\{i,j\}\succ \{k,l\}$ implies $V_{ij}\ge V_{kl}$. Using this total order we can associate to any $g\in \GG_n$ a unique
 minimal spanning tree $\bm T(g)$ of $g$ via the so-called Kruskal algorithm. Namely, $\bm T(g)$ is the spanning tree
 of $g$ constructed by choosing its edges one after the other from the set $E_g$ in such way that at each step the chosen
 edge is always the lowest edge (in the order $\succ$) that does not create a cycle.

\begin{defi}\label {hyp}
\\For a given pair potential $V$ and a fixed $(x_1,\dots,x_n)\in \mathbb{R}^{dn}$, the map  $\bm T: \GG_n \to \TT_n$ associates to any $g\in \GG_n$  its unique minimal spanning tree $\bm T(g)$
w.r.t. the order $\succ$  defined above via the Kruskal algorithm.
\end{defi}
We can  now define the
map
$\bm M$  from ${\mathcal T}_n$ to ${\mathcal G}_n$ as follows.
\begin{defi}\label{d4}
The map $\bm M:{\mathcal T}_n\to{\mathcal G}_n$ is such that  $\bm M( \tau )$ is the graph whose edges are
 the $\{i,j\}$ such that $\{i,j\} \succ \{k,l\} $ for every edge $\{k,l\} \in E_\tau$ belonging to the path  from $i$ to $j$ through $\tau$.
\end{defi}

\\Thus we have constructed
$$\xymatrix{{\mathcal G}_n \ar@<.5ex>[r]^{\bm T} & {\mathcal T}_n \ar@<.5ex>[l]^{\bm M}}$$
Observe that $\tau \subset \bm M(\tau)$ and $\bm T(g) \subset g$.
We now show that these two maps ${\bm T}$, ${\bm M}$ satisfy the hypothesis of  Proposition \ref{twomaps}.

\begin{lem}\label{l1}
 Let $\bm T$ be the and $\bm M$ be the maps given in Definitions
\ref{hyp} and \ref{d4} respectively. Then
 $$\bm T^{-1}(\tau)=\{g \in {\mathcal G}_n :\, \tau \subset g \subset \bm M(\tau)\}$$
 and therefore $\bm M$  is a partition scheme in ${\mathcal G}_n$.
\end{lem}

\begin{proof}
 Let $g \in \bm T^{-1}(\tau)$. We have $\tau=\bm T(g) \subset g$. Now take $\{i,j\} \in E_g$, and let $e \in E_\tau$
 be any edge belonging to the path from $i$ to $j$ in $\tau$.
 Consider the graph $\tau'$  obtained from $\tau$ after replacing the edge $e$ by $\{i,j\}$. Clearly $\tau'$ is connected and has $n-1$ edges,
 so it is a tree. By minimality of $\tau$ we must have $ \{i,j\}\succ e$, whence $\{i,j\} \in E_{\bm M(\tau)}$. Therefore $g \subset \bm M(\tau)$.

 \\Conversely, let $\tau \subset g \subset \bm M(\tau)$. We must show $\bm T(g)=\tau$. By cardinality, it suffices to show  $\bm T(g) \subset \tau$.
 Proceeding by contradiction, take $\{i,j\} \in E_{\bm T(g)} \setminus E_\tau$. Consider the path $p^\t(\{i,j\})$ in $\tau$ joining $i$ with $j$.
   Since $\bm T(g) \subset \bm M(\tau)$, $\{i,j\}$ is greater (w.r.t. $\succ$) than  any
 edge in the path $p^\t(\{i,j\})$. If we remove $\{i,j\}$ from $\bm T(g)$, the tree $T(g)$ splits into two trees.
 Necessarily, at least one of the edges, say $f$, in
 the path $p^\t(\{i,j\})$ joins a vertex of one tree with a vertex of the other. Thus, by adding this edge we obtain a connected
 graph with $n-1$ edges, a new tree $\t'\subset g$ with differs from $\bm T(g)$ because the edge $\{i,j\}$ has been replaced by $f$ and $\{i,j\} \succ f$  which contradicts the minimality of $\bm T(g)$. $\Box$
\end{proof}
\vv

\\{\bf Remark}. The map  $\bm M:{\mathcal T}_n\to{\mathcal G}_n$ given in Definition \ref{d4} is a
partition scheme in ${\mathcal G}_n$. This is our new partition scheme. The advantage to consider this new partition
scheme instead of the one originally proposed
by Penrose is manifestly clear in the following two lemmas. The first of them
 follows easily from the general Penrose tree graph identity given in Theorem \ref{Penid}. The second key lemma shows the crucial role played by the new partition scheme
 introduced above.
\begin{lem}\label{cor1}  Let $n\ge 2$ and let $\{V_{ij}\}_{\{i,j\}\in E_n}$ be $n(n-1)/2$ real numbers
(one for each unordered pair $\{i,j\}\subset[n]$) taking values in $\mathbb{R}\cup\{+\infty\}$.
Let  $\mathfrak{M}: T_n\to G_n$ be any partition scheme in $G_n$.
Then the following inequality  holds
$$
\left|\sum\limits_{g\in G_{n}}~
\prod\limits_{\{i,j\}\in E_g}\left[  e^{ -V_{ij}} -1\right]\right|~
\le~
\sum_{\t\in T_n}
e^{-\sum_{\{i,j\}\in E_{\mathfrak{M}(\t)}\backslash E^+_\t}V_{ij}}
\prod_{\{i,j\}\in E_\t}\left(1-e^{-  |V_{ij}|}\right)\Eq(tre)
$$
where
$$
E_\tau^+=\{\{i,j\}\in E_\t:~V_{ij}\ge 0\}\Eq(t+)
$$
\end{lem}

\\{\bf Proof}. By Theorem \ref{Penid} we have that
$$
\left|\sum_{g\in {\mathcal G}_n} \prod_{\{i,j\}\in E_g}\left(e^{- V_{ij}}-1\right)\right|~\le ~
\sum_{\t\in {\mathcal T}_n} e^{-  \sum_{\{i,j\}\in E_{\Mi(\t)}\setminus E_\t}V_{ij}}\prod_{\{i,j\}\in E_\t}\left|e^{- V_{ij}}-1\right|
 \Eq(uno)
$$
Observe now that, for any $\t\in \mathcal T_n$
$$
 \prod_{\{i,j\} \in E_\tau} |e^{-V_{ij}}-1|=\Big[ \prod_{\{i,j\} \in E_\tau} (1-e^{- |V_{ij}|}) \Big] \
 e^{- \sum_{\{i,j\}\in E_\t\setminus E^+_\t}V_{ij}}
$$
 so that
$$
 e^{-  \sum\limits_{\{i,j\}\in E_{\Mi(\t)}\setminus E_\t}V_{ij}}
\prod_{\{i,j\}\in E_\t}\left|e^{- V_{ij}}-1\right|
~=~
e^{- \sum\limits_{\{i,j\} \in E_{\mathfrak{M}(\tau)} \setminus E^+_\tau}V_{ij}}
\prod_{\{i,j\} \in E_\tau} (1-e^{-| V_{ij}|})
 \Eq(due)
$$
Inserting now \equ(due) into \equ(uno) we get \equ(tre). $\Box$

% Namely, $\widehat{f}$ is
%such that $\sum_{e \in E_\tau} f(e)$
%has a different value (in the tomonoid $\widehat{\mathbb{K}}$) for different trees $\tau \in { T}_n$.
%  Then, {for every $g \in {G}_n$ there is a unique  spanning tree $\t \subset g$
%  for which $\sum_{e \in E_\t} f(e)$ is minimum}.

%Ineed,
% ${\mathbb N}_0^{m}$ is totally ordered lexicographically, according to
% the usual order of ${\mathbb N}_0$ and then we consider the lexicographical order on ${\mathbb R} \times {\mathbb N}_0^{m}$,
% prioritizing the first coordinate.

 %Start from a total order on the edges $E_n$ and then order
% ${\mathbb Z}^{E_n}$ lexicographically, according to
% the usual order of ${\mathbb Z}$.
% Then, we consider the lexicographical order on ${\mathbb R} \times \mathbb{Z}^{E_n}$, prioritizing the first coordinate.
\vv

%\\In the following statement $E_\tau^+$ denotes the set of edges of the tree $\tau$ with positive energy. That is,
%$$
%E_\tau^+=\{\{i,j\}\in E_\t:~V(x_i-x_j)\ge 0\}
%$$
% We also sometimes set shortly $\bm x_n=(x_1,\dots,x_n)$ for $(x_1,\dots,x_n)\in \mathbb{R}^{dn}$.

\begin{lem}\label{stabgen}  Let $n\ge 2$ and let $\{V_{ij}\}_{\{i,j\}\in E_n}$ be $n(n-1)/2$ real numbers
(one for each unordered pair $\{i,j\}\subset[n]$) taking values in $\mathbb{R}\cup\{+\infty\}$.
Suppose that there exist non negative numbers $\{B_i\}_{i\in [n]}$ such that
for each
$S\subset [n]$ it holds
$$
\sum_{\{i,j\}\subset S}V_{ij}\ge \sum_{i\in S}B_i\Eq(stabgen)
$$
Let  ${\bm M}: T_n\to G_n$ be the  partition scheme in $G_n$ given in Definition \ref{d4}.
Then, for any $\t\in T_n$ the following inequality  holds
$$
 \sum_{\{i,j\} \in \bm E_{M(\tau)} \setminus E_\tau^+} V_{ij} \geq -\sum_{i=1}^nB_i
\Eq(Peb0)
$$
\end{lem}
{\bf Proof}.
 The set of edges $E_\t\setminus E_\tau^+$ forms  the forest $\{\tau_1,...,\tau_k\}$.
 Let us denote ${\mathcal V}_{\tau_s}$ the vertex set of the  tree $\tau_s$
 of the forest.
Assume $i \in {\mathcal V}_{\tau_a}$, $j \in {\mathcal V}_{\tau_b}$.
 If $a \neq b$, the path from $i$ to $j$ through $\tau$ involves an edge $e$ in $E_\tau^+$. Thus, if in addition $\{i.j\} \in
E_{\bm M(\tau)}$, we have that  $\{i,j\}\succ e$ and therefore
 $V_{ij} \geq V_e \geq 0$. If $a=b$, the path from $i$ to $j$ through $\tau$ is contained in $\tau_a$. Thus, if in addition
 $\{i,j\} \notin E_{\bm M(\tau)}$, there must be at least  an edge $e$ in that path such that
 $\{i,j\}\prec e$  and  therefore  $V_{ij} \leq V_e \leq 0$.  This allows to bound:

 $$\sum_{\{i,j\} \in E_{\bm M(\tau)} \setminus E_\tau^+} V_{ij} \geq \sum_{s=1}^k \sum_{\{i,j\} \subset {\mathcal V}_{\tau_s}} V_{ij} \geq
 -\sum_{s=1}^k \sum_{i\in \mathcal V_{\tau_s}}B_i \ge   -\sum_{i=1}^nB_i~~~~~~~~~~~~~~~~~~\Box$$

\\From Lemmas \ref{cor1} and \ref{stabgen} we immediately get the following theorem.
\begin{teo}\label{teoPYgen}
 Let $n\ge 2$ and let $\{V_{ij}\}_{\{i,j\}\in E_n}$ be $n(n-1)/2$ real numbers
(one for each unordered pair $\{i,j\}\subset[n]$) taking values in $\mathbb{R}\cup\{+\infty\}$.
Suppose that there exist non negative numbers $\{B_i\}_{i\in [n]}$ such that
for each
$S\subset [n]$ it holds
$$
\sum_{\{i,j\}\subset S}V_{ij}\ge \sum_{i\in S}B_i\Eq(stabgenS)
$$
 Then the following inequality holds.
$$
\left|\sum_{g \in { G}_n} \prod_{\{i,j\} \in E_g} (e^{-V_{ij}}-1)\right| \le
e^{\sum_{i=1}^nB_i}\sum_{\tau \in { T}_n} \prod_{\{i,j\} \in E_\tau} (1-e^{-\b| V_{ij}|})\Eq(bteo1)
$$
\end{teo}
\\{\bf Proof}. By Lemma \ref{cor1} we have that
$$
\big| \sum_{g \in { G}_n} \prod_{\{i,j\} \in E_g} (e^{-V_{ij}}-1) \Big|=\sum_{\tau \in {T}_n}  e^{-  \sum_{\{i,j\}\in E_{{\bm M}(\t)}\setminus E^+_\t}V_{ij}} \prod_{\{i,j\} \in E_\tau} (1- e^{- |V_{ij}|})
 $$
where  $\bm M$ is the map given  in Definition \ref{d4}.
Using then   Lemma \ref{stabgen}, the inequality \equ(bteo1) follows.  $\Box$

\begin{teo}\label{teoPY2}
 Let $V$ be a stable pair  potential with  stability constant $B$. Then
for any  $n\in \N$ such that $n\ge 2$ and any $(x_1,\dots,x_n)\in \Re^{dn}$ the following inequality holds.
$$
\left|\sum_{g \in { G}_n} \prod_{\{i,j\} \in E_g} (e^{-\b V(x_i-x_j)}-1)\right| \le  e^{\b Bn}\sum_{\tau \in { T}_n} \prod_{\{i,j\} \in E_\tau} (1-e^{-\b| V(x_i-x_j|)})\Eq(bteo1b)
$$
\end{teo}
\begin{proof} Just observe that, once $(x_1,...,x_n) \in {\mathbb R}^{dn}$ are fixed,
the numbers  $V_{ij}=\b V(x_i - x_j)$ satisfy \equ(stabgenS) with $B_i=B$ for all $i\in [n]$. So \equ(bteo1b) follows immedialely
by Theorem \ref{teoPYgen}.  $\Box$
\end{proof}

\vv
\\From Theorem \ref{teoPY2} we have immediately the following Theorem.

\begin{teo}\label{pryu}
 Let $V$ be a stable and tempered pair  potential with  stability constant $B$. Then
the  $n$-order Mayer  coefficient $C_n(\b,\L)$ defined in \equ(cbeta2)
is bounded by
$$
|C_n(\b,\L)|\le e^{\b B n}n^{n-2} {[\tilde C(\b)]^{n-1}\over n!}\Eq(Proyuh)
$$
where
$$
\tilde C(\b)=\int_{\mathbb{R}^{d}} dx ~ |e^{-\b |V(x)|}-1|\Eq(cbetapy)
$$
Therefore the Mayer series \equ(pressm) converges absolutely, uniformly in $\L$,
 for any complex  $\l$ inside the disk
$$
|\l| <{1\over e^{\b B+1} \tilde C(\b)}\Eq(RadPY)
$$
I.e. the convergence radius of the Mayer series \equ(pressm) admits the following lower bound
$$
R_{V}\ge {1\over e^{\b B+1} \tilde C(\b)}\Eq(convradpy)
$$
\end{teo}

\\{\bf Proof}. From Theorem \ref{teoPY2} (and also recalling Proposition \ref{integr}) we have that
$$
\begin{aligned}
|C_n(\b,\L)| & \le{1\over n!}{1\over |\L|}\int_\L dx_1\dots \int_\L dx_n |\Phi^T(x_1,\dots, x_n)|\\
&\le {1\over n!}{1\over |\L|}\int_\L dx_1\dots \int_\L dx_n  e^{\b Bn}\sum_{\tau \in { T}_n} \prod_{\{i,j\} \in E_\tau} (1-e^{-\b| V(x_i-x_j|)})\\
&\le {e^{\b Bn}\over n!}\sum_{\tau \in { T}_n}{1\over |\L|}\int_\L dx_1\dots \int_\L dx_n
\prod_{\{i,j\} \in E_\tau} (1-e^{-\b| V(x_i-x_j|)})\\
&\le ~{e^{\b Bn}}{1\over n!}\left[\int_{\mathbb R^d} \Big[1-e^{-\b| V(x)|}\Big]\,dx\right]^{n-1} \sum_{\tau \in { T}_n} 1\\
& =
~{e^{\b Bn}}{n^{n-2}\over n!}\left[\tilde C(\b)\right]^{n-1}
\end{aligned}
$$
%which is  the bound (\ref{bteo1}).
This concludes the proof of Theorem \ref{pryu}. $\Box$.

\vv
\\The improvement on the lower bound of the convergence radius of the Mayer series for stable and tempered potentials
given by Theorem \ref{pryu} with respect to Theorem \ref{peru} is
twofold. First, the factor $e^{\b B+1}$ in \equ(convradpy)  replaces
the factor $e^{2\b B+1}$ in \equ(convrad). Second, the factor $\tilde C(\b)$ in \equ(convradpy)  replaces the factor $C(\b)$ in
\equ(convrad) and clearly, recalling their definitions \equ(cbetapy) and \equ(cbetapr)  respectively, we have $\tilde C(\b)\le C(\b)$ where
the equality only holds  if $V$ is non-negative  (purely repulsive). Moreover observe that while $\tilde C(\b)$
grows at most linearly  in $\b$, the factor $C(\b)$ grows exponentially with $\b$.
So, if we denote by $R_{PR}= [e^{2\b B+1} C(\b)]^{-1}$ the
Penrose-Ruelle lower bound for the convergence radius given in Theorem \ref{peru}
and by $R^*= [e^{\b B+1} \tilde C(\b)]^{-1}$ the lower bound given by Theorem \ref{pryu} for the same convergence radius  we get that
${R^*/R_{PR}}= e^{\b B}[{C(\b)/\tilde C(\b)}]$. This ratio, always greater than one,  is the product of two factors,
 $e^{\b B}$ and $[{\tilde C(\b)/ C(\b)}]$,  both
growing exponentially fast
with $\b$. To give an idea,   for a gas of particles interacting via the classical Lennard-Jones potential
$$
V(|x|)= {1\over |x|^{12}}-{2\over |x|^6}
$$
at inverse temperature $\b=1$, using the value $B_{\rm LJ}=8.61$
 for its stability constant  (see \cite{JI}), the  lower bound \equ(convradpy)
is at least $8.5\times 10^4$ larger than the Penrose-Ruelle lower bound, while for $\b=10$
is at least $7.26\times 10^{43}$ larger than the Penrose-Ruelle lower bound.
%\color{red} Agregar alguna referencia sobre estimacion de B? \color{black}

\section{Convergence of the Virial expansion}
We recall that via the formulas
$$
\b P_\L(\b,\l) ~=~{1\over |\L|}\ln \Xi_\L(\b,\l)= \sum_{n=1}^{\i}C_n(\b,\L)\l^n\Eq(pressu)
$$
$$
\r_\L(\b,\l)~=~{1\over |\L|}{\l}{\partial\over \dpr
\l}{\ln \Xi_\L(\b,\l)} =\sum_{n=1}^{\i}nC_n(\b,\L)\l^n\Eq(densi)
$$it is possible to express the pressure of the system in the grand canonical ensemble in power of the
density $\r(\b,\l)$. Indeed, by the discussion of the previous section  the r.h.s. of \equ(densi)  is a function of $\l$ which is analytic for $|\l|<{1/[e^{\b B+1}\tilde C(\b)]}$
and its derivative respect to $\l$ calculated at $\l=0$ is equal to 1.
So one can eliminate $\l$ in \equ(pressu) and \equ(densi) to obtain the so-called {\it Virial expansion of the Pressure}, i.e. the pressure as a (formal) power series of the density $\r=\r_\L(\b,\l)$, in the Grand Canonical Ensemble
$$
\b P_\L(\b,\l)= \r_\L(\b,\l) - \sum_{k\ge 1} {k\over k+1} \b_k(\b,\L)[\r_\L(\b,\l)]^{k+1}\Eq(virial)
$$

\\The coefficients $ \b_k(\b,\L)$ are combinations of the Mayer coefficients $C_n(\b,\L)$ with $n\le k$ by algebraic expressions. Namely one can prove that  (see e.g.   formula (29) p. 319 of \cite{PB})
$$
\b_k(\b,\L)=
\sum_{n=1}^k{(-1)^{n-1}} {(k-1+n)!\over k!} \sum_{{ \{m_2,\dots,m_{k+1}\}\atop m_i\in {\mathbb{N}\cup\{0\}},~\sum_{i=2}^{k+1} m_i=n}\atop
\sum_{i=2}^{k+1}(i-1)m_i=k}
\prod_{i=2}^{k+1} {[{C_{i}(\b,\L){i}}]^{m_i}\over m_i!}\Eq(virial1)
$$
It is also possible to prove (see e.g. sec. 6.3 of \cite{Mc}) that the limit
$$
\b_k(\b)= \lim_{\L\to\infty} \b_k(\b,\L)
$$
admits a representation in term of two-connected graphs. Namely,
$$
\b_k(\b)={1\over n!}
\int_{\mathbb{R}^d}d\xx_2\dots\int_{\mathbb{R}^d} d\xx_{n} \sum_{g\in G^*_{n}}
\prod_{\{i,j\}\in E_g}\left[  e^{ -\b V(\xx_i -\xx_j)} -1\right]\Eq(virial2)
$$where here $G^*_n$ is the set of the two-connected graphs with vertex set $[n]$. We recall that a graph $g=(V_g,E_g)$ is two-connected if, for every vertex $x\in V_g$ the graph $g\setminus x$ is connected.

\\To study the convergence  of the series \equ(virial) in principle one could try to  proceed as we did in the previous section
when  we dealed with the convergence of the Mayer series. Namely, we could  try to get a bound on the virial coefficients directly from their finite volume expressions
\equ(virial1) or their infinite volume expression \equ(virial2). Unfortunately no method has been found to do this
and the lower bound  for the convergence radius of the
virial series has been obtained in 1964 by Lebowitz and Penrose \cite{LPe} in a more indirect way. We illustrate here below how to get these bounds. Before doing this we need to introduce some further notations and definitions
which will be useful in order to deduce the lower bound for convergence radius of the Virial series.

\\Given a pair potential $V$ with stability constant $B$ we let
 $$
B_n=\sup_{(x_1,\dots,x_n)\in \mathbb{R}^{dn}}-{1\over n}\sum_{1\le i<j\le n}V(|x_i-x_j|)
$$
so that
$$
B=\sup_{n\ge 2} B_n
$$
We also define
$$
{\bar B}_n = {n\over n-1} B_n
$$
and
$$
\bar B= \sup_{n\ge 2} \bar B_n
$$
Then we have clearly, for all $n\ge 2$ and all $x_1,\dots,x_n)\in \mathbb{R}^{dn}$
$$
\sum_{1\le  i<j\le n} V(x_i-x_j)\ge -(n-1) \bar B \Eq(stabas)
$$
and
$$
\bar B\ge B
$$
We call $\bar B$ the {\it Basuev stability constant} of the potential $V$, after Basuev who was the first to introduce it in \cite{Ba1}.
On the other hand, as noted by Basuev,
$B$ and $\bar B$, if not coinciding, should be  very close for the large majority of
stable pair potentials.
Indeed, let us first state and prove the following propositions.
\vv

\begin{pro}
Let  $V(x)$ be a stable and tempered pair potential  in $d$ dimensions with stability constant $B$ and Basuev stability constant $\bar B$.

\begin{itemize}
\item[(a)] If $$
\bar B= \limsup \bar B_n \Eq(supbarb)
$$
then $B=\bar B$.
\item[(b)] If there exists $m\in \mathbb{N}$ such that $\bar B=\bar B_m$, then
$$
\bar B\le {d+1\over d}\Eq(barbd)
$$
\item[(c)] If $V(x)=V(|x|)$ and $V(|x|)$
reaches a  negative minimum at   $|x|=r_0$
being negative for   $|x|>r_0$, then
\end{itemize}
$$
\bar B\le \begin{cases}{3\over 2} B & {\rm if}~ d=1\\
{7\over 6} B &{\rm  if}~ d=\\
{2d(d-1)+1\over 2d(d-1)} B & {\rm if}~ \ge 3
\end{cases}\Eq(estbar)
$$
\end{pro}
\begin{proof}
Let us  prove {\it (a)}.
Suppose by contradiction that $\bar B- B=\d>0$. If this holds
then necessarily there exists a finite  $m \ge 2$ such that $B=B_m$ otherwise if $B=\limsup B_n$ then $B=\bar B$.
Now, due to the hypothesis \equ(supbarb), for all $\e>0$ there exists $n_0$ such that
for infinitely many $n>n_0$
$$
\bar B_n> \bar B-\e~~~~~~~~~~\Longrightarrow ~~~~~~ B_n> {n-1\over n}(\bar B-\e)= {n-1\over n}(B+\d-\e)~
$$
choose now $\e={\d\over 2}$, then  for infinitely many $n$ we have that
$$
B_n>  {n-1\over n}(B+{\d\over 2}) > B~~~~~~~~~~~~\mbox{ as soon as $n> {2B\over \d}+1$}
$$
\end{proof}
in contradiction with the assumption that  $B=B_m$. Hence we must have $B=\limsup B_n= \limsup \bar B_n=\bar B$.
\vskip.2cm
\\We now prove {\it (b)}. By hypotheis $\bar B= \bar B_m$.
Since $V$ is  stable, it is bounded from below and since $V$ is tempered $\inf V$ cannot be positive.
Let $\inf_{x\in \mathbb{R}^d}V(x)=-C$ with $C\ge 0$. Then for any $\e>0$ there exists $r_0$ such that
$V(r_0)< -(C -\e)$.
Take the configuration $(x_1,x_2,\dots, x_{d+1})\in (\mathbb{R}^d)^{d+1}$ such that $x_1$, $x_2,\dots ,x_{d+1}$ are vertices of a $d$-dimensional
hypertetrahedron  with sides of length $r_0$. Recall that a $d$-dimensional
hypertetrahedron has $d+1$ vertices and $d(d+1)/2$ sides.
Then $U(x_1,x_2,\dots, x_{d+1})< -{d(d+1)\over 2}(C-\e)$, which implies that $B> {d\over 2}(C-\e)$ and by the arbitrariness of $\e$ we get $B\ge {d\over 2}C$.
On the other hand, we also have, for any $(x_1,\dots,x_m)\in (\mathbb{R}^d)^m$ that
$U(x_1,\dots,x_m)\ge -{m(m-1)\over 2} C$ which implies that $\bar B_m=\bar B\le {m\over 2}C$. Hence  we can write ${d\over 2}C\le B\le \bar B=\bar B_m\le {m\over 2}C$ which implies  $m\ge d+1$ and so   $\bar B= \bar B_m= {m\over m-1}B_m\le {m\over m-1}B\le {d+1\over d}B$ so that \equ(barbd) holds.

\vskip.2cm

\\We  prove {\it (c)}. \\
Let $\inf_{r\ge 0}V(r)=-C$ with $C> 0$.
First note that $d(d-1)C$ is always a lower bound for $B$ when $d\ge 3$. Just consider a configuration in which $n$ particle
 (with $n$ as large as we want) are
arranged in close-packed configuration  at the sites of a $d$-dimensional  face-centered cubic lattice with step $r_0$. The energy of such configuration
is (asymptotically
as $n\to \infty$)
less than or equal to $-d(d-1)Cn$ since
in a $d$-dimensional face-centered cubic lattice each site has $2d(d-1)$ neighbors (see e.g. \cite{TJ}) and so there are  (asymptotically)
$d(d-1)n$ pairs of neighbors in the configuration.
On the other hand,
for any $n$-particle configuration $(x_1,\dots,x_n)\in \mathbb{R}^{dn}$,
it holds that $U(x_1,\dots,x_n)\ge -n(n-1)C/2$, i.e. $\bar B_n\le nC/2$. Now,
 if $\bar B= \sup_n\bar B_n$ is attained at $n\to \infty$ then, as previously seen, we have $\bar B=B$. So let us suppose
 that $\bar B= \sup_n\bar B_n$ is attained at a certain finite $m$. Then
we must have that  $d(d-1)C\le B\le \bar B\le mC/2$ and so $m \ge 2d(d-1)$. Now if $m=2d(d-1)$ then $\bar B=B$. So when $\bar B>B$ then
$m > 2d(d-1)$ and so we have $\bar B=\frac{m}{m-1}B_m \leq {2d(d-1)+1\over 2d(d-1)} B$. The case $d=1$ and $d=2$ are treated analogously by just
observing that the close-packed arrangement in $d=1$ is simply the cubic lattice with $2$ neighbors for each site while for $d=2$ is the
triangular lattice with $6$ neighbors for each site. So $d(d-1)$ must be replaced by $1$ for $d=1$ and by $3$ for $d=2$ yielding $m>2$ and
$m>6$ for $d=1$ and $d=2$ respectively. $\Box$
\vskip.2cm

\\It is also possible to express the bounds on the Mayer coefficients in terms of the the Basuev Stability constant $\bar B$.
\begin{pro}
Let $V$ be a stable and regular pair potential with basuev stability constant $\bar B$. Then the Mayer coefficients $C_n(\b,\L)$ defined in \equ(ursm) admit the following bound
$$
|C_n(\b, {\Lambda})|\le e^{\beta \bar B (n-1)}n^{n-2} {[\tilde C(\beta)]^{n-1}\over n!}\Eq(Proyuh2)
$$
\end{pro}
{\bf Proof}.
\noindent
Using the partition scheme $\bm M$ given in Definition \ref{d4} let us prove the following inequality.
Given $\tau \in { T}_n$ abd given $(x_1,...,x_n) \in {\mathbb R}^{dn}$, we have
$$
 \sum_{\{i,j\} \in \bm E_{M(\tau)} \setminus E_\tau^+} v(x_i-x_j) \geq -\bar B(n-1)\Eq(Peb)
$$
Indeed,  reasoning indetically as in the proof of Lemma \ref{stabgen} (with $v(x_i-x_j)$ in place of
$V_{ij}$) we have that
$$
\sum_{\{i,j\} \in E_{\bm M(\tau)} \setminus E_\tau^+}v(x_i-x_j)
\geq \sum_{s=1}^k \sum_{\{i,j\} \subset { V}_{\tau_s}} v(x_i-x_j)
$$
Now, by \equ(stabas), for any $I\subset [n]$ it holds that
$\sum_{\{i,j\} \subset {I}} v(x_i-x_j) \ge -(|I|-1)\bar B$ Therefore
$$
\sum_{\{i,j\} \in E_{\bm M(\tau)} \setminus E_\tau^+}v(x_i-x_j) \geq
 \sum_{s=1}^k -| (V_{\tau_s}|-1)\bar B \ge -(n-1)\bar B
 $$
 Proceeding now as in the proof of Theorem \ref{pryu} we get the inequality \equ(Proyuh2).
$\Box$

\vv
\\Let us now go back to the Virial expansion of the Pressure given in \equ(virial) and let us state and prove the following theorem.

\begin{teo}\label{pryu2}
 Let $V$ be a stable and tempered pair  potential with  Basuev stability constant $\bar B$. Then
the convergence radius $\RR$ of the Virial  series \equ(virial)  admits the following lower bound.
$$
\RR\ge { 0.14477\over  \tilde C(\b) e^{\b \bar B}}\Eq(PYvi)
$$
\end{teo}
{\bf Proof}.
\\We follow the steps done by Lebowitz and Penrose in \cite{LPe} to obtain, via Lagrange inversion, a lower bound for $\RR$.

\\We start by observing that, due to \equ(densi) and the fact that $C_1(\b,\L)=1$, there exists a circle $C$
 of some radius $R< 1/[\tilde C(\b)e^{\b B+1}]$ and center in the origin  $\l=0$ of the complex $\l$-plane such that
  $\r_\L(\b,\l)$
has only one zero  in the disc ${D_R}=\{\l\in \mathbb{C}: |\l|\le R\}$ and this zero occurs precisely at $\l=0$.
  Let now  $\r\in \mathbb{C}$ be such that
$$
|\r|< \min_{\l\in C}|\r_\L(\b,\l)| \Eq(condm)
$$
Then by Rouch\'e's theorem $\r_\L(\b,\l)$ and
 $\r_\L(\b,\l)-\r$ have the same number of zeros (i.e. one)  in the region $D_R=\{\l\in \mathbb{C}: |\l|\le R\}$. In other words, for any complex $\r$ satisfying
  \equ(condm) there is only one $\l\in D_R$ such that  $\r=\r_\L(\b,\l)$ and therefore we can invert the equation   $\r=\r_\L(\b,\l)$ and write
 $\l=\l_\L(\b, \r)$. Thus, according to  Cauchy's argument principle, we can write the pressure $\b P_\L(\b,\l)$
  as a function of the density $\r=\r_\L(\b,\l)$  as
$$
P_\L(\r,\b)={1\over 2\p i}\oint_\g P_\L(\b,\l){d\r_\L(\b,\l)\over d\l}{d\l\over \r_\L(\b,\l)-\r}\Eq(cauchy)
$$
where $\g$ can be  any circle centered at the origin in the  complex $\l$-plane fully contained in the region $D_R$
and such that
$$
|\r|<\min_{\l\in \g} | \r_\L(\b,\l)|\Eq(minro)
$$
By standard complex analysis  $P_\L(\r,\b)$ is analytic in $\r$ in the region \equ(minro). Indeed,
once \equ(minro) is satisfied we can write
$$
{1\over \r_\L(\b,\l)-\r}=\sum_{n=0}^{\infty} {\r^n\over  [\r_\L(\b,\l)]^{n+1}}\Eq(okso)
$$
and inserting \equ(okso) in \equ(cauchy) we get
$$
P_\L(\r,\b)= \sum_{n=1}^\infty c_n(\b,\L)\r^n\Eq(virial3)
$$
with
$$
\begin{aligned}
c_n(\b,\L)& = {1\over 2\p i}\oint_\g P_\L(\b,\l){d\r_\L(\b,\l)\over d\l}{d\l\over  [\r_\L(\b,\l)]^{n+1}} \\
&=
-{1\over 2\p in}\oint_\g P_\L(\b,\l){d\over d\l}\left[{1\over  [\r_\L(\b,\l)]^{n}}\right]d\l\\
&=
{1\over 2\p in}\oint_\g {dP_\L(\b,\l)\over d\l}{1\over  [\r_\L(\b,\l)]^{n}}d\l\\
&= ~ {1\over 2\p in\b} \oint_\g{1\over
[\r_\L(\b,\l)]^{n-1}}{d\l\over \l}
\end{aligned}
$$
Therefore
$$
|c_n(\b,\L)|\le{1\over n\b}{1\over \left[\min_{\l\in \g} |\r_\L(\b,\l)|\right]^{n-1}}\Eq(cn)
$$
Inequality \equ(cn) shows that the convergence radius $\RR$ of the series \equ(virial3) (i.e. of the virial series \equ(virial)) is such
that
$$
\RR\ge \min_{\l\in \g}   |\r_\L(\b,\l)|\Eq(Ok)
$$
Therefore  the game is to find an optimal circle $\g$ in the region $D_R$ which maximizes  the r.h.s. of \equ(Ok).
We proceed as follows.
Recalling \equ(densi)  we have, by the triangular inequality, that
$$
 |\r_\L(\b,\l)| \ge |\l|- \sum_{n=2}^\infty n|C_n(\b,\L)||\l|^n \Eq(triangul)
$$
We now  use estimate \equ(Proyuh2) to bound the sum in the r.h.s. of  \equ(triangul).
Therefore we get
$$
\begin{aligned}
 |\r_\L(\b,\l)| ~& \ge ~|\l|- \sum_{n=2}^\infty {n^{n-1}\over n!}[\tilde C(\b) e^{\b \bar B}]^{n-1}|\l|^n\\
 &=
 2|\l|-{1\over \tilde C(\b) e^{\b \bar B}}\sum_{n=1}^\infty {n^{n-1}\over n!}[\tilde C(\b) e^{\b \bar B}|\l|]^{n}
\end{aligned}
$$
Now following \cite{LPe} let us set $w$ to be the  first positive solution of
$$
we^{-w}= \tilde C(\b) e^{\b \bar B}|\l|\Eq(wew)
$$
If we  take $|\l|< 1/ e^{\b \bar B+1} \tilde C(\b)$ (which is
surely  inside the convergence region since $\bar B\ge B$) then $\tilde C(\b) e^{\b \bar B}|\l|<1/e$ and since the function $we^{-w}$ is
increasing in
the interval $[0,1]$ and takes the value $1/e$ at $w=1$,  there is a unique $w$ in the interval $[0,1]$ which solves \equ(wew).
Now we use the Euler's formula
$$
w=\sum_{n=1}^\infty{n^{n-1}\over n!}(we^{-w})^n
$$
to get
$$
|\r_\L(\b,\l)|\ge {w\over \tilde C(\b) e^{\b \bar B}}\left[2e^{-w}
-1\right]\Eq(fififi)
$$
Observe that the r.h.s. \equ(fififi) is greater than zero when $w$ varies in the interval  $(0,\ln 2)$. Therefore, any  circle $\g$ around the origin with radius $R_\g$ between zero and
$\ln 2/(2e^{\b \bar B+1}\tilde
C(\b))$,  which  corresponds, in force of \equ(wew),  to  any  $w$  in the interval $(0,\ln 2)$, is surely in the region $D_R$. So
 we have
$$
\max_{\g\subset D_R} \min_{\l\in \g}   |\r_\L(\b,\l)| \ge {1\over \tilde C(\b) e^{\b \bar B}}\max_{w\in (0,\ln2)} w\left[2e^{-w}
-1\right]\Eq(fine)
$$
which, recalling  \equ(Ok) and since $\max_{w\in (0,\ln2)} {w}[2e^{-w}-1]\ge 0.14477$, concludes  the proof.

\numsec=5\numfor=1

\part{Discrete systems}

\numfor=1\numsec=4
\vskip.1cm\def\E{\mathcal{E}}\def\C{ \mathbb{C}}
\index{polymer gas}

\chapter[\small{The polymer gas}]{The Abstract Polymer gas}

\section{Setting}\label{defpoly}
\\The abstract polymer gas is a discrete model which plays a very important role in many physical situations,
such as spin systems on the lattice at low or high temperature, or continuous or discrete field theories.
The polymer gas was first introduced by Gruber and Kunz \cite{GK} in 1970. In the original Gruber and Kunz
paper polymer were finite subsets of $\Z^d$ (the unit cubic lattice). Later, Kotecy\'y and Preiss proposed
the abstract model in which polymers were abstract objects belonging to some set $\PP$, the polymer set,
whose unique structure was given by mean of  a symmetric and reflexive relation in $\PP$, that they called the incompatibility relation.
As we will see, this is equivalent to assume
that the interaction between polymers occurs via a hard core pair potential.
In the present chapter we
revisit the abstract polymer gas.

We first  specify a countable set $\PP$ whose elements are all possible polymers (i.e. $\PP$ is the
{\it single particle state space}).
We then associate to each polymer $\g\in \PP$ a complex  number $z_\g$ (a positive number in physical situations)
which is  interpreted as the  {\it activity} of the  polymer $\g$. We will denote $\bm z=\{z_\g\}_{\g\in \PP}$ and for any
$\L\subset \PP$, $\bm z_\L=\{z_\g\}_{\g\in \L}$.

In the general situation polymers interact through a pair potential. Namely, the energy $E$ of a configuration $\g_1,\dots,\g_n$
of $n$ polymers is given by
$$
E(\g_1,\dots,\g_n)= \sum_ {1\leq i<j\leq n}V(\g_i,\g_j)\Eq(U)
$$
where pair potential $V(\g,\g')$ is a symmetric function in $\PP\times \PP$ taking values in
$ \mathbb{R}\cup\{+\infty\}$. We will make the assumption that  the pair interaction $V(\g,\g')$ is purely hard-core.
Namely $V(\g,\g')$ takes values in the set $\{0, +\infty\}$. Observe that
an hard-core  pair potential $V(\g,\g')$ in $\PP$ induces  a {relation} $ \mathcal{R}_V$
in $\PP\times \PP$ (i.e.  $ \mathcal{R}_V$ is a subset of
$\PP\times \PP$). Namely, we say that a pair $(\g,\g')$ belongs to  $ \mathcal{R}_V$ if and only if
$V(\g,\g')=+\infty$.
Clearly  $\mathcal{R}_V$ is symmetric because by assumption $V$ is symmetric.
When $(\g,\g')\in \mathcal{R}_V$ (i.e. $V(\g,\g')=+\infty$), we write $\g\not\sim \g'$ and say that $\g$ and $\g'$ are {\it incompatible}. Conversely, if
$(\g,\g')\notin \mathcal{R}_V$ we say that the polymers $\g$ and $\g'$ are
{\it compatible} and we write $\g\sim\g'$. Note that  if $V$ is such that $ \mathcal{R}_V$ is reflexive, then
$\g\not\sim\g$ for all $\g\in \PP$. The relation $\RR_V$ induced by the pair potential $V$ is called {\it the incompatibility relation}.

%\begin{defi}\label{hardc}
%We say that a potential $V(\g_i,\g_j)$ is purely hard-core
%if the incompatibility relation $\RR_V$ induced by $V$ is reflexive and
%$V(\g,\g')=0$ for all pairs $(\g,\g')\not\in \RR_{V}$ (i.e.
%$V(\g,\g')=0$ for all compatible pairs $\g\sim\g'$
%\end{defi}
%
%

\\Fix now  a finite set $\L\subset \PP$ (the ``volume" of the gas). Then, for $z_\g\ge 0$, the probability to see the configuration
$(\g_{1},\dots ,\g_{n})\in \L^n$ is given by
$$
Prob(\g_{1},\dots ,\g_{n})={1\over \Xi_{\L}} {1\over n!}{{z_{\g_1}}{z_{\g_2}}\dots{z_{\g_n}}
e^{- \sum_ {1\leq i<j\leq n}V(\g_i,\g_j)}}
$$
where the normalization constant $\Xi_{\L}$ is
the grand-canonical partition function  in the volume $\L$ and  is given by
$$
\Xi_{\L}(\bm z_\L)=1+\sum_{n\geq 1}{1\over n!}
\sum_{(\g_{1},\dots ,\g_{n})\in\L^n}
{z_{\g_1}}{z_{\g_2}}\dots{z_{\g_n}}
e^{- \sum_ {1\leq i<j\leq n}V(\g_i,\g_j)}\Eq(partiz)
$$
Note that  configurations $\g_1,\dots,\g_n$ for which  there exists some incompatible  pair $\g_i,\g_j$, i.e. such that $V(\g_i,\g_j)=+\infty$
have zero probability to occur, i.e. are  forbidden.

\\{\bf Remark}. The fact that  $V(\g,\g')\ge 0$  immediately
implies that $\X_\L$ is convergent and
$$
|\Xi_{\L}(\bm z_\L)|\le 1+\sum_{n\geq 1}{1\over n!}
\Bigg[\sum_{\g\subset\L}
|z_{\g}|\Bigg]^n\le \exp\Big\{\sum_{\g\in \L}|z_\g| \Big\}\le \max_{\g\in \L}e^{|z_\g||\L| }
$$
Actually, $V(\g\,\g')\ge 0$   implies that $\X_\L(z)$ is  analytic in the whole $\mathbb{C}^{|\L|}$ ($|\L|$ is the cardinality of $\L$).
\vv
\\The  ``pressure'' of this gas is defined via the formula
$$
P_\L(\bm z_\L)={1\over |\L|}\log \Xi_{\L}(\bm z_\L)\Eq(pressure)
$$
While the partition function $\Xi_{\L}(\bm z_\L)$ diverges as $\L\to \PP$, the pressure $P_\L(\bm z_\L)$ is supposed to have
a finite limit at least when $z$ varies in some finite polydisc
$|z_\g|\le \r_\g$ with $\bm \r=\{\r_\g\}_{\g\in \PP}$ being some positive function
$\bm \r:\PP\to \mathbb{R}^+: \g\mapsto\r_\g$ defined on $\PP$ independent on $\L$.
So that in principle  it should be possible to give an upper bound for $|P_\L(\bm z)|$ which is uniform in $\L$.

The pressure \equ(pressure) can be written as a formal series using the
Mayer trick in the partition function \equ(partiz)
by writing  the Gibbs factor
as
$$
e^{-\!\sum_{1\leq i<j\leq
n}V(\g_i,\g_j)}\,=\,
\prod_{1\leq i<j\leq n}\!\left[(e^{-V(\g_i,\g_j)}-1)+1\right]
$$
Proceeding exactly as in section \ref{mayersec} with the only difference if that now $V(\g_i,\g_j)$ is in place of $V(x_i-x_j)$ and $\sum_{\g_i\in \L}$ is in place of
$\int_\L dx_i$ we obtain

$$
\log\Xi_{\L}(\bm z_\L)= \sum_{n=1}^{\infty}{1\over n!}
\sum_{(\g_{1},\dots ,\g_{n})\in\L^n}
\phi^{T}(\g_1 ,\dots , \g_n)\,z_{\g_1}\dots\,z_{\g_n}\Eq(6)
$$
with
$$
\phi^{T}(\g_{1},\dots ,\g_{n})=
\begin{cases}
1&{\text if}~ n=1\\
\sum\limits_{g\in G_{n}}\prod\limits_{\{i,j\}\in E_g}(e^{-V(\g_i,\g_j)}-1)&{\text if}~ n\ge 2
\end{cases}
\Eq(7)
$$
where  $\sum_{g\in{G}_n}$ is the sum over all connected
graphs between $[n]$.

\\Observe  that,  since $V(\g,\g')\ge 0$, by Proposition \ref{alternate} it holds

$$
\phi^{T}(\g_{1},\dots ,\g_{n})=(-1)^{n-1}|\phi^{T}(\g_{1},\dots ,\g_{n})|\Eq(altern)
$$

\\The equation \equ(6) makes sense  only for those $\bm z\in \C^\PP$ such that  the formal series in the r.h.s. of \equ(6) converge absolutely.
To study  absolute convergence,
we will consider, for any $\L$, the positive term series
$$
|\log \Xi|_{\L}(\bm \r_\L)=
\sum_{n=1}^{\infty}{1\over n!}
\sum_{(\g_{1},\dots ,\g_{n})\in \L^n}
|\phi^{T}(\g_1 ,\dots , \g_n)|\r_{\g_1}\cdots{\r_{\g_n}}
\Eq(6abs)
$$
for $\bm \r\in (0,\infty)^{\PP}$.
Note that
$$
|\log\Xi_{\L}(\bm z_\L)|\le |\log \Xi|_{\L}(\bm \r_\L)
$$
for all $\bm z\in \C^{|\L|}$ in the poly-disc  $\{|z_\g|\le \r_\g \}_{\g\in \L}$, so that if we are able to prove
that the series \equ(6abs) converges, for all $\L$,
at some value   $\bm \r=\{\r_\g\}_{\g\in \PP}\in (0,\infty)^{\PP}$,
then we have also proved that
the series \equ(6) converges absolutely, for all $\L$, whenever $\bm z$ is in the poly-disk $\{|z_\g|\le \r_\g\}_{\g\in \PP}$.
We also observe that,  since $V(\g,\g')\ge 0$, by \equ(altern)
$$
|\log \Xi|_{\L}(\bm \r_\L)=-\log\Xi_{\L}(-\bm \r_\L) \Eq(logmr)
$$
To study the convergence of the pressure \equ(pressure) it is convenient to consider the quantity
%$$
%\Pi_{\g_0}(\bm z_{\L})={\partial\over \partial z_{\g_0}}\log \Xi_{\L}(z_\L)
%\Eq(relpilog)
%$$
%whose formal  expansion is immediately derive from \equ(6)
%$$
%\Pi_{\g_0}(z_{\L})=\sum_{n=0}^{\infty}{1\over n!} \sum_{(\g_1,\g_2,\dots,\g_n)\in \L^n}
%\phi^{T}(\g_0,\g_1 ,\dots , \g_n){z_{\g_1}}\dots{z_{\g_n}}\Eq(PiLa)
%$$
%By  \equ(altern)

%\\We further define a  function
%directly related to \equ(6) (a ``pinned'' sum defined in the whole set $\PP$).
%Namely, for a fixed activity $\g_0$ we define the ``infinite volume" pinned sum
%as follows
%$$
%|\Si_{\g_0}|(\r)=
%\sum_{n=1}^{\infty}{1\over n!}
%\sum_{(\g_{1},\dots ,\g_{n})\in\PP^n\atop \exists i:~\g_i=\g_0}
%|\phi^{T}(\g_1 ,\dots , \g_n)|\;\r_{\g_1}\cdots{\r_{\g_n}}
%\Eq(6b)
%$$
%

%\\Now observe that
%$$
%\sum_{(\g_{1},\dots ,\g_{n})\subset\PP^n\atop \exists i:~\g_i=\g_0}=
%\sum_{(\g_1=\g_0,\g_2,\dots,\g_n)\in \PP^n}{n\over m_{\g_0}(\g_1,\dots,\g_n)}\Eq(!)
%$$
%where, calling ${\rm I}_n=\{1,\dots,n\}$,
%\def\setn{{\rm I}_n}
%$$
%m_{\g_0}(\g_1,\dots ,\g_n))= |\{i\in \setn: \g_i=\g_0\}
%$$
%So we can rewrite $|\Si_{\g_0}|(|\r|)$ as follows

%
%$$
%|\Si|_{\g_0}(\r)=
%\sum_{n=0}^{\infty}{1\over n!} \sum_{(\g_1,\g_2,\dots,\g_n)\in \PP^n}
%{|\phi^{T}(\g_0,\g_1 ,\dots , \g_n)|\over m_{\g_0}(\g_1,\dots,\g_n)+1}
%\;\;{\r_{\g_0}}\r_{\g_1}\cdots{\r_{\g_n}}
%$$
%It is then  convenient to define the function

$$
|\Pi|_{\g_0}(\bm \r)=\sum_{n=0}^{\infty}{1\over n!} \sum_{(\g_1,\g_2,\dots,\g_n)\in \PP^n}
|\phi^{T}(\g_0,\g_1 ,\dots , \g_n)|\,\,\r_{\g_1}\cdots{\r_{\g_n}}\Eq(TP)
$$
%The two formal series
%$|\Pi|_{\g_0}(\r)$ and $|\Si|_{\g_0}(\r)$ are directly related. Observe indeed that
%$$
%|\Si|_{\g_0}(\r)= {\r_{\g_0}}\int_0^1 |\Pi_{\g_0}|(\r(\a))d\a\Eq(relspm)
%$$
%with
%$$
%\r_\g(\a)=\cases{\r_\g &if $\g\neq \g_0$\cr\cr
%\a \r_\g &if $\g=\g_0$
%}
%$$
%Observe now that
%$$
%|\Si|_{\g_0}(\r)\le \r_{\g_0}|\Pi|_{\g_0}(\r)\Eq(PminTP)
%$$
%and, for all $\{|z_\g|\le \r_\g\}_{\g\in \PP}$,
%$$
%|P_\L(z_\L)|={1\over |\L|}|\log \Xi_{\L}(z_\L)|\;\le\; \sup_{\g_0\in \PP}\,\,\r_{\g_0}\,|\Pi|_{\g_0}(\r)
%\Eq(PminTP)
%$$
Note that in \equ(TP) the sum over each polymer $\g_i$ is not anymore restricted to the finite ``volume" $\L$ and it runs over the full polymer space $\PP$.

\\If we are able to show that $|\Pi|_{\g_0}(\bm \r)$ %(and hence $|\Si|_{\g_0}(\r)$)
converges
for some (bounded) positive function  $\bm\r\in [0,\infty)^\PP$, then also the pressure
$P_\L(\bm z_\L)$ converges absolutely, for all $\L$, whenever $z$ is in the poly-disk $\{|z_\g|\le \r_\g\}_{\g\in \PP}$ and in this poly-disk it uniformly in $\L$. Indeed
 in the poly-disk $\{|z_\g|\le \r_\g\}_{\g\in \PP}$ we have
$$
|P_\L(\bm z_\L)|={1\over |\L|} |\log\Xi_{\L}(\bm z_\L)|\le  {1\over |\L|}|\log \Xi|_{\L}(\bm \r_\L) ~=
$$
$$
= {1\over |\L|}\sum_{n=1}^{\infty}{1\over n!}
\sum_{(\g_{1},\dots ,\g_{n})\in \L^n}
|\phi^{T}(\g_1 ,\dots , \g_n)|\r_{\g_1}\cdots{\r_{\g_n}}=
$$
$$
={1\over |\L|}\sum_{\g_0\in \L}\r_{\g_0}\sum_{n=1}^{\infty}{1\over n!}
\sum_{(\g_{1},\dots ,\g_{n-1})\in \L^{n-1}}
|\phi^{T}(\g_0,\g_1 ,\dots , \g_{n-1})|\r_{\g_1}\cdots{\r_{\g_{n-1}}}~=
$$
$$
\le~ {1\over |\L|}\sum_{\g_0\in \L}\r_{\g_0}\sum_{m=0}^{\infty}{1\over (m+1)!}
\sum_{(\g_{1},\dots ,\g_{m})\in \L^{m}}
|\phi^{T}(\g_0 ,\g_1\dots , \g_m)|\r_{\g_1}\cdots{\r_{\g_m}}~\le
$$
$$
{1\over |\L|}\sum_{\g_0\in \L}\r_{\g_0}\sum_{m=0}^{\infty}{1\over m!}
\sum_{(\g_{1},\dots ,\g_{n})\in \L^{m}}
|\phi^{T}(\g_0 ,\g_1\dots , \g_m)|\r_{\g_1}\cdots{\r_{\g_m}}~=
$$
$$
=~  {1\over |\L|}\sum_{\g_0\in \L}\r_{\g_0}|\Pi|_{\g_0}(\bm\r)~\le
\sup_{\g_0\in \L}\,\,\r_{\g_0}\,|\Pi|_{\g_0}(\bm\r)~\le ~\sup_{\g_0\in \PP}\,\,\r_{\g_0}\,|\Pi|_{\g_0}(\bm\r)
$$
In short, for all $\bm z$ in the polydisk $\{|z_\g|\le \r_\g\}_{\g\in \PP}$, it holds
$$
|P_\L(\bm z_\L)|~\le~ \sup_{\g_0\in \PP}\,\,\r_{\g_0}\,|\Pi|_{\g_0}(\bm \r)\Eq(PminTP)
$$
So in the next sections we will focus our attention on the  formal series $|\Pi|_{\g_0}(\bm \r)$
defined in equation \equ(TP). Indeed, when one is able to prove that the series \equ(TP) then, by \equ(PminTP)
he has also proved
the absolute convergence of the pressure of the polymer gas uniformly in the volume $\L$.

%\\To understand the physical interpretation of the series $|\Pi|_{\g_0}(\r)$ and $|\Si|_{\g_0}(\r)$ let us consider the
%finite volume functions  $\Si_{\g_0}(z_{\L})$ and $\Pi_{\g_0}(z_{\L})$ given by
%$$
%\Si_{\g_0}(z_{\L})=\sum_{n=1}^{\infty}{1\over n!}
%\sum_{(\g_{1},\dots ,\g_{n})\subset\L^n\atop \exists i:~\g_i=\g_0}
%\phi^{T}(\g_1 ,\dots , \g_n){z_{\g_1}}\dots{z_{\g_n}}\Eq(sigl)
%$$
%and
%$$
%\Pi_{\g_0}(z_{\L})=\sum_{n=0}^{\infty}{1\over n!} \sum_{(\g_1,\g_2,\dots,\g_n)\in \L^n}
%\phi^{T}(\g_0,\g_1 ,\dots , \g_n){z_{\g_1}}\dots{z_{\g_n}}\Eq(PiLa)
%$$
%These are directly related to $\log \Xi_{\L}(z_\L)$. It is immediate in fact to see that
%
%
%$$
%\Pi_{\g_0}(z_{\L})={\partial\over \partial z_{\g_0}}\log \Xi_{\L}(z_\L)
%\Eq(relpilog)
%$$
%and
%
%$$
%\Si_{\g_0}(z_{\L})=\log \Xi_{\L}(z)- \log \Xi_{\L\backslash\g_0}(z)
%\Eq(diflog)
%$$
%
%\\So we have indeed
%
%$$
%\Si_{\g_0}(z_{\L})=\int_0^1 {\partial\over \partial \a} \log\Xi_{\L}(z_\L(\a))d\a=
%z_{\g_0}\int_0^1 {\partial\over \partial {z_{\g_0}}}\log \Xi_{\L}(z_\L(\a))d\a=
%$$
%$$
%={z_{\g_0}}\int_0^1 \Pi_{\g_0}(z_\L(\a))d\a
%$$
%with
%$$
%z_\g(\a)=\cases{z_\g &if $\g\neq \g_0$\cr\cr
%\a z_\g &if $\g=\g_0$
%}
%$$

\section{Convergence of  the  abstract polymer gas}
\vv

As explained in the previous section, to study  absolute convergence of the pressure
we will just need to prove that the positive term series  defined in \equ(TP)

$$
|\Pi|_{\g_0}(\bm \r)=\sum_{n=0}^{\infty}{1\over n!} \sum_{(\g_1,\g_2,\dots,\g_n)\in \PP^n}
|\phi^{T}(\g_0,\g_1 ,\dots , \g_n)|\,\,\r_{\g_1}\cdots{\r_{\g_n}}
$$

\\is convergent for some $\bm \r\in (0,+\infty)^\PP$.

\\Since the interaction $V(\g,\g')$ is purely hard core, we can use the original Penrose identity to bound the factor Ursell factor $|\phi^{T}(\g_0,\g_1 ,\dots , \g_n)|$. We define the
map $\bm m: T^0_n\to G^0_n$ as in Definition \ref{penspart}. By rooting trees in $T_n^0$ in $0$, given $\t\in T^0_n$, $d(i)$ denotes  the depth  of the vertex $i$ (i.e. its edge distance  from $0$)
and that $i'$  denotes the parent of $i$.
The  map $\mi:T^0_n\to G^0_n$ is therefore the map such that to each tree $\tau\in T^0_n$ associates the graph
$\mi(\tau)\in G^0_n$ formed by adding  to $\tau$
all edges $\{i,j\}$ such that either:
\begin{itemize}
\item[(p1)]  $d_\t(i)=d_\t(j)$ (edges between vertices of the same generation), or
\item[(p2)]  $d_\t(j)=d_\t(i)-1$ and $j>i'$ (edges between vertices with generations differing by one).
\end{itemize}

\\As shown previously  the map $\mi$ is a partition scheme and therefore by Theorem \ref{Penid} we have that
$$
\phi^{T}(\g_0,\g_1 ,\dots , \g_n)= \sum\limits_{g\in G^0_{n}}\prod\limits_{\{i,j\}\in E_g}(e^{-V(\g_i,\g_j)}-1)= \sum_{\t\in T^0_n}w_\t(\g_0\,\g_1,\dots,\g_n)
$$
with
$$
w_\t(\g_0\,\g_1,\dots,\g_n)=
e^{-\sum_{\{i,j\}\in E_{\mi(\t)}\backslash E_\t}V(\g_i,\g_j)}
\prod_{\{i,j\}\in \t}\left(e^{- V(\g_i,\g_j)}-1\right)
$$

\def\v{\vskip.1cm}
\def\vv{\vskip.2cm}
\def\gt{{\tilde\g}}
\def\E{{\mathcal E} }
\def\I{{\rm I}}
\def\GI{\mathbb{G}}
\def\EE{\mathbb{E}}\def\VU{\mathbb{V}}

\\As shown in Section \ref{sechard}  (see there Theorem \ref{pencore} and formula \equ(wtpen)) we have
$$
w_\t(\g_0\,\g_1,\dots,\g_n)= (-1)^n\1_{P(\g_0,\g_1,\dots,\g_n)}(\t)
$$
where $P(\g_0,\g_1,\dots,\g_n)$ is the set of Penrose trees and
 $\1_{\t\in P(\g_0,\g_1,\dots,\g_n)}$ is the characteristic function of the set $P(\g_0,\g_1,\dots,\g_n)$ in $T^0_n$, i.e.
$$
\1_{P(\g_0,\g_1,\dots,\g_n)}(\t)\;=\;
\begin{cases}1 &{\text if}~ \t\in P(\g_0,\g_1,\dots,\g_n)\\
0 & {\text if otherwise}
\end{cases}
$$
We remind its definition below.

\begin{defi} Given $(\g_0,\g_1,\dots,\g_n)\in \PP^{n+1}$, we define
the set of \emph{Penrose trees}
$P(\g_0,\g_1,\dots,\g_n)$, as the set   formed by the trees $\t\in T^0_n$ such that
\begin{enumerate}

\item[{\rm (t0)}]  if   $\{i,j\}\in E_\t$ then  $\g_i\nsim\g_j$

\item[{\rm (t1)}]  if   two vertices $i$ and $j$ siblings
then $ \g_i\sim\g_j$;

\item[{\rm (t2)}]  if two vertices $i$ and $j$ are not siblings but  $d_\t(i)=d_\t(j)$, then $  \g_i\sim\g_j$;

\item[{\rm (t3)}]  if two vertices $i$ and $j$ are s. t. $d_\t(j)=d_\t(i)-1$ and $j>i'$, then
 $\g_i\sim\g_j$.
\end{enumerate}
\end{defi}
Therefore we get

$$
\phi^{T}(\g_0,\g_{1},\dots ,\g_{n})\; =\;(-1)^n
\sum\limits_{\t\in T_n^0} \1_{ P(\g_0,\g_1,\dots,\g_n)}(\t)
\Eq(r.20d)
$$

\vskip.2cm
\\Using \equ(r.20d)  we can now rewrite the formal series \equ(TP) as

$$
|\Pi|_{\g_0}(\bm\r)=\sum_{n=0}^{\infty}{1\over n!} \sum_{(\g_1,\g_2,\dots,\g_n)\in \PP^n}
\sum\limits_{\t\in T_n^0} \1_{ P_G(\g_0,\g_1,\dots,\g_n)}(\t)\,\,{\r_{\g_1}}\dots{\r_{\g_n}}\;=
$$
$$
\,\,\,\,\,\,\,\,\,\,\,\,\,\,\,\,\,\,\,\,\,\,\,\,\,\,
= \;\sum_{n=0}^{\infty}{1\over n!} \sum\limits_{\t\in T_n^0}  \sum_{(\g_1,\g_2,\dots,\g_n)\in \PP^n}
\1_{ P_G(\g_0,\g_1,\dots,\g_n)}(\t)\,\,{\r_{\g_1}}\dots{\r_{\g_n}}\; =
$$
$$
=\;\sum_{n=0}^{\infty}{1\over n!} \sum\limits_{\t\in T_n^0}  \phi_{\g_0}(\t) \;\; \;\;
\,\,\,\,\,\,\,\,\,\,\,\,\,\,\,\,\,\,\,\,\,
\,\,\,\,\,\,\,\,\,\,\,\,\,\,\,\,\,\,\,\,\,\,\,\,\,\,\,\,\,\,\,\,\,\,\,\,\,\,\,
\,\,\,\,~~~~~~~\Eq(pitre)
$$
where
$$
\phi_{\g_0}(\t)=\; \sum_{(\g_1,\g_2,\dots,\g_n)\in \PP^n}
\1_{ P_G(\g_0,\g_1,\dots,\g_n)}(\t)\,\,{\r_{\g_1}}\dots{\r_{\g_n}}\;\Eq(lbdp)
$$
The structure of \equ(pitre) is crucial. This equation shows that the formal series $|\Pi|_{\g_0}(\r)$ can be
reaorganized as a sum over terms associated to labelled trees.

\\We stress that here the factor $\phi_{\g_0}(\t)$ depends on the  {\it labelled} tree $\t$ because of the Penrose condition
(t3) which  indeed depends on the labelling of the tree.

\\We will see below that efficient  bounds on the factor $\phi_{\g_0}(\t)$ defined in \equ(lbdp) above can be obtained  by
choosing  a family  of trees $\tilde P(\g_0,\g_{1},\dots ,\g_{n})$ such that
${P}(\g_0,\g_1,\dots,\g_n)\subset  {\tilde P}(\g_0,\g_{1},\dots ,\g_{n})$.
Below we will consider  three possible choice of $\tilde P{(\g_0,\g_{1},\dots ,\g_{n})}$ which are the choices which yield the
three known criteria for the convergence of cluster  expansion of the abstract polymer gas.

\begin{defi}
Given $(\g_0,\g_1,\dots,\g_n)\in \PP^{n+1}$, the set of weakly the Penrose trees, which we denote by is defined as
${P}^*(\g_0,\g_1,\dots,\g_n)$, is formed by all trees $\t$ with vertex set $\{0,1,\dots,n\}$ and edge set $E_\t$ such that
\begin{enumerate}
\item[$\rm (t0)\;\,$]  if $\{i,j\}\in E_\t$ then $\g_i\not\sim\g_j$

\item[$\rm (t1)^*$]  if $i$ and $j$ are siblings then $\g_i\sim\g_j$;
\end{enumerate}
\end{defi}

\begin{defi}
Given $(\g_0,\g_1,\dots,\g_n)\in \PP^{n+1}$, the set of  the Dobrushin trees, which we denote by is defined as
${P}^{\rm Dob}(\g_0,\g_1,\dots,\g_n)$,
is the set the  trees $\t$ with vertex set $\{0,1,\dots,n\}$ and edge set $E_\t$ such that

\begin{enumerate}
\item[$\rm (t0)\;\,$]  if $\{i,j\}\in E_\t$ then $\g_i\not\sim\g_j$

\item[$\rm (t1)^D$]  if $i$ and $j$ are siblings  then $\g_i\neq\g_j$;
\end{enumerate}
\end{defi}

\begin{defi}\label{kp-trees}
Given $(\g_0,\g_1,\dots,\g_n)\in \PP^{n+1}$, the set of  Koteck\'y-Preiss trees, which we denote by ${P}^{\rm KP}(\g_0,\g_1,\dots,\g_n)$ is formed by
all  trees $\t$ with vertex set $\{0,1,\dots,n\}$ and edge set $E_\t$ such that

\begin{enumerate}
\item[$\rm (t0)\;\,$]  if $\{i,j\}\in E_\t$ then $\g_i\not\sim\g_j$

\end{enumerate}

\end{defi}

\\Note that we have, by definition, that
$$\begn
{P}(\g_0,\g_1,\dots,\g_n) & \subset \;{P}^*(\g_0,\g_1,\dots,\g_n)\\
& \subset \;{P}^{\rm Dob}(\g_0,\g_1,\dots,\g_n)\\
& \subset\; {P}^{\rm KP}(\g_0,\g_1,\dots,\g_n)
\egn
\Eq(fpdkp)
$$

\\So that, by \equ(r.20), we have the bounds
$$
\begn
|\phi^{T}(\g_0,\g_1 ,\dots , \g_n)|&\le\sum_{\t\in T_n^0} \1_{{P}^*(\g_0,\g_1,\dots,\g_n)}(\t)\\
&\le
\sum_{\t\in T_n^0} \1_{P^{\rm Dob}(\g_0,\g_1,\dots,\g_n)}(\t)\\
&\le\sum_{\t\in T_n^0}  \1_{P^{\rm KP}(\g_0,\g_1,\dots,\g_n)}(\t)
\egn
\Eq(pfi)
$$
The latter \equ(pfi), i.e. the worst among the three bound proposed,
is known as the Rota bound. This was the bound which was  used by Cammarota \cite{C},  Brydges \cite{B} and Simon
\cite{S} to obtain a direct proof of the absolute convergence of the pressure of a Polymer gas by directly
estimating the Ursell coefficent.

\\Hence, we can bound the positive term series $|\Pi|_{\g_0}(\bm\r)$ defined in  \equ(TP), using of course the first of the estimates \equ(pfi) which is the best
among the three given above, as
$$
|\Pi|_{\g_0}(\bm\r)\;=\;\sum_{n=0}^{\infty}{1\over n!} \sum_{(\g_1,\g_2,\dots,\g_n)\in \PP^n}
|\phi^{T}(\g_0,\g_1 ,\dots , \g_n)| {\r_{\g_1}}\dots{\r_{\g_n}}\;\le
$$
$$
\le\;\sum_{n=0}^{\infty}{1\over n!} \sum_{(\g_1,\g_2,\dots,\g_n)\in \PP^n}
\sum_{\t\in T_n^0} \1_{{P}^*(\g_0,\g_1,\dots,\g_n)}(\t){\r_{\g_1}}\dots{\r_{\g_n}}\;=
$$
$$
=\;\sum_{n=0}^{\infty}{1\over n!} \sum_{\t\in T_n^0}\sum_{(\g_1,\g_2,\dots,\g_n)\in \PP^n}
\1_{{P}^*(\g_0,\g_1,\dots,\g_n)}(\t){\r_{\g_1}}\dots{\r_{\g_n}}
$$
So we get
$$
|\Pi|_{\g_0}(\bm\r)\;\le\; \Pi^*_{\g_0}(\bm\r) \Eq(PiP*)
$$
where
$$
\Pi^*_{\g_0}(\bm\r) \;=\;
\sum_{n=0}^{\infty}{1\over n!}
\sum_{\t\in T^0_{n}}\phi^*_{\g_0}(\t,\bm\r)\;\Eq(exact)
$$
with
$$
\phi^*_{\g_0}(\t,\bm\r)=\; \sum_{(\g_1,\g_2,\dots,\g_n)\in \PP^n}
\1_{{P}^*(\g_0,\g_1,\dots,\g_n)}(\t)\,\,{\r_{\g_1}}\dots{\r_{\g_n}}\; \Eq(factor)
$$
Analogously, using the second and third bounds \equ(pfi)  we can define two more series
which also majorize   the series $|\Pi|_{\g_0}(\r)$. Namely
$$
\Pi^{\rm Dob}_{\g_0}(\bm\r) \;=\;
\sum_{n=0}^{\infty}{1\over n!}
\sum_{\t\in T^0_{n}}\phi^{\rm Dob}_{\g_0}(\t,\r)\;\Eq(exactd)
$$
and
$$
\Pi^{\rm KP}_{\g_0}(\bm\r) \;=\;
\sum_{n=0}^{\infty}{1\over n!}
\sum_{\t\in T^0_{n}}\phi^{\rm KP}_{\g_0}(\t,\r)\;\Eq(exactkp)
$$

\\with
$$
\phi^{\rm Dob}_{\g_0}(\t,\bm\r)=\; \sum_{(\g_1,\g_2,\dots,\g_n)\in \PP^n}
\1_{P^{\rm Dob}(\g_0,\g_1,\dots,\g_n)}(\t)\,\,{\r_{\g_1}}\dots{\r_{\g_n}}\; \Eq(factor2)
$$
and
$$
\phi^{\rm KP}_{\g_0}(\t,\bm\r)=\; \sum_{(\g_1,\g_2,\dots,\g_n)\in \PP^n}
\1_{P^{\rm KP}(\g_0,\g_1,\dots,\g_n)}(\t)\,\,{\r_{\g_1}}\dots{\r_{\g_n}}\; \Eq(factor3)
$$

\\Here it is important to stress that, differently form the factor $\phi_{\g_0}(\t)$ defined in \equ(lbdp),
the three factors $\phi^{*,\,\rm Dob,\, KP}_{\g_0}(\t)$ do not depend on the labels of the tree $\t$, but only on its topological structure.
This means that the terms in the series $\Pi^*_{\g_0}(\r)$ can further be grouped together in terms of non unlabelled rooted trees.
We will make this concept precise in the next section. We conclude this section by remarking that  inequalities \equ(pfi)
immediately imply
$$
|\Pi|_{\g_0}(\bm\r)\;\le\; \Pi^*_{\g_0}(\bm\r) \;\le\; \Pi^{\rm Dob}_{\g_0}(\bm\r) \;\le\; \Pi^{\rm KP}_{\g_0}(\bm\r)  \Eq(Pitut)
$$

\subsection{Reorganization of the series  $\Pi^*_{\g_0}(\bm\r)$ }\label{treesec}
We now reorganize the sum over rooted labelled trees appearing
in formula \equ(exact) in terms of the plane rooted trees. Such reorganization
is motivated by the observation that
the factors \equ(factor)-\equ(factor3) do not depend on the labels
assigned to the vertices of $\t$ but only its topological structure.
As a  matter of fact,
to each labelled ordered rooted tree $\t\in T^0_{n}$ we can associate a drawning in the plane known as
the ``plane rooted tree'' associated to $\t$.
The drawing of $\t$ is obtained by putting  parents at the left of their children which are ordered
in the top-to-bottom order consistently with the order of their labels.
For example the plane rooted trees with $n+1=5$ vertices associated to the trees $\bf a$ with edge set
$\{0,3\},\{1,3\}, \{2,3\}, \{1,4\}$,  $\bf b$ with edge set $\{0,2\},\{0,3\}, \{1,2\}, \{2,4\}$ and $\bf c$ with edge set
$\{0,2\},\{0,4\}, \{4,3\}, \{1,4\}$
are drawn below
\vv\vv
\includegraphics[width=12cm,height=4.0cm]{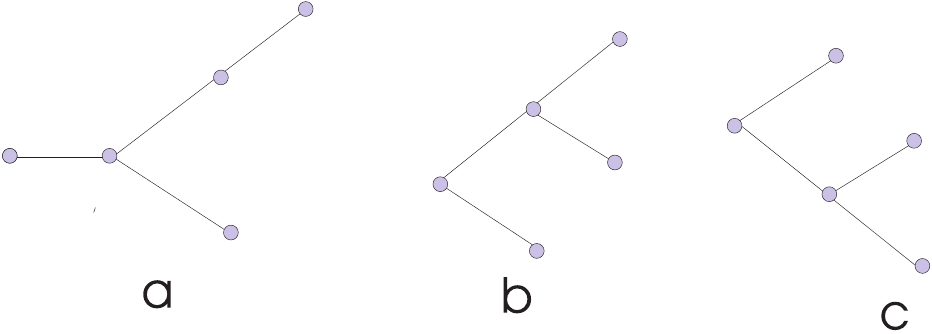}
\begin{center}
Figure
\end{center}
\vv\vv
Observe that $\it b$ and $\it c$ are {\it different} plane rooted tree
(because of  the rule of the ordering of the children from top-to-bottom).

In this way we have defined a map $m:\t\mapsto m(\t)$ which associate to each labelled tree $\t\in T^0_{n}$ a unique  drawing $t=m(\t)$ in the plane,
called the {\it planar rooted tree} associated to $\t$.
We denote by $ {\TT}^0_n=$  the set of all planar rooted trees
with $n+1$ vertices and by  $\TT^{0,k}$ the
set of planar rooted trees with maximal generation number $k$; let also $\TT^0= \cup_{n\ge 0}{\TT}^0_n=\cup_{k\ge 0}\TT^{0,k}$ be
the set of all planar rooted trees.
An element $t\in \TT^0_n$ can also be viewed as an equivalence class of elements
$\t\in T_{n+1}$ with the equivalence relation being that
two elements $\t$ and $\t'$ are equivalent if the produce the same planar rooted tree. So
when we write $\t\in t$ with $t\in \TT^0_n$ we mean that $\t$ is an element of the set of all labelled trees in $T^0_n$
that produce the same plane ordered rooted tree. %An alternative way to define this equivalence relation in $T^0_n$ is by mean of
%permutations of the labels $\{1,2,\dots, n\}$. Namely, we say that
%$\t'\in T^0_n$ is equivalent to   $\t\in T^0_n$ if there exists a permutation $\s$ in
%$\{1,2,\dots, n\}$  preserving the order of the children in each vertex of $\t$ such that $\t'=\s(\t)$
%(here $\s(\t)$ is the graph with vertex set $V_{\s(\t)}=\{0,\s(1),\dots,\s(n)\}$
%and edge set $E_{\s(\t)}=\{\{\s(i),\s(j)\}: \{i,j\}\in E_\t,\}$).
%The set of all equivalence classes is precisely the set of plane
%rooted tree with $n$ vertices distinct from the root.

\\We will use the following notations.
Given a vertex $v\neq 0$ in a rooted tree (with root $0$), we denote by $v'$ is  parent,
we denote by $s_v$ the number of its children and  we  denoted by $v^1,\dots , v^{s_v}$ the
children of $v$. If $s_v=0$ we say that $v$ is an end-point or a {\it leaf} of $\t$.

\\Note that  the set of vertices of a labeled rooted tree  $\t\in T^0_n$ (a plane rooted tree $t\in \TT^0_n$) can be endowed with a total order $\prec$  in a natural way
in such way that  for any $v$ the father $v'$ of $v$ is such that $v'\prec v$ and children of any vertex $v$ of a labeled tree $\t\in T^0_n$  (a plane rooted tree  $t\in \TT^0_n$) are ordered
following the order of their labels (from high to low);  namely,  the ordering of $v^1,\dots , v^{s_v}$ is such that
$v^1\prec v^2\prec \cdots  \prec v^{s_v}$.
%From now on we will always elements $\t\in T^0_n$ as ordered trees in the sense that children
%$v^1,\dots , v^{s_v}$o f
%any (internal) vertex $v$ are ordered accordingly to the order of their labels, i.e. in such way that $v^1< v^2< \cdots  <v^{s_v}$.
%In other words, the  natural order of vertices in a rooted labels tree   is realize by the permutation $\s$ of the
%labels $\{1,\dots ,n\}$ of $\t$
%such that $\s(i)$ is the
%natural label of the  plane rooted tree $t= [\t]$.

Clearly the  map $\t\mapsto m(\t)=t$ is many-to-one and the cardinality of the
pre-image of a planar rooted tree
$t$ (=number of ways of labelling the $n$ non-root vertices of the tree with $n$ distinct labels
consistently with the rule ``from high to low") is given by
$$
\card{\{\t\in T^0_n: m(\t)=t\}}\;=\; {n!\over\prod_{v\succeq 0} s_{v}!}\Eq(rel1)
$$
As a matter of fact, it is very easy to count how many labeled trees $\t\in T^0_n$ belong to the same equivalent class
$t$, i.e. are associated to the same plane rooted tree. One have just to count all permutations
$\s$ of $\{0,1,2,\dots,n\}$ which leaves the root unchanged and which respect the order of the children
in any vertex. Let $\t\in T^0_n$  and let  $t=[\t]$ the plane root tree associated to
$\t$ characterized by the sequence $\{s_v\}_{v\succeq 0}$ then we have
$$
|[\t]|= {n!\over \prod_{v\succeq 0}s_v!}
$$
Indeed, $n!$ is the number of all permutations in the set $\{1,2,\dots, n\}$,
while for each vertex $v$, $s_v!$ are the permutation of the children. So ${n!/ \prod_{v\succeq 0}s_v!}$
is the number of permutations of the vertices of $\t$ different from the root which do not  change the order of the children in every vertex.

\def\T{\mathbb{T}}

\\Now observe that the factor \equ(factor) depend actually only on the plane rooted tree associated to $\t$.
We have indeed
$$
\phi^*_{\g_0}(\t,\bm\r)= \phi^*_{\g_0}([\t],\bm\r)=\phi^*_{\g_0}(t,\bm\r)=
\prod_{v\succeq 0}
\Bigg[
\sum_{(\g_{v^1},\dots,\g_{v^{s_v}})\in \PP^{s_v}\atop \g_{v^i}\not\sim \g_v,\,\,\,\g_{v^i}\sim\g_{v^j}}
\r_{\g_{v^1}}\dots \r_{\g_{v^{s_v}}}
\Bigg]\Eq(explw1)
$$
Note that, since in each vertex $v$ the sum over polymers $\g_{v_1},\dots,\g_{v_{s_v}}$
associated to children of $v$ depends on the polymer $\g_v$ associated to $v$,
in the expression  above the order of the product is relevant and it is organized in such way that
products corresponding to ancestors are  at the left of products corresponding
to descendants.
In the r.h.s. of \equ(explw1) it  also adopted the convention that the sum in brackets is equal to 1 for a vertex $v$ such that $s_v=0$.

\\We now ready to reorganize the sum in the the r.h.s. of \equ(exact)
$$
\Pi^*_{\g_0}(\bm\r) =
\sum_{n\ge0}{1\over n!}
\sum_{\t\in T^0_{n}}\phi^*_{\g_0}(\t,\bm\r)\;=
\sum_{n\geq 0}{1\over n!}\sum\limits_{t\in \TT^0_{n}}\sum_{\t\in t}\phi^*_{\g_0}(t,\r)=
$$
$$
=
\sum_{n\geq 0}{1\over n!}\sum\limits_{t\in \TT^0_{n}}\phi^*_{\g_0}(t,\r)\sum_{\t\in t}1=
\sum_{n\geq 0}{1\over n!}\sum\limits_{t\in \TT^0_{n}}\phi^*_{\g_0}(t,\r)|t|=
$$
$$
= \sum_{n\geq 0}\sum\limits_{t\in \TT^0_{n}}\Big[\prod_{v\succeq 0}{1\over s_v!}\Big]\phi^*_{\g_0}(t,\r)=
\sum_{n\geq 0}\sum\limits_{t\in \TT^0_{n}}\prod_{v\succeq 0}
\Bigg[{1\over s_v!}
\sum_{(\g_{v^1},\dots,\g_{v^{s_v}})\in \PP^{s_v}\atop \g_{v^i}\not\sim \g_v,\,\,\,\g_{v^i}\sim\g_{v^j}}
\r_{\g_{v^1}}\dots \r_{\g_{v^{s_v}}}
\Bigg]
$$

\\In conclusion we have obtained
$${
\Pi^*_{\g_0}(\bm\r)\;=}$$
$$
{  =\;\sum\limits_{t\in \TT^0}
\prod_{v\succeq 0}
\Bigg\{{1\over s_v!}\sum_{(\g_{v^1},\dots,\g_{v^{s_v}})\in \PP^{s_v}}
\prod^{s_v}_{i=1}\ind{\g_{v^i}\nsim\g_{v}}\prod_{1\le i<j\le s_v}\ind{\g_{v^i}\sim\g_{v^j}}\r_{\g_{v^1}}\dots \r_{\g_{v^{s_v}}}
\Bigg\}
}\Eq(aab)
$$
In a completely analogous way we also can obtain
$${
\Pi^{\rm Dob}_{\g_0}(\bm\r)\;=}$$
$$
{  =\;\sum\limits_{t\in \TT^0}
\prod_{v\succeq 0}
\Bigg\{{1\over s_v!}\sum_{(\g_{v^1},\dots,\g_{v^{s_v}})\in \PP^{s_v}}
\prod^{s_v}_{i=1}\ind{\g_{v^i}\nsim\g_{v}}\prod_{1\le i<j\le s_v}\ind{\g_{v^i}\neq\g_{v^j}}\r_{\g_{v^1}}\dots \r_{\g_{v^{s_v}}}
\Bigg\}
}\Eq(aabd)
$$
and
$${
\Pi^{\rm KP}_{\g_0}(\bm\r)\;=\;\sum\limits_{t\in \TT^0}
\prod_{v\succeq 0}
\Bigg\{{1\over s_v!}\sum_{(\g_{v^1},\dots,\g_{v^{s_v}})\in \PP^{s_v}}
\prod^{s_v}_{i=1}\ind{\g_{v^i}\nsim\g_{v}}\r_{\g_{v^1}}\dots \r_{\g_{v^{s_v}}}
\Bigg\}
}\Eq(aabkp)
$$

\subsection{Trees and convergence}
We start by defining a proper domain in the space of function $(0,\infty)^{\PP}$, i.e. the space of the functions
$\bm \m: \PP\to (0,\infty): \g\mapsto \m_\g$.
We are agree that given two functions $\bm\m$ and $\bm\n$ in $(0,\infty)^{\PP}$, we say that $\bm \m<\bm\n$ if
and only if $\m_\g<\n_\g$ for all $\g\in \PP$.

\\Given $n\in \mathbb{N}$ and
$(\g_0,\g_1,\ldots,\g_n)\in\PP^{n+1}$,  let  $b_n(\g_0;\g_1,\ldots,\g_n)\ge 0$.
Once numbers $b_n(\g_0;\g_1,\ldots,\g_n)$ are given for all $n\in \mathbb{N}$ and
$(\g_0,\g_1,\ldots,\g_n)\in\PP^{n+1}$, we can
define a function
$$
\varphi^b:  (0,\infty)^{\PP}\to  (0,\infty]^{\PP}: \bm u\mapsto \varphi^b(\bm u)
$$
with entries
$$
[\varphi^b(\bm u)]_\g \doteq \varphi^b_{{\g}} (\bm u)\;=\; 1+\sum_{n\geq 1}
\sum_{(\g_{1},\dots ,\g_{n})\in\PP^n} b_n(\g;\g_1,\ldots,\g_n)\,
u_{\g_1}\dots u_{\g_n}
\Eq(r.7.1)
$$
We define the set
$$
\mathcal{D}^b=\Big\{\bm u\in (0,\infty)^{\PP}: \;\;\varphi^b_{{\g}} (\bm u)<+\infty, \,\,\forall\g\in \PP\Big\}
$$
Then the restriction   of $\varphi^b$ to $\mathcal{D}^b$ is a function  in $(0,\infty)^\PP$,
i.e.
$$
\varphi^b_{{\g}} (\bm u)\;<\;\infty,\,\,\forall\g\in \PP, \,\,\,\,\,\,\,\,\,\,
\mbox{whenever $\bm u\in \mathcal{D}^b$}
$$
Note also that, if $\bm u\in \mathcal{D}^b$ and $\bm u'<\bm u$ then also $\bm u'\in \mathcal{D}^b$.

\\Let now $\bm \m: \PP\to (0,\infty)$ be a function in the set
$\mathcal{D}^b\subset (0,\infty)^{\PP}$
and  let $\bm r \in (0,\infty)^{\PP}$ be defined such that
its  entries  $r_\g$, as $\g$ varies in $\PP$, are given by
$$
r_\g= {\m_\g\over  \varphi^b_{{\g}}(\bm \mu)}\Eq(muR)
$$

\\Note that $\bm r\in \mathcal{D}^b$ because $\bm r\le\bm\m$ by construction, since $\varphi^b_\g(\bm \m)\ge 1$ for all $\g\in \PP$. Morevover
the assumption $\bm \m\in (0,\infty)^{\PP}$ is equivalent to say
$$
\m_\g>0 \,\,\,\,\,\,\,\,\,\,\,\,\,\,\,\mbox{for all $\g_0\in \PP$}  \Eq(ass1)
$$
while assumption $\bm\m\in \mathcal{D}^b$  means that
$$
\varphi^b _{\g} (\bm \mu)<+\infty
\,\,\,\,\,\,\,\,\,\,\,\,\,\,\,\mbox{for all $\g\in \PP$}\Eq(ass2)
$$
Assumptions \equ(ass1) and \equ(ass2) guarantee that $\bm r\in (0,\infty)^{\PP}$, i.e.
$$
r_\g >0 \,\,\,\,\,\,\,\,\,\,\,\,\,\,\,\mbox{for all $\g\in \PP$}
$$

\\Let us now consider, for any $\bm \r\in \mathcal{D}^p$, the map $T^{\bm\r}=\bm \r\varphi^b$. $T^{\bm \r}$ is the map
$$
T^{\bm \r}: (0,\infty)^\PP\to (0,\infty]^\PP: \bm u \mapsto T^{\bm \r}(\bm u)
$$
with entries
$$
[T^{\bm \r}(\bm u)]_\g\doteq T^{\bm \r}_{\g}(\bm u)= \r_{\g} \varphi^b_{\g}(\bm u)\;\;\;\;\;\;\;\;\;\;\;\;\;\;\;\;\;\;\;\;\;\;\;\;\g\in \PP
$$
From \equ(muR) we get
$$
\bm \m = T^{\bm r}(\bm \m)\Eq(fixed)
$$
I.e. $\bm \m$ is fixed point for the map $T^{\bm r}$. So, by  \equ(fixed) we have that, for all $\g_0\in\PP$
$$
\m_{\g_0}= T^{\bm r}_{\g_0}(\bm \mu) =
$$
$$
=r_{\g_0} \;+ \;r_{\g_0}\sum_{\g_{1}\in\PP} b_1(\g_0;\g_1)\,
\m_{\g_1} + r_{\g_0}\sum_{(\g_{1},\g_{2})\in\PP^2} b_2(\g_0;\g_1,\g_2)\,
\m_{\g_1}\m_{\g_2}+\;\dots
$$
$$
\dots \;+\,
r_{\g_0}\sum_{(\g_{1},\dots ,\g_{n})\in\PP^n} b_n(\g_0;\g_1,\ldots,\g_n)\,
\m_{\g_1}\dots\m_{\g_n}~+ \dots\Eq(mTm)
$$
Equation \equ(mTm), recalling the definition \equ(r.7.1) of $\varphi_{\g}^b(\bm \mu)$ can be visualized in the diagrammatic form

\setlength{\unitlength}{.9cm}
$$
\begin{picture}(15,2.0)
\thicklines
\put(0.5,1.5){$\bullet$}
\put(0.5,1.0){$\scriptstyle\g_0$}
\put(1,1,5){$\doteq$}
\put(1.5,1.5){$\mu_{\g_0}$}
\put(2.6,1.5){=}
\put(3.1,1.5){{$T^{\bm r}_{\g_0}(\bm \mu)$}}
\put(4.4,1.5){$\doteq$}
\put(5,1.5){$\circ$}
\put(5,1.3){$\scriptstyle\g_0$}
\put(5.4,1.5){$+$}
\put(5.9,1.5){$\circ$} %
\put(5.9,1.3){$\scriptstyle\g_0$} %
\put(6.05,1.61){\line(1,0){1}} %
\put(7.05,1.5){$\bullet$} %
\put(7.05,1.3){$\scriptstyle\g_1$} %
\put(7.4,1.5){$+$}   %
\put(8,1.5){$\circ$}
\put(8,1.3){$\scriptstyle\g_0$}
\put(8.15,1.65){\line(2,1){0.8}} %
\put(8.95,1.95){$\bullet{\scriptstyle \g_1}$} %
\put(8.15,1.6){\line(2,-1){0.8}}
\put(8.95,1.13){$\bullet{\scriptstyle \g_2}$}
\put(9.7,1.5){$+$}
\put(10.2,1.5){$\cdots$}
\put(10.9,1.5){$+$}  %
\put(11.4,1.5){$\circ$}
\put(11.3,1.3){$\scriptstyle\g_0$}
\put(11.53,1.65){\line(1,1){0.9}}
\put(12.4,2.5){$\bullet{\scriptstyle \g_1}$}
\put(11.5,1.65){\line(2,1){0.9}}
\put(12.4,2.03){$\bullet{\scriptstyle \g_2}$}
\put(12.45,1.45){\vdots}
\put(11.55,1.55){\line(1,-1){0.85}}   %%%%%%%%%%%%%%
\put(12.4,0.55){$\bullet{\scriptstyle \g_n}$}
\put(13.1,1.5){$+$}
\put(13.7,1.5){$\cdots$}
\end{picture}
$$
where
$$
\circ_{\g_0}= r_{\g_0} ~~~~~~~~~~~\bullet_{\g_i}= \m_{\g_i}
$$
and, for any $n\ge 1$
\setlength{\unitlength}{.9cm}
$$
\begin{picture}(15,3.0)
\thicklines
\put(2,1.5){$\circ$}
\put(2,1.0){$\scriptstyle\g_0$}
\put(2.13,1.65){\line(1,1){0.9}}
\put(3,2.5){$\bullet{\scriptstyle \g_1}$}
\put(2.1,1.65){\line(2,1){0.9}}
\put(3,2.03){$\bullet{\scriptstyle \g_2}$}
\put(3.05,1.45){\vdots}
\put(2.15,1.55){\line(1,-1){0.85}}   %%%%%%%%%%%%%%
\put(3,0.55){$\bullet{\scriptstyle \g_n}$}
\put(4,1.5){$=$}
\put(5,1.5){$r_{\g_0}\sum\limits_{(\g_{1},\dots ,\g_{n})\in\PP^n} b_n(\g_0;\g_1,\ldots,\g_n)\,
\m_{\g_1}\dots\m_{\g_n}$}
\end{picture}
$$

\\The iteration $[{T^r}]^2(\mu)=T^r(T^r(\mu))$ corresponds to
replacing each of the bullets by each one of the diagrams of the
expansion for $T^r$.

\\This leads to plane rooted trees of up to two
generations, with open circles at first-generation vertices and bullets
at second-generation ones.  The
$k$-th iteration of $T$ involves all possible \emph{plane rooted trees} previously  seen with generation up to $k$.
In each tree  of the expansion, vertices of the last generation  are occupied
by bullets and all the others by open circles.
Let us recall that $\TT^{0,k}$ denotes  the
set of trees with maximal depth (or maximal generation number) equal to $k$.
A straightforward inductive argument shows that
$$
[T^{\bm r}]^k_{\g_0}(\bm \mu) \;=\; r_{\g_0} \Bigl[\sum_{\ell=0}^{k-1}
\Phi^{(\ell)}_{\g_0}(\bm r) + R^{(k)}_{\g_0}(\bm r,\bm\mu)\Bigr]
\Eq(r.8.1)
$$
with

$$
\Phi^{(\ell)}_{\g_0}(\bm r) \;
=\,\sum_{t \in\TT^{0,\ell}}\,
\prod_{v\succeq 0}
\Bigg\{\sum_{(\g_{v^1},\dots,\g_{v^{s_v}})\in \PP^{s_v}}b_{s_v}(\g_v;\g_{v^1},\ldots,\g_{v^{s_v}})\,
r_{\g_{v^1}}\dots r_{\g_{v^{s_v}}}
\Bigg\}\Eq(explw3)
$$
while
$$
R^{(k)}_{\g_0}(\bm r,\bm \mu)
=\,\sum_{t \in\TT^{0,k}}\,
\prod_{v\succeq 0}
\Bigg\{\sum_{(\g_{v^1},\dots,\g_{v^{s_v}})\in \PP^{s_v}}b_{s_v}(\g_v;\g_{v^1},\ldots,\g_{v^{s_v}})\,
\chi^{t}_{\g_{v^1}}\dots \chi^t_{\g_{v^{s_v}}}
\Bigg\}\Eq(explw4)
$$
where
$$
\chi^{t}_{\g_v}=
\begin{cases}
r_{\g_v} & {\text if} \,\,\,d_t(v)< k\\
\m_{\g_v} & {\text if} \,\,\,d_t(v)= k
\end{cases}
\Eq(defchi)
$$
with, we recall, $d_t(v)$ indicating the depth (distance from the root) of $v$ in $t$.
In other words $R^{(k)}_{\g_0}(\bm r,\bm \mu)$ has an expression similar to $\Phi^{(k)}_{\g_0}(\bm r)$ but
with the activities of the vertex of the $k$-th generation
weighted by $\bm \mu$. Here we agree that if $v$ is such that $s_v=0$ then  $b_0(\g_v)\equiv 1$.
Now, by \equ(mTm) we have
$$
[T^{\bm r}]^k_{\g_0}(\bm \mu) \;= \;\m_{\g_0}
$$
which implies immediately, via \equ(r.8.1),
$$
 r_{\g_0} \sum_{\ell=0}^{k-1}
\Phi^{(\ell)}_{\g_0}(\bm r)\; \le\; \m_{\g_0}\;\,\,\,\,\,\,\,\,\,\,\,\,\,  \mbox{for all $k\in \mathbb{N}$ }  \Eq(parz)
$$
Equation \equ(parz) immediately implies the following proposition.
\begin{pro}\label{prop:1}
\\Let $\bm\m$ be a function
$\bm \m\in (0,\infty)^\PP$ and, for any $\g\in \PP$ let $\varphi^b_{{\g}} (\bm \mu)$ be function defined in \equ(r.7.1)
supposed to satisfy \equ(ass2). Let $\bm r \in(0,\infty)^\PP$ be defined by \equ(muR).
Then, for all $\bm \r\le \bm r$

\begin{enumerate}

\item[i)]
The series
$$
\Phi^{b}_{\g_0} (\bm \r)\;\bydef\;
\,\sum_{t \in\TT^{0}}\,
\prod_{v\succeq 0}
\Bigg\{\sum_{(\g_{v^1},\dots,\g_{v^{s_v}})\in \PP^{s_v}}b_{s_v}(\g_v;\g_{v^1},\ldots,\g_{v^{s_v}})\,
\r_{\g_{v^1}}\dots \r_{\g_{v^{s_v}}}
\Bigg\}
\Eq(r.11)
$$
converges  for each $\g_0\in\PP$ and admits, for each $\g_0\in\PP$, the bound
$$ \Phi^b_{\g_0}(\bm \r)\le \Phi^b_{\g_0}(\bm r)\le \varphi^b_{{\g_0}}(\bm \mu)
\Eq(r.12)
$$

\item[ii)]
$$
\r_{\g_0}\Phi^b_{\g_0}(\bm \r)= \lim_{n\to\infty} [T^{\bm \r}]^n_{\g_0}(\bm \r)\Eq(r.12a)
$$
and
$\r_{\g_0}\Phi^b_{\g_0}(\bm \r)$
is solution of the equation \equ(mTm), i.e. is fixed point of the map $T^{\bm \r}$, i.e.
$$
\r_{\g_0}\Phi^b_{\g_0}(\bm \r)= T^{\bm \r}_{\g_0}\Big(\r_\g \Phi^b_{\g}(\bm \r)\Big)\Eq(r.12b)
$$
\end{enumerate}
\end{pro}
This proposition can be viewed as a generalization of the Lagrange inversion formula for series depending on infinite (countable) variables.
\vv
\\{\bf Proof}.

\\i). By \equ(r.8.1) we get
$$
r_{\g_0} \sum_{\ell=0}^{n}\Phi^{(\ell)}_{\g_0}(\bm r)\,\,\le\,\,[T^{\bm r}]_{\g_0}^{n+1}(\bm \mu)\;\;\;\;\;\;\;\forall n\in \mathbb{N}
$$
but, by definition \equ(mTm) we have, for any $k\in \mathbb{N}$ that $[T^{\bm r}]^k_{\g_0}(\bm \mu)=\m_{\g_0}$. So we obtain
$$
r_{\g_0}\sum_{\ell=0}^{n}\Phi^{(\ell)}_{\g_0}(\bm r)\;\le\;
\m_{\g_0} \,\,\,\,\,\,\,\,\,\,\,\,\,\,\mbox{for all $n$}
$$
i.e., by  \equ(muR),
$$
\sum_{\ell=0}^{n}\Phi^{(\ell)}_{\g_0}(\bm r)\;\le\;
 \varphi^b_{{\g_0}}(\bm \mu) \,\,\,\,\,\,\,\,\,\,\,\,\,\,\mbox{for all $n$}
$$
which implies
$$
\Phi^b_{\g_0}(\bm r)\le\varphi^b_{{\g_0}}(\bm \m)
$$
Therefore, by monotonicity, for any $\bm \r\le \bm r$
$$
\Phi^b_{\g_0}(\bm \r)\le \Phi^b_{\g_0}(\bm r)\le\varphi^b_{{\g_0}}(\bm \m)
$$

\v
\\ii) By \equ(r.8.1) we have that
$$
[T^{\bm \r}]_{\g_0}^k(\bm \r) \;=\; \r_{\g_0} \Bigl[\sum_{\ell=0}^{k-1}
\Phi^{(\ell)}_{\g_0}(\bm \r) + R^{(k)}_{\g_0}(\bm \r,\bm \m)|_{\bm \m=\bm \r}\Bigr]
$$
But, recalling definition \equ(defchi)
$$
R^{(k)}_{\g_0}(\bm \r,\bm \m)|_{\bm \m=\bm \r}
\;=
$$
$$
=\,\sum_{t \in\TT^{0,k}}\,
\prod_{v\succeq 0}
\Bigg\{\sum_{(\g_{v^1},\dots,\g_{v^{s_v}})\in \PP^{s_v}}b_{s_v}(\g_v;\g_{v^1},\ldots,\g_{v^{s_v}})\,
\r_{\g_{v^1}}\dots \r_{\g_{v^{s_v}}}
\Bigg\}\,\,=\,\,\Phi^{(k)}_{\g_0}(\bm \r)
$$
So
$$
[T^{\bm \r}]^k_{\g_0}(\bm \r) \;=\; \r_{\g_0}\sum_{\ell=0}^{k}
\Phi^{(\ell)}_{\g_0}(\bm \r)
$$
Hence, for any $\bm\r\le \bm r$
$$
\lim_{k\to\infty} [T^{\bm \r}]^k_{\g_0}(\bm \r) \,=\,\lim_{k\to\infty}\r_{\g_0}\sum_{\ell=0}^{k}
\Phi^{(\ell)}_{\g_0}(\bm \r) \,=\,\r_{\g_0}\Phi^b_{\g_0}(\bm \r)
$$
\v
\\Finally, for any  $\bm \r\le \bm r$
$$
\begn
\r_{\g_0}\Phi^b_{\g_0}(\bm \r) &= \lim_{n\to\infty} [T^{\bm\r}]^{n+1}_{\g_0}(\bm \r)\\
&=\,\lim_{n\to\infty} T^{\bm \r}_{\g_0}\Bigl( [T^{\bm \r}]^{n}_{\g_0}(\bm \r)\Bigr)\\
&=
\,\r_{\g_0}\lim_{n\to\infty}\varphi^b_{{\g_0}}\Bigl( [T^{\bm \r}]^{n}_{\g_0}(\bm \r)\Bigr)\\
&= \, \r_{\g_0}\varphi^b_{\g}\Bigl(\lim_{n\to\infty} [T^{\bm \r}]^{n}_{\g_0}(\bm \r)\Bigr)\\
&=  \, \r_{\g_0}\varphi^b_{{\g_0}}\Bigl(\r_{\g_0}\Phi^b_{\g_0}(\bm \r)\Bigr)\\
&=\, T^{\bm \r}_{\g_0}\Bigl(\r_{\g_0}\Phi^b_{\g_0}(\bm \r)\Bigr)
\egn
$$
$\Box$
\vv

\subsection{Convergence criteria}

\\\underline{\it Fern\'andez-Procacci criterion}.

\\Let us now  choose
$$
b_n(\g_0;\g_1,\ldots,\g_n) = b^{*}_n(\g_0;\g_1,\ldots,\g_n)\doteq
{1\over n!}\prod_{i=1}^n \ind{\g_i\not\sim\g_0}\;\; \prod_{1\le i<j\le n} \ind{\g_{i}\sim\g_{j}}\Eq(defb)
$$
and thus
$$
\varphi^{b^{*}}_{\g_0} (\bm \mu)\;=\; 1+\sum_{n\ge 1} \frac{1}{n!}\,
\sum_{(\g_{1},\dots ,\g_{n})\in\PP^n\atop \g_i\not\sim \g_0,\,\,\g_{i}\sim\g_{j}}
{\mu_{\g_1}}\dots{\mu_{\g_n}}\;=\;
\Xi_{\PP_{\g_0}}(\bm \mu)
\Eq(r.15.2)
$$

\\Then, proposition \ref{prop:1} tells us that the series
$$
\Phi^{b^{*}}_{\g_0} (\bm \r) =
$$
$$
=
\,\sum_{t \in\TT^{0}}\,
\prod_{v\succeq 0}
\Bigg\{\sum_{(\g_{v^1},\dots,\g_{v^{s_v}})\in \PP^{s_v}}{1\over s_v!}
\prod_{i=1}^{s_v} \ind{\g_{v^i}\not\sim\g_v} \prod_{1\le i<j\le s_v} \ind{\g_{v^i}\sim\g_{v^j}}\,
\r_{\g_{v^1}}\dots \r_{\g_{v^{s_v}}}
\Bigg\}
\Eq(r.11fp)
$$
converges as soon as $\bm \r\le \bm r^*$ with
$$
r^*_{\g}={\m_\g\over \Xi_{\PP_{\g}}(\bm \mu)} \Eq(muRv)
$$

\\comparing \equ(aab) with \equ(r.11fp)  we immediately
see that
$$
\Pi^*_{\g_0}(\bm \r)=\Phi^{b^{*}}_{\g_0}(\bm \r)
$$
So
we immediately get, by Proposition \ref{prop:1},
the following criterion for the convergence of cluster expansions.

\begin{teo}\label{coro:1}
Choose $\bm \m\in \mathcal{D}^{b^*}\subset (0,\infty)^\PP$ and let $\bm r^* \in(0,\infty)^\PP$ s.t.
$$
r^*_{\g_0}={\m_{\g_0}\over \Xi_{\PP_{\g_0}}(\bm \mu)}
$$
Let $ \bm \r$ such that
$$
\r_{\g}\;\le\; r^*_{\g} \,\,\,\,\,\,\,\,\,\,\,\, \forall\g\in\PP\Eq(FPcrit)
$$
Then the series $|\Pi|_{\g_0}(\bm \r)$ defined in  \equ(TP)
is finite for each $\g_0\in\PP$ and
$$
|\Pi|_{\g_0}(\bm \r)\le  \Xi_{\PP_{\g_0}}(\bm \mu)
\Eq(r.12b0)
$$
and hence
$$
\r_{\g_0}|\Pi|_{\g_0}(\bm \r)\le  \m_{\g_0}\Eq(r.12bb)
$$
\\for each $\g_0\in\PP$.
\end{teo}

\proof By proposition \ref{prop:1} we have immediately that the series $\Pi^*_{\g_0}(\bm \r)$ defined in  \equ(exact)
is finite for each $\g_0\in\PP$ and for all $\bm \r$ such that $\r_{\g}\le r^*_{\g}$ where $r^*_{\g}$ is defined in \equ(FPcrit). Moreover
$$
\Pi^*_{\g_0}(\bm \r)\le \Xi_{\PP_{\g_0}}(\bm \mu)
$$
for each $\g_0\in\PP$.  Now recalling \equ(PiP*) we obtain that the same is true also for the series $|\Pi|_{\g_0}(\bm \r)$. $\Box$

\vv

\\\underline{\it Dobrushin criterion}.

\\If we now  choose
$$
b_n(\g_0;\g_1,\ldots,\g_n)=b^{\rm Dob}_n(\g_0;\g_1,\ldots,\g_n)
\doteq {1\over n!}\prod_{i=1}^n \ind{\g_i\not\sim\g_0}\;\; \prod_{1\le i<j\le n} \ind{\g_{i}\neq\g_{j}}
$$
and thus
$$
\varphi^{b^{\rm Dob}} _{\g_0} (\bm\mu)\;=\; 1+\sum_{n\geq 1}
\frac{1}{n!}\,\sum_{(\g_{1},\dots ,\g_{n})\in\PP^n\atop
\g_0\nsim\g_i\,,\, \g_i\neq\g_j}
{\mu_{\g_1}}\dots{\mu_{\g_n}}\;=\;
\prod_{\g\nsim\g_0} [1+\mu_{\g}]\;,
\Eq(r.25)
$$
\\Then again proposition \ref{prop:1} tells us that the series
$$
\Phi^{b^{\rm Dob}}_{\g_0} (\bm \r)\; =
$$
$$=
\,\sum_{t \in\TT^{0}}\,
\prod_{v\succeq 0}
\Bigg\{\sum_{(\g_{v^1},\dots,\g_{v^{s_v}})\in \PP^{s_v}}{1\over s_v!}
\prod_{i=1}^{s_v} \ind{\g_{v^i}\not\sim\g_v} \prod_{1\le i<j\le s_v} \ind{\g_{v^i}\neq\g_{v^j}}\,
\r_{\g_{v^1}}\dots \r_{\g_{v^{s_v}}}
\Bigg\}
\Eq(r.11d)
$$
converges as soon as $\bm \r\le \bm r^{\rm Dob}$ with
$$
r^{\rm Dob}_{\g}={\m_\g\over \prod_{\g\nsim\g_0} [1+\mu_{\g}]} \Eq(muRvD)
$$
\\Comparing \equ(aabd) with \equ(r.11d)  we get
$$
\Pi^{\rm Dob}_{\g_0}(\bm \r)=\Phi^{b^{\rm Dob}}_{\g_0}(\bm \r)
$$
So
Proposition \ref{prop:1} also yields
the following (weaker) criterion for the convergence of cluster expansions.
\begin{cor}[Dobrushin]\label{coro:1dob}
Choose $\bm \m\in {\mathcal D}^{b^{\rm Dob}}$ and let $\bm r^{\rm Dob} \in(0,\infty)^\PP$ s.t.
$$
r^{\rm Dob}_\g={\m_\g\over \prod_{\gt\nsim\g} [1+\m_\gt]} \Eq(muRvd)
$$
Let $ \bm \r \in[0,\infty)^\PP$ such that
$$
\r_\g\;\le\; r^{\rm Dob}_\g \,={\m_\g\over \prod_{\gt\nsim\g} [1+\m_\gt]}\,\,\,\,\,\,\,\,\,\,\, \forall\g\in\PP\Eq(Dcrit)
$$
Then the series $|\Pi|_{\g_0}(\bm \r)$ defined in  \equ(TP)
is finite for each $\g\in\PP$ and
$$ |\Pi|_{\g}(\bm \r)\le \prod_{\gt\nsim\g} [1+\mu_\gt]
\Eq(r.12bd)
$$
or
$$
\r_\g|\Pi|_{\g}(\bm \r)\le  \m_\g\Eq(same)
$$
for each $\g\in\PP$.
\end{cor}
For the benefit of the readers we stress that
 in the literature the Dobrushin condition is generally written in a different (but equivalent) form. In particular,
in the Dobrushin paper \cite{D} and also in \cite{M} the condition \equ(Dcrit) is written as follows.

$$
\r_\g\;\le\; r^{\rm Dob}_\g \;=(e^{{\tilde\m}_\g}-1)
e^{-\sum_{\gt\in\PP:\,\gt \not\sim \g}\tilde\m_\gt}\,\,\,\,\,\,\,\,\,\,\,\, \forall\g\in\PP
$$

\\This is clearly the same condition \equ(Dcrit) by defining
$$
\tilde\m_\g=\log[1+\m_\g] \Eq(chang)
$$

\def\gt{{\tilde\g}}

\vv\vv

\\\underline{\it Kopteck\'y-Preiss criterion}.

\\Finally, if we now choose
$$
b_n(\g_0;\g_1,\ldots,\g_n)=b^{\rm KP}_n(\g_0;\g_1,\ldots,\g_n)\doteq{1\over n!} \prod_{i=1}^{n}\ind{\g_i\not\sim \g_0}
$$
and
$$
\varphi^{\rm KP} _{\g_0} (\bm\mu)\;=\; 1+\sum_{n\geq 1} \,\frac{1}{n!}
\sum_{(\g_{1},\dots ,\g_{n})\in\PP^n\atop
\g_0\nsim\g_i\,,\,  1\le i \le n}
{\mu_{\g_1}}\dots{\mu_{\g_n}}\;=\;
\exp\Bigl[\sum_{\g\nsim\g_0} \mu_\g\Bigr]
\Eq(r.27)
$$
\\Then once again proposition \ref{prop:1} tells us that the series
$$
\Phi^{b^{\rm KP}}_{\g_0} (\bm \r)\;=\;
\,\sum_{t \in\TT^{0}}\,
\prod_{v\succeq 0}
\Bigg\{\sum_{(\g_{v^1},\dots,\g_{v^{s_v}})\in \PP^{s_v}}{1\over s_v!}
\prod_{i=1}^{s_v} \1_{\g_{v^i}\not\sim\g_v}\;\;
\r_{\g_{v^1}}\dots \r_{\g_{v^{s_v}}}
\Bigg\}
\Eq(r.11kp)
$$
converges as soon as $\bm \r\le \bm r^{\rm KP}$ with
$$
r^{\rm KP}_{\g}={\m_\g\over \exp\Bigl[\sum_{\gt\nsim\g} \mu_\gt\Bigr]} \Eq(muRvKP)
$$
\\Comparing \equ(aabkp) with \equ(r.11kp)  we get
$$
\Pi^{\rm KP}_{\g_0}(\bm \r)=\Phi^{b^{\rm KP}}_{\g_0}(\bm \r)
$$
So we get the criterion of Koteck\'y and Preiss. Namely,

\begin{cor}[Koteck\'y-Preiss]\label{coro:1kp}
Choose $\bm \m\in {\mathcal D}^{b^{\rm KP}}$ and let $\bm r^{\rm KP} \in(0,\infty)^\PP$ s.t.
$$
r^{\rm KP}_{\g_0}={\m_{\g_0}\over \exp\Bigl[\sum_{\g\nsim\g_0} \mu_\g\Bigr]} \Eq(muRvkp)
$$
Let $ \bm \r \in[0,\infty)^\PP$ such that
$$
\r_\g\;\le\; r^{\rm KP}_\g \;=\;
{\m_{\g}\over\exp\Bigl[\sum_{\gt\nsim\g} \mu_\gt\Bigr]},\,\,\,\,\,\,\,\,\,\,\,\, \forall\g\in\PP\Eq(KPcrit)
$$
Then the series $|\Pi|_{\g_0}(\bm \r)$ defined in  \equ(TP)
is finite for each $\g_0\in\PP$ and
$$ |\Pi|_{\g_0}(\bm \r)\le  \exp\Bigl[\sum_{\g\nsim\g_0} \mu_\g\Bigr]
\Eq(r.12bkp)
$$
or
$$
\r_{\g_0}|\Pi|_{\g_0}(\bm \r)\le  \m_{\g_0}
$$
for each $\g_0\in\PP$.
\end{cor}
\vv\vv
{\bf Remark}. Again for the benefit of the readers we stress that
 in the literature the Kotecky-Preiss  condition is generally written in a different (but equivalent) form. Namely in the original  Kotecky-Preiss paper convergence is guaranteed by choosing $\bm \r$  such that
 there is a function  $a_\g$ such that
$$
\sum_{\gt\nsim \g} \r_\gt e^{a_\gt}\le a_\g ~~~~~~~~~~~~~~~~~~~~\forall\g\in\PP
$$

\\This is clearly the same condition \equ(muRvkp) by  setting
$$
\m_\g=\r_\g e^{a_\g} \Eq(changKP)
$$
\vv\vv
\\Summing up, available convergence conditions are of the form
$$
\r_\g \,\;\le\; r_\g= \,{\mu_\g\over \varphi_{\g}(\m) }
\Eq(cdg.7)
$$
with
$$
\varphi_{\g}(\bm \m) \;=\; \left\{\begin{array}{ll}
\exp\Bigl[\sum_{\gt\nsim\g} \mu_\gt\Bigr] & \mbox{ (Koteck\'y-Preiss)}\\[10pt]
\prod_{\gt\nsim\g} \bigl(1 + \mu_\gt\bigr)  & \mbox{ (Dobrushin)}\\[10pt]
\Xi_{\PP_{\g}}(\bm \m)  & \mbox{ (Fern\'andez-Procacci)}
\end{array}\right.
\Eq(cdg.8)
$$

\\Each condition is strictly weaker than the preceding one. Namely,
since, for fixed $\bm \m\in (0,\infty)^\PP$,
$$
\Xi_{\PP_{\g}}(\bm \m) \le \prod_{\gt\nsim\g} \bigl[1 + \mu_\gt\bigr] \le \exp\Bigl[\sum_{\g\nsim\g_0} \mu_\gt\Bigr]
$$
we get
$$
\bm r^*_\g\ge \bm r^{\rm Dob}_\g\ge \bm r_\g^{\rm KP}
$$
So the criterion given by the Corollary \ref{coro:1kp} (i.e. Kotecky-Preiss condition) yields the worst estimate for convergence radius for the cluster expansion;
the Dobrishin Criterion of Corollary \ref{coro:1dob} gives an  estimate which is better (i.e. larger) than
that given by the Koteck\'y-Preiss criterion for the same radius and finally the criterion \ref{coro:1} give the best estimate
for convergence radius for the cluster expansion
among the three proposed.

\subsection{Elementary Examples}\label{elexe}
\vv
In this section we give some elementary in order to illustrate how the criterion \equ(FPcrit) represents a sensible improvement
on previous criteria in applications.
\vv

\\{\it Example 1. The  Domino model on $\mathbb{Z}^2$.}
\v
\\This model has also been considered by Dobrushin in \cite{D}. The elements of the polymer space $\PP$
are in this case nearest neighbor bonds of the bidimensional
cubic lattice. For any $\g\in \PP$ we put $\r_\g=\r$, where $\r>0$ (all polymers have the same activity).
Two polymers are incompatible if and only if they have non empty intersection.
We can choose by symmetry that the function $\m_\g$  appearing in the Kotecky-Preiss, Dobrushin and Fern\'andez-Procacci criteria  are constant at the value $\m$.

\\The Kotecky-Preiss criterion \equ(KPcrit) for the domino model then reads as
$$
\r_\g\le \m_\g e^{-\sum_{\tilde\g\not\sim\g}\m_{\tilde\g}}~~\Longleftrightarrow~~
\r \le \m e^{-7\m}
$$
which yelds at best
$$
\r\le {1\over 7e}\approx 0.0525
$$
On the other hand the Dobrushin condition \equ(Dcrit) reads
$$
\r_\g\le {\m_\g\over \prod_{\tilde\g\not\sim \g}[1+\m_\gt]}
~~\Longleftrightarrow~~\r\le {\m\over (1+\m)^7}
$$
which yields at best

$$
\r\le {{1\over 6}\over (1+{1\over 6})^7}\approx 0.0566
$$

\\Finally, the condition \equ(FPcrit) gives
$$
\r_\g\le {\m_\g\over \Xi_{\PP_{\g}}(\bm \m)}
~~\Longleftrightarrow~~\r\le {\m\over1 +7\m +9\m^2}
$$
which yields at best

$$
\r\le {1\over 13}\approx 0.0769
$$

\vv
\\{\it Example 2. The  lattice gas on a bounded degree graph $\GI=(\VU,\EE)$ with hard core
self repulsion and hard core pair interaction and the triangular lattice on then plane}
\v
\\Let  $\GI=(\VU,\EE)$ be a bounded degree infinite graph with vertex set $\VU$ and edge
set $\EE$, and maximum degree $\D$. A polymer system is obtained by choosing $\PP=\VU$ and
by defininig the incompatibility relation $\nsim$ by saying that two polymers $\g$ and $\g'$ (i.e. two vertices
of $\GI$) are incompatible if and only if either $\g=\g'$ or $\{\g,\g'\}\in \EE$. This polymer gas realization
is called the self repulsive hard core lattice gas on $\GI$. In this case
the polymers are the vertices of $\GI$ and two  polymers $\{x,y\}\subset \VU$ are incompatible if either
$y=x$ (self repulsion) or  $\{x,y\}\in \EE$ (hard core pair interaction). In general, since polymers have no structure
(they are just vertices in a graph) one  can suppose that the activity  of a polymer $x\in \PP$ is  a constant,
i.e. $\r_x=\r$ for all $x\in \VU$. Of course  expect that the convergence radius depends
strongly on the topological structure of $\GI$. We first consider  the worst case i.e.
when the graph $\GI$ is such that the nearest neighbors of any vertex are pairwise compatible.
This happens e.g. if $\GI$ is a tree or if it is the cubic lattice $\mathbb{Z}^d$.
The Kotecky-Preiss condition for this model then reads as
$$
\r_{x}\le \m_xe^{-\sum_{y \not\sim x}\m_y}~~\Longleftrightarrow~~
\r \le \m e^{-(\D+1)\m}
$$
which yields at best
$$
\r\le {1\over (\D+1)e}\Eq(treekp)
$$
On the other hand the Dobrushin  condition reads
$$
\r_{x}\le {\m_x \over \prod_{y\not\sim x}[1+\m_y]}
~~\Longleftrightarrow~~\r\le {\m\over (1+\m)^{\D+1}}
$$
which yields at best

$$
\r\le {{1\over \D}\over (1+{1\over \D})^{\D+1}}={\D^\D\over (\D+1)^{\D+1}}={1\over \D+1}{1\over (1+{1\over \D})^\D}\Eq(treed)
$$

\\Finally, the condition \equ(FPcrit) gives
$$
\r_{x}\le {\m_x\over \Xi_{\PP_x}(\bm \m)}
~~\Longleftrightarrow~~\r\le {\m\over 1 +(\D+1)\m +\sum_{k=2}^\D{\D\choose k}\m^k }
={\m\over \m +(1+\m)^\D}
$$
which yields at best

$$
\r\le {{1\over \D-1}\over {1\over \D-1} +\left(1+{1\over \D-1}\right)^\D}=
{1\over 1+{\D^\D\over (\D-1)^{\D-1}}} ={1\over \D(1+{1\over \D-1})^{\D-1} +1}\Eq(treeb)
$$

\\To illustrate that Theorem \ref{coro:1} permits to  improve this last bound \equ(treeb) if we know more about the
topological structure of $\GI$, we now consider a case of the triangular lattice in $d=2$ (a regular graph with degree $\D=6$,  where our bound turns to be
more efficient than the Dobrushin bound  and  the Shearer-Sokal  bound \equ(treekp). For the triangular lattice the tree bound \equ(treeb)
gives
$$
\e< {5^5\over 6^6}\approx 0,067
$$
while our bound
gives

$$
\e\le {c\over \Xi_\g(c)}= {c\over 1 + 7c + 9c^2 + 2c^3}
$$
The maximum occurs when $4c^3+ 9c^2-1=0$, which is somewhere between 1/3 and 3/10. For example choosing $c=1/3$ (which is not the best choice) we obtain
$$
\e\le {c\over \Xi_\g(c)}= {{1\over 3}\over 1 + {7\over 3} +{1} + {2\over 27}}\approx 0,075
$$

\section[Gas of non-overlapping subsets]{Gas of non overlapping finite subsets}
\def\L{\Lambda}
In this section we will study  a particular realization of the polymer gas
which appears in  the most part of the examples in statistical mechanics.

\\We will suppose that it is given an infinite  countable set $\VU$, and we define the space of polymers as
$$
\PP_\VU=\{R\subset \VU :  |R|<\infty\}
$$
and the incompatibility relation in $\PP_\VU$ is defined as
$$
\g\not\sim\gt\,\,\,\,\Longleftrightarrow\,\,\,\, \g\cap\tilde\g\neq\0
$$
Note that now polymers have
a cardinality, so that we can speak about big polymers and small polymers. Of course, as before, to a
polymer $\g$  is associated an activity $\z(\g)$. We assume in general that  $\z(\g)\in \mathbb{C}$ as far as  $\g\in \PP_\VU$ and we set
$$
|\z(\g)|=\r(\g)\Eq(zeqr)
$$
Note that here we allow the value $\z(\g)=0$ for some $\g$ in order to stay more general. For example in the polymer
expansion of high temperature spin systems, the polymer space is always $\PP_\VU$
for some suitable $\VU$ but it happens that $\z(\g)=0$ whenever $|\g|=1$.

In most of the physics realizations $\VU$ is the vertex set of
an infinite
graph $\GI=(\VU,\EE)$ with edge set $\EE$.
For example $\VU= \Z^d$ and $\EE$ is the set of nearest neighbor in $\Z^d$.  When $\VU$ is the vertex
set of a graph $\GI$ then $\VU$  has a natural metric structure induced by the graph distance in $\GI$.
This metric structure on $\VU$ allow us to  talk about how spread is a polymers (a polymer is spread if its points are
far apart) and we can say now if two polymers
$\g$ and $\g'$ are close or far apart. So, from the abstract context
we can pass to more concrete realizations which have richer structures. Namely,
if one suppose that polymers are finite subsets of an underlying  countable set $\VU$ with $\not\sim=\cap$,
the any polymer has an activity and a cardinality, so we an distinguish between big and small polymers.
If we further suppose that the underlying set $\VU$ is the vertex set of some graph $\GI$ we  also can
talk about distance between polymers and spread polymers.
In any case,  in the whole section below we will not suppose any graph structure for the set $\VU$.
Our abstract polymer space is just the set of all finite subsets of a countable set with the incompatibility
relation being the the non void intersection.

Let now $\L$ be a finite set of $\VU$. A configuration
of polymer gas in $\L$ is given once we specify  the set of polymers which are present in $\L$.
Of course this polymers must be pairwise  compatible, i.e. a configuration in $\L$ is an unordered
$n$-ple $\{\g_1,\dots, \g_n\}$ such that $\g_i\cap \g_j=\0$ for all $i,j=1,\dots, n$.
The ``probability"\footnote{(4.79) is a real probability only if $\z(\g)\in [0,+\infty)$} to see the configuration $\{\g_1,\dots, \g_n\}$ in the box $\L$ is defined as
$$
{\rm Prob}_{\bm \z}(\g_1,\dots,\g_n)= \Xi_\L^{-1} \prod_{i=1}^n\z(\g_i)\Eq(ppp)
$$
where $\Xi_\L$ is the partition function defined as
$$
\Xi_\L(\bm \z)=1+\sum_{n\ge 1}\sum_{\{\g_1,\dots,\g_n\}:~\g_i\subset \L\atop \g_i\cap \g_j=\0}
\z(\g_1)\dots \z(\g_n)\Eq(GKp)
$$

\subsection{Convergence via the abstract polymer criteria}\label{431}
We compare the three conditions for this model. Starting with the Kotecky-Preiss condition, choosing
$\m(\g)=\r(\g) e^{a|\g|}$  (recall: $|\z(\g)|=\r(\g)$), the condition \equ(KPcrit) becomes the well known inequality
$$
\sum_{\gt\in \PP_\VU\atop\tilde\g\not\sim\g} \r(\gt)\; e^{a|\gt|}\:\le\: a|\g|,\,\,\,\,\,\,\,\,\,\forall\g\in\PP\Eq(KPht)
$$
Now using that $\tilde\g\not\sim\g$ means for the present model $\gt\cap\g\neq\emptyset$ we have that
$$
\sum_{\tilde\g\not\sim\g} \r(\gt) e^{a|\gt|}\le |\g|\sup_{x\in \VU}\sum_{\tilde\g\ni x}\r(\gt) e^{a|\gt|}
$$
Hence \equ(KPht) becomes the well known condition
$$
\sup_{x\in \VU}\,\sum_{\g\in \PP_\VU\atop\g\ni x} \,\r(\g)\, e^{a|\g|}\,\,\le \,\,a\Eq(KPht1)
$$

\\On the other hand the Dobrushin condition \equ(Dcrit) can be written as
$$
\r(\g)\le {c(\g)\over \prod_{\gt\in\PP_\VU:\,\tilde\g\not\sim \g}[1+c(\tilde\g)]}
$$

\\choosing  again $c(\g)=|\r(\g)| e^{a|\g|}$ the condition above becomes
$$
\prod_{\gt\in\PP_\VU\atop\tilde\g\not\sim\g} (1+\r(\gt) e^{a|\gt|})\le e^{a|\g|},\,\,\,\,\,\,\,\,\,\forall\g\in\PP
$$
i.e.
$$
\sum_{\gt\in\PP_\VU\atop\tilde\g\not\sim\g} \log(1+\r(\gt) e^{a|\gt|})\le a|\g|\,\,\,\,\,\,\,\,\,\forall\g\in\PP\Eq(Dht)
$$
i.e.
$$
\sup_{x\in \VU}\sum_{\g\in\PP_\VU\atop\g\ni x}\log(1+\r(\g) e^{a|\g|})\le a
\Eq(Dht1)
$$which is slightly better than \equ(KPht1).

\\Finally, the condition \equ(FPcrit), putting again $c(\g)=\r(\g) e^{a|\g|}$, becomes
$$
\Xi_{\,\PP_\g}(c)\le e^{a|\g|}\Eq(usht)
$$
where $\PP_\g=\{\g'\in \PP_\VU: \g'\cap \g\neq\emptyset\}$ and
$$
\Xi_{\,\PP_\g}(c)=1+ \sum_{n= 1}^{|\g|}{1\over n!}
\sum_{(\g_1,\dots,\g_n)\in \PP_\VU^n\atop\g_i\not\sim\g, \g_i\sim\g_j} \prod_{i=1}^n\r(\g_i) e^{a|\g_i|}
$$

\\We again use the fact that $\g_i\not\sim\g_j\Longleftrightarrow\g_i\cap\g_j\neq\emptyset$ and $\g_i\sim\g_j\Longleftrightarrow\g_i\cap\g_j=\emptyset$
to estimate the factor
$$
\sum_{(\g_1,\dots,\g_n)\in \PP_\VU^n\atop\g_i\not\sim\g, \g_i\sim\g_j} \prod_{i=1}^n\r(\g_i) e^{a|\g_i|}
$$
Note that this factor is zero whenever $n>|\g|$, since there is no way to choose $n$ subsets $\g_i$ of $\VU$ such that they are all pairwise compatible
(i.e. non intersecting)
and all incompatible (i.e. intersecting) with a fixed subset $\g$ of $\VU$ with a number of elements equal to  $|\g|$.
On the other hand, when the sum above is not zero, i.e. for $n\le |\g|$, it can be bounded at least by (a very rough bound)

$$
\begn
\sum_{(\g_1,\dots,\g_n)\in \PP^n\atop\g_i\not\sim\g, \g_i\sim\g_j} \prod_{i=1}^n\r(\g_i) e^{a|\g_i|}
& \le
|\g|(|\g|-1)\cdots (|\g|-n+1)\Bigg[\sup_{x\in \VU}\,\sum_{\g\in  \PP_\VU\atop x\in \g}{\r(\g)}e^{a|\g|}\Bigg]^n\\
& =~
{|\g|\choose n}n! \Bigg[\sup_{x\in \VU}\,\sum_{\g\in  \PP_\VU\atop x\in \g}{\r(\g)}e^{a|\g|}\Bigg]^n
\egn
$$
Thus

$$
\Xi^{\,\g}_{\,\PP_\VU}(c)\le 1+ \sum_{n=1}^{|\g|}
{|\g|\choose n}\Bigg[\sup_{x\in \VU}\,\sum_{\g\in  \PP_\VU\atop x\in \g}{\r(\g)}e^{a|\g|}\Bigg]^n=
\Bigg[1+\sup_{x\in \VU}\,\sum_{\g\in  \PP_\VU\atop x\in \g}{\r(\g)}e^{a|\g|}\Bigg]^{|\g|} \Eq(4321)
$$
Thus \equ(usht) can be written as
$$
\Bigg[1+\sup_{x\in \VU}\,\sum_{\g\in  \PP_\VU\atop x\in \g}{\r(\g)}e^{a|\g|}\Bigg]^{|\g|}\le e^{a|\g|}
$$
i.e.
$$
\sup_{x\in \VU}\,\sum_{\g\in  \PP_\VU\atop x\in \g}{\r(\g)}e^{a|\g|}\le e^{a}-1\Eq(usht1)
$$
Note that by \equ(r.12b0) we also get the upper bound
$$
\Pi_{\g}(\r)\le   e^{a|\g|}
$$
We have therefore proved the following theorem

\begin{teo}\label{teo42}
Let $\VU$ be a countable set and let $\PP_\VU=\{\g\subset \VU: |\g|<\infty\}$ be the  polymer set  with incompatibility relation:
$\g\not\sim\g'\,\,\Leftrightarrow\,\, \g\cap\g'\neq\0$ and (complex) activity $\z(\g)$.
Assume that there is a positive number $a>0$,  such that, for all $x\in \VU$
$$
\sum\limits_{\g\in \PP_\VU\atop x\in\g}
|\z(\g)| ~  e^{a|\g|}\le e^a-1\Eq(fpset)
$$
then, for   all finite $\L\subset \VU$, we have  that ${1\over |\L|}\ln \Xi_\L(\bm \z)$,  where $ \Xi_\L(\bm \z)$ is the partition function define in \equ(GKp), can be written as an absolutely convergent series uniformly in $\L$.
\end{teo}

\subsection{The subset gas: proof by Induction}

We now illustrate the reasoning by induction to prove convergence of the cluster expansion
for the subset gas in the space $\PP_\VU$ in the special case in which the activity is non negative (i.e. $\r(\g)\ge 0$ for all $\g\in \PP_\VU$. We use the Miracle-sole
approach.

\begin{teo}
Let $\VU$ be a countable set and let $\PP_\VU=\{\g\subset \VU: |\g|<\infty\}$ be the  polymer set  with incompatibility relation:
$\g\not\sim\g'\,\,\Leftrightarrow\,\, \g\cap\g'\neq\0$ and activity $\r(\g)\ge 0$.
Assume that there is a positive number $a>0$,  such that, for all $x\in \VU$
$$
\sum\limits_{\g\in \PP_\VU\atop x\in\g}
\r(\g) ~  e^{a|\g|}\le e^a-1\Eq(gkindu)
$$
(here of course $\r\ge0$)
then, for all $\L\in \VU$, we have
$$
|\Si^\L_{x}|(\r)\le a
\Eq(6bmigk)
$$
where $x\in \L$ and
$$
|\Si^\L_{x}|(\r)= \sum_{n=1}^{\infty}{1\over n!}
\sum_{(\g_{1},\dots ,\g_{n})\in\L^n\atop \exists i:~x\in \g_i}
|\phi^{T}(\g_1 ,\dots , \g_n)|\;{\r(\g_1)}\dots{\r(\g_n)}
$$

\end{teo}

\\{\bf Proof}. First observe that, by the alternate sign property of $\phi^{T}(\g_1 ,\dots , \g_n)$ we have
that

$$
|\Si^\L_{x}|(\r)= -\log \Xi_{\L}(-\r)+  \log \Xi_{\L\backslash\{x\}}(-\r)\Eq(funda)
$$

\\We will perform the proof by induction on $\L$. We assume that, under the condition \equ(gkindu), the inequality
\equ(6bmigk) is satisfied for a given $\L$ and any of its subsets.  Take now $x\in \Z^d\backslash\L$. We want to bound
$$
|\Si^{\L\cup\{x\}}_{x}|(\r)= -\log \Xi_{\L\cup\{x\}}(-\r)+  \log \Xi_{\L}(-\r)= - \log{\Xi_{\L\cup\{x\}}(-\r)\over \Xi_{\L}(-\r)}
$$

\\But now we know that
$$
\Xi_{\L\cup\{x\}}(-\r)= \Xi_{\L}(-\r) -\sum_{S\subset \L}\r(\{x\}\cup S) \Xi_{\L\backslash S}(-\r)
$$
So
$$
{\Xi_{\L\cup\{x\}}(-\r)\over \Xi_{\L}(-\r)} =
1 -\sum_{S\subset \L}\r(\{x\}\cup S) {\Xi_{\L\backslash S}(-\r)\over \Xi_{\L}(-\r)}
$$
Now, let $S=\{y_1,\dots,y_k\}$ with of course $k=|S|$, then
$$
{\Xi_{\L\backslash S}(-\r)\over \Xi_{\L}(-\r)}= {\Xi_{\L\backslash \{y_1\}}(-\r)\over \Xi_{\L}(-\r)}
{\Xi_{\L\backslash \{y_1,y_2\}}(-\r)\over \Xi_{\L\backslash \{y_1\}}(-\r)}
\cdots
{\Xi_{\L\backslash \{y_1,\dots,y_k\}}(-\r)\over \Xi_{\L\backslash \{y_1,\dots,y_{k-1}\}}(-\r)}
$$
The  r.h.s. of equation above is a product of $k=|S|$ terms of the form
$$
{\Xi_{\tilde\L\backslash \{z\}}(-\r)\over \Xi_{\tilde\L}(-\r)}
$$
with $\tilde \L\subset \L$ and $z\in \tilde \L$. So by the induction hypothesis for each of these terms
we have
 $$
\log {\Xi_{\tilde\L\backslash \{z\}}(-\r)\over \Xi_{\tilde\L}(-\r)}= \log\Xi_{\tilde\L\backslash \{z\}}(-\r)-
\log\Xi_{\tilde\L}(-\r)=
$$
$$
= - \log\Xi_{\tilde\L}(-\r)+ \log\Xi_{\tilde\L\backslash \{z\}}(-\r)=|\Si^{\tilde\L}_{z}|(\r)\le a
$$
So, adopting the convention $\L\backslash \{y_1,\dots,y_{i-1}\}=\L$ for $i=1$,
$$
\log {\Xi_{\L\backslash S}(-\r)\over \Xi_{\L}(-\r)}=\sum_{i=1}^k\log
{\Xi_{\L\backslash \{y_1,\dots,y_i\}}(-\r)\over \Xi_{\L\backslash \{y_1,\dots,y_{i-1}\}}(-\r)}\le k\,a=|S|\,a
$$
and hence
$$
{\Xi_{\L\backslash S}(-\r)\over \Xi_{\L}(-\r)}\le e^{a|S|}
$$

Thus

$$
|\Si^{\L\cup\{x\}}_{x}|(\r)= -\log \Xi_{\L\cup\{x\}}(-\r)+  \log \Xi_{\L}(-\r)= - \log{\Xi_{\L\cup\{x\}}(-\r)\over \Xi_{\L}(-\r)}=
$$
$$
-\log \Bigg[1 -\sum_{S\subset \L}\r(\{x\}\cup S) {\Xi_{\L\backslash S}(-\r)\over \Xi_{\L}(-\r)}\Bigg]\le
-\log \Bigg[1 -\sum_{S\subset \L}\r(\{x\}\cup S) e^{a|S|}\Bigg]=
$$
$$
= -\log \Bigg[1 -\sum_{S\subset \L}\r(\{x\}\cup S) e^{a|S|}\Bigg]= -\log \Bigg[1 -{1\over e^a}\sum_{S\subset \L}\r(\{x\}\cup S) e^{a(|S|+1)}\Bigg]
$$
$$
\le -\log \Bigg[1 -{1\over e^a}(e^a-1)\Bigg]=a
$$
To complete the induction we have to show that the  claim is true when $\L$ contains just one point,
i.e. $\L=\{x\}$. If $\L=\{x\}$ then we have only to check the function

$$
\begn
|\Si^{\{x\}}_{x}|(\r) & = \sum_{n=1}^{\infty}{1\over n!}
\sum_{(\g_{1},\dots ,\g_{n})\in (\{x\})^n\atop \exists i:~x\in \g_i}
|\phi^{T}(\g_1 ,\dots , \g_n)|\;{\r(\g_1)}\dots{\r(\g_n)}\\
& =
\sum_{n=1}^{\infty}{1\over n!}
|\phi^{T}(\{x\} ,\dots , \{x\})|\;[{\r(\{x\})}]^n\\
&=
\sum_{n=1}^{\infty}{1\over n}
\;[{\r(\{x\})}]^n\\
&\le \sum_{n=1}^{\infty}{1\over n}
\;\Big[1-{1\over e^a}\Big]^n\\
&= -\log \Big[1-  \Big(1-{1\over e^a}\Big)\Big]\\
&=a
\egn
$$
where in the third line we have used that $|\phi^{T}(\{x\} ,\dots , \{x\})|=\,(n-1)!$ and
in the fourth line we used that, by \equ(gkindu), $\r(\{x\})\,\,\,\le 1-{1\over e^a}$. So
the induction is completed. $\Box$

\vv\v\vv

\chapter{Two systems in the cubic lattice}
\numsec=5\numfor=1
To show the utility of the polymer expansion in the study os discrete systems, we will consider in this chapter two systems on the $d$-dimensional cubic lattice $\mathbb{Z}^d$. We agree that $\mathbb{Z}^d$ is equipped with the usual ``graph distance" metric $|\cdot |$. Namely, given two
points $x$ and $y$ in  $\mathbb{Z}^d$, then $|x-y|$ denotes the graph distance between $x$ and $y$ (i.e. the number of nearest neighbor edges of the shortest path between $x$ and $y$).
\section{The Self repuslive Lattice gas}
The self repulsive  lattice gas in  $\Z^d$ is a system  formed by ``particles"  which may or may not  occupy   the  vertices $x$
of the unit cubic lattice $\Z^d$. In other words we associate to  any site $x$ of $\Z^d$ a variable $n_x$ taking values in the set $\{0,1\}$. When $n_x=1$ a particle is occupying the site $x$, and when $n_x=0$, the site $x$ is empty. The name self-repulsive is because  each site  can be occupied by at most one particle. Let us then suppose that each particle  occupying the site $x$ has an activity
$\l_x$  assumed to be constant, i.e. $\l_x=\l$ for all $x\in \Z^d$. Finally, let us suppose that particles in the self repulsive lattice gas interact via a pair potential $V(x,y)$ which we assume to be  symmetric $V(x,y)=V(y,x)$ (and, of course,  self repulsive, i.e. only one particle can occupy a site $x\in \Z^d$). Namely:
$$
V(x,x)=+\infty\,\,\,\,\,\,\,\,\,\ \mbox{for all $x\in \Z^d$}\Eq(la0)
$$
We will also assume that the pair potential
is summable (regular) in the following sense
$$
\sup_{x\in \Z^d}\sum_{y\in \Z^d\atop y\neq x}|V(x,y)|= J<\infty\Eq(la1)
$$
Note that the assumptions \equ(la1) and  \equ(la1) automatically guarantee that the pair potential
is stable with stability constant not greater that $J/2$. Indeed
for any $n\in \N$ and $(x_1,\dots,x_n)\in (\Z^d)^n$,
$$
\sum_{1\leq i<j\leq n}V(x_i ,x_j)\ge - {1\over 2} \sum_{(i,j)\in [n]^2\atop i\neq j}|V(x_i ,x_j)| = -
{1\over 2}\sum_{i=1}^n\sum_{i\in [n]\atop j\neq 1}|V(x_i ,x_j)|\geq -{J\over 2}n, \Eq(la2)
$$
\\The Grand Canonical partition function of such a lattice gas is
$$
Z_{\La}(\b ,\l)
=~\sum_{n=0}^{\i}{\l^n\over n!} \sum_{(x_1 ,\dots
,x_n)\in\La^n} ~e^{-\b\sum_{1\le i<j\le n}V(x_i ,x_j)}
$$
The specific case in which
$$
V_{\rm h.c.}(x,y)=\begin{cases}
+\infty & {\rm if} ~|x-y|\le 1\\
0 & {\rm otherwise}
\end{cases}
$$
belongs to the class of   pair potentials  satisfying assumptions \equ(la1) and \equ(la2).
so by the discussion done in Section \ref{elexe}, if $V(x_i ,x_j)=V_{\rm h.c.}(x_i ,x_j)$ then the finite volume pressure
$p_\L(\l)={1\over |\L|}\ln Z_{\La}(\l)$ (now independent from $\b$) converge absolutely  and it is bounded by $(2d-1)^{-1}$ uniformly in $\L$
as soon as
$$
\l\le \l_c\equiv {1\over 1+{(2d)^{2d}\over (2d-1)^{2d-1}}} \Eq(raggiozd)
$$
In other words, $p_\L(\b,\l)$ is analytic in the disk $|\l|\le \l_c$.
\subsection{Covergence by direct Mayer expansion}
\\As far as the general case is concerned (i.e. $V(x,y)$ regular and stable pair potential $\b$), we can obtain the  convergent condition for the high tempretaure/low activity phase of such lattice gas by directly perform a Mayer expansion of the
grand canonical partition function similarly to what we did in Secrion\ref{mayersec}. Let us sketch how.
\index{partition function!grand canonical}
\\The Grand Canonical partition function is
$$
Z_{\La}(\b ,\l)
=~\sum_{n=0}^{\i}{\l^n\over n!} \sum_{(x_1 ,\dots
,x_n)\in\La^n} ~e^{-\b\sum_{1\le i<j\le n}V(x_i ,x_j)}
$$
where of course we are assuming in that
$V(x,x) =+\i$. By Mayer expansion on the factor
$e^{-\b\sum_{i<j}V(x_i ,x_j)}$ we find, as usual
$$
\ln Z_{\La}(\b ,\l) ~=~ \sum_{n=1}^{\i}{\l^n\over n!}\sum_{(x_1
,\dots ,x_n)\in\La^n} ~\sum_{g\in { G}_n}\prod_{\{i,j\}\in
E_g}(e^{-\b V(x_i ,x_j)}-1)
$$
Then, by Theorem \ref{teoPY2} and tree graph inequality \equ(bteo1b), we have, using stability also condition \equ(la2),
$$
|\sum_{g\in { G}_n}\prod_{\{i,j\}\in E_g}(e^{-\b V(x_i ,x_j)}-1)|\le
e^{ {\b  J\over 2}n}\sum_{\t\in { T}_n}\prod_{\{i,j\}\in E_\t}(1- e^{-\b |V(x_i ,x_j)|})
$$
hence
$$
|\ln Z_{\La}(\b ,\l)| \le \sum_{n~=~1}^{\i}{|\l|^n\over
n!}\sum_{(x_1 ,\dots ,x_n)\in\La^n} ~e^{ {\b  J\over 2}n}\sum_{\t\in {
T}_n}\prod_{\{i,j\}\in E_\t}(1- e^{-\b |V(x_i ,x_j)|})\le
$$
$$\le
\sum_{n~=~1}^{\i}{|\l|^n\over n!} ~e^{ {\b  J\over 2}n}\sum_{\t\in { T}_n}
\sum_{(x_1 ,\dots ,x_n)\in\La^n}\prod_{\{i,j\}\in E_\t}(1- e^{-\b |V(x_i ,x_j)|})
$$
but similarly to what we saw in Proposition \ref{integr}, formula \equ(treintg) it is not difficult to show that
$$
\sum_{(x_1 ,\dots ,x_n)\in\La^n}\prod_{\{i,j\}\in E_\t}(1- e^{-\b |V(x_i ,x_j)|})
\le |\L| \left[\sup_{x\in \Z^d}\sum_{y\in \Z^d}(1-e^{-\b |V(x,y)|})\right]^{n-1}
$$
note that the sum over $y$ after the sup includes also $y=x$ where $V(x,x)~=~\i$.
Observing that
$$
\sup_{x\in \Z^d}\sum_{y\in \Z^d\atop y\neq x}(1-e^{-\b |V(x,y)|})\le \sup_{x\in \Z^d}\sum_{y\in \Z^d\atop y\neq x}\b |V(x,y)|\le \b J
$$
we get
$$
\sum_{(x_1 ,\dots ,x_n)\in\La^n}\prod_{\{i,j\}\in E_\t}(1- e^{-\b |V(x_i ,x_j)|})
\le |\L| \left[1+ \b J\right]^{n-1}
$$
and therefore
$$
\begn
|\ln Z_{\La}(\b ,\l)| & \le \sum_{n=1}^{\i}{|\l|^n\over n!} ~e^{ {\b  J\over 2}n}|\L|[1+ \b J]^{n-1}n^{n-2}\\
 & \le \sum_{n=1}^{\i}{|\l|^n\over n^2} ~e^{ {\b  J\over 2}n}|\L|[1+ \b J]^{n-1}{n^{n}\over n!}\\
& \le \sum_{n=1}^{\i}{|\l|^n\over n^2} ~e^{ {\b  J\over 2}n}|\L|[1+ \b J]^{n-1}e^n
\egn
$$
thus in this case the condition for the convergence is
$$
|\l| e^{ {\b  J\over 2}+1}[1+ \b J] <1 \Eq(!!!b)
$$
This condition is quite unsatisfactory since it says that, for any temperature, even very high (i.e even for $\b$ very small),
one needs to set the activity smaller than $1/e$ to ensure convergence.

\subsection{High temperature  polymer expansion of the lattice gas}
\\We now obtain the convergence condition for the same $\ln Z_{\La}(\b ,\l)$ performing a first example of high temperature polymer expansion. We will get a much satisfactory  bound and this will illustrate
quite well how convenient can be to perform polymer expansion in discrete systems when
possible!

\\The
{\it Grand canonical partition function} of the same  lattice gas enclosed
in $\La\subset \Z^d$, with activity $\l$, inverse temperature $\b$,
interacting via a pair potential $V(x,y)$ ($x$ and $y$ sites in
$\La$) such that $V(x,x)~=~\infty$ (a site $x$ can be occupied at most by one particle) can be written as follows

$$
\begin{aligned}
Z_{\La}(\b ,\l)~
& =~\sum_{n=0}^{\i}{\l^n\over n!} \sum_{(x_1 ,\dots
,x_n)\in\La^n} ~e^{-\b\sum_{i<j}V(x_i ,x_j)}\\
&=~\sum_{n=0}^{|\La|}{\l^n\over n!} \sum_{(x_1
,\dots ,x_n)\in\La^n\atop x_i\neq x_j} ~e^{-\b\sum_{i<j}V(x_i
,x_j)}\\
&= ~\sum_{n=0}^{|\La|} \l^n\sum_{\{x_1
,\dots ,x_n\}\subset \L}  ~e^{-\b\sum_{i<j}V(x_i
,x_j)}\\
&=~\sum_{S\subset \L} \l^{|S|} ~e^{-\b\sum_{\{x,y\}\subset S}V(x,y)}
\end{aligned}
$$
Define then for $x\in \La$ the variable  $n_x$ taking values in the set $\{0,1\}$  ($n_x$ can be interpreted as the occupation number of the site $x$: $n_x=0$ means that
the site is empty and $n_x=1$ means that the site is occupied).

\\We can rewrite the partition function above by considering that to each $S\subset \L$ we can associate a unique function $\bm n_\L:\L\to \{0,1\}: x\mapsto n_x$ such that $\bm n^{-1}(1)=S$ (we remind that $\bm n_\L$ is usually called a configuration of the lattice gas in $\L$ and it specifies which sites are occupied by  a particle and which sites are empty). We let $N_\L$ be the set of all such functions. Clearly, if $|N_\L|$ denotes the cardinality of $N_\L$ and
$|\L|$ is the number of sites in $\L$, then  $|N_\L|=2^{|\L|}$. Moreover, in general, if $R\subset \mathbb{Z}^d$, we denote by
$\bm n_R$ a possible configuration  in $R$ and by $N_R$ the set
of all possible configurations $n_R$.

\\As remarked earlier, there is a one to one correspondence between $\bm n_\L\in N_\L$ and $S\subset \L$ by defining $S(n_\L)=\{x\in \L: n_x=1\}$.
Therefore we can write
$$
\begn
Z_{\La}(\b ,\l) &
=~\sum_{S\subset \L}\,\l^{|S|} ~e^{-\b\sum_{\{x,y\}\subset S}V(x,y)}\\
&=~\sum_{\bm n_\L\in N_\L}
\l^{\sum_{x\in\La}n_x}
e^{-\b\sum_{\{x,y\}\subset\La}n_xn_y V(x,y)}
\egn
\Eq(1.7z)
$$

\\Expand now the exponential in  \equ(1.7z)
\vskip.2cm
$$
e^{-\b\sum_{\{x,y\}\subset\La}n_x n_y V(x,y)}~=~
\prod_{\{x,y\}\subset\La} [e^{-\b n_x n_y
V(x,y)}-1+1]~=~
$$
$$
~=~ \sum_{s=1}^{|\L|}\sum_{\{R_1 ,\dots ,R_s\}\in\pi(\La)}\r(R_1)\cdots\r(R_r)
$$
where $\pi(\La)~=~$ set of all partitions of $\La$ (so that $\{R_1 ,\dots ,R_s\}$ is a partition of $\L$), and
\vskip.3cm
$$
\r(R)~=~
\begin{cases} 1 &{\text if} |R|~=~1\\
\displaystyle{\sum\limits_{g\in G_R}\prod\limits_{\{x,y\}\in
E_g}[e^{-\b n_x n_y V(x,y)}-1]} &{\text if} ~ |R|\geq 2\
\end{cases}
$$
where $G_R$ is the set of connected graph with vertex set $R$.

\nin Thus \equ(1.7z) can be written as, \vskip.5cm
$$
\begn
Z_{\La}(\b ,\l)  &=~  \sum_{\bm n_\L\in N_\L}
\l^{\sum_{x\in\La}n_x}\sum_{s=1}^{|\L|}
\sum_{R_1 ,\dots,R_s\in\pi(\La)}\r(R_1)\cdots\r(R_s)\\
&=~ \sum_{s=1}^{|\L|}\sum_{R_1 ,\dots ,R_s\in\pi(\La)}
\sum_{\bm n_{R_1}\in N_{R_1}}\r(R_1)\l^{\sum\limits_{x\in R_1}n_x} \cdots
\!\!\!\sum_{\bm n_{ R_s}\in N_{R_s}}\!\!\!\r(R_s)\l^{\sum\limits_{x\in R_s}n_x}
\egn
$$
Define now

$$
{\tilde \r}(R)\,~=~\, \sum_{n_{R}\in N_R}\;\r(R)\;\l^{\sum_{x\in R}\;n_x}
$$
\vskip.5cm
\\Now observe that, for any $g\in G_R$, the factor
$$
\prod_{\{x,y\}\in E_g}[e^{-\b n_x n_y V(x,y)}-1]\neq 0
$$
is different from zero only for the configuration $\bm n_R$ such that $n_x =1$ for all $x\in R$. Therefore
$$
{\tilde \r}(R)~=~
\begin{cases}
1+\l &{\rm if}~ |R|=1\\
\l^{|R|}\sum\limits_{g\in
G_R}\prod\limits_{\{x,y\}\in E_g}[e^{-\b V(x,y)}-1] &{\rm if}~ |R|\geq 2
\end{cases}
$$
Defining  now
$$
\z(R)~=~
\begin{cases}
1 &{\text if}~ |R|~=~1\\
{\l^{|R|}\over (1+\l)^{|R|}}
\sum\limits_{g\in G_R}\prod\limits_{\{x,y\}\in E_g}[e^{-\b V(x,y)}-1] &{\text if}~ |R|\geq 2
\end{cases}
\Eq(defact)
$$
we obtain
\bea\nonumber Z_{\La}(\b ,\l) & = & (1+\l)^{|\La|}\sum_{s=1}^{|\L|}
\sum_{R_1 ,\dots ,R_s\in\pi(\La)} \z(R_1)\cdots\z(R_s)\\\nonumber
\\\nonumber
&= & (1+\l)^{|\La|}\sum_{s\ge 0}\sum_{\{R_1 ,\dots ,R_s\}\subset \L\atop |R_i|\ge 2,~R_i\cap R_j=\0} \z(R_1)\cdots\z(R_s)
\eea
where the term $s=0$ in the last sum is equal to 1 and corresponds to the partition of $\L$ in $|\L|$ subsets each of cardinality 1. Clearly
$$
\begn
\Xi_{\La}(\b ,\l)
& \doteq
\sum_{s\ge 0}\sum_{\{R_1 ,\dots ,R_s\}\subset \L\atop |R_i|\ge 2,~R_i\cap R_j=\0} \z(R_1)\cdots\z(R_s)\\
&=~1+\sum_{n\geq 1}{1\over n!} \sum_{R_1 ,\dots
,R_n\atop |R_i|\geq 2,\, R_i\cap R_j~=~\emptyset}
\z(R_1)\cdots\z(R_n)
\egn
$$
I.e. $\Xi_{\La}(\b ,\l)$ is the grand canonical partition function of a polymer gas in which polymers are
finite subsets of $R\subset \mathbb{Z}^d$ such that $|R|\ge 2$ and with activity $\z(R)$.

\\We finally  have
$$
Z_{\La}(\b ,\l)=  (1+\l)^{|\La|}\Xi_{\La}(\b ,\l)
$$
and hence
$$
\ln Z_{\La}(\b ,\l)=|\L| \ln (1+\l) + \ln \Xi_{\La}(\b ,\l)
$$

\\Hence the pressure of the lattice gas ${1\over |\L|}\log Z_{\La}(\b ,\l)$ converges absolutely  uniformly in $\La$ if ${1\over |\L|}\log
\Xi_{\La}(\b ,\l)$ does.

\\I.e., by \equ(usht1) conditions if,
$$
\sum_{n\ge 2}e^{an} \sup_{x\in \Z^d}\sum_{R\subset \Z^d:\atop x\in R, ~|R|=n} |\z(R)|\le e^a-1\Eq(conlat)
$$
or a little bit  ``roughly''  (i.e choosing $a=\ln 2$ which is not an optimal choice)
$$
\sum_{n\ge 2}2^{n} \sup_{x\in \Z^d}\sum_{R\subset \Z^d:\atop x\in R, ~|R|=n} |\z(R)|~\le ~1\Eq(KP)
$$

\\Recalling definition \equ(defact) of the activity $\z(R)$
and setting ${\tilde \l} ~=~ {\l\over (1+\l)}$ we have
\def\lt{{\tilde \l}}
\vskip.5cm
$$
\sup_{x\in Z^d}\sum_{R\subset \Z^d:\atop x\in R, ~|R|=n} |\z(R)|~=~
{|\lt|}^n\sup_{x\in Z^d}\sum_{R\subset \Z^d:\atop x\in R, ~|R|=n}
\bigg|\sum_{g\in G_R}\prod_{\{x,y\}\in E_g}[e^{-\b V(x,y)}-1]\bigg|
$$
Now note that
$$
\sum_{R\subset \Z^d:\atop x\in R, ~|R|=n} ~=~{1\over (n-1)!}\sum_{(x_1,\dots ,x_n)\in \Z^{dn}\atop x_1=x,\;x_i\ne x_j\ {\rm for}\ i\ne j}
$$
Hence
$$
\sup_{x\in \Z^d}\sum_{R\subset \Z^d:\atop x\in R, ~|R|=n} |\z(R)|~=
$$
$$= {{|\lt|}^n\over (n-1)!}\sup_{x\in\Z^d}
\sum_{(x_1,\dots ,x_n)\in \Z^{dn}\atop x_1=x,\;x_i\ne x_j\ {\rm for}\ i\ne j} \bigg|\sum_{g\in
G_n}\prod_{\{i,j\}\in E_g}[e^{-\b V(x_i,x_j)}-1]\bigg| \Eq(reorg)
$$

\\Once again,  by Theorem \ref{teoPY2}, tree graph inequality \equ(bteo1) and using stability condition
\equ(la2) we have
the estimate
$$
 \bigg|\sum_{g\in
G_n}\prod_{\{i,j\}\in E_g}[e^{-\b V(x_i,x_j)}-1]\bigg|\leq
e^{ {\b  J\over 2}n}\sum_{\t\in T_{n}}
\prod_{\{i,j\}\in E_\t}(1-e^{-\b |V(x_i,x_j)|})
$$
Now observe that, for any $\t\in T_n$
$$
\begn
\sum_{(x_1,\dots ,x_n)\in \Z^{dn}\atop x_1=x,\;x_i\ne x_j\ {\rm for}\ i\ne j}
\prod_{\{i,j\}\in E_\t}(1-e^{-\b |V(x_i,x_j)|})& \le~ \Big[\sup_{x\in \Z^d}\sum_{y\in \Z^d\atop y\neq x}\b |V(x,y)|\Big]^{n-1}\\
&=~
[\b J]^{n-1}
\egn
$$
Therefore, recalling also  Cayley formula \equ(cay2), we get

$$
\begn
\sup_{x\in Z^d}\sum_{R\subset \Z^d:\atop x\in R, ~|R|=n} |\z(R)|& \le~
 {n^{n-2}\over (n-1)!}[\b J]^{n-1} (|\lt| e^{\b J\over 2})^n\\
& \le ~ |\lt| e^{\b J\over 2} \Big[|\lt| e^{{\b J\over 2}+1}\,\b J\Big]^{n-1}
\egn
$$

\\Hence the condition for convergence \equ(conlat) is satisfied if
$$
2|\lt| e^{\b J\over 2}\sum_{n\ge 2} \Big[2|\lt| e^{{\b J\over 2}+1}\,\b J\Big]^{n-1}~\le ~1\Eq(!!!)
$$
which, after some calculus,  is satisfied if
$$
|\lt| e^{{\b J\over 2}+1}\le {1\over \b J}{1\over 1+ \sqrt{1+{4\over e\b J}}}\Eq(OKKK)
$$
This condition is much better  than \equ(!!!b).
Indeed, denoting by   $\b_c$
the unique solution of the equation
$$
e^{{\b J\over 2}+1}~=~{1\over \b J}{1\over 1+ \sqrt{1+{4\over e\b J}}}
$$
and recalling that $\lt =\l/(1+\l)$,
we get that, as soon as $\b\le \b_c$,
the pressure
of the lattice gas is an absolute convergent expansion for all
$\l$  real. The temperature $\b_c$ is a first example of critical temperature: below $\b_c$ the lattice  gas  is  in a pure phase for any  activity $\l>0$.

\section{The BEG model at low temperature}
\vv

\def\XX{{\rm X}}\def\YY{{\rm Y}}
As an example, we consider the Blume-Emery-Griffiths (BEG)  model
 in the low temperature disordered
phase. The model is defined on the cubic unit lattice in
$d$-dimensions $\Z^d$ by supposing that in each vertex $x\in \mathbb{Z}^d$
there is a spin variable $\s_x$ taking values in the set
$\{0,-1,+1\}$. Given  a volume $\L\subset \mathbb{Z}^d$ (typically, a cube centered at the origin), a spin configuration in $\L$
is a function $\bm \s: \L \to \{0,-1,+1\}: x\mapsto \s_x$, We denote by  $\Sigma_\L$ the set of all spin configurations in $\L$ (note that $|\Si_\L|=3^{|\L|}$). Given
$\bm \s\in \Si_\L$,  its  energy $ H_\L(\bm \s)$  is given by the Hamiltonian
$$
H_\L(\bm \s)=-\sum_{\{x,y\}\subset \L\atop |x-y|=1}[\s_x\s_y+ \YY\s^2_x\s^2_y]+ \XX\sum_{x\in \L}\s_x^2
\Eq(Ham)
$$
where  in general $\XX$ and  $\YY$ are real parameters. In other words spins interacts via a nearest neighbor pair potential.
We will further suppose that
$$
\XX>d(1+|\YY|) \Eq(dis)
$$
With this hypothesis it is not difficult to check that  the minimum of  $H_\L(\bm \s)$ is reached at the configuration $\bm \s$ such that   $\s_x=0$ for all $\x\in \L$. This
region is called the disordered phase. Via the  Hamiltonian \equ(Ham) we can  assign a probability $Prob(\bm \s)$ to each configuration $\bm \s$ in $\L$.  Namely,
$$
Prob(\bm \s)= {e^{ -\b H_\L(\bm \s)}\over Z_\L(\b)}
$$
where $\b>0$ is the inverse temperature and
$$
Z_\L(\b)=\sum_{\bm \s\in \Si_\L}e^{-\b H (\bm \s)}\Eq(parttt)
$$
is the partition function of the system. The thermodynamics of the BEG model  is recovered from the the partition function $Z_\L(\b)$. namely, the (finite volume) free energy
of the model is given by
$$
f_\L(\b)= {1\over |\L|} \ln Z_\L(\b)\Eq(freng)
$$

\\We'll show in this section  that the  free energy \equ(freng) of the BEG model  is analytic as function of $\b$ for $\b$ sufficiently large (i.e low temperature). This is achieved
 by rewriting  $Z_\L(\b)$ defined in  \equ(parttt) as the partition function of a polymer system of the type considered in the previous sections. Hence the free energy of the BEG model
coincides with the pressure of this polymer system.
Then, using theorem \ref{teo42}, we will prove that, in the disordered phase \equ(dis)  and with the assumptions
\equ(dis), the pressure of such  polymer system converges absolutely for $\b$ large enough (i.e.  for sufficiently low temperatures).

\\In order to do that, let us  define,
for a fixed spin configuration $\bm \s\in \Si_\L$,
the subset of $\L$ given
by  $P=\{x\in \L: \s_x\neq 0\}$. We view this set as the union of its  connected components,
i.e. $P=\cup_{i=1}^n p_i$ with each set $p_i\subset \L$  being connected in the sense that
for each partition $A,B$ of $p_i$ (i.e. $A\cup B=p_i$ and $A\cap B=\emptyset$) there exist $x\in A$ and $y\in B$ such that
$|x-y|=1$.
The configuration $\bm \s$ induces a (non zero) spin configuration $\bm s_{p_i}$ on each connected component
$p_i$ of $P$ which is a function
$\bm s_{p_i}:p_i\to \{-1,+1\}: x\mapsto s_x$. The pairs ${\bm p}_i=(p_i,\bm s_{p_i})$ are the polymers associated to the configuration $\bm \s$.

\\We denote by $\QQ$ the set of all subsets of $\Z^d$ which are connected and finite.
Given $p\in \QQ$, we denote by $S_p$ the set of all (non zero)  spin configurations $\bm s_p$ in $p$ (note that $|S_p|=2^{|p|})$.
By construction the correspondence $\bm \s\leftrightarrow \{\bm p_1,\dots,\bm p_n\}$ is one to one.
The distance between two polymers $\bm p=(p,\bm s_p)$ and $\tilde{\bm p}=(\tilde p,\bm s_{\tilde p})$
 is the number $d(p,\tilde p)=\min_{x\in p,\,\,y\in \tilde p}|x-y|$. Note that if $ \{\bm p_1,\dots,\bm p_n\}$ are the polymers
associated to the configuration $\bm \s$, then necessarily $d(p_i,p_j)\ge 2$ for all $\{i,j\}\subset \{1,\dots, n\}$.

\\With these definitions, given $\bm \s\in \Si_\L$ such that $\bm \s\leftrightarrow \{\bm p_1,\dots,\bm p_n\}$,   we can rewrite the Hamiltonian \equ(Ham)  as
$$
H_\L(\bm s)= \sum_{i=1}^n
\Big[\XX|p_i|-A(\bm p_i)\Big]
$$
where
$$
A(\bm p_i)=\sum_{\{x,y\}\subset p_i\atop |x-y|=1}  [ s_xs_y+\YY]\Eq(Ap)
$$
Observe now that to sum over configuration $\bm\s\in \Si_\L$ is equivalent to sum over polymers configurations
$\{\bm p_1,\dots,\bm p_n\}$ in $\L$ such that $n\ge 0$ ($n=0$, i.e. no polymers, is the lowest energy configuration $\s_x=0$ for all $x\in \L$)
and $d(p_i,p_j)\ge 2$ for all pairs $\{i,j\}\subset \{1,\dots,n\}$.
Hence the partition function of the system, at inverse temperature $\b$ and with free boundary conditions, is rewritten as
$$
\begn
Z_\L(\b) & =~\sum_{\s_\L}e^{-\b H (\s_\L)}\\
& =~1~+~\sum_{n\ge 1}{1\over n!}\sum_{(\bm p_1,\dots,\bm p_n)\in \PP^n \atop p_i\subset \L,\; d(p_i,p_j)\ge 2}\r_{\bm p_1}\dots\r_{\bm p_n}
\egn
\Eq(parttt2)
$$
where
$$
\r_{\bm p}=e^{-\b \big[\XX|p|-A(\bm p)\big]}\Eq(actp)
$$
and $\PP$ is the set
$$
 \PP=\Big\{\bm p=(p,\bm s_p): \mbox{$p\in \QQ$ and $\bm s_p\in S_p$}\Big\}\Eq(polyspace)
$$
Thus we have rewritten the partition function $Z_\L(\b)$ of the BEG model as the partition function of an hard-core  polymer gas in which the space of polymers is the set
given in \equ(polyspace). Each polymer $\bm p\in \PP$ has activity $\r_{\bm p}$ given by \equ(actp) and the hard core interaction $V(\bm p_i,\bm p_j)$  between pair  of polymers
$p_i, p_j$ is
$$
V(\bm p_i,\bm p_j)=\begin{cases}0&{\rm if}~ d(p_i,p_j)\ge 2\\
+\infty & {\rm otherwise}
\end{cases}
\Eq(extlongr)
$$

\\With these definitions it is immediate to see that r.h.s. of \equ(parttt2) can be written as
$$
\begn
Z_\L(\b)& = 1+\sum_{n\ge 1}{1\over n!}
\sum_{(\bm p_1,\dots,\bm p_n)\in \PP^n_\L}
%{(\bm p_1,\dots,\bm p_n)\in \P^_\L^n}
\r_{\bm p_1}\dots\r_{\bm p_n}
e^{-\sum_{1\le i<j\le n}V(\bm p_i,\bm p_j)}\\
&=
1+\sum_{n\ge 1}{1\over n!}
\sum_{(\bm p_1,\dots,\bm p_n)\in \PP^n_\L\atop \bm p_i\sim \bm p_j}
\r_{\bm p_1}\dots\r_{\bm p_n}
\egn
\Eq(parttt3)
$$
where we have denoted shortly $\PP_\L=\{(p,\bm s_p)\in \PP: \;p\subset \L\}$ and  $ \bm p_i\sim \bm p_j$ means that $d(p_i,p_j)\ge 2$ (and similarly
ahead we use the symbol  $ \bm p_i\nsim \bm p_j$ when  $d(p_i,p_j)< 2$).

\\So, by Theorem \ref{coro:1},  the pressure of this polymer gas (i.e. the free energy of our BEG model) is absolutely convergent
if there exist $\{\m_{\bm p}\}_{\bm p\in \PP}$ such that
such that
$$
\r_{\bm p}\le {\m_{\bm p}\over \Xi_{\bm p}(\bm \mu) }
 \;\;\;\,\,\,\,\,\,\,\forall \g\in \PP \Eq(muRv2)
$$
with
$$
\Xi_{\bm p}(\bm \mu)  = 1+\sum_{n\ge 1} \frac{1}{n!}\,
\sum_{(\bm p_{1},\dots ,\bm p_{n})\in\PP^n\atop \bm p_i\not\sim \bm p,\,\,\bm p_{i}\sim\bm p_{j}}
{\mu_{\bm p_1}}\dots{\mu_{\bm p_n}}
$$
Given $p\in \QQ$ let $\partial p=\{x\in \Z^d\setminus p: \exists y\in p~s.t.~ |x-y|=1 \}$  be the the boundary of $p$ and let
$\bar p = p\cup \partial p$.
\def\pp{{\bm p}}\def\pt{{\tilde{\bm p}}}
Given $\a>0$, we set
$$
\m_\pp=\r_\pp e^{\a|\bar p|}\def\pp{{\bm p}} \Eq(mpp)
$$
Hence, inserting \equ(mpp) in \equ(muRv2), we obtain that
the pressure of such contour gas can be written in terms of an absolutely convergent series if, for some $\a>0$
$$
\Xi_{\bm p}(\bm \mu) \le e^{\a|\bar p|}\Eq(converg2)
$$
Now, similarly to what we did in Section \ref{431} we can bound
$$
\Xi_{\bm p}(\bm \mu)\le \Bigg[1+\sup_{x\in \Z^d}\,\sum_{\bm p\in  \PP\atop x\in p}\r_{\bm p}e^{\a|\bar p|}\Bigg]^{|\bar p|}
$$
so that condition \equ(converg2) is satisfied if
$$
\sup_{x\in \Z^d}\sum_{\bm p\in  \PP\atop x\in p}\r_{\bm p}e^{\a|\bar p|}\le e^\a-1\Eq(alfs)
$$
Now, recalling \equ(actp),  observe that
$$
A(\bm p_i)=\sum_{\{x,y\}\subset p_i\atop |x-y|=1}  [s_xs_y+\YY]\le d(1+|\YY|)|p_i|
$$
and so
$$
\r_{\bm p}\le e^{-\b D|p|}
$$
where
$$
D=  \XX-d(1+|\YY|)
$$
and observe that $D>0$ by hypotheis \equ(dis). Therefore, choosing for simpliciy $\a=\ln2$, condition \equ(alfs) is satisfied if
$$
\sup_{x\in \Z^d}\sum_{\bm p\in  \PP\atop x\in p}e^{-\b D|p|}2^{|\bar p|}\le 1\Eq(alfs2)
$$
Now observe that surely $|\bar p|\le 2d|p|+2\le 3d|p|$. Therefore we get that  condition \equ(alfs2) is satisfied if
$$
\sum_{n=1}^\infty (2^{3d+1}e^{-\b D})^n \sup_{x\in \Z^d}\sum_{\bm p\in  \PP\atop x\in p,\;|p|=n}1 \le 1\Eq(quasi)
$$
Now we have that
$$
\sup_{x\in \Z^d}\sum_{\bm p\in  \PP\atop x\in p,\;|p|=n}1=2^{n} \sup_{x\in \Z^d}\sum_{ p\in  \QQ\atop x\in p, \;|p|=n}1=2^nC_n
$$
where $C_n$ is the number of connected sets of vertices of $\Z^d$ with cardinality $n$
containing a fixed point $x$ of $\Z^d$ (this number does not depend on $x$ and can be chosen to be the orgin)
 and the extra factor $2^n$  counts the number of functions $\s_p$
from $p$ to $\{-1,+1\}$ when $|p|=n$. $C_n $ can be easily bounded by $C^n$ for some $C$, e.g. one can
take $C_n\le (4d)^n$. So condition \equ(quasi) becomes
$$
\sum_{n=1}^\infty (d2^{3d+3}e^{-\b D})^n\le 1
$$
which is satisfied if
$$
\b \ge \b^* \doteq {\ln (d2^{3d+4})\over D}\Eq(finale)
$$
Therefore  for all  $\b\ge \b^*$ the free energy  \equ(freng) of the BEG mode under the condition \equ(dis) is analytic in $\b$ an no phase transitions occur.

\chapter{Spin systems in a lattice}
\numsec=6\numfor=1

We denote by $ \mathbb{Z}^d$ the unit lattice
in $d$ dimensions, i.e. an element  $x\in  \mathbb{Z}^d$ is  an
ordered  $d$-ple $x~=~(n_1 ,\dots ,n_d): n_i\in\mathbb{Z}$. We will
consider $ \mathbb{Z}^d$ embedded in the standard way in
$\mathbb{R}^d$. Let $\L$ be a  finite  subset of ${\Z}^d$,
e.g.usually one takes al set $\L$ the set $\L~=~\{ (n_1 ,\dots
,n_d): n_i\in\mathbb{Z}; -{L\over2} \leq n_i \leq {L\over2},
\forall i~=~1,\dots ,d\}$, i.e. all sites of the unit lattice inside
a cube with the center in the origin and with size $L$.
\index{spin system}
\\The definition of a {\it spin system} on ${\Z}^d$ begins by associating to each
site $x\in \L$ a  variable $\phi_x$, (the "spin" at the site $x$).
The spin $\phi_x$  takes values in some space $\Omega$ (which is the same for all $x\in \Z^d$) equipped
with a probability measure $d\m(\phi_x)$ and with a norm
$\|\phi_x\|$. The single spin probability space $\O$ depends of the type of
system we are treating and, in general, its structure can be very
different from case to case. The probability measure $d\m(\phi_x)$
is called sometimes the {\it single spin a priori distribution}.

\\The simplest choice for which $\O$ is assuming that $\O$ is a {\it finite set},
with $d\m$ assigning equal probability to any value of the spin in
this set $\O$. E.g.  the choice $\O~=~\{1, -1\}$ (i.e. the spin
$\phi_x$ at a site $x$ can take only the two values $\phi_x~=~\pm
1$) and $d\m(\phi_x) ~=~{1\over 2}\d_{\phi_x,+1}+{1\over
2}\d_{\phi_x,-1}$, (i.e. the probability that the spin $\phi_x$ at
the site $x$ has the value $+1$ is ${1\over 2}$ and so for
$\phi_x~=~-1$) corresponds to {\it Ising systems}. For such systems
the random variable $\phi_x$ is usually denote in the literature
by  $\s_x$ and when  $\s_x~=~1$ people say that the spin is ``up"
while $\s_x~=~-1$ is referred as spin ``down".

\index{lattice gas}
\\The choice $\O~=~\{0,1\}$  and
$d\m(\phi_x) ~=~{1\over 2}\d_{\phi_x,+1}+{1\over 2}\d_{\phi_x,0}$
corresponds to the {\it Lattice gas} considered in the previous sections. In this case the
variable $\phi_x$ is usually denoted by $n_x$ and interpreted as
an occupation number, namely if $n_x=1$ then we say that a
particle is present at the site $x$ and the site $x$ is occupied,
while if $n_x=0$ then no particle is present at $x$ and the site
$x$ is empty.

\\The choice  $\O~=~\{1,2,\dots, q\}$ with $q$ positive integer
and $d\m(\phi_x) ~=~{1\over q} [\d_{\phi_x,1}+{1\over
q}\d_{\phi_x,2}+\dots + {1\over q}\d_{\phi_x,q}]$ correspond to
the so called {\it Potts model}.

\\A less simple choice for $\O$ is to assume $\O$ infinite but  compact. A typical example
is the (classical) Heisenberg model for which $\Omega~=~
\{\phi_{x}\in \mathbb{R}^d: \Vert\phi_{x}\Vert ~=~1\}$, i.e.
$\Omega$ is the  unit sphere in $d$ dimensions, and $d\m(\phi_x)~=~[
{\rm volume}~ \O]^{-1} d\O(\phi_x)$, where $d\O(\phi_x)$ is the
usual Rienmann measure on the sphere.

\\Spin systems for which   $\Omega$ is a compact space
(in the sense of the norm $\|\phi_x\|$), are usually called {\it
bounded spin systems}. But we could also consider the case in
which $\Omega$ is  not a compact space. Spin systems for which the
single spin space $\O$ is a non compact space are usually known as
{\it unbounded spin systems}. In general bounded spin systems are
treated easier than unbounded spin systems.

\\As an example of a  unbounded spin system we will consider ahead the so called
$\l\phi^4$ field theory, for which $\Omega~=~ \mathbb{R}$, (i.e. the
spin take values in $\mathbb{R}$) and
$$d\m(\phi_x)~=~{Z}^{-1}
e^{-U(\phi_x)} d\phi_x$$
where  $U(\phi_x)~=~ \l\phi^4$ with $\l>0$
and $Z~=~\int_{ \mathbb{R}}e^{-U(\phi_x)} d\phi_x$. In this case
$\phi_x$ is interpreted as the value of the field at the site $x$.

\vskip.2cm

\\A {\it configuration} $\phi$ of a general spin  system in $\Z^d$ is given once it has been specified
the value of $\phi_x\in\Omega$ for all $x\in \Z^d$. More
precisely, a configuration $\phi$ is a function $\phi: \Z^d\to
\O$. The set of all configuration $\phi$ of the system is the set
$\Phi=\O^{\Z^d}$. So a configuration  $\phi$ is
an element of a very big space $\Phi$ and in general it is
mathematically problematic to define a probability measure on this
space. Even the product measure $\prod_{x\in \Z^d} \m(\phi_x)$ can
be difficult to define  rigorously  depending on the structure of
the space $\O$ (e.g. in case that $\O$ in non compact).

\\Hence, next step towards the definition of a spin system in the lattice it to fix some {\it finite}
subset $\L\subset \Z^d$ and then consider our system of spins
restricted to this set $\L$. So a {configuration} $\phi_\L$ of the
system restricted to  $\L$ is given once it has been specified the
value of $\phi_x\in\Omega$ for all $x\in \L$.
We denote by $\Phi_\L=\prod_{x\in \L}\O$ the set of all spin configurations restricted to  $\L$.
Now the product measure
$d\m(\phi_{\L})=\prod_{x\in\L}d\m(\phi_x)$ in $\Phi_\L$ is
rigorously defined without no particular difficulties. If $X$ is
any subset of  $\L$, then symbols as $\Phi_X$ and $d\m(\phi_{X})$
are defined in the obvious way.

%\\Concluding the definition of a spin system on $\Z^d$,
%we  define  an {\it interaction} between spins $\phi_x$ at
%different sites of $\mathbb{Z}^d$. An interaction is given once,
%for each {\it finite} subset $X\subset \Z^d$, we specify a real
%valued function $\Psi(X, \phi_X)$, called {\it many body
%interaction}, which depends in general on the set $X$ and the
%values of the spins $\phi_x$ at site $x\in X$.

\\The boundary conditions are in general chosen by
fixing a spin configuration $\phi_{\L^c}$ for the spins outside  $\L$ (i.e. in the set $\L^c~=~\VU-\L$). Such choice can be sometimes
limited due to the structure of the  single spin space $\O_x$. But other choices are also possible. E.g. two very popular choices
of the boundary conditions are the so-called {\it free  boundary conditions}  and {\it periodic boundary conditions}.
In particular,  the free boundary condition choice is to say that
``there is no outside", i.e.  simply to study the spin system in the finite volume $\L$ with the identification
$\VU=\L$.   Also the periodic boundary conditions follows the same philosophy. However to be stated $\VU$ has to possesses some additional structure.
For example, if $\VU=\Z^d$, one has imposed boundary  conditions by saying that $\L$ is a torus in  $\Z^d$.

%Finally, the Hamiltonian or interaction between spins, denoted here by $H(\phi_\L, \phi_{\L^c} )$, is
%in general  expressed in terms of a {\it potential}
%$\Psi$, i.e.  an assignment to each nonempty subset $X\subset \VU$  of real valued measurable function
%$\Psi(X,\phi_X)$ definied in $\Omega_{X}$.  Once   $\Psi$ has been given, and
%boundary conditions are give through a fixed
%configuration $\phi_{\L^c}$ in $\VU\backslash \L$,
%the Hamiltonian $H(\phi_\L, \phi_{\L^c} )$ of the discrete
%spin system confined in the volume $\L$ is given  (when boundary conditions are give through a fixed
%configuration $\phi_{\L^c}$ in $\VU\backslash \L$) by
%$$
%H(\phi_\L, \phi_{\L^c} )= \sum_{X: X\cap \La\ne \emptyset\atop|X|\geq 2}\Psi(X,\phi_X)\Eq(pote)
%$$
%In the majority  of the  applications presented in  this book the potential  $\Psi$
%will be defined only for pairs of spins, namely $\Psi(X,\phi)\ne 0$ only for $|X|=2$. In this case $\Psi$ is called a
%{\it pair potential}.

Finally, the Hamiltonian or interaction between spins, denoted here by $H(\phi_\L, \phi_{\L^c} )$, is
in general  expressed in terms of a {\it potential}
$\Psi$, i.e.  an assignment to each nonempty subset $X\subset \VU$  of real valued measurable function
$\Psi(X,\phi_X)$ definied in $\Omega_{X}$.  Once   $\Psi$ has been given, and
boundary conditions are give through a fixed
configuration $\phi_{\L^c}$ in $\VU\backslash \L$,
the Hamiltonian $H(\phi_\L, \phi_{\L^c} )$ of the discrete
spin system confined in the volume $\L$ is given  (when boundary conditions are give through a fixed
configuration $\phi_{\L^c}$ in $\VU\backslash \L$) by
$$
H(\phi_\L, \phi_{\L^c} )= \sum_{X: X\cap \La\ne \emptyset\atop|X|\geq 2}\Psi(X,\phi_X)\Eq(pote)
$$
In the majority  of the  applications presented in  this book the potential  $\Psi$
will be defined only for pairs of spins, namely $\Psi(X,\phi)\ne 0$ only for $|X|=2$. In this case $\Psi$ is called a
{\it pair potential}.

Note that the
Hamiltonian  \equ(pote) can be also written as
$$
V(\phi_\L, \phi_{\L^c} )~=~ H_{\rm bulk}(\phi_\L) + W(\phi_\L,
\phi_{\L^c})  \Eq(bulk)
$$
where
$$
H_{\rm bulk}(\phi_\L)~=~ \sum_{X\subset \L}
\Psi(X,\phi_X)
$$
and
$$
W(\phi_\L, \phi_{\L^c})~=~ \sum_{X: X\cap \L\neq \emptyset\atop X\cap\L^c\neq\emptyset}
\Psi(X,\phi_X)
$$
where evidently the term $H_{\rm bulk}$ represent the ``internal"
energy inside $\L$ of the system, which does not depend on boundary
conditions, while $W(\phi_\L, \phi_{\L^c})$, which is the boundary
condition term, represents the interaction  of the system in
$\L$ with the ``outside". Note that in the
free boundary conditions one sets $W(\phi_\L,\phi_{\L^c})=0$,  which indeed corresponds to suppose that there is no outside.

\noindent
Once we give explicitly $(\Omega ,\m)$, the boundary conditions  and the Hamiltonian,
then the measure $\Pi$ of the lattice system is defined by
$$
d\Pi(\phi_\L)= {d\m(\phi_{\La})
e^{-\b H(\phi_\L, \phi_{\L^c} )}\over \Xi_{\La}(\b)}\Eq(Dp)
$$
where $\b>0$  is a parameterusually interpreted in statistical mechanics as the inverse temperature of the system
and the normalization constant $\Xi_\L(\b)$, given by
$$
Z_{\La}(\b)=\int d\m(\phi_{\La})
e^{-\b H(\phi_\L, \phi_{\L^c} )}
\Eq(partdis)
$$
is usually called the {\it Partition functions} of the system.

\section{Some examples}\label{examples}
1.  The celebrated {\it Ising model} falls easily in the structure defined above
.
Here the set $\VU$ is typically the lattice $\Z^d$ in $d$ dimensions,
and the model is defined on a box $\L\subset \mathbb{Z}^d$.
\def\ss{{\bf \s}}
The role of the spin $\phi_x$, $x\in \L$, is played by a random variable
usually denoted by $\s_x$ which can
take one of the two possible value $\s_x ~=~\pm 1$. A configuration
$\ss_\L\in \{-1,1\}^{|\L|}$ of the system is given once the value of the
spin $\s_x$ for each site $x\in \L$ is specified. Hence $\ss_{\L}$ is a set
$\ss_\L ~=~ (\s_{x_1},\dots,\s_{x_{|\L|}})$ of $|\L|$ numbers $\pm
1$. Equivalently one can say that a configuration $\s_\L$ of the
system is a function $\s_\L: \L\to \{+1,-1\}$.
Observe that the total number of configurations of the system in
$\L$ is $2^{|\L|}$.
\def\dL{\partial \L}

\index{Ising model}
Given a configurations $\ss_\L$ of the system in $\L$,  the {\it energy}
of such configuration is
$$
H_\L(\ss_\L) ~=~ -J\sum_{\{x,y\}\subset \L \atop x,y ~{\rm nearest ~neighbors} }
\s_x\s_y- h\sum_{x\in \L}\s_x+ {\cal B}(\s_\L)\Eq(hamis)
$$
where $h$ is an ``external magnetic field'', $J$ is a real
constant representing the intensity of the nearest neighbors interaction,
and the model is said {\it ferromagnetic} when $J>0$ and {\it antiferromagnetic}
otherwise.
${\cal B}(\s_\L)$ represents the interaction of the
spins inside the box $\L$ with a fixed configuration outside $\L$ (i.e. in $\Z^d
\backslash \L$).

\\The statistical mechanics
of the Ising model is obtained by assigning to any configuration $\s_\L$
a probability to occur $P(\s_\L)$. This probability is given by
$$
P(\s_\L)~=~ { e^{ -\b H_\L(\ss_\L)}\over Z_{\L}(\b,h)}\Eq(prob)
$$
where
$$
Z_{\L}(\b,h) ~=~  \sum_{\ss_\L} e^{ -\b H_\L(\ss_\L)}
$$
is the partition function. The term $e^{ -\b H_\L(\ss_\L)}$ is called
{\it Boltzmann weight} of the configuration $\ss_\L$.

The structure of the probability distribution above shows that if the system
is ferromagnetic the spins tend to align, i.e. the configurations have
higher probability if they have a large number of
nearest neighbors assuming the same direction of the spin,
while in the antiferromagnetic case the configuration
with nearest neighbors of the opposite sign are privileged.

\noindent
In order to point out the ingredients of the general structure of the
discrete systems described above, we will consider for simplicity free boundary conditions,
namely we put ${\cal B}(\s_\L)=0$ in \equ(hamis). Then observe that

$$
Z_{\L}(\b,h) ~=~  \sum_{\ss_\L}\Big[\prod_{x\in \L}e^{ h\s_x} \Big] e^{ -\b J\sum_{\{x,y\}\subset \L \atop x,y ~{\rm near. ~neigh.} }
\s_x\s_y }=
$$

$$
 ~=~  [2\cosh(\b h)]^{|\L|}\Xi_{\L}(\b,h)
$$
where
$$
\Xi_{\L}(\b,h)=\int d\m(\s_\L)e^{ -\b J\sum_{\{x,y\}\subset \L \atop x,y ~{\rm near. ~neigh.} }
\s_x\s_y }\Eq(1.7pi)
$$
It is clear now that giving to the term $J\sum_{\{x,y\}\subset \L \atop x,y ~{\rm near. ~neigh.} }
\s_x\s_y$ the role of the potential $V$
defined in \equ(pote),
we have that the Ising model is a lattice system in which $\O_x=\{1,-1\}$,
the single spin  distribution is the discrete measure given by
$$
\m(\s_x=+1) = {e^{\b Jh}\over \cosh\b h},~~~~~ \m(\s_x=-1) = {e^{-\b Jh}\over \cos
h\b Jh}
$$
and the interactions   $\Psi(X, \phi_X)$ are defined only for $X$ of the form $X
=\{x,y\}$ with
$x$ and $y$ nearest neighbors.

The Ising model, despite its simplicity, has a very rich structure which make it
one of the most interesting model in statistical mechanics. We shall investigate
some of its properties in chapter ??.
\vskip.5cm

\\2. The {\it Potts model} is a generalization of the Ising model.
The system is again defined on a box $\L\subset \mathbb{Z}^d$
which is a square subset of the lattice $\Z^d$ in $d$ dimensions,
but here the spin $\s_x$ varies in the set $\{1,2,\dots, q\}$
where $q$ is an integer. The Hamiltonian is given by

$$
H_{\L}(\sigma_{\L}) =
-J\sum_{\{x,y\}\subset \L \atop x,y ~{\rm nearest ~neighbors} }
\d_{\sigma_x \sigma_y}  \Eq(2.1.0)
$$

where
$\d_{\sigma_x\sigma_y}$ is the Kronecker symbol which is equal to one when $\sigma_x=\sigma_y$ and
zero otherwise.
Again, the {\it coupling}  $J$ is called
{\it ferromagnetic} if $J> 0$ and {\it anti-ferromagnetic} if
$J<0$.

 As in the case of the Ising model the {\it statistical mechanics} of the system
 can be done by
introducing the {\it Boltzmann weight} of a configuration $\sigma_\L$, defined a
s
$\exp\{-\b H_{\L}(\sigma_{\L})\}$. Then the
{\it probability} to find the system in the configuration $\sigma_\L$ is given b
y
$$
{\rm Prob}(\sigma_\L)= {e^{-\b H_{\L}(\sigma_{\L})}\over
Z_{\L}(q,\b)}\Eq(1.2.0)
$$
Again the normalization constant in the denominator
is called the  {\it partition function} and is given by
$$
Z_{\L}(q,\b)=\sum_{\sigma_{\L}\in \{1,2,...,q\}^{|\L|} }e ^{- \b H_{\L}(\sigma_{
\L})} \Eq(2.4b.0)
$$
The case $\b J=-\i$ is the {\it anti-ferromagnetic} and
{\it zero temperature} Potts model with $q$-states. In this case
configurations with non zero probability are only those in
which adjacent spins have different values (or colors) and
$Z_{\L}(q,\b)$ becomes simply the number of all allowed
configurations (proper colorings).

The Potts model is obviously in the class above. The spin $\s_x$ plays the role
of $\phi_x$,
$\O_x=\{1,\dots, q\}$, the single spin distribution
is the uniform distribution on  $\{1,\dots, q\}$
and as above the interaction is defined only for $X$ of the form $X=\{x,y\}$ wit
h
$x$ and $y$ nearest neighbors.

Namely, we can write
$$
Z_{\L}(q,\b)=q^{|\L|} \int d\m(\s_\L) \in e ^{- \b H_{\L}(\sigma_{\L})}
$$
where  $\O_x=\{1,2,\dots,q\}$, and
the single spin  distribution is the discrete uniform measure  on $\O_x=\{1,2,\dots,q\}$ given by
$$
\m(\s_x=i) = {1\over q}~~~~~~~~~~ \forall i\in \{1,2,\dots,q\}
$$
and the interactions   $\Psi(X, \phi_X)$ are defined only for $X$ of the form $X
=\{x,y\}$ with
$x$ and $y$ nearest neighbors.

\vskip.5cm
\\3. Let us now consider the example of {\it lattice gases} interacting through
pair potential. This will be a crucial example in the following sections.

\\In statistical mechanics (see also below, section ??), the gran canonical
partition function of a lattice gas of identical particles
enclosed in a cubic volume $\La\subset\Z^d$ with activity $\l$,
interacting via a pair potential
$V(x,y)$ ($x$ and $y$ sites in $\La$)
is given by
$$
Z_{\La}(\b ,\l)=\sum_{n=0}^{\infty}{\l^n\over n!}\sum_{(x_1,...,x_n)\in\L^n}e^{-\b\sum_{i,j}V(x_i,x_j)}
\Eq(latt)
$$
As far as lattice gas is concerned, it  is natural to assume that at most one particle can sit in each site
$x$,
hence we suppose that  $V(x,x)=+\infty$. In this case \equ(latt) can be written in
 the following way

$$
Z_{\La}(\b ,\l)=\sum_{\s_{\La}\in \{0,1\}^{|\L|}}
\l^{\sum_{x\in\La}\s_x}
e^{-\b\sum_{\{x,y\}\subset\La}\s_x\s_y V(x,y)}
\Eq(lattsigma)$$
where
$\sum_{\s_{\L}\in \{0,1\}^{|\L|}}$ means sum over all configurations
$\{\s_x\}$ with $x\in I$ and
the variables $\s_x$ are such that $\s_x=1$ if the site is occupied
by a particle
and $\s_x=0$ if it is empty.

It is easy to see that this model falls in the class of discrete systems defined
 in
\equ(Dp) and \equ(partdis). Namely we can write
$$
Z_{\La}(\b ,\l)= (1+\l)^{|\L|}\int d\m(\s_\L) e ^{- \b\sum_{\{x,y\}\subset\La}\s_x\s_y V(x,y)}
$$
where  $\O_x=\{0,1\}$,
the single spin  distribution is
$$
\m(\s_x=0) = {1\over 1+\l},~~~~~ \m(\s_x=1) = {\l\over 1+\l}
$$
and the interactions   $\Psi(X, \phi_X)$ are defined only for  $X=\{x,y\}$
and have the form $\Psi(X, \phi_X)=\s_x\s_y V(x,y)$.

\\The model can be easily generalized to a lattice gas with a many body potential
writing
$$
\Xi_{\La}(\b ,\l)=\sum_{\s_{\La}\in \{0,1\}^{|\L|}}
\l^{\sum_{x\in\La}\s_x}
e^{-\b\sum_{X\subset\La}\s_X V(X)}
\Eq(lattsigmamb)$$
where $\s_X=\prod_{x\in X}\s_x$
\vskip.5cm
\\4. An example of discrete system in which the single spin state space is continuous
but compact is the so-called {\it Heisemberg model}. The system has the same structure
of the Ising model, but the spins are unit vector in $\mathbb{R}^d$, i.e. $\s_x\in S_1^d
$,
where $S_1^d$ is the unit sphere in $\mathbb{R}^d$. We will denote with $a\cdot b$ the scalar product
between the two vectors $a$ and $b$.
Hence the Hamiltonian of the system is given by

$$
H_\L(\ss_\L) ~=~ -J\sum_{\{x,y\}\subset \L \atop x,y ~{\rm nearest ~neighbors} }
\s_x\cdot\s_y-
\sum_{x\in \L}h\cdot\s_x
$$
where $h\in \mathbb{R}^d$ is an ``external magnetic field'', $J$
represents again the intensity of the nearest neighbors interaction,
and the model is said {\it ferromagnetic} when $J>0$ and {\it antiferromagnetic}
otherwise.
We are using once again free boundary condition for sake of simplicity.

\\The partition function of the Heisemberg model is
$$
Z_{\L}(\b,h) ~=~  \int_{(S_1^d)^{|\L|}} \prod_{x\in \L}d\s_xe^{ -\b H_\L(\ss_\L)}
$$
The Heisemberg model falls in the class of the discrete systems defined above.
We can indeed write
$$
Z_{\L}(\b,h) ~=~ M^{|\L|} \int d\mu(\s_x)\prod_{x\in \L}d\s_xe^{ -J\sum_{\{x,y\}\subset \L \atop x,y ~{\rm nearest ~neighbors} }
\s_x\cdot\s_y}
$$
where
$$d\mu(\s_x)={d\s_x\ e^{\b h\cdot\s_x}\over M}\qquad M= \int_{S_1^d}d\s_x\ e^{\b
 h\cdot\s_x}$$
and $d\s_x$ denotes the Lebesgue measure on the surface of the sphere $S_1^d$.
\vskip.5cm
\\5. As a final example we introduce here a discrete model in which the single spin state space
in non-compact. This model is very important, both in statistical mechanics and
in field theory.

\\We take as $\VU$  the lattice $\Z^d$
equipped with
the usual Euclidean distance. Let $\L\subset \Z^d$ be a cubic subset of ${\Z}^d$.

\\The spin variable $\phi_x$, called also
{\it field},  takes values in ${\mathbb{R}}$.
A configuration
$\phi_\L$
is a function $\phi:\L\to{\mathbb{R}}$.

\def\dLae{\partial\L_{e}}
\def\dLai{\partial\L_{i}}

\def\ti{\omega}
\\We define the boundary conditions of the system fixing a value of the
field $\phi_x$ for all $x\in \Z^d\setminus \L$
\\We consider the discrete model described by the (Gibbs) measure
$$\m_\L^{\ti}(\cdot)={1\over Z(\L ,\ti)}
\int \prod_{x\in \L}d\phi_x e^{-H(\phi_{\L},\ti)}
(\cdot)\Eq(1.1)$$
where $d\phi_x$ is the Lebesgue measure in $\mathbb{R}$ and
the partition function $Z(\L ,\ti)$ is defined by
$$
Z(\L ,\ti)=
\int \prod_{x\in \L}d\phi_x e^{-H(\phi_{\L},\ti)}
\Eq(1.2)$$
The Hamiltonian of the system is
$$
H(\phi_{\L},\ti)=
\sum_{x\in \L}U(\phi_x)-
\sum_{\{x,y\}\cap\L\ne\emptyset}J_{xy}\phi_x \phi_y
+\sum_{x\in \L}h_x\phi_x=$$
$$=\sum_{x\in\L}\left[U(\phi_x)-
\phi_x\ti_x\right]-\sum_{\{ x,y\}\subset \L}J_{xy}\phi_x\phi_y
\Eq(1.3)$$

\\$U(x)$ is an even polynomial of degree $2k$, $k> 1$,
of the form
$$U(x)=x^{2k}+\sum_{i=0}^{k-1}u_{2i}x^{2i}\Eq(1.4)$$
with $u_{2i}\in{\mathbb{R}}$;
the pair potential \\$J_{xy}$ is such that
$$
\sup_{x\in {\bf Z}^d}\sum_{y\neq x} |J_{xy}|=J<+\infty\Eq(1.4.1);
$$

\\$h_x\in \mathbb{R}$ represents the external field, and in the last line
of \equ(1.3) we defined
$$\ti_x=-h_x+\sum_{y\in\L^{c}}J_{xy}\phi_y\Eq(1.5.1)$$
The boundary fields $\phi_y$ with $y\in \L^c$ must be
chosen in such way that $\sum_{y\in\L^{c}}J_{xy}\phi_y$ is finite
for all $x$.

\\This model is again in the class defined above. We can write
$$
Z(\L ,\ti)= M^{|\L|}
\int \prod_{x\in \L}d\m(\phi_x) e^{-\sum_{\{ x,y\}\subset \L}J_{xy}\phi_x\phi_y}
\Eq(1.2)
$$
The single spin distribution  is in this
 case
dependent form the site $x$:
$$d\m(\phi_x)={d\phi_x e^{-U(\phi_x)-
\phi_x\ti_x}\over M}$$
where
$$M=\int d\phi_x e^{-U(\phi_x)-
\phi_x\ti_x}$$

\\The system we just defined has a lot of very interesting features, and
some of them will be extensively discussed in details in section ??
\def\PP{\mathbb{P}}

\section{Polymer gas representation of $Z_{\L}(\b)$}\label{polyexp}

We  need to recall here some definitions and notations about graphs and more generally hypergraphs.

If $A$ is any finite
set, we denote by $|A|$ the number of elements of $A$.
An  {\it hypergraph}
$g$ is a pair $g=(V,E)$ where $V$ is a finite or countable set and $E$ is a collection
of distinct subsets of $V$ with cardinality greater or equal than 2. The elements $x\in V$ are called vertices
of the hypergraph $g$  and $V$ is the set of vertices of $g$. The elements $e\in E$ are called edges of the hypergraph
and $E$ is the set of edges of $g$.  A graph $g$
such that
$|e|=2$ for all $e\in E$
is called  a graph.
Let $g=(V,E)$ be a hypergraph, let $U\subset V$ and let  $E|_{U}=\{e\in E: e\subset U\}$. Then the graph $g|_{U}=(U,E|_{U})$ is
called {\it  the restriction} of $G$ to $U$.  A hypergraph $g'=(V',E')$ is a subgraph of $g=(V,E)$  if $V'\subset V$, and $E'\subset E|_{V'}$.

An hypergraph  $g=(V_g,E_g)$ is said to be {\it connected}
if for any pair $B, C$ of  subsets of $V$ such that
$B\cup C =V_g$ and $B\cap C =\emptyset$, there is a $e\in E_g$ such
that $e\cap B\neq\emptyset$ and $e\cap C\neq\emptyset$. A connected component of an hypergraph  $g=(V_g,E_g)$
is a a subset $U\subset V$ such that  $g|_{U}$ is connected and for any $U'$ such that $ V\supseteq U'\supset U$, $g|_{U'}$ is not connected. Note that
a connected component of $g=(V_g,E_g)$ can be a singleton. The set of all connected components of $g=(V_g,E_g)$ form a partition
of $V_g$.
We denote by ${\cal G}_{V}$ the set of all  hypergraphs
with vertex set $V$ and by ${G}_{V}$ the set of all  connected hypergraphs
with vertex set $V$.

\vskip.2cm
In order to make things easier we will
consider just free boundary condition. This means that the
partition function of our spin system is given by
$$Z_{\L}(\b)~=~\int d\m(\phi_{\L})
\exp\left[-\b\sum_{X\subset \L} \Psi(X,\phi_X)
\right]\Eq(41.10)$$

\\Let us now use the
{\it Mayer expansion trick} and rewrite our partition function in the following way
$$
Z_{\L}(\b)~=~\int d\m(\phi_{\L}) \prod_{X\subset \L}\exp\left[-\b \Psi(X,\phi_X)\right]~=~$$
$$~=~
\int d\m(\phi_{\L}) \prod_{X\subset
\L}[(e^{-\b\Psi(X,\phi_X)}-1)+1]~=~
$$
Now, posing $U_X = e^{-\b\Psi(X,\phi_X)}-1$ we
see that

$$
\prod_{X\subset
\L}[U_X+1]= \sum_{g\in {\cal G}_{\La}} \prod_{X\in
E_g} U_X
$$
(here we agree that the graph with no edges, i.e. $g$ such that with $E_g=\emptyset$ yields the contribution
1 to the sum above) and hence we get

$$Z_{\L}(\b)~=~
\int d\m(\phi_{\L})\sum_{g\in {\cal G}_{\La}}\prod_{X\in
E_g}(e^{-\b\Psi(X,\phi_X)}-1)
$$
Note again  that the graph
with empty edge set  yields the contribution
1 in equation above, since by definition $\int d\m(\phi_{\L})1~=~1$.

We now can group together connected components
in $\sum_{g\in {\cal G}_{\La}}$ as follows. First note every graph $g\in {\cal G}_{\La}$ determines
a partition of $\L$ into its connected components. Hence
$$
\sum_{g\in {\cal G}_{\La}}\prod_{X\in
g}(e^{-\b\Psi(X,\phi_X)}-1)~=~ \sum_{k=0}^{|\L|}\sum_{\{R_1,\dots R_k\}\atop
{\rm partition~ of~ \L}} \prod_{i ~=~1}^{k}\sum_{g\in {{
G}_{R_i}}}\prod_{X\in
E_g}(e^{-\b\Psi(X,\phi_X)}-1)
$$
where by convention the term corresponding to $k=0$ is set equal to 1 and  $\sum_{g\in {{
G}_{R}}}\prod_{X\in
g}(e^{-\b\Psi(X,\phi_X)}-1)=1$ if $|R|=1$.

Hence we can also write
$$
Z_{\L}(\b)=\int d\m(\phi_{\L})\sum_{k=0}^{|\L|}\sum_{\{R_1,\dots R_k\}: ~R_i\subset \La\atop
R_i\cap R_j~=~\emptyset,~|R_i|\geq 2} \prod_{i ~=~1}^{k}\sum_{g\in {{
G}_{R_i}}}\prod_{X\in
g}(e^{-\b\Psi(X,\phi_X)}-1)
$$
where $\{R_1,\dots R_k\}$ is a non ordered $k$-ple of mutually non
intersecting subsets of $\La$ with cardinality greater than 2.
 Now, recalling that $d\m(\phi_{\L})$ is a product
measure, we get
$$Z_{\L}(\b)~=~
\sum_{k=1}^{|\L|}\sum_{\{R_1,\dots R_k\}: ~R_i\subset \La\atop R_i\cap
R_j~=~\emptyset,~|R_i|\geq 2}\r_\b(R_1)\dots \r_\b(R_k)
$$
with
$$
\r_\b(R)~=~\int\prod_{x\in R}d\m(\phi_x) \sum_{g\in { G}_R}
\prod_{X\in g}(e^{-\b\Psi(X,\phi_X)}-1)
$$
and evidently
$$
\sum_{k=0}^{|\L|}\sum_{\{R_1,\dots R_k\}: ~R_i\subset \La\atop R_i\cap
R_j~=~\emptyset,~|R_i|\geq 2} ~~\r_\b(R_1)\dots \r_\b(R_k) =
~~~~\sum_{k\geq 0}{1\over k!} \sum_{(R_{1},\dots
,R_{k}): ~R_i\subset \La\atop R_i\cap R_j=\emptyset,~ |R_{i}|\ge 2}~~~~
\r_\b(R_1)\dots\r_\b(R_k)
$$
where in the l.h.s we are summing over {\it unordered}  $k$-uples $\{R_1,\dots R_k\}$ while in the r.h.s we are summing over
{\it ordered}  $k$-uples $(R_1,\dots R_k)$.

\index{partition function!polymer gas representation}
\\So we have shown that the free boundary condition
partition function $Z_{\L}(\b)$ defined in \equ(41.10) can always
be written as
$$Z_{\L}(\b)~=~\sum_{k\geq 0}{1\over k!}
\sum_{(R_{1},\dots
,R_{k}): ~R_i\subset \La\atop R_i\cap R_j~=~\emptyset ,
~|R_{i}|\ge 2} \r_\b(R_1)\dots\r_\b(R_k) \Eq(43.1)$$

It is now easy to convince one-self that expression above is actually
the {\it grand canonical} partition function of a
gas of particle of different species in which the particles (or ``polymers") $R$  are {\it subsets of
$\La$} with cardinality greater or equal 2, with activity $\r(R)$. Note that the activity
is not the same for every polymer, i.e.  different polymers have in principle
different activities. These polymers  are interacting via a two-body
interaction which prevent the overlapping between different polymers. Indeed, the
only interaction between polymers  is a {\it repulsive
hard core} pair interaction. Namely, if we define for any pair of
polymers $R_i ,R_j$
$$
U(R_i ,R_j)~=~\begin{cases} 0 &if R_i\cap R_j ~=~\emptyset\\ +\infty &if
R_i\cap R_j \neq \emptyset\end{cases} \Eq(4.1b)$$ we have proved the following elementary but fundamental theorem

\begin{teo}\label{polyexp}
The partition function  $Z_{\L}(\b)$ defined in \equ(41.10) can be written  as
$$
Z_{\L}(\b)~=~\sum_{k\geq 0}{1\over k!} \!\sum_{R_{1},\dots
,R_{k}\atop |R_{i}|\ge 2} \!\prod_{j=1}^k\r_\b(R_j)e^{-\sum_
{1\leq i<j\leq n}U(R_i,R_j)} \Eq(4.2z)
$$
where
$$
\r_\b(R)~=~\int\prod_{x\in R}d\m(\phi_x) \sum_{g\in { G}_R}
\prod_{X\in g}(e^{-\b\Psi(X,\phi_X)}-1)\Eq(3.3)
$$
\end{teo}
{In conclusion theorem \ref{polyexp} states that the partition function of discrete spin systems  can be rewritten
as the grand canonical partition function of  a polymer gas in which the polymers $R$ are finite subsets
of $\VU$ with cardinality greater or equal than 2 and a polymer $R$ has activity $\r(R)$. Such way to rewrite the partition
function of a discrete system is called {\it the high temperature polymer gas expansion of the system}}.

\\Note that all the information about the original system is now contained in the structure of the
activity $\r(R)$ of a given polymer, given by \equ(3.3).

Note also that $\r_\b(R)\to 0$ as $\b\to 0$. Hence the activity of
polymers is small when $\b$ is small. So it is reasonable to
expect the convergence of this polymer gas when the parameter
$\b$ is sufficiently small. This is the reason why such an
expansion is called  ``high temperature expansion".

We want to outline here that we meet again in the structure of polymer activity \equ(3.3), the combinatorial
problems we have seen in section \ref{combpro}. Here such problems are even worst because the  factors
$(e^{-\b\Psi(X,\phi_X)}-1)$ in \equ(3.3) when dealing with many-body interactions and hypergraphs.
As far as discrete  spin systems with purely pair interactions are concerned,
we will get rid of this challenging
combinatorial problem  by means of the previously seen tree graph identities.

Many body interaction may also be treated, with a completely different technique, as we will see in the  section \ref{cassoli}.

\subsection{General spin systems interacting via pair potentials}
General many body
interactions are usually quite difficult be treated, in particular
in the case in which the single spin probability space $\O$ is non
compact. Hereafter we will consider only {\it pair interaction},
i.e. we will assume that $\Psi(X, \phi_X)~=~0$ whenever $|X|\neq 2$.
Hence our interactions will be always pair interactions of the
form $\Psi(\{x, y\}, \phi_x, \phi_y)$ with $x$ and $y$ varying in
$\Z^d$. This function $\Psi(x, y, \phi_x, \phi_y)$ is supposed to
satisfies some properties. In particular we will demand that there
exist a positive symmetric function $J(\{x,y\})$ (i.e. $J_{yx}=J_{xy}$) such that
$$
|\Psi(\{x, y\}, \phi_x, \phi_y)|\le J_{xy}
\|\phi_x\|\|\phi_y\|\Eq(1ai)
$$
and
$$
\sup_{x\in \Z^d} \sum_{y\in \Z^d\atop y\neq x}
J_{xy}~\doteq~J<\infty\Eq(1biu)
$$
It is crucial to point out that conditions \equ(1ai) and \equ(1biu) automatically implies that the potential
$\Psi(\{x, y\}, \phi_x, \phi_y)$ is stable in the following sense. Given any finite set $R\subset \Z^d$ (with cardinality $|R|\ge 2$),  it holds that
$$
\sum_{\{x ,y \}\subset R}\Psi(\{x ,y \},\phi_{x} ,\phi_{y})
\ge -{J\over 2}\sum_{x\in
R}\|\phi_{x}\|^2\Eq(stabilityb)
$$
Indeed,
$$
\begin{aligned}
\sum_{\{x,y\}\subset R}\Psi(\{x,y\},\phi_{x},\phi_{y}) & \geq
-  \sum_{\{x,y\}\subset R}|J_{xy}\|\phi_{x}\|\|\phi_{y}\|\\
& \ge~
-\sum_{\{x,y\}\subset R}J_{xy}{1\over 2}
\Big[\|\phi_{x}\|^2+\|\phi_{y}\|^2\Big]\\
& = - \sum_{\{x,y\}\subset R}J_{xy}
\|\phi_{x}\|^2\\
& =
{1\over 2}\sum_{x\in R} \|\phi_{x}\|^2\sum_{y\in R\atop y\neq x} J_{xy}\\
& \ge -{J\over 2}\sum_{x\in
R}\|\phi_{x}\|^2
\end{aligned}
$$

\\Once a  function $\Psi$ with the properties \equ(1ai) and \equ(1biu) above has been given,  we can
construct the {\it Hamiltonian} of the system. To do this we first
fix a spin configuration $\phi_{\L^c}$ for the spins outside the
box $\L$ (i.e. in the set $\L^c~=~\Z^d-\L$). This correspond to
choose the {\it boundary conditions} for the system restricted to
$\L$. Such choice is totally free in the case of bounded spin
system, while, for consistency reasons, we will see in a moment
that there are limitations for such choice in the case of
unbounded spin systems.

\\Once the boundary condition $\phi_{\L^c}$ has been fixed outside
$\L$, then, for any configuration $\phi_\L$ inside $\L$, the
Hamiltonian of the system is given by \index{Hamiltonian}

$$
H^{\phi_{\L^c}}(\phi_\L)~=~ \sum_{\{x,y\}\cap\L\neq\emptyset}
\Psi(\{x,y\},\phi_x,\phi_y)\Eq(hamilt)
$$
The number $H(\phi_\L)$ represents  the energy associated to the
configuration $\phi_\L$ of the spin system.
In order to make things easier we will consider hereafter just free boundary condition (no influence of the exterior on the system confined in the volume $\L$). Namely we will suppose that our Hamiltonian is

$$
H(\phi_\L)~=~ \sum_{\{x,y\}\subset \L}
\Psi(\{x,y\},\phi_x,\phi_y)\Eq(hamfre)
$$

\\Now that we have defined the spin systems (and its single spin a priori distribution), the
space of configurations, the Hamiltonian (i.e. the energy
associated to each system configuration) and the boundary
conditions we are ready to study the statitic/thermodynamic
behavior of the system. The last step is to define a probability
associate to each configuration $\phi_\L$ of the system in $\L$.
So we define the probability to find the system in  a given
configuration $\{\phi_{x}\}_{x\in \L}$ of spins, with energy
$H(\phi_\L)$ as
$$
{\rm Prob}(\phi_\L)~=~{e^{-\b H(\phi_\L) }d\m(\phi_{\L})\over
\int d\m(\phi_{\L}) e^{-\b H}}\Eq(1.0)
$$
the parameter  $\b^{-1}$ performs the link with thermodynamics and
it is interpreted as the temperature of the system. The constant
$Z_{\L}(\b)~=~\int d\m(\phi_{\L}) e^{-\b H}$ in \equ(1.0) is the
{\it partition function} of the system \noindent

$$Z_{\L}(\b)~=~\int d\m(\phi_{\L})
\exp\left[-\b\sum_{\{x,y\}\subset\L}
\Psi(\{x,y\},\phi_x,\phi_y) \right]\Eq(41.1)$$
\index{partition function}

\\The partition function \equ(41.1) is the key quantity in order
to study thermodynamic properties of the system. E.g. the
thermodynamic {\it free energy} of the system is defined as
$$
f(\b)~=~\lim_{\L\to\infty}{1\over |\L|}\log Z_{\L}(\b)\Eq(enlib)
$$
where of course the limit $\L\to \infty$ has to be interpreted in
the sense of  Fisher or Van Hove or even in a more restrictive
fashion, e.g. cubic boxes with size $N\to \infty$.

\\In particular, the question  about  the {\it analyticity} of
$f(\b)$ as a function of $\b$ is very important. Accordingly to
the physics belief: {\it No phase transition can occur in the
system wherever  $f(\b)$  analytic}. Concerning this question, we
could say that there is a sort of ``meta-theorem" stating that
any reasonable spin system has free energy analytic if $\b$ is
sufficiently small. Note that $\b$ sufficiently small means
substantially that the system is weakly interacting or in other
words spin are nearly independent.

\\Now we will now show that $\log Z_{\L}(\b)$ can be written always  in term of
a formal series expansion which  converges absolutely if $\b$ is
sufficiently small. Such an expansion is called {\it the high
temperature expansion} of the lattice system. We will also study
directly (via cluster expansion methods)  this series expansion in
the next section.

\subsection{Polymer representation of $Z_{\L}(\b)$ for the pair potential case}
Again, in order to make things easier we will
consider just free boundary condition. This means that the
partition function of our spin system is given by
$$Z_{\L}(\b)~=~\int d\m(\phi_{\L})
\exp\left[-\b\sum_{\{x,y\}\subset \L} \Psi(\{x,y\},\phi_x,\phi_y)
\right]\Eq(41.10)$$

\\Let us now use the
``Mayer expansion trick" and rewrite our partition function in the following way
$$
Z_{\L}(\b)~=~\int d\m(\phi_{\L}) \exp\left[-\b\sum_{\{x,y\}\subset
\L} \Psi(\{x,y\},\phi_x,\phi_y)\right]~=~$$
$$~=~
\int d\m(\phi_{\L}) \prod_{\{x,y\}\subset
\L}[(e^{-\b\Psi(\{x,y\},\phi_x,\phi_y)})-1)+1]~=~
$$
$$~=~1+
\int d\m(\phi_{\L})\sum_{g\in {\cal G}_{\La}}\prod_{\{x,y\}\in
E_g}(e^{-\b\Psi(\{x,y\},\phi_x,\phi_y)}-1)~=~
$$
where recall that ${\cal G}_{\La}$ denotes the set of all graphs
$g$ with vertex set $\L$ (either connected or not connected). Note also that the graph
with empty edge set has been separated and yields the contribution
1 in equation above, since by definition $\int d\m(\phi_{\L})1~=~1$.
As we did in Section \ref{mayersec}, we can group together connected components
in $\sum_{g\in {\cal G}_{\La}}$ namely
$$
\sum_{g\in {\cal G}_{\La}}~=~ \sum_{\{R_1,\dots R_k\}\in \La\atop
R_i\cap R_j~=~\emptyset,~|R_i|\geq 2} \prod_{i ~=~1}^{k}\sum_{g\in {{
G}_{R_i}}}
$$
where $\{R_1,\dots R_k\}$ is a non ordered $k$-ple of mutually non
intersecting subsets of $\La$ with cardinality greater than 2.
Note in fact that the sets with cardinality 1 yields contribution
equal to 1. Now, recalling that $d\m(\phi_{\L})$ is a product
measure, we get
$$\int d\m(\phi_{\L})\sum_{g\in {\cal G}_{\La}}\prod_{X\in g}(e^{-\Psi(\{x,y\},\phi_x,\phi_y)}-1)~=~
\sum_{\{R_1,\dots R_k\}\in \La\atop R_i\cap
R_j~=~\emptyset,~|R_i|\geq 2}\r_\b(R_1)\dots \r_\b(R_k)
$$
with
$$
\r_\b(R)~=~\int\prod_{x\in R}d\m(\phi_x) \sum_{g\in { G}_R}
\prod_{X\in g}(e^{-\b\Psi(\{x,y\},\phi_x,\phi_y)}-1)
$$
and evidently
$$
\sum_{{\{R_1,\dots R_k\}\in \La\atop R_i\cap
R_j=\emptyset}\atop|R_i|\geq 2} ~~\r_\b(R_1)\dots \r_\b(R_k) =
~~~~\sum_{k\geq 1}{1\over k!} \sum_{{R_{1},\dots
,R_{k}\subset\L\atop R_i\cap R_j=\emptyset}\atop |R_{i}|\ge 2}~~~~
\r_\b(R_1)\dots\r_\b(R_k)
$$
\index{partition function!polymer gas representation}
Namely, if we define for any pair of
sets $R_i ,R_j$
$$
V(R_i ,R_j)~=~\begin{cases}0 & {\rm if}~ R_i\cap R_j ~=~\emptyset\\
 +\infty & {\rm if}~
R_i\cap R_j \neq \emptyset \end{cases}
\Eq(4.1b)$$
We can rewrite the
grand-canonical partition function \equ(3.1) as
$$
Z_{\L}(\b)~=~1+\sum_{k\geq 1}{1\over k!} \!\sum_{R_{1},\dots
,R_{k}\atop |R_{i}|\ge 2} \!\r_\b(R_1)\dots\r_\b(R_k) e^{-\sum_{1\leq i<j\leq n}V(R_i,R_j)} \Eq(4.2z)
$$
where
$$
\r_\b(R)~=~\int\prod_{x\in R}d\m(\phi_x) \sum_{g\in { G}_R}
\prod_{X\in g}(e^{-\b\Psi(\{x,y\},\phi_x,\phi_y)}-1)\Eq(3.3)
$$
So we have shown that the
partition function $Z_{\L}(\b)$ defined in \equ(41.10) can always
be written as the {\it grand canonical} partition function of a
gas of ``polymers'' in which polymers  are {\it subsets $R$  of
$\La$} with cardinality greater or equal 2, with activity $\r_\b(R)$ given by \equ(3.3),
and with pair  interaction  given by \equ(4.1b), i.e.,  they cannot overlap.

\\Threfore,
according to  Theorem \ref{teo42}, whenever  the condition
$$
\sup_{x\in Z^d}\sum\limits_{\g\in \PP\atop x\in R}
|\r_\b(R)| ~  e^{a|R|}\le e^a-1\Eq(fpsetsp)
$$is satisfied,
the free energy is analytic in $\b$.

\\Note that the information about the original system is contained in the structure of the
activity $\r_\b(R)$ of a given polymer $R$. Note also that $\r_\b(R)\to 0$ as $\b\to 0$. Hence the activity of
polymers is small when $\b$ is small. So it is reasonable to
expect for  convergence of this polymer gas when the parameter
$\b$ is sufficiently small. This is the reason why such an
expansion is called the high temperature expansion of the spin
system.

\subsection{The bounded spin system case}
Let us  suppose that our spin sysytem is bounded, namely  it holds that
$$
\|\phi_x\|\le c<+\infty \Eq(bousp)
$$
then we can state the following theorem.

\begin{teo}\label{spinbound}
Consider  a spin system at inverse temperature $\b$
with  partition function $Z_{\La}(\b)$ given by \equ(41.1), interacting via a pair potential satisfying
\equ(1ai) and \equ(1biu). Suppose that \equ(bousp) holds  and let $\b$ such that
$$
\b\le { 0,058\over c^2 J}
$$
Then
 $|\La|^{-1}\log Z_{\La}(\b)$ is an
absolutely convergent series uniformly in $\La$.
\end{teo}

\\{\bf Remark.} \equ(1ai) and \equ(1biu) are satisfied for pratically all
studied  bounded spin system (e.g. Ising-type models, Potts model,
Heisenberg model).

\vskip.2cm
\\{\bf Proof}. All we need to do in order to prove the theorem is to check
condition \equ(fpsetsp).  First note that \equ(fpsetsp) is satisfied if
$$
\sum_{n=2}^\infty e^{an} \sup_{x\in \Z^d}\sum\limits_{R: \in \PP \atop |R|=n, \;x\in R}
|\r_\b(R)|\le e^a-1\Eq(fpsetsp2)
$$For that we need need an upper bound on the modulous of the  activity
$\r_\b(R)$ given in \equ(3.3) when $|R|=n$. To bound $|\r_\b(R)|$  we use Theorem \ref{teoPYgen} and in particular the tree graph inequality
\equ(bteo1). Supposing that $R$ with cardinality $|R|=n$ is the unordered set $\{x_i,\dots,x_n\}\subset \Z^d$ and setting, for $\{x_i,x_j\}\subset R$,
$$
\Psi(\{x_i,x_y\},\phi_{x_i},\phi_{x_j})= \Psi_{ij}
$$
we have
$$
\begin{aligned}
|\r_\b(R)| & =  \Bigg|\int\prod_{x\in R}d\m(\phi_x) \sum_{g\in { G}_R}
\prod_{X\in g}(e^{-\b\Psi(\{x,y\},\phi_x,\phi_y)}-1)\Bigg|\\
& =  \bigg|\int \prod_{i=1}^n d\m(\phi_{x_i}) \sum_{g\in G_n} \prod_{\{i,j\}\in E_g}(e^{-\b\Psi_{ij}}-1)\bigg|\\
&\le  \int \prod_{i=1}^n d\m(\phi_{x_i}) \bigg|\sum_{g\in G_n} \prod_{\{i,j\}\in E_g}(e^{-\b\Psi_{ij}}-1)\bigg|
\end{aligned}
$$
Now observe that in force of \equ(stabilityb) and \equ(bousp) we have that
$$
\sum_{\{x,y\}\subset R}\Psi(\{x,y\},\phi_{x},\phi_{y})\ge  -{c^2J\over 2} |R|
$$
i.e., if we pose $R=\{x_i,\dots, x_n\}$ and $\Psi(\{x_i,x_y\},\phi_{x_i},\phi_{x_j})= \Psi_{ij} $,  we have that the real numbers $\{\Psi_{ij}\}_{\{i,j\}\in E_n}$ are such that
$$
|\Psi_{ij}|\le   c^2J_{x_ix_j}
$$
and
$$
\sum_{1\le i<j\le n}\Psi_{ij}\ge  -{c^2J\over 2} n
$$
Therefore, according to Theorem \ref{teoPYgen}
$$
\bigg|\sum_{\t\in T_n} \prod_{\{i,j\}\in E_g}(e^{-\b\Psi_{ij}}-1)\bigg|\le e^{{c^2\b J\over 2} n}
\sum_{\t\in T_n}\prod_{\{i,j\}\in E_\t}(1-e^{-\b |\Psi_{ij}|})
$$
Therefore, since $\int d\m(\phi_{x_i}) =1$, we have
$$
\begin{aligned}
|\r_\b(R)| & \le e^{{c^2\b J\over 2} n}\sum_{\t\in T_n} \prod_{\{i,j\}\in E_\t}(1-e^{-\b  c^2J_{x_ix_j}})\\
& \le  e^{{c^2\b J\over 2} n}(\b  c^2)^{n-1}\sum_{\t\in T_n} \prod_{\{i,j\}\in E_\t}J_{x_ix_j}
\end{aligned}
$$
and thus
$$
\begin{aligned}
\sup_{x\in \Z^d}\sum\limits_{R: \in \PP \atop |R|=n, \;x\in R}
|\r_\b(R)| & \le  e^{{c^2\b J\over 2} n}(\b  c^2)^{n-1} \sup_{x\in \Z^d}\sum\limits_{\{x_1, \dots, x_n\}\subset \Z^d\atop  x_1=x}
\sum_{\t\in T_n} \prod_{\{i,j\}\in E_\t}J_{x_ix_j}\\
& =
e^{{c^2\b J\over 2} n}(\b  c^2)^{n-1} \sup_{x\in \Z^d}{1\over (n-1)!}\sum_{(x_i,\dots, x_n)\in \Z^{dn}\atop x_1=x,\,x_i\neq x_j}
\sum_{\t\in T_n} \prod_{\{i,j\}\in E_\t}J_{x_ix_j}\\
&\le e^{{c^2\b J\over 2} n}{(\b  c^2)^{n-1} \over (n-1)!}\sum_{\t\in T_n}\sup_{x\in \Z^d}\sum_{(x_i,\dots, x_n)\in \Z^{dn}\atop x_1=x,\,x_i\neq x_j}\prod_{\{i,j\}\in E_\t}J_{x_ix_j}
\end{aligned}
$$
and since (recall  Proposition \ref{integr} and \equ(1biu)), for any $\t\in T_n$
$$
\sup_{x\in \Z^d}\sum_{(x_i,\dots, x_n)\in \Z^{dn}\atop x_1=x,\,x_i\neq x_j}\prod_{\{i,j\}\in E_\t}J_{x_ix_j}\le J^{n-1}\Eq(intbiu)
$$
we get in the end that (recall: $\sum_{\t\in T_n}1=n^{n-2}$)
$$
\sup_{x\in \Z^d}\sum\limits_{R: \in \PP \atop |R|=n, \;x\in R}
|\r_\b(R)| \le e^{{c^2\b J\over 2} n} (c^2\b   J)^{n-1}{n^{n-2} \over (n-1)!}
$$
Therefore condition \equ(fpsetsp2) is satisfied if (as we did earlier, we choose for simplicity $a=\ln 2$ and we bound ${n^{n-2} \over (n-1)!}\le e^{n-1}$)
$$
\sum_{n=2}^{\infty}\Big(2e^{{c^2\b J\over 2}}\Big)^n (c^2e\b   J)^{n-1}\le 1
$$
i.e., setting $x=c^2\b J$
$$
\sum_{n=2}^{\infty}\Big(2e^{{x\over 2}}\Big)^n (ex)^{n-1}\le 1
$$
i.e.
$$
2e^{{x\over 2}}\sum_{n=2}^{\infty} \Big(2e^{{x\over 2}+1}x\Big)^{n-1}\le 1
$$
i.e
$$
 {4e^{{x}+1}x\over 1- 2e^{{x\over 2}+1}x}\le 1
$$
which is satisfied if
$x \le 0,058$.

\\In conclusion \equ(fpsetsp2) is satisfied for all $\b$ such that
$$
\b J\le { 0,058\over c^2}
$$
which ends the proof of Theorem \ref{spinbound}. $\Box$

\subsection{The unbounded spin system case}
Concerning unbounded spin system we will prove the following theorem
\vskip.2cm
\begin{teo}\label{spinunb}
Consider  a spin system at inverse temperature $\b$
with  partition function $Z_{\La}(\b)$ given by \equ(41.1), interacting via a pair potential satisfying
\equ(1ai) and \equ(1biu). Suppose that
the norm $\|\cdot \|: \Omega \to [0,+\infty)$ is a function such that, for $n\in \mathbb{N}$
$$
\int e^{+{\b J\over 2}\|\phi\|^2}\|\phi\|^{n}d\m(\phi)\le n! [C(\b J)]^n
\Eq(condimes)
$$
and let $\b$ such that
$$
\b  J C^2(\b J)\le {e-1\over 4e^2}
\Eq(pippo)
$$
Then  $|\La|^{-1}\log
Z_{\La}(\b)$ is an absolutely convergent series uniformly in
$\La$.
\end{teo}

\\{\bf Proof}.

\\Again, we just need to check the condition \equ(fpsetsp). Thus, given $R\subset \Z^d$ such that $|R|\ge 2$,  let us first estimate
the modulus of the polymer activity is given by \equ(3.3). We have as before
$$
|\r_\b(R)|~\le ~\int\prod_{x\in R}d\m(\phi_x) \Big|\sum_{g\in
G_R}\prod_{\{x,y\}\in g}(e^{-\b\Psi(\{x,y\},\phi_{x},\phi_{y})}-1)\Big|
$$

\\As before let  $R$ such that  $|R|=n$~; so $R$ is the unordered set $\{x_i,\dots,x_n\}\subset \Z^d$ and let set, for $\{x_i,x_j\}\subset R$,
$$
\Psi(\{x_i,x_y\},\phi_{x_i},\phi_{x_j})= \Psi_{ij}
$$
Thus, as before, we have
$$
|\r_\b(R)| \le  \int \prod_{i=1}^n d\m(\phi_{x_i}) \bigg|\sum_{g\in G_n} \prod_{\{i,j\}\in E_g}(e^{-\b\Psi_{ij}}-1)\bigg|
$$
Now, by  \equ(1ai) and \equ(stabilityb), with
 $R=\{x_1,\dots, x_n\}$ and $\Psi(\{x_i,x_y\},\phi_{x_i},\phi_{x_j})= \Psi_{ij} $,  we have that the real numbers $\{\Psi_{ij}\}_{\{i,j\}\in E_n}$ are such that
$$
|\Psi_{ij}|\le   J_{x_ix_j}\|\phi_{x_i}\|\|\phi_{x_j}\|
$$
and, for any $S\subset [n]$,
$$
\sum_{\{i,j\}\subset S}\Psi_{ij}\ge  -{J\over 2} \sum_{i\in S}\|\phi_{x_i}\|^2
$$
Therefore, according to Theorem \ref{teoPYgen}
$$
\bigg|\sum_{g\in G_n} \prod_{\{i,j\}\in E_g}(e^{-\b\Psi_{ij}}-1)\bigg|\le e^{{\b J\over 2}  \sum_{i=1}^n\|\phi_{x_i}\|^2}
\sum_{\t\in T_n}\prod_{\{i,j\}\in E_\t}(1-e^{-\b  J_{x_ix_j}\|\phi_{x_i}\|\|\phi_{x_j}\|})
$$
Now we get
$$
\begin{aligned}
|\r_\b(R)| &  \le \int \prod_{i=1}^n d\m(\phi_{x_i})e^{{\b J\over 2}\|\phi_{x_i}\|^2}
\sum_{\t\in T_n} \prod_{\{i,j\}\in E_\t}(1-e^{-\b  J_{x_ix_j}\|\phi_{x_i}\|\|\phi_{x_j}\|})\\
& \le  \b^{n-1}\sum_{\t\in T_n} \int \prod_{i=1}^n d\m(\phi_{x_i})e^{{\b J\over 2}\|\phi_{x_i}\|^2} \prod_{\{i,j\}\in E_\t}J_{x_ix_j}\|\phi_{x_i}\|\|\phi_{x_j}\|
\end{aligned}
$$
Then observe that if $\t\in T_n$ has degrees $d_1, \dots d_n$ at  vertices $1,2,\dots,n$ respectively, then
$$
\prod_{\{i,j\}\in E_\t}\|\phi_{x_i}\|\|\phi_{x_j}\|= \prod_{i=1}^n \|\phi_{i}\|^{d_i}
$$
and thus
$$
|\r_\b(R)|   \le b^{n-1} \!\!\!\!\sum_{d_1, \dots, d_n:\; d_i\ge 1\atop d_1+\dots +d_n=2n-2}
 \prod_{i=1}^n \int d\m(\phi_{x_i})e^{{\b J\over 2}\|\phi_{x_i}\|^2} \|\phi_{x_i}\|^{d_i}
\!\!\!\!\sum_{\t\in T_n\atop d_1, \dots, d_n~{\rm fixed}}
\prod_{\{i,j\}\in E_\t}J_{x_ix_j}
$$
whence, using again \equ(intbiu),  and recalling also Cayley formula \equ(cay1), we get
$$
\begin{aligned}
\sup_{x\in \Z^d}\sum\limits_{R: \in \PP \atop |R|=n, \;x\in R}
|\r_\b(R)| & \le   {(\b  J)^{n-1}\over n-1}\sum_{d_1, \dots, d_n:\; d_i\ge 1\atop d_1+\dots +d_n=2n-2}
\prod_{i=1}^n \int d\m(\phi_{x_i}) {\|\phi_{x_i}\|^{d_i}e^{{\b J\over 2}\|\phi_{x_i}\|^2}\over (d_1-1)!}\\
& \le {(\b  J)^{n-1}\over n-1}\sum_{d_1, \dots, d_n:\; d_i\ge 1\atop d_1+\dots +d_n=2n-2}
\prod_{i=1}^n d_iC(\b J)^{d_i}\\
&\le {(\b  J)^{n-1}\over n-1}[C(\b J)]^{2n-2}\sum_{d_1, \dots, d_n:\; d_i\ge 1\atop d_1+\dots +d_n=2n-2}
\prod_{i=1}^n d_i
\end{aligned}
$$
where in the second inequality   we have used hypothesis \equ(condimes). Observe now that
$$
\begin{aligned}
\sum_{d_1, \dots, d_n:\; d_i\ge 1\atop d_1+\dots +d_n=2n-2}
\prod_{i=1}^n d_i & = \sum_{s_1, \dots, s_n:\; s_i\ge 0\atop s_1+\dots +s_n=n-2}
\prod_{i=1}^n (1+s_i)\\
& \le  \sum_{s_1, \dots, s_n:\; s_i\ge 0\atop s_1+\dots +s_n=n-2}\prod_{i=1}^n e^{s_i}\\
& =e^{n-2} \sum_{d_1, \dots, d_n:\; d_i\ge 1\atop d_1+\dots +d_n=2n-2}1\\
&= e^{n-2}{2n-3\choose n-1}\\
& \le (4e)^{n-1}
\end{aligned}
$$
Therefore we get
$$
\sup_{x\in \Z^d}\sum\limits_{R: \in \PP \atop |R|=n, \;x\in R}
|\r_\b(R)|\le {(4e\b  J)^{n-1}\over n-1}[C(\b J)]^{2n-2}
$$
In conclusion we have that condition  \equ(fpsetsp) is satisfied if  (as we did earlier, we choose for simplicity $a=\ln 2$)
$$
\sum_{n=2}^\infty {(4e\b  J C^2(\b J))^{n-1}\over n-1}\le 1
$$
i.e. if
$$
-\ln[1-4e\b  J C^2(\b J)]\le 1
$$
i.e. if
$$
\b  J C^2(\b J)\le {e-1\over 4e^2}
$$
This concludes the proof of Theorem \ref{spinunb}. $\Box$

\section[Many-body interactions]{Discrete Spin systems with many-body interactions}\label{cassoli}

We have found a general condition, the equation
\equ(gkindu), that guarantees the absolute convergence uniformly in $\L$
of $\Sigma^\L_x$, and hence of $\log\Xi_\L(\r)$.
The ways we found so far to check that the polymer activities
satisfy \equ(gkindu) may be used only for pair interaction.
We have now to face a new problem: starting from
the structure of the polymer activity $\r_\g$
in the case of general many body interactions, which is given by
$$
\r_\g~=~\int\prod_{x\in \g}d\m(\phi_x) \sum_{g\in { G}_\g}
\prod_{X\in g}(e^{-\b\Psi(X,\phi_X)}-1)\Eq(robeta1)
$$
we want to give general criteria in order to satisfy the condition \equ(gkindu),
that in this case take the form
$$
\sum\limits_{\g\in \PP_\VU\atop x\in \g}
|\r_{\g}| ~  e^{a|\g|}\le e^a-1\Eq(1.6.0)
$$

\\Here we want to present a simple technical device,
introduced in 1981 by M. Cassandro and E. Olivieri,
which exploits the fact that for bounded spin systems on the
lattice the sum over the Mayer graphs can be directly bounded,
for $N$ body potentials
absolutely summable in an appropriate sense.
This method has two main advantages: it is relatively simple,
since it avoids the combinatorial complexity arising in the tree
identities presented above,
and it may be applied to a large class of discrete systems.
In particular, for systems with many-body interactions it is difficult
to obtain better estimates of the interactions $\b\Psi(X,\phi_X)$ for which
\equ(1.6.0) is satisfied.

\\The method can be also applied to systems with two body interactions but
the results are obviously far from optimal, as we shall see below.
\vglue.5truecm

\noindent
\begin{teo}[Cassandro, Olivieri]\label{hier}
Consider the polymer activity \equ(robeta1).
If it is possible to give for any connected graph
$g$ such that ${\rm supp}\ g=n$ the following bound

$$
\int\prod_{x\in {\rm supp}\ g}d\m(\phi_x)
\prod_{X\in g}\left|e^{-\b\Psi(X,\phi)}-1\right|\le{\bar\s}^{n}
\prod_{X\in g}W(X)\Eq(1.6.1)
$$
with
$$
\sum_{X\ni x} W(X)\le K,\qquad\qquad{\bar\s}<1\Eq(1.6.2)
$$
and the following relation between ${\bar\s}$ and $K$  holds

$$K< {|\log\bar\s|\over 4}\Eq(1.6.3)$$
then the condition \equ(1.6.0) is satisfied for some $a>0$.
\end{teo}
\vskip.2cm
\noindent
%{\it Remark}. If $J$ is not translational invariant then the
%factor $1/4$ in r.h.s. of (\ref{condfinal}) has to  be %replaced
%by $1/16$
\vglue.3truecm

\noindent
{\bf Remark}

\\It is important, in order to be able to apply this theorem to concrete models,
to understand the structure of the estimate needed in order to satisfy  \equ(1.6.0).
Given a graph, the contribution of that graph to the polymer activity have to be
bounded by the product of a {\it vertex contribution} $\bar\s<1$, meaning that
the bigger ${\rm supp}\ g$ the smaller its contribution, and an {\it edge contribution}
$J(X)$, whose estimate has to be suitably related to $\bar\s$.
\vglue.3truecm

\noindent
{\bf Proof}.

\noindent
We will bound directly l.h.s. of \equ(1.6.0)
$$
\sum_{\g\ni x}
|\r_\g|e^{a|\g|}=\sum_{\g\ni x}e^{a|\g|}
\left|
\int\prod_{x\in \g}d\m(\phi_x)
\sum_{g\in {\cal G}_\g}
\prod_{X\in g}[e^{-\b\Psi(X,\phi)}-1]
\right|\leq
$$
$$\le\sum_{\g\ni x}
\sum_{g\in {\cal G}_\g}{\bar\s}^{|\g|}e^{a|\g|}
\prod_{X\in g}W(X)
\le
\sum_{\g\ni x}
\sum_{g\in {\cal G}_\g}{\s}^{|\g|}
\prod_{X\in g}W(X)
\Eq(1.6.4)
$$
Hence the theorem follows if we are able to show
the inequality
$$
F(\s)=\sum_{\g\ni x}
\sum_{g\in {\cal G}_\g}{\s}^{|\g|}
\prod_{X\in g}W(X)< e^a-1
\Eq(1.6.5)$$
where
$$\s ={\bar\s}e^{a}$$

We now rewrite the sum over connected graphs $g$
passing through the lattice point $x$ using Cassandro-
Olivieri hierarchy.
The latter is constructed observing that, given a connected graph
$g$, its links can be always ordered, for some
positive integer $t$,  in the following way
$$
g=\{C_{0},C^1_1 ,\dots ,C^1_{k_1},\dots ,
C^t_1 ,\dots ,C^t_{k_t}\}
$$
where $C_{0}$ is a link of $g$ such that $x\in C_0$.
For $1\leq s \leq t$,
the links $C^s_1 ,\dots ,C^s_{k_s}$ represent the
$s^{\rm th}$ hierarchy and we have the following relations
between different hierarchies:

$$
C^s_i\cap\left[\cup_{i=1}^{k_{s-1}}C^{s-1}_{i}\right]
\neq  \emptyset\nonumber$$
$$C^s_i\cap\left[\cup_{l=0}^{s-2}\cup_{i=1}^{k_l}C^{l}_{i}\right]
=  \emptyset\Eq(1.6.6)
$$
Thus the first hierarchy
$C^1_1 ,\dots ,C^1_{k_1}$
of links of $g$ is the collection of
all links of $g$ which have a non empty intersection with $C_0$;
the second hierarchy $C^2_1 ,\dots ,C^2_{k_2}$ is the collection
of links of $g$ which have a non empty intersection with
$\cup_{i}C^1_i$, but has an empty intersection with the set $C_0$,
and so on.

We will also denote, for $s\geq 1$
$$
\D_s =
\cup_{i=1}^{k_s}C^s_i
\backslash\left[\left(\cup_{i=1}^{k_s}C^s_i\right)
\cap\left(\cup_{i=1}^{k_{s-1}}C^{s-1}_{i}\right)\right]
\Eq(1.6.7)$$
while $\D_0 = C_0$.
The set $\D_s$ represents therefore the set of the new
points reached by the $s$-th hierarchy. By definition
$|\D_s|>0$ for $s<t$, because the bonds of the $(s+1)^{\rm th}$
hierarchy have intersection only with the set $\D_s$,
else they would belong to an earlier hierarchy.
Note that $\{\D_0 ,\D_1 ,\dots ,\D_t \}$ is a
partition of the set ${\rm supp}\ g$ and thus
$$
|\D_0|+ |\D_1| +\cdots +|\D_t|= |{\rm supp}\ g|
\Eq(1.6.8)$$
thus, for a given $g=\{C_{0},C^1_1 ,\dots ,C^1_{k_1},\dots ,
C^t_1 ,\dots ,C^t_{k_t}\}$, we can always write
$$\s^{|{\rm supp}\ g|}\le
\s^{|\D_0|}\s^{|\D_1|}\cdots\s^{|\D_{t-1}|}
$$

Hence we can get the estimate
$$
F(\s)\leq
\sum_{C_0\ni 0}W(C_0)\s^{|C_0|}+
\sum_{C_0\ni 0}W(C_0)\s^{|\D_0|}
\sum_{t=1}^{\infty}
\sum_{k_{1} =1}^{\infty}\sum_{C^{1}_1 ,\dots C^{1}_{k_{1}}}
\prod_{j=1}^{k_{1}}W(C^{1}_{j})\s^{|\D_{1}|}
\cdots
$$
$$
\cdots
\sum_{k_{t-1} =1}^{\infty}\sum_{C^{t-1}_1 ,\dots C^{t-1}_{k_{t-1}}}
\prod_{j=1}^{k_{t-1}}W(C^{t-1}_{j})\s^{|\D_{t-1}|}
\sum_{k_{t} =1}^{\infty}\sum_{C^{t}_1 ,\dots C^{t}_{k_t}}
\prod_{j=1}^{k_{t}}W(C^{t}_{j})
\Eq(1.6.9)$$
Let us start to perform the last sum in r.h.s. of \equ(1.6.9),
i.e. the sum over the $t^{\rm th}$ hierarchy, supposing to keep
fixed all sets concerning the previous hierarchies.
We obtain
$$
\begin{aligned}
\sum_{k_{t} =1}^{\infty}\sum_{C^{t}_1 ,\dots C^{t}_{k_t}}
\prod_{j=1}^{k_{t}}W(C^{t}_{j}) & \leq
\sum_{k_{t} =1}^{\infty}{1\over k_{t}!}
\left(\sum_{C:\, C\cap\D_{t-1}\neq \emptyset}
W(C)\right)^{k_t}\\
& \leq
\sum_{k_{t} =1}^{\infty}{1\over k_{t}!}
\left(|\D_{t-1}|K\right)^{k_t}\\
&\leq
e^{K|\D_{t-1}|}-1
\end{aligned}
$$
where the sum $\sum_{C:\, C\cap\D_{t-1}\neq \emptyset}
W(C)$ is bounded by fixing one point in $\D_{t-1}$
(and this gives the factor $|\D_{t-1}|$)
and then by summing over all $C$ passing for such point.
Hence
$$
F(\s)\leq
\sum_{C_0\ni 0}W(C_0)\s^{|C_0|}+
\sum_{C_0\ni 0}W(C_0)\s^{|\D_0|}
\sum_{t=1}^{\infty}
\sum_{k_{1} =1}^{\infty}\sum_{C^{1}_1 ,\dots C^{1}_{k_{1}}}
\prod_{j=1}^{k_{1}}W(C^{1}_{j})\s^{|\D_{1}|}
\cdots
$$
$$
\cdots
\sum_{k_{t-1} =1}^{\infty}\sum_{C^{t-1}_1 ,\dots C^{t-1}_{k_{t-1}}}
\prod_{j=1}^{k_{t-1}}W(C^{t-1}_{j})
\s^{|\D_{t-1}|}(e^{K|\D_{t-1}|}-1)
$$
Note now that $e^K\s<1$
certainly holds if \equ(1.6.3) does and $a$ is choosen small enough,
namely $a< (3/4)|\log{\bar \s}|$.
To find a bound of the maximum of
the factor $\s^{|\D_{t-1}|}(e^{K|\D_{t-1}|}-1)$
it is convenient to write
$$
\s^{|\D_{t-1}|}(e^{K|\D_{t-1}|}-1)=e^{(-|\log \s|+K)|\D_{t-1}|)}(1-e^{-K|\D_{t-1}|})
$$
The maximum in $|\D_{t-1}|$ of this function is assumed for $|\D_{t-1}|=x$ such that
$e^{-Kx}={|\log\s|-k\over|\log\s|}$, and using $e^K\s<1$ and hence
$e^{(-|\log \s|+K)|\D_{t-1}|)}<1$, this gives easily
$$
\s^{|\D_{t-1}|}(e^{K|\D_{t-1}|}-1)\leq {K\over |\log\s|}
$$
We sum now the next hierarchy $t-1$ as before keeping fixed
everything concerning the previous hierarchies and obtain
$$
\sum_{k_{t-1} =1}^{\infty}\sum_{C^{t-1}_1 ,\dots C^{t-1}_{k_t}}
\prod_{j=1}^{k_{t-1}}W(C^{t-1}_{j})
\s^{|\D_{t-1}|}(e^{K|\D_{t-1}|}-1)
$$
$$\leq
{K\over |\log\s|}
\sum_{k_{t-1} =1}^{\infty}{1\over k_{t-1}!}
\left(\sum_{C:\, C\cap\D_{t-2}\neq \emptyset}
W(C)\right)^{k_t}\leq
{K\over |\log\s|}
(e^{K|\D_{t-2}|}-1)
$$
Iterating until $t=1$, we obtain
$$
F(\s)\leq
\sum_{C_0\ni 0}W(C_0)\s^{|C_0|}+
\sum_{C_0\ni 0}W(C_0)
\sum_{t=1}^{\infty}
\left({K\over |\log\s|}\right)^{t}
\leq
K\left[\s+{K\over |\log\s|-K}\right]
$$
hence we get \equ(1.6.0), and then the proof of the theorem,
if the following condition is satisfied
$$
K\left[\s+{K\over |\log\s|-K}\right]<e^a-1\Eq(1.6.10)
$$
The relation \equ(1.6.10) give a condition between $K$, $a$ and $\bar\s$.
If for instance one suppose to know the value of $\bar\s$, the problem is
to find the value of $a$ which maximizes the constant $K$. This problem
can be solved numerically. A rough estimate, once the value of $\bar\s$
is fixed, is to
choose e.g. $a=(1/2)|\log{\bar \s}|$, which corresponds to
the choice
$\s=\sqrt{\bar\s}$ (i.e. $|\log\s| = {1\over 2}|\log{\bar\s}|$),
obtaining easily
$$
K< {|\log\bar\s|\over 4}
\Eq(1.6.11)$$
\\$\Box$

\section{Some applications of theorem \ref{hier}}

\subsection{The Ising model}

\\A rough but fast estimate of the analyticity region in $\b$ (high temperature regime)
of the partition function
of the Ising model  can be obtained using theorem \ref{hier}.

\\In order to apply theorem \ref{hier} to the partition function \equ(1.7pi) we have to
find $\bar\s$ and $J$ such that

$$
\int \prod_{x\in {\rm supp}\ g} d\m(\s_x)
\prod_{\{x,y\}\in g\atop x,y ~{\rm near. ~neigh.} }
\left|e^{\b J\s_x\s_y }-1\right|\le{\bar\s}^{n}
\prod_{\{x,y\}\in g}W\Eq(1.6.12)$$

This can be done observing that, for $\b<1$, $\left|e^{\b J\s_x\s_y }-1\right|<2\b J$.
The l.h.s. of \equ(1.6.12) can be estimated then as $(2\b J)^{|E_g|}$. Since $|E_g|\ge n-1$ and $n\ge 2$,
we can choose $\bar\s=(\b J)^{1/3}$, $J=2(\b J)^{1/3}$ and then, having each site $2d$ nearest neighbors,
$K=4d(\b J)^{1/3}$. The condition \equ(1.6.3) becomes then
$$4d(\b J)^{1/3}<{|\log(\b J)|\over 12}\Rightarrow {(\b J)^{1/3}\over|\log(\b J)|}<{1\over 48d}\Eq(condorr)$$
which is satisfied for $(\b J)<1$ small enough.
By a rapid computation
\equ(condorr) is satisfied by $\b\approx 4\times 10^{-5}$.
However, we will see in the next chapter that this estimate is far from optimal.

\subsection{Direct proof of the Israel theorem}
\zeq

Discrete spin systems with many body interactions have been studied since the early
'70. After the pioneering work by Gallavotti and Miracle-Sole, which proved analiticity
of the logarithm of the partition function when
the spin space state has only a finite number of values, Israel proved in [I] a general
condition, valid for all compact spins,
to guarantee that uniformly in the volume the partition function stays different
from zero, and hence indirectly the analyticity of its logarithm. This analyticity result can be
found easily and directly with a proof based only on combinatorial
arguments by means of theorem \ref{hier}.

\\Consider a lattice system described by the
partition function \equ(robeta1)
with
many-body potentials such that

$$
\b|\Psi(X,\phi_{X})|\leq w(X),\quad
{\rm \ for\ all}\ X \Eq(6.2.1)$$

\\with $w(X)$ positive function defined on the sets of $Z^d$
such that, for some $a>0$
$$
\sum_{X\ni 0}w(X)e^{a|X|}=I(a)<\infty
\Eq(6.2.2)$$
which is a condition analogous to [I].
We also define
$$
{\bar I}=\sup_{X\subset\La}w(X)
$$
Then we have
\begin{teo}
The logarithm of the partition
function  \equ(robeta1)
of a lattice system with $\Psi$ satisfying
\equ(6.2.1) and \equ(6.2.2),
converges absolutely uniformely
in the volume, provided the following relation between $I(a)$
and $a$
is satisfied
$$
I(a)<{e^{-{\bar I}}\over 4}\ a
\Eq(6.2.3)$$
\end{teo}

\\{\bf Proof}.

\noindent
For any $\g\subset Z^d$ and for
any graph $g\in {\cal G}_{\g}$ we have
$$
\begin{aligned}
\left|
\int\prod_{x\in \g}d\m(\phi_x)
\prod_{X\in g}[e^{-\b\Psi(X,\phi)}-1]
\right| & \leq
\prod_{X\in g} w(X)e^{ w(X)}\\
& \leq
e^{-a|R|}\prod_{X\in g} e^{a|X|}w(X)e^{\bar I}
\end{aligned}
$$
Hence we can apply theorem \ref{hier} posing $\bar\s=e^{-a}$ and
$W(X)=e^{a|X|}w(X)e^{\bar I}$.
$\Box$
\vglue.3truecm
\noindent

\subsection{ The lattice gas with $N$ body interactions}

\\The lattice gas interacting through many body interaction
can be also  treated using theorem \ref{hier}. We will prove the following
theorem.
\vglue.5truecm
\begin{teo}
Consider a lattice gas
at temperature $\b^{-1}$ and activity $z$,
described by the
grand canonical partition function
$$
\Xi_{\La}(\b , z)=1+ z|\Lambda|+ \sum_{n\geq 2}z^{n}
\sum_{X\subset \Lambda\atop |X|=n}
e^{-\b\sum_{Y\subset X\atop |Y|\geq 2}V(Y)}
\Eq(1.6.13)$$
with
$$
\sum_{X\ni 0}|V(X)|=J<\infty,
\Eq(1.6.14)$$
then  $|\La|^{-1}\log \Xi_{\La}(\b ,z)$ can be written as a series
absolutely convergent uniformely in $\La$, for all $\b$ and $z$
satisfying the following condition
$$
z(\exp[4\b Je^{\b J}]-1)<1
\Eq(1.6.15)$$
\end{teo}
\vglue.3truecm

\noindent
{\bf Proof}.

\\We can trasform $\Xi_{\La}(\b , z)$ in a form suitable for
polymer expansion. We can write
$$
\Xi_{\La}(\b , z)=\sum_{\s_{\Lambda}}
z^{\sum_{x\in \Lambda}\s_{x}}
\exp\{-\b\sum_{X\subset \Lambda\atop |X|\geq 2}V(X)
\prod_{x\in X}\s_x\}
\Eq(1.6.16)$$
where the "spin" variables $\s_x$ at the site $x\in\La$
can take the two values $\s_x =0,1$, and
$\sum_{\s_{\Lambda}}$ means the sum over all
possible configurations of the valus of spins in $\La$.
Then, by Mayer expansion on the $\exp$
$$
\Xi_{\La}(\b , z)=\sum_{\g_1 ,\g_2 ,\dots ,\g_k\in\pi(\Lambda)}
\r_{\g_1}\r_{\g_2}\cdots\r_{\g_{k}}
$$
where $\pi(\La)$ is the set of the partitions of $\La$ and
$$
\r_\g=\begin{cases}1+z &if |\g|=1\\
\sum_{\s_\g}z^{\sum_{x\in \g}\s_x}\left[
\sum_{g\in\cal G_{\g}}\prod_{X\in g}\left
(e^{-\b V(X)\prod_{x\in X}\s_x}-1\right)
\right] &if |\g|\geq 2
\end{cases}
$$
Now note that, for any $g\in\cal G_{\g}$,
$$
\prod_{X\in g}\left
(e^{-\b V(X)\prod_{x\in X}\s_x}-1\right)=
\begin{cases}\prod_{X\in g}\left
(e^{-\b V(X)}-1\right) &if \s_{x}=1\,\,\forall x\in \g\\
0 &otherwise
\end{cases}
$$
hence defining
$$
\z(\g)=\begin{cases}1 &if |\g|=1\\
\left({z\over 1+z}\right)^{|\g|}\left[
\sum_{g\in\cal G_{\g}}\prod_{X\in g}\left
(e^{-\b V(X)}-1\right)
\right] &if |\g|\geq 2
\end{cases}
$$
we can write
$$
\Xi_{\La}(\b , z)=(1+z)^{|\La|}
\sum_{\{\g_1 ,\dots ,\g_k\}\atop |\g_i|\geq 2\,\g_i\cap \g_j =\emptyset}
\z(\g_1)\z(\g_2)\cdots\z(\g_{k})
$$
and thus also
$$
\log \Xi_{\La}(\b , z)=|\La|\log(1+z)+
\sum_{n\geq 1}{1\over n!}
\sum_{\g_1 ,\dots ,\g_k\atop |\g_i|\geq 2}
\phi^{T}(\g_1 ,\dots ,\g_k)
\z(\g_1)\z(\g_2)\cdots\z(\g_{k})
$$
Hence the lattice gas pressure is analytic when the infinite
sum in r.h.s. of equation above is absolutely convergent
uniformely in $\La$.
We have
$$
\left|
\left({z\over 1+z}\right)^n
\prod_{X\in g}[e^{-\b V(X)}-1]
\right|\leq
$$
$$
\left({z\over 1+z}\right)^n
\prod_{X\in g} \b|V(X)|e^{\b J}
\leq
\left({z\over 1+z}\right)^{n}
\prod_{X\in g} |{\tilde V}(X)|
$$
where ${\tilde V}(X)=\b|V(X)|e^{\b J}$.
Hence we can apply theorem 2 with $K=\b Je^{\b J}$ and
$$
{\bar\s}=\left({z\over 1+z}\right)
$$
and the convergence occurs when
$$
\b Je^{\b J}<{1\over 4}\left|\log\left({z\over 1+z}\right)\right|
$$
which is the condition \equ(1.6.15). $\Box$

\\Note that the relation between
$z$ and $\b$ in \equ(1.6.15),
namely the fact that $z$ can be taken
large if $\b$ is sufficently small,
is of the form needed to apply
the particle-hole symmetry pointed out in
\cite{GMR}, enlarging the analiticity region.

\chapter{Ising model}
\numsec=7\numfor=1
\def\L{\Lambda}
The Ising model is  a lattice system enclosed in a box $\L\subset \mathbb{Z}^d$. The box
$\L$ is a finite set, generally a square of size $L$ which
contains $|\L|~=~L^d$ sites of the lattice $\mathbb{Z}^d$.
\def\ss{{\bf \s}}
In each site $x\in \L$ there is a random variable $\s_x$ which can
take one of the two possible value $\s_x ~=~\pm 1$. A configuration
$\ss_\L$ of the system is given when one declares the value of the
spin $\s_x$ for each site $x\in \L$. Hence $\ss_{\L}$ is a set
$\ss_\L ~=~ (\s_{x_1},\dots,\s_{x_{|\L|}})$ of $|\L|$ numbers $\pm
1$. Equivalently one can say that a configuration $\s_\L$ of the
system is a function $\s_\L: \L\to \{+1,-1\}: x\mapsto \s_x$.
Observe that the total number of configurations of the system in
$\L$ is $2^{|\L|}$. We denote by $\O_\L$ the set of all possible spin configurations in $\L$.
\def\dL{\partial \L}

\index{Ising model}
\\Given a configurations $\ss_\L$ of the system in $\L$,  the energy
of such configuration is
$$
H_\L(\ss_\L) ~=~ -J\sum_{|x-y|=1\atop\{x,y\}\subset \L}\s_x\s_y- h\sum_{x\in \L}\s_x+ {\cal B}(\s_\L)\Eq(Hh)
$$
where $h$ is an ``external magnetic field'', $J>0$ is a positive
constant and ${\cal B}(\s_\L)$ represent the interaction of the
spins inside the box $\L$ with the world outside. Of course ${\cal
B}(\s_\L)$ is rather arbitrary.  We will list in a moment some
typical boundary conditions. Let us just remark here that ${\cal
B}(\s_\L)$ has to be in any case  a ``surface term'', i.e.
$$
\lim_{\L\to \i}{ \max_{\s_\L}|{\cal B}(\s_\L)|\over
|\L|}~=~0\Eq(contor)
$$

\\The statistical mechanics is obtained by assigning to any configuration $\s_\L$
a probability to occur $P(\s_\L)$. This probability is given by
$$
P(\s_\L)~=~ { e^{ -\b H_\L(\ss_\L)}\over Z_{\L}(\b,h)}\Eq(prob)
$$
where
$$
Z_{\L}(\b,h) ~=~  \sum_{\ss_\L\in\O_\L} e^{ -\b H_\L(\ss_\L)}
$$
is the partition function in the grand canonical ensemble. If
$f(\s_\L)$ is a function depending on the configuration of the
system, its mean value in the grand canonical ensemble is
$$
\<f(\s_\L)\>_\L ~=~ \sum_{\s_\L\in \O_\L} P(\s_\L)f(\s_\L)~=~  \sum_{\s_\L\in \O_\L} { e^{
-\b H_\L(\ss_\L)}f(\s_\L)\over Z_{\L}(\b,h)}
$$
The finite volume free energy of the model is given by
$$
 f_\L(\b,h)~=~ {1\over |\L|}\ln Z_{\L}(\b,h)\Eq(elibL)
$$
Note that, at finite volume, this function is
analytic in $\b,h$ for all $\b>0$ and  for all $h\in
(-\i,+\i)$.

\\Thermodynamic is recovered evaluating the following limit
$$
f(\b,h)~=~\lim_{\L\to \i}{1\over |\L|}\ln
Z_{\L}(\b,h)\Eq(elibinf)
$$
The function $f(\b,h)$ is called the infinite volume free energy of the system. It
is easy to show that the limit exists, but in general is not
expected to be analytic in the whole physical region. It is also
worth to stress that \equ(contor) ensures that the limit above in
independent on the boundary conditions $\cal B$. As a matter of
fact, let
\def\BB{{\cal B}}
$\BB(\s_\L)$ and $\BB '(\s_\L)$ two different boundary conditions
and let us denote
$$
H^{\rm bulk}_\L(\ss_\L)~=~ -J\sum_{|x-y|=1\atop\{x,y\}\subset \L }\s_x\s_y- h\sum_{x\in \L}\s_x
$$
then we have
$$
{Z_{\L, \BB}(\b,h)\over Z_{\L,\BB'}(\b,h)}~=~ { \sum_{\ss_\L} e^{
-\b H^{\rm bulk}_\L(\ss_\L)}e^{ -\b\BB(\s_\L)} \over \sum_{\ss_\L}
e^{ -\b H^{\rm bulk}_\L(\ss_\L)}e^{ -\b\BB'(\s_\L)} }\le
\max_{\s_\L}e^{ \b|\BB(\s_\L)|+ \b|\BB'(\s_\L)|}
$$
i.e
$$
{Z_{\L, \BB}(\b,h)\over Z_{\L,\BB'}(\b,h)}\le e^{\b \max_{\s_\L}(
|\BB(\s_\L)|+ |\BB'(\s_\L)|)}
$$
analogously one can get
$$
e^{-\b \max_{\s_\L}( |\BB(\s_\L)|+ |\BB'(\s_\L)|)}\le {Z_{\L,
\BB}(\b,h)\over Z_{\L,\BB'}(\b,h)}
$$
whence
$$
{1\over |\L|}|\ln Z_{\L, \BB}(\b,h) - \ln Z_{\L, \BB'}(\b,h)|\le
{1\over |\L|}\b \max_{\s_\L}( |\BB(\s_\L)|+
|\BB'(\s_\L)|)\Eq(find)
$$
and taking the limit $\L\to \infty$ and using \equ(contor) we get
the result.

\\We list below some typical boundary conditions.

\\1) {\it Open (or free)  boundary conditions}. This is the case
$$
\BB_0(\s_\L)~=~0
$$
\vv
\\2) {\it Periodic boundary conditions}. This is the case in which $\L$ is a Torus,
i.e. spins of opposite faces interact via the constant $-J$, as if
they were nearest neighbors. Clearly this can be obtained by a
suitable choice of $\BB_{p}(\s_\L)$
\vv
\index{boundary condition}
\\3) {\it $+$ boundary conditions}. This is the following case. Let
$\dL=\{x\in \mathbb{Z}^2\backslash\L: |x-y|=1~~{\rm for~ some}~
y\in \L\}$ be the ``external boundary'' of $\L$ and put $\s_x~=~+1$
for all $x\in \dL$. Then
$$
B_+(\s_\L) ~=~ -J\sum_{x\in \dL}\sum_{y\in \L\atop |x-y|=1}\s_y
$$
Physically this means to fix the spins outside $\L$ (those who can
interact with spin inside!) to the value $\s_x=+1$.

\\The fact that the free energy is independent from boundary
conditions does not mean that the system is stable respect to
boundary conditions. Precisely the instability of the system
respect to boundary conditions is an evidence for phase transition.
Even if the free energy is independent of boundary conditions
it still can occur a discontinuity of some derivative of the free energy at some point
in the $(\b, h)$ region of the physical parameters. We will show below that this indeed happens.

\\We will consider ahead the magnetization of the
system, which measures if spin are mostly oriented up or down. This quantity is
the partial derivative of the free energy respect to the magnetic field.
$$
M_{\L}^{\BB}(\b, h)~=~ {1\over |\L|}\sum_{x\in \L}\<\s_x\>_{\L, \BB}~=~
\b^{-1}{\partial\over\partial h}f_\L(\b, h)
$$
We will see that  the magnetization in the
thermodynamic limit can take different  values according to different boundary conditions.

\subsection {High temperature expansion}
We now develop a (high temperature)
polymer expansion for the partition function of the two
dimensional Ising model {\it with zero magnetic field} (i.e. we set $h=0$ in \equ(Hh)). We suppose from now on that
$\L$ is a square of size $L$ which thus contains $L^2$ lattice
sites. We also choose free boundary conditions (i.e. we set $\BB(\s_\L)=0$ in \equ(Hh)).
The Hamiltonian of the zero field, free boundary conditions Ising
model is thus
$$
H_\L(\ss_\L) ~=~ -J\sum_{|x-y|=1\atop\{x,y\}\subset \L }\s_x\s_y ~=~ -J \sum_{b\in B(\L)}\tilde \s_b
$$
where $B(\L)$ is the set of all pairs $b=\{x,y\}\subset \L$ such
that $|x-y|~=~1$ (nearest neighbor pairs) and for  $b=\{x,y\}$ we
put $\tilde \s_b= \s_x\s_y$. Recall that, since $\L$ is a $L\times
L$ box in $\mathbb{Z}^2$, then $|B(\L)|~=~ 2L(L-1)$

\\The partition function of the Ising model at zero magnetic field is  thus
$$
Z_{\L}(\b) ~=~  \sum_{\ss_\L\in \O_\L}\prod_{b\in B(\L)}e^{+\b J\tilde\s_b}
$$
We want to prove that the free energy
$$
f_\L(\b)= {1\over |\L|}\ln Z_{\L}(\b) \Eq(free0)
$$
is an analytic function of $\b$ is $\b$ is sufficiently small (i.e. in the high temperature regime) uniforlmly in the volume $\L$.
Observe that $\tilde\s_b~=~\pm 1$, hence
$$
e^{+\b J\tilde\s_b}~=~ \cosh(\b J\tilde\s_b)+ \sinh(\b J\tilde\s_b)~=~
 \cosh(\b J)+ \tilde\s_b\sinh(\b J)~=
 $$
 $$
 =~ \cosh(\b J)[1+\tilde\s_b\tanh(\b J)]
$$
hence
$$
Z_{\L}(\b) ~=~   [\cosh(\b J)]^{2L(L-1)} \sum_{\ss_\L\in \O_\L}\prod_{b\in
B(\L)}[1+\tilde\s_b\tanh(\b J)]
$$
Developing the product $\prod_{b\in B(\L)}[1+\tilde\s_b\tanh(\b
J)]$ we get terms of the type
$$
[\tanh(\b J)]^k\tilde\s_{b_1}\dots\tilde\s_{b_k}
$$
which has a clear geometric interpretation. The set of bonds ${b_1},\dots,{b_k}$ form  a graph (connected
or not) in $\L$ whose links are nearest neighbors.

\\When one perform the sum over $\ss_\L$
we get that
$$
\sum_{\ss_\L}\tilde\s_{b_1}\dots \tilde\s_{b_k}
$$
is zero whenever there is a not paired spin. See Figures 11 and 12.

\begin{figure}
\begin{center}
\includegraphics[width=5cm,height=5cm]{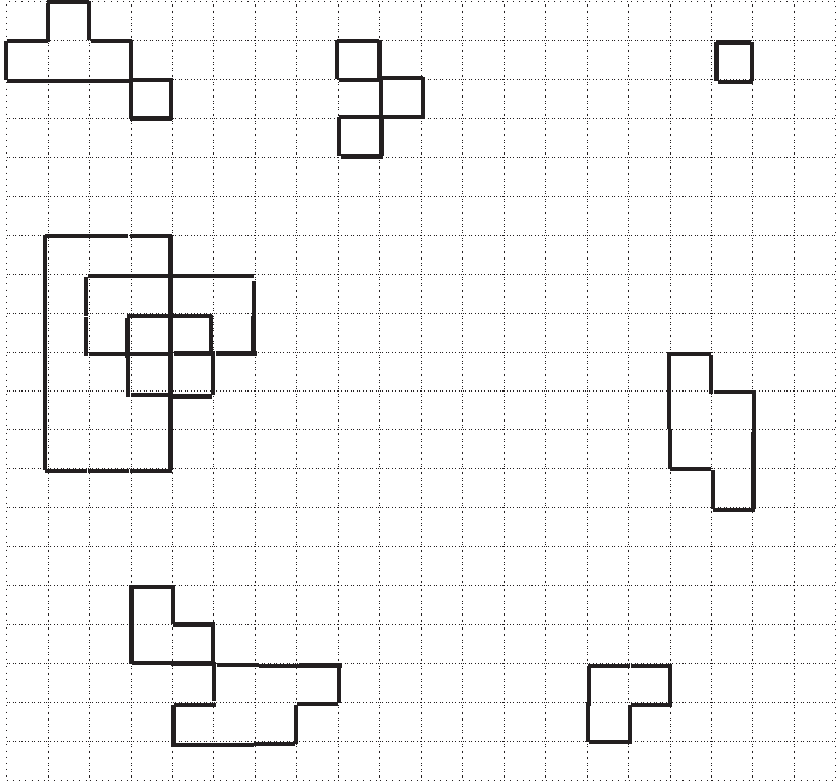}
\end{center}
\begin{center}
Figure 11.  A  non  vanishing  graph with seven connected
components
\end{center}
\end{figure}

\\The only graphs which
yield a non vanishing contribution to the partition function are
those whose vertices have incidence number zero, two or four (Figure
11), while the contribution of all other graphs is zero once the sum over
configurations $\s_\L$ has been done (Figure 12).
If the graph
$\tilde\s_{b_1}\dots \tilde\s_{b_k}$ is non vanishing then
$$
\sum_{\ss_\L}\tilde\s_{b_1}\dots \tilde\s_{b_k}~=~ 2^{L^2}
$$
\def\sg{{\rm supp}\,\g}
We can naturally split a non vanishing graph in non intersecting
connected components which we will call {\it lattice animals}. For
example in figure 11 it is drawn a non vanishing graph formed by
seven non intersecting lattice animals. A lattice animal $\g$ is
thus nothing but a graph $g$ with edge set $E_\g~=~\{b_1, \dots ,b_k\}\subset B(\L)$
formed by
nearest neighbor links $b~=~\{x, y\}$
and with vertex set
$V_g=\cup_{i=1}^k b_i\subset \L$  which is a connected graph  in $V_\g$ (in usual sense). The allowed lattice animals are only those $\g$ with
incidence number at the vertices equal to two or  four. Let us denote by $\LL$ the set of possible lattice animals in $\mathbb{Z}^d$ and by
$\LL_\L$ the set of all possible lattice animals in $\L$.

\\Two lattice animals $\g$ and $\g'$ are non overlapping (i.e. compatible), and we write $\g\sim\g'$
{ \it if and only if } $V_\g\cap V_{\g'}'=\emptyset$. We will denote shortly
$|\g|= |E_\g|$ (i.e. $|\g|$ is the number of nearest neighbor bonds which constitute $\g$,
i.e. if $\g=\{b_1, \dots ,b_k\}$ then $|\g|=k$.
\\Note that only such lattice animals (i.e. just those
with incidence number at the vertices equal to 2 or to 4) survive
because we are using free boundary conditions. Note also that
 lattice animal $c\in \LL$ with incidence number equal to 2 in anyone of its vertices is a simple cycle. For a cycle $c$
we have that $|V_c|=|c|$, while for a lattice lattice animal $\g$ which is not a cycle we have $|V_\g|< |\g|$.
\begin{figure}
\begin{center}
\includegraphics[width=5cm,height=5cm]{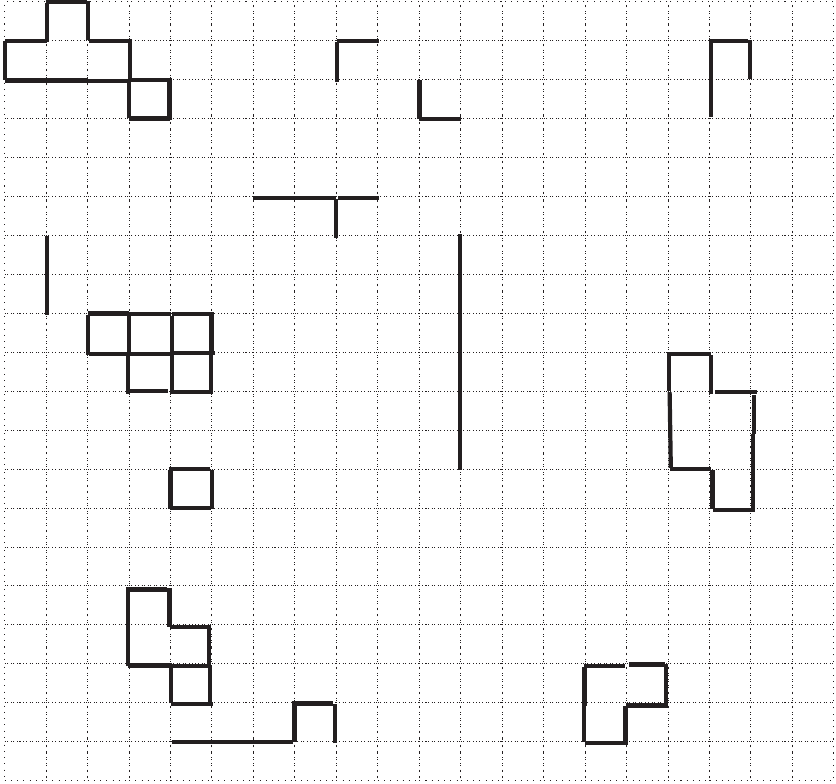}
\end{center}
\begin{center}
Figure 12.  A  vanishing graph.
\end{center}
\end{figure}

\\In conclusion we can write
$$
Z_{\L}(\b) ~=~   [\cosh(\b
J)]^{2L(L-1)}2^{L^2}\;\Xi_{\L}(\b)\Eq(1is)
$$
where
$$
\Xi_\L(\b)~=~
 1+
\sum_{n\ge 1}\sum_{\{\g_1, \dots \g_n\}\subset \LL_\L\atop |\g_i|\ge 4,\;
\g_i\sim \g_j}\xi(\g_1)\dots \xi(\g_n)\Eq(is2)
$$
where $\g$ denote a allowed lattice animal (whence the condition
$|\g|\ge 4$) with activity
$$
\xi(\g)~=~ [\tanh(\b J)]^{|\g|}\Eq(actis)
$$
Thus the partition function of the Ising   model \equ(1is) can be
written, modulo a constant (the factor $[\cosh(\b
J)]^{2L(L-1)}2^{L^2}$), as the partition function of a hard core
polymer gas, i.e the term \equ(is2). In this case polymers
are  lattice animals, i.e. elements of $\LL$ with the incompatibility relation
$\g\nsim\g'$ if and only if $V_{\g}\cap V_{\g'}\neq \0$.

\\Let us apply to this polymer gas the convergence  criterion \equ(FPcrit) of Theorem \ref{coro:1}  to the
polymers system with partition function $\Xi_\L(\b)$ given by \equ(is2) above. Namely we need to find numbers $\m(\g)\in (0,+\infty)$ such that
$$
\x(\g)\le {\m(\g)\over \Xi_{\LL_\g}(\bm \m)}\Eq(is3)
$$
with $\LL_\g=\{\g'\in \LL: \g'\nsim \g\}$.

\\As we did previoulsy we set $\m(\g)=\x(\g)e^{a{|\g|}}$ so that condition becomes
$$
 \Xi_{\LL_\g}(\bm \m)\le e^{a{|\g|}}\Eq(what?)
$$
where

$$
\Xi_{\LL_\g}(\bm \m)=1+ \sum_{n= 1}^{|V_\g|}{1\over n!}
\sum_{(\g_1,\dots,\g_n)\in \LL^n_\g\atop \g_i\sim\g_j} \prod_{i=1}^n|\x(\g_i)| e^{a|\g_i|}
$$
Consider now the factor
$$
F_n(\g,a)\doteq \sum_{(\g_1,\dots,\g_n)\in \LL^n_\g\atop\g_i\sim\g_j} \prod_{i=1}^n|\x(\g_i)| e^{a|\g_i|}\Eq(factorb)
$$
We have thus to choose  $n$ lattice animals   $\g_1,\dots \g_n$ all incompatible with a given lattice animal $\g$ and all pairwise compatible.
We recall that  two lattice animals are incompatible of they share a vertex os $\Z^2$.
The factor \equ(factorb) is zero whenever $n>|V_\g|$, $\g_i\nsim\g$ since
$\g$ has  $|V_\g|$ vertices and  thus we can arrange at most $|V_\g|$ lattice animals pairwise compatible each one sharing a different vertex of $V_\g$.
Therefore, when the factor above is not zero, i.e. for $n\le |V_\g|$, it can be bounded at least by (again a very rough bound)

$$
\begn
F_n(\g,a)
& \le
|V_\g|(|V_\g|-1)\cdots (|V_\g|-n+1)\Bigg[\sup_{x\in \Z^2}\,\sum_{\g\in  \LL\atop x\in \g}|{\x(\g)}|e^{a|\g|}\Bigg]^n\\
&=
{|V_\g|\choose n}n! \Bigg[\sup_{x\in \Z^2}\,\sum_{\g\in  \LL\atop x\in \g}|{\x(\g)}|e^{a|\g|}\Bigg]^n
\egn
$$
Thus

$$
\begn
\Xi_{\LL_\g}(\bm \m)& \le 1+ \sum_{n=1}^{|V_\g|}
{|V_\g|\choose n}\Bigg[\sup_{x\in \Z^2}\,\sum_{\g\in  \LL\atop x\in \g}|{\r(\g)}|e^{a|\g|}\Bigg]^n\\
&=
\Bigg[1+\sup_{x\in \Z^2}\,\sum_{\g\in  \LL\atop x\in \g}|{\x(\g)}|e^{a|\g|}\Bigg]^{|V_\g|}
\egn
$$
Thus \equ(what?) can be written as
$$
\Bigg[1+\sup_{x\in \Z^2}\,\sum_{\g\in  \LL\atop x\in \g}|{\x(\g)}|e^{a|\g|}\Bigg]^{|V_\g|}\le e^{a|\g|}\Eq(vgg)
$$
and since, $|V_\g|\le |\g|$ for any $\g\in \LL$  (the equality holding only if $\g$ is a cycle) we have that \equ(vgg) is surely satisfied if
$$
\sup_{x\in \Z^2}\,\sum_{\g\in  \LL\atop x\in \g}|{\x(\g)}|e^{a|\g|}\le e^{a}-1\Eq(is4)
$$
Observe finally that, due to the symmetry of the problem the function
$$
f(x)=  \sum_{\g\in  \LL\atop x\in \g}|{\x(\g)}|e^{a|\g|}
$$
is constant as $x$ varies in $\Z^2$. Therefore \equ(is4) is equivalent to the condition
$$
\sum_{\g\in  \LL\atop 0\in \g}|{\x(\g)}|e^{a|\g|}\le e^{a}-1\Eq(is5)
$$
where $0$ is the origin in $\Z^2$.

\\The condition \equ(is5) is  a (high temperature) convergence condition for the analyticity of the free energy of the Ising model at zero magnetic field
and free boundary conditions.
Now recalling \equ(actis) and due to the symmetry of the problem
$$
\sum_{\g\in  \LL\atop 0\in \g}[\tanh(\b J)]^{|\g|}e^{a|\g|}
=
 \sum_{n\ge 4}[\tanh(\b J)]^{n}e^{an}\sum_{\g\in  \LL\atop 0\in \g, \;|\g|=n}1= \sum_{n\ge 4}[\tanh(\b J)]^{n}e^{an}C_n
$$
with
$$
C_n= \sum_{\g\in  \LL\atop 0\in \g, \;|\g|=n}1
$$
being the number of lattice animals in $\LL$  made by $n$ nearest neighbor bonds containing the origin.
I.e. we need to count all lattice animals $\g$ with a given cardinality $|\g|=n$ that pass through the origin.
To do this just observe that the nearest neighbor bonds of  a lattice animal  form a graph with degree 2 or 4, i.e. a graph
with even degree vertices. It is long known that any  graph with even degree at its vertices admits an Eulerian cycle (i.e. a graph cycle that crosses each edge exactly once).
Therefore all lattice animals containing the origin can be formed by performing a cycle starting at the origin.
To form a cycle starting at the origin in $\Z^2$ we can take each time 3 directions (also at the beginning since the cycle can be traveled in
two direction)
This immediately implies that
$$
C_n= \sum_{\g\in  \LL\atop 0\in \g, \;|\g|=n}1\le 3^n
$$
Therefore condition \equ(is5) is surely satisfied if
$$
\sum_{n\ge 4}[3\tanh(\b J)]^{n}e^{an}\le e^a-1
$$
i.e. posing $x= |3\tanh(\b J)|$  as soon as
$$
e^{4a} [3\tanh(\b J)]^4+(e^{2a}-e^a) [3\tan(\b J)]- (e^a-1)\le 0
$$
which yields  (for $a=0.21$)
$$
\tanh(\b J)\le 0.1525
$$
Hence convergence of the free energy, uniformly in the volume $\L$ occurs as soon as
$$
\b \le \b_0 \doteq {1\over J} \tanh^{-1} (0.1525) \approx {0.151\over J}\Eq(beta0)
$$

\subsection {Low temperature expansion}
We consider now the zero magnetic  field Ising model with $+$ boundary conditions in a box
$\L$ which is assumed to be a $L\times L$ square with $L^2$ sites.
Thus, given a configuration $\s_\L$ of the spins inside $\L$,  the
Hamiltonian is
$$
H^+_\L(\ss_\L)~=~ -J\sum_{\{x,y\}\subset \L \atop |x-y|~=~1}\s_x\s_y -
J\sum_{y\in \partial \L}\sum_{x\in \L, \atop |x-y|~=~1} \s_x
$$
Another way to write the Hamiltonian of the zero field + boundary
condition Ising model is as follows is
$$
H^+_\L(\ss_\L)~=~ -J\sum_{\{x,y\}\cap \L \neq \emptyset \atop
|x-y|~=~1}\s_x\s_y
$$
recalling that $\s_y=+1$ whenever $y\in \dL$.

We now rewrite the partition function $Z^+_{\L}(\b)$  via a
contour gas in the following way. For a fixed configuration
$\s_\L$ draw a unit segment perpendicular to the center of each
bond $b$ of nearest neighbors having opposite spins at its
extremes (in three dimensions this segment becomes a unite square
surface). Such unit segments have extremities which are in a
lattice shifted respect to the original lattice by a factor $1/2$
along both $x$ and $y$ axis. This lattice is called the dual lattice
and denoted by $\mathbb{Z^*}^2$. Hence this set of unit segments
forms a graph of nearest neighbor links in $\L^*\subset
\mathbb{Z^*}^2$ where $\L^*$ is a $(L+1)\times (L+1)$ square in
the dual lattice. With the particular choice of + boundary
conditions such graphs are exactly the same of Figure 11, with the
only difference that now they live in $\L^*$, i.e. in a square
with $(L+1)\times(L+1)$ sites in the dual lattice. See Figure 13.

\begin{figure}
\begin{center}
\includegraphics[width=7cm,height=7cm]{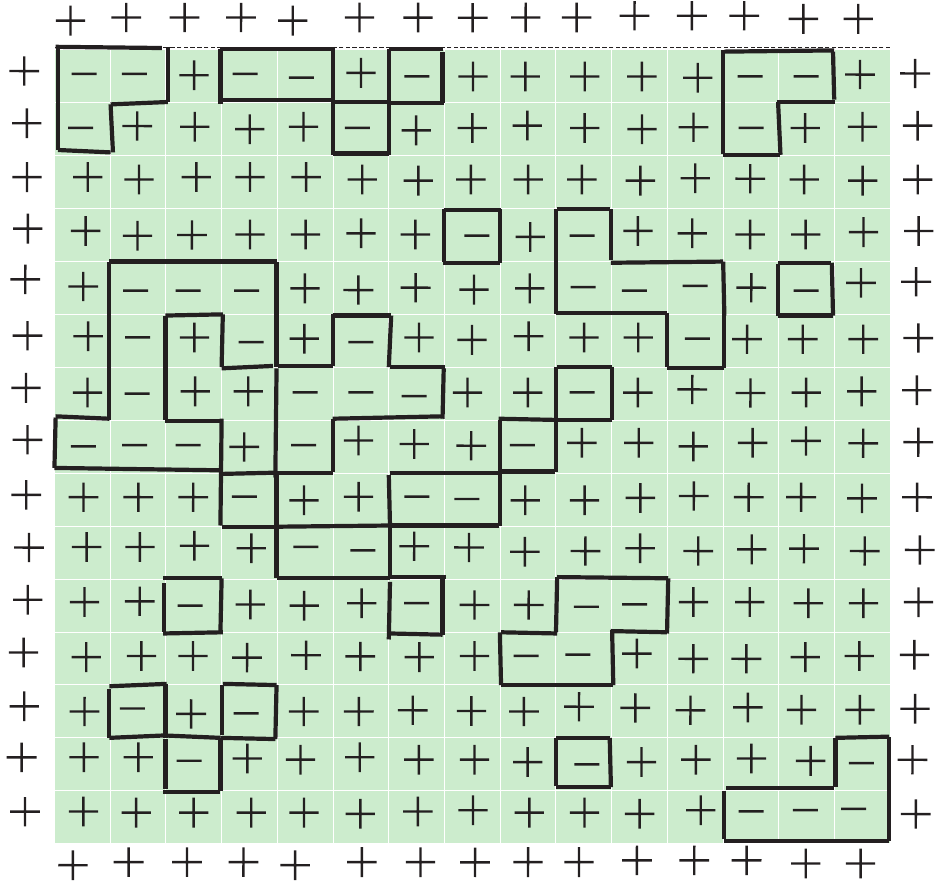}
\end{center}
\begin{center}
Figure 13.  The grey square is $\L^*$. The outer spins are those
fixed by the boundary conditions.
\end{center}
\end{figure}
To any configuration $\s_\L$ in $\L$  we can thus  associate a
graph $\G$ in the dual lattice made by nearest neighbor links with
incidence number equal either 0, 2 or 4. The graph $\G$ splits into its
connected components in the usual manner. Hence $\G=\{\g_1,\g_2,
\dots, \g_k\}$ with $\g_i\cap\g_j=\emptyset$. In this case (low
temperature expansion), such connected components $\g$  are called
{\it contours} and $|\g|$ is now called perimeter of the contour.
Note that now a contour configuration correspond to a spin
configuration. I.e., given the spin configuration $\s_\L$ in $\L$
we can recover the contour configuration $\G=\{\g_1,\g_2, \dots,
\g_k\}$ and conversely, given the contour configuration
$\G=\{\g_1,\g_2, \dots, \g_k\}$  we can recover the spin
configuration $\s_\L$ which originated it.

\\I.e. there is a one to one correspondence $\s_\L\to \{\g_1, \dots ,\g_k\}$.
\\We can express the Hamiltonian, which is a function of $\s_\L$, as a function of
the contour configuration $\G=\{\g_1,\g_2, \dots, \g_k\}$ uniquely
determined by $\s_\L$. This is done very easily, just observing
that, for a fixed configuration $\s_\L$ we have
$$
H^+_\L(\ss_\L)~=~ -J\sum_{\{x,y\}\cap \L \neq \emptyset \atop
|x-y|~=~1}\s_x\s_y ~=~ -J\left[B^+ - B^-\right]
$$
where
$$
B^+~=~\sharp ~{\rm of\; nearest\; neighbor\; pairs}~ \s_x\s_y ~{\rm
with}~ \s_x~=~\s_y
$$
$$
B^-~=~\sharp ~{\rm of\; nearest\; neighbor\; pairs}~ \s_x\s_y ~{\rm
with}~ \s_x~=~-\s_y
$$
If now  $ \{\g_1, \dots ,\g_k\}$ is the contour configurations
corresponding to $\s_\L$ we have that the number of nearest
neighbor pairs with opposite signal is by construction equal to
$\sum_{i=1}^k |\g_i|$, where recall that $|\g_i|$ is the perimeter
of $\g$
$$
B^-~=~ \sum_{i=1}^k |\g_i|
$$
\\Denoting with
$$
\tilde B_\L~=~2L(L-1)+ 4L~=~2L(L+1)
$$
the number of nearest neighbor pairs with non void intersection
with $\L$, we also have
$$
B^+ ~=~ \tilde B_\L-B^-
$$
hence
$$
H^+_\L(\ss_\L)~=~H_{\L^*}(\g_1, \dots ,\g_k)~=~-J \tilde B_\L+ 2J
\sum_{i=1}^k |\g_i|
$$
Hence the partition function of the Ising model with $h=0$ and
with $+$-boundary conditions can be written as

$$
Z^+_{\L}(\b)~=~e^{ +\b J \tilde B_\L}\left[ 1+ \sum_{n\ge
1}\sum_{\{\g_1,\dots , \g_n\}\subset \LL_{\L^*}\atop |\g_i|\ge 4, ~\g_i\sim \g_j}
e^{-2\b J\sum_{i=1}^n|\g_i|}\right]\Eq(is10)
$$
the factor 1 inside the bracket correspond to the situation in
which no contour is present or equivalently when all spin are +1.
As usual $\g_i\sim\g_j$ means that
$V_{\g_i}\cap V_{\g_j}~=~\emptyset$ where recall that $V_{\g_i}$ and $V_{\g_j}$
are subsets of  $\L^*$.

\\The expression in square bracket in \equ(is10) it is the grand canonical partition
function of a hard core lattice polymer gas in which the  polymers
are as before contours $\g_i$ but now  with activity $\exp[-2\b J |\g_i|]$, i.e.
$$
\Xi_{\L^*}(\b)~=~ 1+ \sum_{n\ge 1}\sum_{\{\g_1,\dots , \g_n\}\subset \LL_{\L^*}\atop |\g_i|\ge 4, ~\g_i\sim \g_j}e^{-2\b
J\sum_{i=1}^n|\g_i|}~=~
$$
$$~~~~~~~~~~~~~~
1+ \sum_{n\ge 1}{1\over n!} \sum_{(\g_1,\dots , \g_n)\in \LL^n_{\L^*}\atop |\g_i|\ge 4, ~\g_i\sim \g_j} \r(\g_1)\dots \r(\g_n)
$$
where now the activity of contours is
$$
\r(\g)~=~ e^{-2\b J|\g|}\Eq(is11)
$$

\\Note that the activity of a contour
goes to zero if the perimeter $|\g|$ of the contour increases, so
that big contours are depressed, but also the
activity of a contour goes to zero when $\b$ goes to infinity, i.e
all contours tend to be depressed when the temperature is lower
and lower. Hence at very low temperature one is expected to see a
sea of $+$ with very small and very few contours which surround islands
of -. This means that the in Ising model at zero magnetic field
and very low temperature spins are predominantly oriented +, and
this fact occurs independently on the size $\L$ where the system
is confined. This give a quite clear a detailed picture of the low
temperature phase of the Ising model with + boundary condition and
$h~=~0$.

\\We can now formulate a condition for the analyticity of the free energy of the
Ising model at $h~=~0$ with +boundary conditions in a completely
analogous way as we did for the case of high temperature. The only
thing that changes is that now the activity of contours (which in
the high temperature case where called lattice animals, but they
are the same objects) is defined by \equ(is11). Again we have to
check a formula identical to \equ(is5) where now in place of
$\xi(\g)$ we have to put $\r(\g)$ defined in \equ(is11). Hence we
reduce ourselves to check for which $\b$ the following inequality holds
$$ \sum_{n\ge 4}
[ 3e^a\, e^{-2\b J}]^{n}\le e^a-1
$$
which is satisfied, for $a=0.21$,  if
$$
 3e^{-2\b J}\le 0.4577
$$
i.e.
$$
\b \ge \b_1 ~\doteq~{1\over 2J} \ln\left[3\over 0.4577\right]\approx {0.94\over J}\Eq(beta1)
$$
we have that the finite volume free energy $f_\L(\b)$ is analytic
in $\b$ for all $\b\ge \b_1={0.94\over J}$ uniformly in $\L$. Note that in this
case $f_\L(\b)$ is calculated with + boundary condition, but the
limit $f(\b)~=~\lim_{\L\to\i}f_\L(\b)$ which is also analytic in
$\b$ for all $\b>\b_1$, does not depend on boundary conditions.

\\In conclusion we have proved the following theorem
\vskip.5cm
\\{\bf Theorem}.
The free energy of the Ising model at zero magnetic field
$$
f(\b, h=0)~=~\lim_{\L\to \i}{1\over |\L|}\ln Z_{\L}(\b,h=0)
$$
is analytic in $\b$ for all $\b\in(0,\b_0]\cup[\b_1, +\i)$ where
$\b_0$ is defined in \equ(beta0) and $\b_1$ is defined in \equ(beta1).
We now show that there is at least a point in the interval $(\b_0,
\b_1)$ where $f(\b, h=0)$ is non analytic.

\subsection{Existence of phase transitions}
First recall the  definition of  the magnetization of the system as
\index{phase transition}
$$
M_\L(\b)~=~{1\over |\L|}\sum_{x\in \L}\<\s_x\>_\L \Eq(magn)
$$
The quantity $M_\L(\b)$ measures the mean orientation of spins of
the finite system confined in $\L$. If $M_\L(\b)>0$ then it means
that spin + are predominant while if $M_\L(\b)<0$ this means that
spin - are predominant. It is just a simple calculus to show the
identity
$$
M_\L(\b)~=~\b^{-1}\left.{\partial f_\L(\b, h)\over\partial
h}\right|_{h~=~0}
$$
which tells us that the (finite volume) magnetization is, modulo a
factor $\b^{-1}$, the first derivative respect to $h$ of the
finite volume free energy of the system.

By a Theorem called ``Lee Yang theorem'' is possible to show that
$f(\b, h)$ is analytic for all  $\b>0$ when $h\neq 0$ hence the only
singularity of this function can occur when $h=0$.

Hence in order to show existence of transition phase, our strategy
will be to show that, depending on boundary condition the infinite
volume magnetization $M(\b)~=~\lim_{\L\to\i} M_\L(\b)$ is not stable
for variations of the boundary conditions.

Thus we will show first that  the magnetization $M^+_\L(\b)$ of
the Ising model with + boundary conditions (the superscript  ``+''
remember us that we are using + boundary conditions) is a number
arbitrarily near 1 for $\b$ sufficiently big. The result will also
immediately imply that the magnetization $M^-_\L(\b)$ of the Ising
model with - boundary conditions  is a number arbitrarily near -1
for $\b$ sufficiently big.

Let us denote $\<\cdot \>_\L^+$ the mean values when the Hamiltonian
is taken with $h=0$ and with + boundary condition. It is easy to
see that, for any $x\in \L$
$$
\<\s_x\>_\L^+ ~=~ 1-2\Big[ {\rm
Prob}_\L^+(\s_x=-1)\Big]\Eq(is13)
$$
where ${\rm Prob}_\L^+(\s_x=-1)$ denotes the probability that the
spin $\s_x$ at the site $x$ is equal to $-1$ for the Ising model
with zero magnetic field and + boundary conditions. As a matter of
fact, by definition

$$
\begn
\<\s_x\>_\L^+ &=~
(+1) ~\Big[{\rm Prob}_\L^+(\s_x=+1)\Big]~~ +  ~~(-1)~\Big[ {\rm
Prob}_\L^+(\s_x=-1)\Big]\\
&=~1\;-\, 2\Big[{\rm Prob}_{\L}^+(\s_x=-1)\Big]
\egn\Eq(proba)
$$

\\Thus in order to evaluate $\<\s_x\>_\L^+$ it is sufficient to evaluate
${\rm Prob}_\L^+(\s_x=-1)$.

\\We have seen before that the Ising model with + boundary condition and zero magnetic
field can be mapped in a contour gas, and that there is a one to
one correspondence between spin configurations $\s_\L$ and contour
configurations $\g_1, \dots \g_k$. Now  a $\s_\L$
such that the spin at $x$ is $-1$ is such that the site $x$ is the
``interior'' of at least one  contour associated to $\s_\L$, i.e. $x$ is
surrounded at least by a contour (actually by an odd number of
contours). Hence, if we denote ${\rm Prob}_\L^+(\exists \g\odot
x)$ the probability that at least one contour surrounds $x$, we
have surely
$$
{\rm Prob}_\L^+(\s_x=-1)\le {\rm Prob}_\L^+(\exists \g\odot x)
$$
But now
$$
{\rm Prob}_\L^+(\exists \g\odot x)~=~{ \sum_{\{\g_1,\dots , \g_n\}:
\exists \g_i\odot x\atop |\g_i|\ge 4, ~\g_i\sim \g_j}
e^{-2\b J\sum_{i=1}^n|\g_i|}\over \sum_{\{\g_1,\dots , \g_n\}\atop
|\g_i|\ge 4, ~\g_i\sim \g_j } e^{-2\b
J\sum_{i=1}^n|\g_i|}}
$$
where the denominator is the grand canonical partition function of
the contour gas and in the sum in the denominator is also included
the empty graph which contribute with the factor 1. Now we can
write
$$
\begin{aligned}
{\rm Prob}_\L^+(\exists \g\odot x)& =~{ \sum_{\{\g_1,\dots , \g_n\}:
\exists \g_i\odot x\atop |\g_i|\ge 4, ~\g_i\sim \g_j}
e^{-2\b J\sum_{i=1}^n|\g_i|}\over \sum_{\{\g_1,\dots , \g_n\}\atop
|\g_i|\ge 4, ~\g_i\sim \g_j } e^{-2\b
J\sum_{i=1}^n|\g_i|}}\\\\
& =~{\sum_{\g\odot x}e^{-2\b J|\g|} \sum_{\{\g_1,\dots , \g_n\}:
\g_i\sim \g\atop |\g_i|\ge 4, ~\g_i\sim \g_j}
e^{-2\b J\sum_{i=1}^n|\g_i|}\over \sum_{\{\g_1,\dots , \g_n\}\atop
|\g_i|\ge 4, ~\g_i\sim \g_j} e^{-2\b
J\sum_{i=1}^n|\g_i|}}\\\\
&\le\sum_{\g\odot x}e^{-2\b J|\g|}
\end{aligned}
$$
I.e. in conclusion we get
$$
{\rm Prob}_\L^+(\exists \g\odot x) \le\sum_{\g\odot x}e^{-2\b
J|\g|}\le \sum_{n\ge 4}e^{-2\b J n}\sum_{\g\odot x}1
$$
It is now easy to bound $\sum_{\g\odot x}$ with $n3^n$. As a
matter of fact, let us denote with $x'$ the point of intersection
of a contour $\g$ (such that $|\g|~=~n) $which surrounds $x$ with
the horizontal axis passing through $x$. Let us ask ourselves how
many possible $x'$ we can get. Obviously not more than $n$. Now
the possible contours with fixed perimeter $|\g|~=~n$ which pass
through a fixed $x'$ are at most $3^n$, so  $\sum_{\g\odot x} 1\le
n 3^n$.
Hence
$$
{\rm Prob}_\L^+(\exists \g\odot x) \le \sum_{n\ge 4}e^{-2\b J n} n
3^n
$$
Setting $$x_\b=3e^{-2\b J},$$the sum in l.h.s. converges  to if $x_\b<1$ and in this case we have
$$
{\rm Prob}_\L^+(\exists \g\odot x) \le\sum_{n\ge 4}n x_\b^n= {x_\b^4(4-3x_\b)\over (1-x_\b))^2}~\doteq~ g(\b)\Eq(gibeta)
$$
Recalling \equ(proba), in order to guarantee that $\<\s_x\>^+_\L>0$ uniformly in $\L$ we need to impose that
$$
 g(\b)<{1\over 2}
$$
which holds if
$$
x_\b<0.4795
$$
i.e. if
$$
\b> \b'_{1}\doteq ~ {0.917\over J}
$$
Note that $g(\b)$ defined in \equ(gibeta) above goes to zero as $\b\to\i$.
Thus when $\b>\b'_{1}$ we have that
$$
\<\s_x\>_\L^+ ~\ge ~ 1-2g(\b)
$$
and recalling \equ(magn)
$$
M_\L(\b,h=0)\ge
1-2g(\b)
$$
and therefore
$$
\lim_{\L\to\i}M_\L(\b, h=0)\ge
1-2g(\b)
$$
Note that $\<\s_x\>_\L^+ $ tends to 1 as $\b\to\i$ and this estimate
is uniform in $\L$.

\\On the other side, if we bound $\<\s_x\>_\L^-$ (i.e. the mean value of the spin at
a site $x$ with - boundary conditions) we have
$$
\<\s_x\>_\L^-~=~ (+1) ~\Big[{\rm Prob}_\L^-(\s_x=+1)\Big]~~ +  ~~(-1)~\Big[ {\rm
Prob}_\L^-(\s_x=-1)\Big]~=~
$$
$$
=~ 2~\Big[{\rm Prob}_{\L}^-(\s_x=+1)\Big]~-~1
$$
but now we have obviously that ${\rm Prob}_\L^-(\s_x=+1)~=~{\rm
Prob}_\L^+(\s_x=-1)$ and hence
$$
\<\s_x\>_\L^-~\le ~ 2g(\b)-1
$$
In this case $\<\s_x\>_\L^- $ tends to -1 as $\b\to\i$.

\\In conclusion the system show spontaneous magnetization for $\b$
sufficiently high uniformly in the volume $\L$. I.e. we have
shown that e.g.
$$
\lim_{\L\to\i}\<\s_x\>_\L^+\neq 0 ~~~~~~~~~{\rm if}  ~~ \b> {0.917\over J}
$$
In other
 words the system is not stable even at the infinite volume limit
to boundary conditions.

\\As a last computation we show that when $\b$ is small there is no
such instability. As a matter of fact, we can express
$\<\s_x\>_\L^+$ in term  of high temperature lattice animals. Recall that the
Hamiltonian on the Ising model with + boundary conditions and zero
magnetic field is
$$
H^+_\L(\ss_\L)~=~ -J\sum_{\{x,y\}\cap \L \neq \emptyset \atop
|x-y|~=~1}\s_x\s_y
$$
with $\s_y~=~+1$ whenever $y\in \dL$. And the partition function is
$$
Z^+_\L(\b)~=~\sum_{\s_\L}e^{ +\b J \sum_{\{x,y\}\cap \L \neq
\emptyset \atop |x-y|~=~1}\s_x\s_y } ~=~
\sum_{\s_\L}\prod_{\{x,y\}\cap \L \neq \emptyset \atop |x-y|~=~1}
e^{ +\b J \s_x\s_y }~=~ $$
$$
=~\sum_{\s_\L}\prod_{b \atop b\cap \L \neq
\emptyset } e^{ +\b J \tilde\s_b}
$$
where now the nearest neighbor pair $b$ are those strictly
contained in $\L$ plus the nearest neighbor  pairs $\{x,y\}$ for
which $x\in \L$ and $y\in \dL$. In this last case, since we are
using + boundary conditions, $\tilde \s_b = \s_x\s_y=+\s_x$. As
before we can write

$$
Z^+_\L(\b)~=~\sum_{\s_\L}\prod_{b \atop b\cap \L \neq \emptyset }
\cosh \b J \left [1 + \tilde\s_b \tanh\b J\right]
$$
and supposing as usual that $\L$ is a square of size $L$ we have
that the number of $b$ such that $b\cap \L \neq \emptyset$ is
$2L(L-1)+4L~=~2L(L+1)$, hence

$$
Z^+_\L(\b)~=~[\cosh \b J]^{2L(L+1)}\sum_{\s_\L}\prod_{b \atop b\cap
\L \neq \emptyset } \left [1 + \tilde\s_b \tanh\b J\right]
$$
As before the development of the product
$$
\prod_{b \atop b\cap \L \neq \emptyset } \left [1 + \tilde\s_b
\tanh\b J\right]\Eq(prodo)
$$
gives rise to terms of the form
$$
(\tanh \b J)^k \tilde\s_1\dots \tilde\s_k
$$
which can be associated to graphs in $\L\cup\dL$. It is important
to stress that when $b\subset \L$ then we call $b$ internal bond,
i.e.  $b~=~\{x,y\}$ with $x\in \L$ and $y\in \L$ and
$\tilde\s_b~=~\s_x\s_y$ where {\it both $\s_x$ and $\s_y$} are
summed in the summation $\sum_{\s_\L}$ over spin configurations.
On the other hand when
$b\cap \dL\neq \emptyset$, then  we  call $b$ boundary bond i.e.
$b~=~\{x,y\}$ with $x\in \L$ and $y\in \dL$ and
$\tilde\s_b~=~\s_x\s_y$ where only $\s_x$ is  summed in the
summation $\sum_{\s_\L}$, while  $\s_y~=~1$

So this time graphs which will not vanish after the sum over spin
configuration are not only closed polygons type lattice animals as before. There are also new type of lattice animals. New
non vanishing  lattice animals  are also those which start in a site
$y\in\dL$ and end in another site $y'\in\dL$. See Figure 14 below.
\begin{figure}
\begin{center}
\includegraphics[width=7cm,height=7cm]{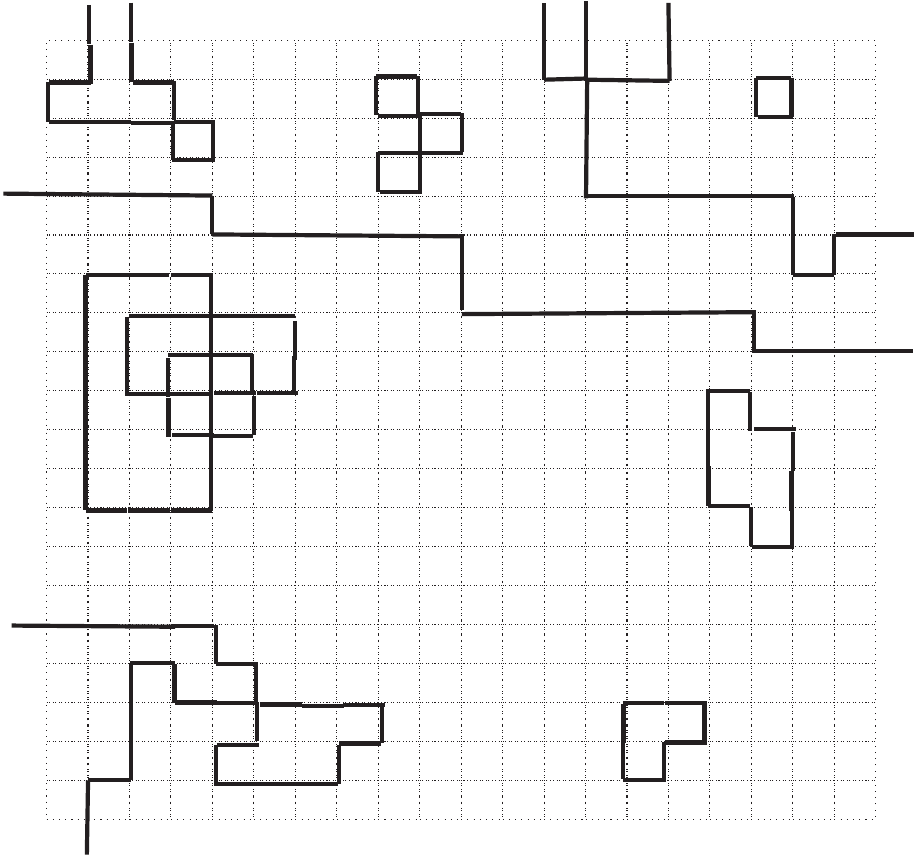}
\end{center}
\begin{center}
Figure 14.  A non vanishing configuration of lattice animals in
$Z_\L^+(\b)$.
\end{center}
\end{figure}
Any non vanish factor $(\tanh\b)^k J\tilde\s_1\dots \tilde\s_k$,
independently of the nature of the bonds $b$ (i.e. internal bonds
or boundary bonds)  becomes, once the sum over configuration $\s_\L$ is preformed,  $2^{L^2}(\tanh\b J)^k $. Note that  all non vanishing factors  $(\tanh\b)^k J\tilde\s_1\dots \tilde\s_k$
are all positive even
if we were imposing $-$ boundary conditions (so that now for a
boundary bond $b=\{x, y\}$ with $y\in \dL$ we should have
$\s_y=-1$), because any  new type (connected) lattice animals has always an even number of paths ending at the boundary so that
it has
only an {\it even} number of such points where $\s_y=-1$. So that negative values
of spins at the boundary cancels. In conclusion we have again

$$
Z^+_{\L}(\b) ~=~   [\cosh(\b
J)]^{2L(L-1)}2^{L^2}\;\Xi^+_{\L}(\b)\Eq(1iss)
$$
where
$$
\Xi^+_\L(\b)~=~
 1+
\sum_{n\ge 1}\sum_{\{\g_1, \dots \g_n\}: \;|\g_i|\ge 4\atop
\g_i\sim \g_j}\xi(\g_1)\dots \xi(\g_n)\Eq(iss2)
$$
where again $\g$ denote a allowed lattice animal (just recall that
there are in this case  more new lattice animals allowed respect
to the case of free boundary conditions) with activity
$$
\xi(\g)~=~ [\tanh(\b J)]^{|\g|}
$$
Let us now express the mean value $\<\s_x\>^+_\L$ (where $x\in \L$)
in terms of lattice animals. By definition we have
$$
\<\s_x\>^+_\L~=~{ \sum_{\s_\L}\s_xe^{ +\b J \sum_{\{x,y\}\cap \L \neq
\emptyset \atop |x-y|~=~1}\s_x\s_y }\over Z^+_{\L}(\b)}\Eq(sigm)
$$
The numerator of the expression above can be easily rewritten in
term of lattice animals as
$$
\sum_{\s_\L}\s_xe^{ +\b J \sum_{\{x,y\}\cap \L \neq \emptyset
\atop |x-y|~=~1}\s_x\s_y } ~=~
$$
$$=~[\cosh(\b J)]^{2L(L-1)}2^{L^2}
\sum_{n\ge 1}\sum_{\{\g_1, \dots \g_n\}: \atop
\g_i\sim\g_j, \;\exists \g_x}\xi(\g_1)\dots \xi(\g_n)\Eq(numer)
$$
where now  in $\sum$ the notation $\exists \g_x$ means that among lattice animals
$\g_1, \dots \g_n$  at least one of them
starts at $x$ and end at a boundary site, see Figure  15 for such
kind of lattice animals (actually are also allowed lattice animals
which have incidence number $3$ in $x$). Let as indicate with $\LL_x$ the set os all this lattice animals

\begin{figure}
\begin{center}
\includegraphics[width=7cm,height=7cm]{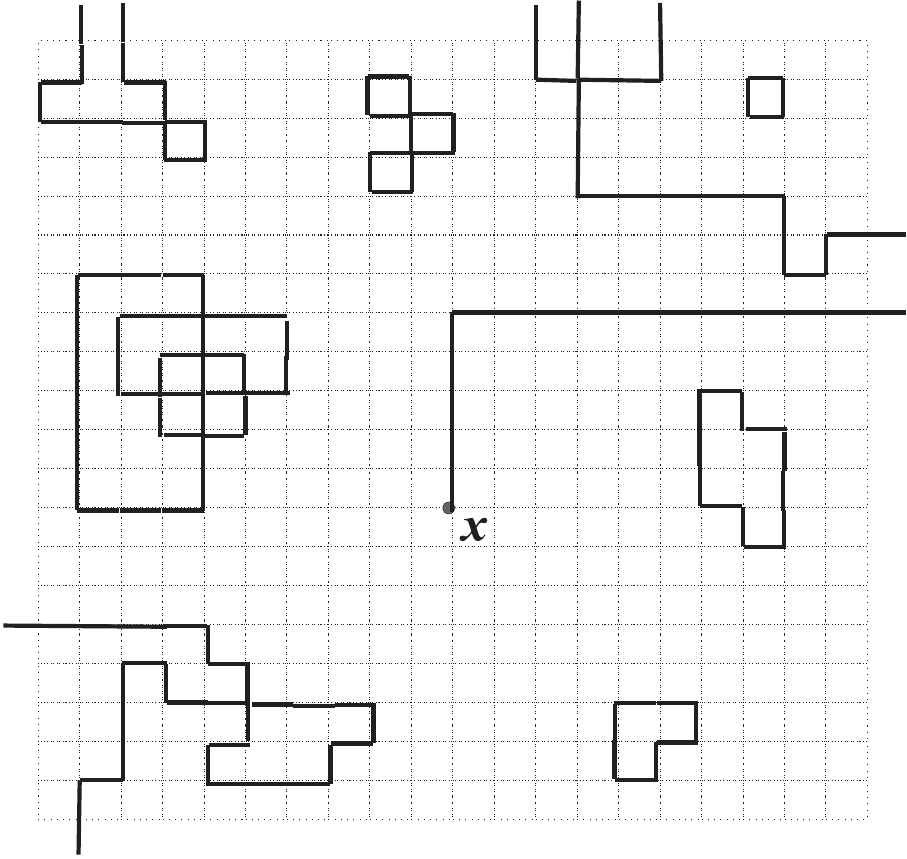}
\end{center}
\begin{center}
Figure 15.  A non vanishing configuration of lattice animals in
the numerator of $\<\s_x\>^+_\L$.
\end{center}
\end{figure}

\\Note also that the factor $1$ (i.e that corresponding to the empty
lattice animal) is no more present in the numerator \equ(numer).

\\Thus we can write

$$
\begn
\<\s_x\>^+_\L&=~{ \sum_{\{\g_1, \dots \g_n\}: \atop
\g_i\sim \g_j,~\exists \g_x}\xi(\g_1)\dots \xi(\g_n) \over 1+
\sum_{\{\g_1, \dots \g_n\}: \;|\g_i|\ge 4\atop \g_i\sim
\g_j}\xi(\g_1)\dots \xi(\g_n)}\\\\
& \le~ {\sum\limits_{\g\in \LL_x
\atop V_\g\cap\dL\neq \emptyset}\xi(\g) \sum_{\{\g_1, \dots \g_n\}:
\;|\g_i|\ge 4\atop \g_i\sim \g_j,\;{ \g_i\sim \g }}\xi(\g_1)\dots \xi(\g_n) \over 1+ \sum_{\{\g_1,
\dots \g_n\}: \;|\g_i|\ge 4\atop \g_i\sim
\g_j}\xi(\g_1)\dots \xi(\g_n)}\\\\
&\le ~\sum_{\g\in \LL_x \atop V_\g\cap\dL\neq
\emptyset}\xi(\g)
\egn
$$
In conclusion we have that
$$
\<\s_x\>^+_\L~\le~\sum_{\g\in \LL_x\atop V_\g\cap\dL\neq
\emptyset}\xi(\g)
\Eq(sigma)
$$
Observe now that if $\LL^1_x$ denotes the set of paths starting at $x$ (i.e. lattice animals having degree 1 at $x$) we can reorganize the sum above as follows
$$
\sum_{\g\in \LL_x\atop V_\g\cap\dL\neq\emptyset}\xi(\g)\le\sum_{\g\in \LL^1_x\atop V_\g\cap\dL\neq \emptyset}\xi(\g)\left[1+ \left(\sum_{\gt\in \LL^1_x\atop V_\gt\cap\dL\neq \emptyset}\xi(\gt)\right)^2\right]
$$

If we denote by $d(x,\dL)$ the minimum distance (in nearest
neighbor bonds) between $x$ and the boundary $\partial\L$, we have clearly
that a $\g$ in the sum must have at leat $|\g|~=~ d(x,\dL)$ bonds. So
$$
\sum_{\g\in \LL^1_x\atop V_\g\cap\dL\neq \emptyset}\xi(\g) ~~\le~~ \sum_{n\ge d(x,\dL)}(\tanh\b J)^n \sum_{\g\in \LL^1_x:\;|\g|=n\atop V_\g\cap\dL\neq \emptyset}1 ~\le ~ {4\over 3}\sum_{n\ge d(x,\dL)}(3\tanh\b J)^n
$$
where we have estimated
$$
\sum_{\g\in \LL^1_x:\;|\g|=n\atop V_\g\cap\dL\neq \emptyset}\le 4\cdot3^{n-1}
$$
The series $\sum_{n\ge d(x,\dL)}(3\tanh\b J)^n$  converges for
$
3\tanh\b J<1
$
i.e for
$$
\b<\b_0'\doteq {1\over J}\tanh^{-1}(1/3)\approx {0.34\over J}\Eq(corj)
$$
and in conclusion
$$
\sum_{\g\in \LL^1_x\atop V_\g\cap\dL\neq \emptyset}\xi(\g)~\le ~ {4\over 3}{(3\tanh\b J)^{d(x,\dL)}\over 1-3\tanh\b J}
$$
Therfore
we get finally, for $\b$ as small as  in  condition \equ(corj)
$$
\<\s_x\>^+_\L~~\le~~ {4\over 3}{(3\tanh\b J)^{d(x,\dL)}\over 1-3\tanh\b J}\left[1+\left( {4\over 3}{(3\tanh\b J)^{d(x,\dL)}\over 1-3\tanh\b J}\right)^2\right]~\le
$$
$$
\le ~{100\over 3}{(3\tanh\b J)^{d(x,\dL)}\over (1-3\tanh\b J)^3}
$$
where in the las line we have used that $3\tanh\b J<1$.
Hence the magnetization at zero magnetic field and at finite volume is bounded by
$$
\begn
M^+_\L(\b,h=0) &=~{1\over |\L|}\sum_{x\in \L}\<\s_x\>^+_\L\\
&\le~ {100\over 3(1-3\tanh\b J)^3}{1\over |\L|}\sum_{x\in \L}[3\tanh\b J]^{d(x,\dL)}\\
& \le
{100\over 3(1-3\tanh\b J)^3}{1\over |\L|}\sum_{n=1}^{L/2} [3\tanh\b J]^n\sum_{x\in
\L\atop d(x,\dL)=n}1\\
& \le
{300\tanh\b J\over 3(1-3\tanh\b J)^4}{|\partial \L|\over |\L|}
\egn
$$
hence we get, for $\b<\b_0'$
$$
\lim_{\L\to\i}M^+_\L(\b,h=0)~=~ ~0
$$
Note also, for  any $\L$
$$
\lim_{\b\to 0}M^+_\L(\b, h=0)~=~ ~0
$$
Thus when the temperature is sufficiently high the magnetization
of the system, even with + boundary conditions, tends to be zero
in the thermodynamic limit. This in contrast with the result that
we have shown for low temperature, where $M^+_\L(\b, h=0)$ is definitely
away from zero and is near to one.

\\This result say to us that the Ising model does not present the phenomenon of the
spontaneous magnetization when the temperature is sufficiently
high, or in other words the bulk system is not sensible to change
of boundary conditions in the thermodynamic limit. On the
contrary, when the temperature is very low, the system shows
indeed spontaneous magnetization and its bulk  becomes sensible to
boundary conditions even in the thermodynamic limit.

\\The interpretation of this fact is that the Ising model has a phase transition
at zero magnetic field and at some critical value $\b_c$ of the
inverse temperature. Below $\b_c$ the system behaves like in the
high temperature regime and above $\b_c$ the system behave like
the low temperature regime.

\\Note also that $\lim_{\b\to 0}M^+_\L(\b, h=0)$  and  $\lim_{\b\to 0}M^-_\L(\b, h=0)$  has different values, the first near 1 and the second near -1, and therefore
the thermodynamic limit $M(\b, h=0)$ of the magnetization (i.e. the derivative of the free energy $f(\b,h)$) does not have a definite value. This
is an evidence that $f(\b,h)$ has discontinous derivative respect to $h$ at $h=0$ when the inverse temperature is greater than $\b>\b_1$.
On the other hand we have seen that $\lim_{\b\to 0}M^+_\L(\b, h=0)=\lim_{\b\to 0}M^-_\L(\b, h=0)=0$ when $\b<\b'_0$, therefore there is a value $\b_c$ such that
for $\b<\b_c$ the  infinite volume magnetization is zero and for $\b>\b_c$ the magnetization is different from zero. It can be shown that
the free energy is non analytic in  $\b=\b_c$ as a function of $\b$.

\subsection{The critical temperature}
As a very last exercise we show  that if there is a unique non
analytic point of the free energy  $f(\b)$ at $h~=~0$ in the interval
$(0, \i)$ then this point can stay only on a well defined $\b_c$.

Consider the partition function  of the Ising model at zero magnetic field
assuming  without loss of generality that $J=1$ to simplify things.
The low temperature expansion in terms of closed closed contours of this partition function
 in a
box $\L$ which  is a square of size $L-1$ where $L$
with  + boundary conditions is such that the contours live in the square $\L^*$ of size $L$ and the
partition function is
$$
Z^+_{(L-1)\times (L-1)}(\b)~=~e^{ +\b  2L(L-1)} \sum_{\g_1,\dots ,
\g_n:~\g_i\sim \g_j\atop \g_i\in L\times L}\prod_{i=1}^n e^{-2\b|\g_i|} \Eq(isi1)$$

\\On the other hand, via high temperature expansion, the partition
function of the zero magnetic field, free boundary conditions
Ising model in a box of size $L$ can be written as
$$
Z^{\rm open}_{L\times L}(\b)~=~ \cosh(\b)^{2L(L-1)}2^{L^2}
\sum_{\g_1 ,\dots ,\g_n: ~\g_i\sim\g_j\atop\g\in L\times L}
\prod_{i=1}^n[\tanh{\b}]^{|\g_i|}\Eq(isi2)
$$
Compare now \equ(isi1) and \equ(isi2) and note that the sums on
closed polymers in both equations are identical.

Let  $\f(\b)$ be  the function of $\b$ defined by
$$
\f(\b)~=~- {\ln[\tanh{\b }]\over 2}\Eq(bstar)
$$
so that
$$
e^{ -2\f(\b)}~=~ \tanh{\b}
$$

Then we can write, by \equ(isi1) and \equ(isi2)

$$
{Z^+_{(L-1)\times (L-1)}(\f(\b))\over e^{ +2\f(\b)  (L-1)(L-2)}}~=~
{Z^{\rm open}_{L\times L}(\b)\over  \cosh(\b J)^{2L(L-1)}2^{L^2}}
$$
and taking the logarithm on both sides
$$
\ln Z^+_{(L-1)\times (L-1)}(\f(\b)){ -2\f(\b)  (L-1)(L-2)}~=~
$$
$$
=~\ln
Z^{\rm open}_{L\times L}(\b)-{2L(L-1)}\ln \cosh(\b)-L^2\ln 2
$$
Dividing by $L^2$ and taking the limit $L\to \i$
$$
\lim_{L\to\i}\Bigg\{{(L-1)^2\over L^2}{1\over (L-1)^2)}\ln
Z^+_{(L-1)\times(L-1)}(\b^*)~
 - ~2\b^* J {(L-1)(L-2)\over
L^2}\Bigg\}~=~
$$
$$
~=~\lim_{L\to\i}\left\{ {1\over L^2} \ln Z^{\rm open}_{L\times
L}(\b)- {2L(L-1)\over L^2}\ln \cosh(\b J) -\ln 2\right\}
$$
i.e., since free energy $f(\b)~=~\lim_{\L\to \i}\ln Z^{\t}_\L(\b)$
does not depend on boundary conditions
$$
f(\f(\b)) - 2\f(\b)~=~  f(\b)- \ln[2\cosh^2(\b)]
$$
hence we get
$$
f(\b) ~=~ f(\f(\b)) + \ln  [2\cosh^2(\b )] - 2\f(\b)
$$
i.e.
$$
f(\b) ~=~ f(\f(\b)) + \ln [ 2\cosh^2(\b)] + \ln[\tanh{\b}] \Eq(key)
$$
We now suppose that $f(\b)$ can have at most one singularity, and
we conclude via \equ(key) that this singularity, if exists, can
occur only at $\b_c$ where $\b_c$ is the solution of the equation $\b=\f(\b)$.

\\Note that $\f(\b)$ (see \equ(bstar)) is analytic as a
function of $\b$ for all $\b\in (0,\i)$. Now, by \equ(key),
$f(\b)-f(\f(\b))$ is analytic for all $\b\in (0,\i)$. Let now
$\b'\in (0,\i)$ such that $\b'\neq \b_c$, then by \equ(key)
$f(\b)$ is analytic in $\b~=~\b'$. In fact suppose, by absurd, that
$f(\b)$ is non analytic  in $\b~=~\b'$, then $f(\f(\b'))$ must
also be non analytic, since only in this way $f(\b') - f(\f(\b'))= \ln  2\cosh(\b ') +2 \ln[\tanh{\b' }] $ can be non
singular. But if $\b'\neq \b_c$  then $f(\f(\b'))=~f(\b'')$
with $\b''\neq \b'$. Hence $f(\b)$ would have two singularities
which is in contradiction with the hypothesis. This equation is
saying that if $f(\b)$ has a singularity point in the domain
$\b\in (0,\i)$,  this point can occur only when $\b=\b_c$, i.e.
where $\b_c$ is the solution of
$$
e^{ -2\b }~=~ \tanh{\b }
$$
which  gives $\b_c={1\over 2}\ln(1+\sqrt{2})$ which
is the correct value calculated via the Onsager solution.

\def\xx{{\bf x}}
\def\vv{{\bf v}}

\printindex

\end{document}